%% file: main.tex
\documentclass[a4paper, oneside, 12pt, openany]{book}
\pdfoutput=1

\usepackage{packages}
\usepackage{macros}

\newbool{shade-envs}
\setboolean{shade-envs}{true}

\ifthenelse{\boolean{shade-envs}}{%
  \colorlet{ShadeOfPurple}{blue!5!white}
  \colorlet{ShadeOfYellow}{yellow!5!white}
  \colorlet{ShadeOfGreen} {green!5!white}
  \colorlet{ShadeOfBrown} {brown!10!white}
  \colorlet{ShadeOfGray}  {gray!10!white}
}
{%
  \declaretheoremstyle[
      spaceabove=6pt,
      spacebelow=6pt,
      bodyfont=\normalfont,
      qed=\(\lozenge\)
  ]{definitionwithbox}
  \declaretheoremstyle[
      headfont=\itshape,
      bodyfont=\normalfont,
      qed=\(\lozenge\)
      ]{remarkwithbox}
}

\numberwithin{equation}{section}
\numberwithin{table}{section}
\makeatletter
\let\c@equation\c@table
\makeatother

\ifthenelse{\boolean{shade-envs}}{%
  \declaretheorem[sibling=equation]{theorem}
  \declaretheorem[sibling=theorem]{lemma,proposition,corollary}
  \declaretheorem[sibling=theorem,style=definition]{definition}
  \declaretheorem[sibling=theorem,style=definition]{example}
  \declaretheorem[sibling=theorem,style=remark]{remark}
  \tcbset{shadedenv/.style={
      colback={#1},
      frame hidden,
      enhanced,
      breakable,
      boxsep=0pt,
      left=2mm,
      right=2mm,
      add to width=1.1mm,
      enlarge left by=-0.6mm}
  }
  \tcolorboxenvironment{theorem}    {shadedenv={ShadeOfPurple}}
  \tcolorboxenvironment{lemma}      {shadedenv={ShadeOfPurple}}
  \tcolorboxenvironment{proposition}{shadedenv={ShadeOfPurple}}
  \tcolorboxenvironment{corollary}  {shadedenv={ShadeOfPurple}}
  \tcolorboxenvironment{definition} {shadedenv={ShadeOfYellow}}
  \tcolorboxenvironment{example}    {shadedenv={ShadeOfGreen}}
  \tcolorboxenvironment{remark}     {shadedenv={ShadeOfBrown}}
  \tcolorboxenvironment{proof}      {shadedenv={ShadeOfGray}}
}{
  \declaretheorem[sibling=equation]{theorem}
  \declaretheorem[sibling=theorem]{lemma,proposition,corollary}
  \declaretheorem[sibling=theorem,style=definitionwithbox]{definition}
  \declaretheorem[sibling=theorem,style=definitionwithbox]{example}
  \declaretheorem[sibling=theorem,style=remarkwithbox]{remark}
  }
\declaretheorem[sibling=theorem,style=remark,numbered=no]{claim}

\begin{document}

\frontmatter

\include{frontmatter/titlepage}
\restoregeometry%

\include{frontmatter/abstract}
\include{frontmatter/acknowledgements}

\setcounter{tocdepth}{2}
\tableofcontents

\mainmatter%

\include{mainmatter/introduction}

\include{mainmatter/univalent-foundations}
\include{mainmatter/basic-domain-theory}
\include{mainmatter/continuous-and-algebraic}
\include{mainmatter/applications}
\include{mainmatter/predicativity-in-order-theory}

\include{mainmatter/formalisation}

\include{mainmatter/conclusion}

\nocite{*}

\backmatter%
\printbibliography[heading=bibintoc]%
\printnomenclature%
\printindex

\end{document}

%% file: frontmatter/titlepage.tex
\newgeometry{total={180mm,267mm}} 
\begin{titlepage}
\begin{center}
  \vspace*{\stretch{0.5}}

  \large 

  {\Huge\textsc{Domain Theory in Constructive and Predicative Univalent Foundations}\par}

  \vspace{\stretch{0.2}}

  by

  \vspace{\stretch{0.2}}

  {\huge\textsc{Tom de Jong}}

  \vspace{\stretch{0.5}}

  A thesis submitted to the University of Birmingham for the degree of\\
  \textsc{Doctor of Philosophy}

  \vfill

  \flushright{
  School of Computer Science \\
  College of Engineering and Physical Sciences \\
  University of Birmingham \\\phantom{force-newline}\\

  \begin{tabular}{rr}
  Submitted: & 30 September 2022 \\
  Defended:  & 20 December  2022 \\
  Accepted:  & 1 February  2023 \\
  \end{tabular}}
\end{center}
\end{titlepage}


%% file: frontmatter/abstract.tex
\chapter{Abstract}

We develop domain theory in constructive and predicative univalent foundations
(also known as homotopy type theory). That we work predicatively means that we
do not assume Voevodsky's propositional resizing axioms. Our work is
constructive in the sense that we do not rely on excluded middle or the axiom of
(countable) choice.
Domain theory studies so-called directed complete posets (dcpos) and Scott
continuous maps between them and has applications in a variety of fields, such
as programming language semantics, higher-type computability and topology.
A common approach to deal with size issues in a predicative foundation is to
work with information systems, abstract bases or formal topologies rather than
dcpos, and approximable relations rather than Scott continuous functions.
In our type-theoretic approach, we instead accept that dcpos may be large and
work with type universes to account for this.
A~priori one might expect that complex constructions of dcpos, involving
countably infinite iterations of exponentials for example, result in a need for
ever-increasing universes and are predicatively impossible. We show, through a
careful tracking of type universe parameters, that such constructions can be
carried out in a predicative setting.
We~illustrate the development with applications in the semantics of programming
languages: the soundness and computational adequacy of the Scott model of PCF,
and Scott's \(D_\infty\) model of the untyped \(\lambda\)-calculus.
Both of these applications make use of Escard\'o's and Knapp's type of partial
elements.
Taking inspiration from work in category theory by Johnstone and Joyal, we also
give a predicative account of continuous and algebraic dcpos, and of the related
notions of a small (compact) basis and its rounded ideal completion. This is
accompanied by concrete examples, such as the small compact basis of
Kuratowski finite subsets of the powerset.
The fact that nontrivial dcpos have large carriers is in fact unavoidable and
characteristic of our predicative setting, as we explain in a complementary
chapter on the constructive and predicative limitations of univalent
foundations. We prove no-go theorems for a general class of posets that includes
dcpos, bounded complete posets, sup-lattices and frames.
In~particular, we show that, constructively, locally small nontrivial dcpos
necessarily lack decidable equality.
Our account of domain theory in univalent foundations has been fully formalised
with only a few minor exceptions. The ability of the proof assistant \Agda\ to
infer universe levels has been invaluable for our purposes.


%% file: frontmatter/acknowledgements.tex
\chapter{Acknowledgements}\label{chap:acknowledgements}
\markboth{Acknowledgements}{Acknowledgements}

First and foremost, it is my pleasure to express my deepest thanks to my
supervisor Mart\'in Escard\'o. This thesis simply wouldn't exist without his
ideas, patient teaching, and ever-helpful feedback and support.
Secondly, I am grateful to my co-supervisor Benedikt Ahrens who really helped me
feel at home in Birmingham and has always been available to offer sage
advice even after he moved away from Birmingham.

I feel fortunate to have been a part of the very vibrant Theory group at
Birmingham and wish to thank all its members, and the participants of our Theory
PhD student seminar \emph{Bravo} in particular.

Of the (former) PhD students, Alex Rice, Auke Booij, Ayberk Tosun, Calin Tataru,
George Kaye and Todd Waugh Ambridge deserve a special mention.
Alex and Calin were excellent teaching assistants and great office mates before
they moved to Cambridge. The framework developed to mark \emph{Functional
  Programming} assignments is largely due to them.
Auke was very welcoming and helpful when I first moved to Birmingham.
With George I had the pleasure of organising the Theory seminar, upgrading the
Theory website and introducing web profiles for PhD students.
I have had interesting discussions with Ayberk and Todd and would like to thank
them for their efforts as teaching assistants in the modules we taught
together. Moreover, like George, they became my friends in Birmingham.

Of the Theory staff members, I would like to mention Anupam Das, Eric Finster,
Paul Levy and Sonia Marin in particular. Anupam is responsible for much of the
social activities in the Theory group and seemingly never tires of playing
devil's advocate. It~was a pleasure to work with Eric in \emph{(Advanced)
  Functional Programming} and at \emph{MGS'22}.
Both Sonia and Paul were simply two very friendly presences in the School of
Computer Science and I wish to thank them for all our hallway conversations.
Moreover, I appreciate that Paul has agreed to chair my viva.

Next, I would like to express my thanks to my thesis group committee members
Steve Vickers and Vincent Rahli, and its past members Nicolai Kraus and Noam
Zeilberger, for their interest, questions and comments. Furthermore, I wish to
thank Vincent for accepting to be my internal examiner.

Last but not least on the list of people at University of Birmingham are Jason
Fenemore and his successor Angeliki Bompetsi who must be mentioned for their
outstanding work in supporting PhD students and dealing with Worklink.

A special thanks is reserved for Andrea Vezzosi who, upon request, kindly
implemented the \verb|lossy-unification| flag in \Agda. This heuristic
significantly speeds up \Agda's typechecking in some cases and it allowed me to
complete my formalisation of domain theory.

I am grateful to Andrej Bauer for agreeing to be my external examiner and for
hosting me during my research visit to Ljubljana, where I had many inspiring
conversations with Egbert Rijke to whom I am indebted. I quickly felt at home in
Ljubljana thanks to them, Ajda Lampe, Alex Simpson, Filip Koprivec, Katja
Ber\v{c}i\v{c}, L\'eo Mangel, and in particular, Jure Taslak.

I would also like to thank the organisers, lecturers, participants, and
especially, my fellow teaching assistants of the \emph{HoTTEST Summer School
  2022}: it has been a great pleasure to learn and work with you.

My fourth and final year as a PhD student was financially supported by the
\emph{Homotopy Type Theory Dissertation Fellowship}, funded by Cambridge Quantum
Computing and Ilyas Khan. I wish to express my sincere thanks to them and the
committee members Steve Awodey, Thierry Coquand, Emily Riehl, and Mike Shulman
for allowing me to work on this topic for another year.

I moved to Nottingham towards the end of my PhD and I greatly appreciated
Nicolai Kraus' help and support in moving. I am also very grateful to Josh Chen
and Stefania Damato for making my first weeks in Nottingham so enjoyable.

I am indebted to Harry Smit for many wonderful conversations, especially when
they concerned life as an academic. A special thanks to Menno de Boer for his
enthusiasm, our discussions on homotopy type theory and for proofreading parts
of this thesis.

These acknowledgements have deliberately been restricted to people that were
directly involved in my academic endeavours, and fall short of mentioning many
friends and family that I am thankful to. I wish to make a single important
exception: Madelon, thank you for all your love, support and understanding. I am
deeply grateful that you were, and will be, by my side while I pursue(d) my
academic interests abroad.


%% file: mainmatter/introduction.tex
\chapter{Introduction}\label{chap:introduction}

\emph{Univalent foundations}~\cite{Voevodsky2015}, also known as \emph{homotopy
  type theory}~\cite{HoTTBook} and often abbreviated as \emph{HoTT} or
\emph{HoTT/UF}, is a recent enhancement of intensional Martin-L\"of Type
Theory~(MLTT)~\cite{MartinLof1984} and has many complementary uses. %
\index{univalent foundations}%
\index{homotopy type theory|see {univalent foundations}}%
\index{HoTT|see {univalent foundations}}%
\index{HoTT/UF|see {univalent foundations}}%
For example, it is a language for \((\infty,1)\)-toposes~\cite{Shulman2019}, it
allows for a natural development of synthetic homotopy
theory~\cite{HoTTBook,Rijke2019,HoTT-Coq,HoTT-Agda,Agda-UniMath} and synthetic
group theory~\cite{BezemEtAl2022,Agda-UniMath}, and it functions as a modern
foundation~\cite{HoTTBook,Escardo2019,Agda-UniMath} for general mathematics
providing an alternative to traditional set-theoretic approaches.
Moreover, thanks to the type-theoretic basis of univalent foundations, it is
possible to implement proofs in HoTT/UF in proof assistants such as
\Agda~\cite{Agda}, \CubicalAgda~\cite{CubicalAgda}, \Coq~\cite{Coq} and
\Lean~\cite{Lean}, among others~\cite{Arend,cooltt,redtt,RedPRL}, allowing for a
formalised, computer-checked development of mathematics. %
\index{proof assistant}%
\index{proof assistant!Agda@\Agda}%
\index{proof assistant!Coq@\Coq!}%
\index{Agda@{\Agda}|see {proof assistant}}%
\index{Coq@{\Coq}|see {proof assistant}}%

This thesis is concerned with homotopy type theory as a foundations for
(formalised) mathematics. Specifically, we develop a formalised account of
\emph{domain theory}, an important area in theoretical computer science, in
univalent foundations.
In fact, we present a fully \emph{constructive} and \emph{predicative} treatment
of domain theory within this setting. A precise overview of what is covered can
be found in~\cref{sec:outline-and-summary-of-contributions}, but for now we
emphasise that our development is illustrated and proved to be useful through
the exposition of two applications in the semantics of programming languages:
the soundness and computational adequacy of the Scott model of
PCF~\cite{Plotkin1977,Scott1993}, and Scott's \(D_\infty\)
model~\cite{Scott1982b} of the untyped \(\lambda\)-calculus, which are fully
formalised in \Agda~\cite{TypeTopologyDomainTheory,Hart2020} and
\Coq~\cite{UniMathScottModelOfPCF,deJong2021a}.%
\index{semantics}\index{programming language}%
\index{Scott model!of PCF}\index{Scott model!of PCF!soundness}%
\index{Scott model!of PCF!computational adequacy}%
\index{Scott model!of the untyped \(\lambda\)-calculus}%
\index{formalisation}\index{proof assistant!Agda@\Agda}\index{proof assistant!Coq@\Coq}

\paragraph{Domain theory}\index{domain theory}%
Domain theory~\cite{AbramskyJung1994} studies a particular class of posets and
has applications in a variety of fields, such as: programming language
semantics~\cite{Plotkin1977,Scott1982b,Scott1993}, higher-type
computability~\cite{LongleyNormann2015} and topology~\cite{GierzEtAl2003}. For
instance, domain-theoretic insights have led to the discovery of surprising
algorithms that exhaustively search infinite sets in finite
time~\cite{Berger1990,Escardo2008}. More generally, domain theory can be used to
prove correctness of algorithms through denotational semantics.

\paragraph{The univalent point of view}
In intensional Martin-L\"of Type Theory, the identity type of a type is
uniformly and inductively defined.
Thus, for every type \(X\), we have a type \(x = y\) of identities, or
\emph{identifications}. %
\index{type!identity}%
One of the key features of univalent foundations, and of the \emph{univalence
  axiom} specifically, is that type of identifications captures (in a precise
sense) the correct notion of equality, cf.\ \cite[Section~9.8]{HoTTBook},
\cite{CoquandDanielsson2013}, \cite[Section~3.33]{Escardo2019}
and~\cite{AhrensEtAl2020}.\index{univalence}
For example, if we have two elements \(G\) and \(H\) of the type of groups, then
the type \(G = H\) is equivalent (in Voevodsky's sense~\cite{Voevodsky2015}) to
the type of group isomorphisms from \(G\) to \(H\).
In particular, the type of identifications \(x = y\) can contain many elements,
so equality in univalent foundations is not necessarily a truth value.
This naturally leads to higher structures in univalent foundations. In fact,
another key insight of Voevodsky~\cite{Voevodsky2015} was that the
stratification of types according to the complexity of their identity types
into \emph{(sub)singletons} (truth values), \emph{sets}, 1-groupoids,
2-groupoids, etc.\ can be internalised and defined inside MLTT.

Consequently, the mathematical distinction between a \emph{property} and
(additional) \emph{data} or \emph{structure} is also internalised. Sometimes we
know how to express something as an equipment with extra structure, but we are
interested in obtaining a property instead.
For this, we turn to the \emph{propositional truncation}: the universal method
of making a type into a subsingleton.
The propositional truncation is an example, and in this thesis the \emph{only}
example, of a \emph{higher inductive type}~\cite[Section~6]{HoTTBook}.%
\index{higher inductive type}

For developing domain theory we typically do not need the theory of higher
groupoids. Accordingly, weaker consequences of the univalence axiom (function
extensionality and propositional extensionality, to be precise) are often
sufficient for our purposes.
An important exception, besides its use in the theory of ordinals
(\cref{sec:small-suprema-of-ordinals}), is the fundamental notion of
\emph{\(\V\)-smallness}: if we want to know that it expresses a property, then
univalence is sufficient and (in some precise sense) necessary, as we explain
in~\cref{sec:small-and-locally-small-types}.
Even if univalence is often not needed, it will hopefully become clear
throughout this thesis, that the univalent point of view is prevalent in our
work. For example, we recall that if a type \(X\) is equipped with a
subsingleton-valued reflexive and antisymmetric binary relation, then \(X\) is a
set, meaning its elements can be equal in at most one way.
Moreover, our development fundamentally relies on the aforementioned
propositional truncation and also features several applications of a
theorem~\cite[Theorem~5.4]{KrausEtAl2017} that characterises when we can map
from a propositionally truncated type to a set.

\paragraph{Constructivity}\index{constructivity|(}
Constructivity has historically always been important in the type theoretic
tradition. Indeed, Martin-L\"of invented his type theory to serve as a
constructive foundation of mathematics~\cite{MartinLof1975}.
More recently, its extension, univalent foundations, has been given a
computational interpretation through cubical type theory~\cite{CohenEtAl2018}
and this has been implemented in practice as the proof assistant and functional
programming language \CubicalAgda~\cite{CubicalAgda}.

That we work constructively means that we do not assume excluded middle,
or weaker variants, such as Bishop's LPO~\cite{Bishop1967}, or the axiom of
choice (which implies excluded middle), or its weaker variants, such as the
axiom of countable choice.
An~advantage of working constructively and not relying on these additional
logical axioms is that our development is valid in every
\((\infty,1)\)-topos~\cite{Shulman2019} and not just those in which the logic is
classical.

Our commitment to constructivity has several manifestations throughout this
thesis. For example, it means we cannot simply add a least element to a set to
obtain the free pointed \emph{directed complete poset (dcpo)}. Instead of adding
a single least element representing an undefined value, we must work with a more
complex type of partial elements~(\cref{sec:lifting}). Another example is the
distinction between continuity and pseudocontinuity of
dcpos~(\cref{sec:pseudocontinuity}); the notions coincide when the axiom of
choice is assumed.
Moreover, the absence of countable choice is discussed in connection to
semidecidability and the Scott model of PCF in~\cref{sec:semidec-pcf}.

Constructive mathematics is naturally at home in theoretical computer science,
because constructive proofs give rise to
algorithms~\cite{Bishop1970,MartinLof1982}.
We illustrate this point through a domain-theoretic
example. In~\cref{sec:Scott-model-of-PCF} we give a constructive proof that the
Scott model of PCF is computationally adequate. The constructive nature of the
proof yields (in theory, at least) an interpreter: if we can prove that a given
program (of base type) is total, then we can compute its numerical outcome
through computational adequacy.
In classical domain theory it is of fundamental
interest~\cite{Scott1970,Smyth1977} how its aspects can be formulated in an
effective or computational manner.
The computational nature of constructive mathematics might enable one to use our
constructive development of domain theory to obtain algorithms without having to
develop a separate account of effectively given dcpos.%
\index{algorithm}

The~advent of proof assistants (many based on type theory, including
Martin-L\"of Type Theory and the Calculus of
Constructions~\cite{CoquandHuet1988}) has narrowed the gap between mathematics
and computer science further and we discuss the implementations of our work in
the proof assistants \Agda\ and \Coq\ below and in more detail
in~\cref{chap:formalisation}.
\index{constructivity|)}

\paragraph{Predicativity}\index{predicativity|(}\index{predicativity|seealso {resizing}}

Our work is predicative in the sense that we do not assume Voevodsky's
\emph{resizing} rules~\cite{Voevodsky2011,Voevodsky2015} or axioms. In
particular, powersets of small types are large.
Before we explain some of the ramifications of this for the domain-theoretic
development, we reflect on some of the reasons for working without resizing
principles.

\index{resizing!propositional}
First and foremost, it is currently an open problem whether propositional
resizing axioms can be given a computational interpretation, as has been done
for the univalence axiom and higher inductive types in cubical type
theory~\cite{CohenEtAl2018}.
Thus, in line with our constructive agenda and to retain a computational
interpretation (in for instance, \CubicalAgda~\cite{CubicalAgda}) we work in the
absence of propositional resizing axioms.
Since higher inductive types may be seen as particular resizing principles, it
is also noteworthy that the only higher inductive type needed in our work is the
propositional truncation.
Another reason for being interested in predicativity is the fact that
propositional resizing axioms fail in some models of univalent type theory. This
is discussed further in \cref{sec:related-work:predicativity}.
Furthermore, it is expected, by analogy to predicative and impredicative set
theories, that adding resizing axioms significantly increases the
proof-theoretic strength of univalent type
theory~\cite[Remark~1.2]{Shulman2019}.
Lastly, one may have philosophical reservations regarding impredicativity. For
example, some constructivists may accept predicative set theories like
Aczel's~CZF\index{CZF} or Myhill's~CST, but not Friedman's impredicative set
theory~IZF~\cite{Beeson1985}.
Or, paraphrasing Shulman's narrative~\cite{Shulman2011}, one can ask why
propositions (or (\(-1\))-types) should be treated differently, i.e.\ given that
we have to take size seriously for \(n\)-types for \(n > -1\), why not do the
same for (\(-1\))-types?

A common approach to deal with domain-theoretic size issues in a predicative
foundation is to work with information systems~\cite{Scott1982a,Scott1982b},
abstract bases~\cite{AbramskyJung1994} or formal
topologies~\cite{Sambin1987,Sambin2003,CoquandEtAl2003} rather than dcpos, and
approximable relations rather than \emph{Scott continuous functions}. %
\index{information system}\index{basis!abstract}\index{formal topology}%
\index{approximable relation}%
Instead, we work directly with dcpos and Scott continuous functions. In dealing
with size issues, we draw inspiration from category theory and make crucial use
of type universes and type equivalences to capture \emph{smallness}.
For example, in our development of the Scott model of PCF, the dcpos have
carriers in the second universe \(\U_1\) and least upper bounds for directed
families indexed by types in the first universe \(\U_0\).
Moreover, up to equivalence of types, the order relation of the dcpos takes
values in the lowest universe \(\U_0\).
Seeing a poset as a category in the usual way, we can say that these dcpos are
large, but locally small, and have small filtered colimits.
The fact that the dcpos have large carriers is in fact unavoidable and
characteristic of predicative settings, as proved
in~\cref{chap:predicativity-in-order-theory}.
Because the dcpos have large carriers it is a priori not clear that complex
constructions of dcpos, involving countably infinite iterations of exponentials
for example, do not result in a need for ever-increasing universes and are
predicatively possible. We show that they are possible through a careful
tracking of type universe parameters, and this is also illustrated by
applications, such as the Scott model of PCF and \(D_\infty\). Since keeping
track of these universes is prone to mistakes, we have implemented much of our
work in \Agda; its ability to infer universe levels has been invaluable.%
\index{proof assistant!Agda@\Agda}\index{universe!parameters}%
\index{predicativity|)}

\paragraph{Formalisation}\index{formalisation}

Type theories are the basis of many successful proof assistants, such as
\Agda~\cite{Agda}, \Coq~\cite{Coq} and \Lean~\cite{Lean}.
Much of the work in this thesis has been formalised in \Agda\ using
Escard\'o's~\TypeTopology~\cite{TypeTopology} development and this has helped
considerably to guide our predicative and constructive development of domain
theory.
A full discussion of the formalisation efforts, including the work
in~\Coq/\UniMath, can be found in~\cref{chap:formalisation}.

\section{Related work}

We give a brief overview of related work on (constructive and/or predicative)
domain theory and of predicativity in general.
In short, the distinguishing features of our work are: (i) the adoption of
homotopy type theory as a foundation, (ii) a commitment to predicatively and
constructively valid reasoning, (iii) the use of type universes to avoid size
issues concerning large posets.

\subsection{Domain theory}\index{domain theory|(}

The standard works on domain theory, e.g.~\cite{AbramskyJung1994,GierzEtAl2003},
are based on traditional impredicative set theory with classical logic.
A constructive, topos valid, and hence impredicative, treatment of some domain
theory can be found in~\cite[Chapter~III]{Taylor1999}.

Domain theory has been studied predicatively in the setting of formal topology
\cite{Sambin1987,Sambin2003,CoquandEtAl2003} in
\cite{MaiettiValentini2004,Negri2002,SambinValentiniVirgili1996} and the more
recent categorical papers~\cite{Kawai2017,Kawai2021}. In this predicative
setting, one avoids size issues by working with information
systems~\cite{Scott1982a,Scott1982b}, abstract bases~\cite{AbramskyJung1994} or
formal topologies, rather than dcpos, and approximable relations rather than
Scott continuous functions. %
\index{information system}\index{basis!abstract}\index{formal topology}%
\index{approximable relation}%
Hedberg~\cite{Hedberg1996} presented some of these ideas in Martin-L\"of Type
Theory and formalised them in the proof assistant \ALF~\cite{Magnusson1995}, a
precursor to \Agda. A~modern formalisation in \Agda\ based on Hedberg's work was
recently carried out in Lidell's master thesis~\cite{Lidell2020}.

Our development differs from the above line of work in that it studies
dcpos directly and uses type universes to account for the fact that
dcpos may be large.\index{dcpo}
An advantage of this approach is that we can work with (Scott continuous)
functions rather than the arguably more involved (approximable) relations.
For the treatment of continuous (and algebraic) dcpos we turn to the work
of~\citeauthor{JohnstoneJoyal1982}~\cite{JohnstoneJoyal1982} which is situated
in category theory where attention must be paid to size issues even in an
impredicative setting.
In constructive set theory, this corresponds to working with partially ordered
classes~\cite{Aczel2006} as opposed to sets, where our notion of a small basis
for a dcpo (\cref{sec:small-bases}) is similar to \citeauthor{Aczel2006}'s
notion of a set-generated~\cite[Section~6.4]{Aczel2006} dcpo.

Another approach to formalising domain theory in type theory can be found
in~\cite{BentonKennedyVarming2009,Dockins2014}. Both formalisations study
\(\omega\)-chain complete preorders, work with setoids, and make use of \Coq's
impredicative sort~\texttt{Prop}. %
\index{omega-completeness@\(\omega\)-completeness}%
\nomenclature[Prop]{$\texttt{Prop}$}{special type of propositions in the
  Calculus of Constructions and \Coq}%
A setoid is a type equipped with an equivalence relation that must be respected
by all functions. The particular equivalence relation given by equality is
automatically respected of course, but for general equivalence relations this
must be proved explicitly. %
\index{setoid}%
The aforementioned formalisations work with preorders, rather than posets,
because they are setoids where two elements \(x\) and \(y\) are related if
\(x \leq y\) and \(y \leq x\). %
\index{preorder}
Our~development avoids the use of setoids thanks to the adoption of the
univalent point of view. Moreover, we work predicatively and we work with the
more general directed families rather than \(\omega\)-chains, as we intend the
theory to be also applicable to topology and algebra~\cite{GierzEtAl2003}.

There are also constructive accounts of domain theory aimed at program
extraction~\cite{BauerKavkler2009,PattinsonMohammadian2021}.
Both these works study \(\omega\)-chain complete posets (\(\omega\)-cpos) and
define notions of \(\omega\)-continuity for them. %
\index{omega-completeness@\(\omega\)-completeness}%
The former~\cite{BauerKavkler2009} is notably predicative, but makes use of
additional logical axioms: countable choice, dependent choice and Markov's
Principle, which are validated by a realisability interpretation. %
\index{constructivity}\index{choice!axiom of countable}%
The latter~\cite{PattinsonMohammadian2021} uses constructive logic to extract
witnesses but employs classical logic in the proofs of correctness by phrasing
them in the double negation fragment of constructive logic.
By~contrast, we study (continuous) dcpos rather than (\(\omega\)-continuous)
\(\omega\)-cpos and is fully constructive without relying on additional
principles such as countable choice or Markov's Principle.

Finally, yet another approach is the field of \emph{synthetic domain
  theory}~\cite{Rosolini1986,Rosolini1987,Hyland1991,Reus1999,ReusStreicher1999}. %
\index{domain theory!synthetic}%
Although the work in this area is constructive, it is still impredicative, as it
is based on topos logic; but more importantly it has a focus different from that
of regular domain theory. The aim is to isolate a few basic axioms and find
models in (realisability) toposes where every object is a domain and every
morphism is continuous. These models often validate additional axioms, such as
Markov's Principle and countable choice, and moreover (necessarily) falsify
excluded middle. We have a different goal, namely to develop regular domain
theory constructively and predicatively, but in a foundation compatible with
excluded middle and choice, while not relying on them or on Markov's Principle
or countable choice.
\index{domain theory|)}

\subsection{Predicativity}\label{sec:related-work:predicativity}

We summarise work on (im)predicativity in univalent foundations as well as work
on the limits of predicative mathematics and its relation to the results
presented in~\cref{chap:predicativity-in-order-theory}.

\subsubsection{Resizing in models of univalent foundations}
\index{resizing!in models}

As mentioned in the introduction, propositional resizing axioms fail in some
models of univalent type theory. A notable example of such a model is Uemura's
cubical assembly model~\cite{Uemura2019}. What is particularly striking about
Uemura's model is that it does support an impredicative universe \(\U\) in the
sense that if \(X\) is \emph{any} type and \(Y : X \to \U\), then
\(\Pi_{x : X}Y(x)\) is in \(\U\) again even if \(X\) isn't, but that
propositional resizing fails for this universe.
We also highlight Swan's (unpublished) results~\cite{Swan2019a,Swan2019b} that
show that propositional resizing axioms fail in certain presheaf (cubical)
models of type theory. Interestingly, Swan's argument works by showing that the
models violate certain collection principles if we assume Brouwerian continuity
principles in the metatheory.

By contrast, we should mention that propositional resizing is validated in many
models when a classical metatheory is assumed. For example, this is true for any
type-theoretic model topos~\cite[Proposition~11.3]{Shulman2019}. In particular,
Voevodsky's simplical sets model~\cite{KapulkinLumsdaine2021} validates excluded
middle and hence propositional resizing.
We~note, however, that in other models it is possible for propositional resizing
to hold and excluded middle to fail, as shown
by~\cite[Remark~11.24]{Shulman2015}.

\subsubsection{Resizing rules versus axioms}
\index{resizing!axioms and rules}

This thesis concerns resizing \emph{axioms}, meaning we ask a given type to be
equivalent to one in a fixed universe \(\U\) of ``small'' types.
Voevodsky~\cite{Voevodsky2011} originally introduced resizing \emph{rules} which
add judgements and hence modify the syntax of the type theory to make the given
type inhabit \(\U\), rather than only asking for an equivalent copy
in~\(\U\). It is not known whether Voevodsky's resizing rules are consistent
with univalent foundations in the sense that no-one has constructed a model of
univalent type theory extended with such resizing rules.
It is also an open problem~\cite[Section~10]{CohenEtAl2018} whether we have
normalisation for cubical type theory extended with resizing rules.
In fact, as far as we know, it is an open problem for plain Martin-L\"of Type
Theory as well.

\subsubsection{Limits of predicativity}\index{predicativity|(}

While~\cref{chap:basic-domain-theory,chap:continuous-and-algebraic-dcpos,chap:applications}
are devoted to demonstrating the possibility of developing domain theory
predicatively in univalent foundations,
\cref{chap:predicativity-in-order-theory} instead explores what cannot be done
in our predicative setting.
Curi had a similar goal and investigated the limits of predicative mathematics
in CZF\index{CZF}~\cite{AczelRathjen2010} in a series of
papers~\cite{Curi2010a,Curi2010b,Curi2015,Curi2018,CuriRathjen2012}.
In particular, Curi shows (see~\cite[Theorem~4.4 and
  Corollary~4.11]{Curi2010a}, \cite[Lemma~1.1]{Curi2010b} and
  \cite[Theorem~2.5]{Curi2015}) that CZF cannot prove that various nontrivial
posets, including sup-lattices, dcpos and frames, are small. This result is
obtained by exploiting that CZF is consistent with the anti-classical
generalised uniformity principle
(GUP)~\cite[Theorem~4.3.5]{vandenBerg2006}.%
\index{smallness}\index{dcpo!small}\index{uniformity principle}

Our related \cref{nontrivial-impredicativity,positive-impredicativity} is of a
different nature in two ways.
Firstly, the theorem is in the spirit of reverse constructive
mathematics~\cite{Ishihara2006}: Instead of showing that GUP implies that there
are no non-trivial small dcpos, we show that the existence of a non-trivial
small dcpo is equivalent to weak propositional resizing, and that the
existence of a positive small dcpo is equivalent to full propositional
resizing. Thus, if we wish to work with small dcpos, we are forced to assume
resizing axioms.
Secondly, we work in univalent foundations rather than CZF\index{CZF}. This may
seem a superficial difference, but a number of arguments in Curi's
papers~\cite{Curi2015,Curi2018} crucially rely on set-theoretical notions and
principles such as transitive set, set-induction, and the weak regular extension
axiom (wREA), which cannot even be formulated in the underlying type theory of
univalent foundations.
Moreover, although Curi claims that the arguments of~\cite{Curi2010a,Curi2010b}
can be adapted to some version of Martin-L\"of Type Theory, it is presently not
clear whether there is any model of univalent foundations which validates
GUP. However, one of the anonymous reviewers of~\cite{deJongEscardo2022}
suggested that Uemura's cubical assemblies model~\cite{Uemura2019} might
validate it. In particular, the reviewer hinted that
\cite[Proposition~21]{Uemura2019} may be seen as a uniformity principle.
\index{predicativity|)}

\section{Outline and summary of contributions}
\label{sec:outline-and-summary-of-contributions}

We develop domain theory~(\cref{chap:basic-domain-theory}) in predicative and
constructive univalent foundations~(\cref{chap:univalent-foundations}). %
We include the theory of continuous and algebraic dcpos and rounded ideal
completions~(\cref{chap:continuous-and-algebraic-dcpos}), as well as
applications in the semantics of programming
languages~(\cref{chap:applications}), namely soundness and computational
adequacy of Scott's model of PCF, and Scott's \(D_\infty\) model of the untyped
\(\lambda\)-calculus.
We use type universes to deal with size issues arising in our predicative
setting. Moreover, we show that dcpos are predicatively necessarily large
in~\cref{chap:predicativity-in-order-theory}.
The development of domain theory, including the applications, is supported by a
formalisation, as discussed in~\cref{chap:formalisation}. In particular, \Agda's
ability to automatically infer universe levels has been invaluable to us.%
\index{proof assistant!Agda@\Agda}\index{universe!parameters}

\subsection{Summary of contributions}

We briefly describe our contributions per chapter and record what parts of this
thesis are based on our publications. The full bibliographical details of the
publications can be found in~\cref{publications}.
Moreover, each chapter features a section at the end with further
bibliographical notes.

\paragraph{\cref{chap:univalent-foundations}} %
\emph{\nameref*{chap:univalent-foundations}}

Our exposition of \emph{univalent foundations} is fairly standard and largely
follows~\cite{HoTTBook}, and in particular~\cite{Escardo2019}.
Two exceptions
are~\cref{sec:propositional-truncations-from-set-quotients,sec:set-replacement}
which are original contributions where we show small \emph{set quotients} and a
\emph{set replacement} principle to be equivalent.
\cref{sec:quotients-replacement-prop-trunc-revisited}, on set quotients,
\emph{propositional truncations} and their universe levels, as a whole was
included in our work~\cite{deJongEscardo2021b,deJongEscardo2022}.
Other exceptions are the main results on \emph{indexed \(\WW\)-types} with
\emph{decidable equality} in~\cref{sec:indexed-W-types} which are due to Jasper
Hugunin~\cite{Hugunin2017,HuguninMail2017}, and were included in our
paper~\cite{deJong2021a}.

\paragraph{\cref{chap:basic-domain-theory}} %
\emph{\nameref*{chap:basic-domain-theory}}

We present the basic definitions of domain theory: \emph{directed complete
  posets (dcpos)} and \emph{Scott continuous functions}.
It must be remarked that our definitions make use of type universes and are
size-aware: we ask for suprema of directed families indexed by types in some
fixed universe.
We proceed with several basic examples and with constructions of dcpos:
\emph{products}, \emph{exponentials}, \emph{lifting} and \emph{bilimits}.
Because we work constructively we use Escard\'o's and
Knapp's~\cite{EscardoKnapp2017,Knapp2018} lifting monad to construct the free
dcpo with a least element on a set.
This chapter is a revision of our two
papers~\cite{deJong2021a,deJongEscardo2021a}, see the
\nameref{sec:basic-domain-theory-notes} for further details.

\paragraph{\cref{chap:continuous-and-algebraic-dcpos}} %
\emph{\nameref*{chap:continuous-and-algebraic-dcpos}}

This chapter has its roots in~\cite{deJongEscardo2021a}, but the treatment has
been considerably expanded and revised. In particular, we disentangled the
notions of \emph{continuity} and having a \emph{(small) basis} in this thesis.
Taking inspiration from the categorical treatment of~\cite{JohnstoneJoyal1982},
we give predicatively adequate definitions of continuous and \emph{algebraic}
dcpos, and discuss issues related to the absence of the axiom of choice.
We~also present predicative adaptations of the notions of a basis and the
\emph{rounded ideal completion}~\cite{AbramskyJung1994}.
Our development is illustrated with several examples: we describe small compact
bases for the lifting and the powerset, and consider the ideal completion of the
dyadics.

\paragraph{\cref{chap:applications}} %
\emph{\nameref*{chap:applications}}

We describe two applications of domain theory to the \emph{semantics of
  programming languages}.
The first application is a predicative reconstruction of
Scott's~\cite{Scott1972} famous \(D_\infty\) model of the untyped
\(\lambda\)-calculus, and was included in our paper~\cite{deJongEscardo2021a}.
The~use of exponentials and bilimits of dcpos is crucial in the construction of
\(D_\infty\) and we describe how Scott's original proof is adapted to
predicative and proof relevant setting of univalent foundations.
The second application is the \emph{Scott model}~\cite{Plotkin1977,Scott1993} of
the typed programming language \emph{PCF}, including its \emph{soundness} and
\emph{computational adequacy}, and was the subject of our
publication~\cite{deJong2021a}.
The Scott model of PCF highlights our use of the lifting monad in particular.
We also discuss issues concerning semidecidability and countable choice.

\paragraph{\cref{chap:predicativity-in-order-theory}} %
\emph{\nameref*{chap:predicativity-in-order-theory}}

We complement the above development by exploring the predicative and
constructive limits of order theory in univalent foundations.
We show that nontrivial dcpos are \emph{necessarily large} and \emph{necessarily
  lack decidable equality} in our constructive and predicative setting.
In particular, the carriers of the dcpos of the Scott model of PCF can only live
in the lowest universe \(\U_0\) if we work impredicatively.
The fact that nontrivial dcpos are necessarily large has the important
consequence that \emph{Tarski's theorem} (and similar results) cannot be applied
in nontrivial instances, even though it has a predicative proof.
Further, we explain, by studying the large sup-lattice of \emph{ordinals}, that
generalisations of Tarski's theorem which allow for large structures are provably
false.
Finally, we elaborate on the connections between requiring suprema of
\emph{families} and of \emph{subsets} in our predicative setting.
This chapter is taken mostly verbatim from our preprint~\cite{deJongEscardo2022}
which itself is based on our conference paper~\cite{deJongEscardo2021b}.

\paragraph{\cref{chap:formalisation}} %
\emph{\nameref*{chap:formalisation}}

Our development of domain theory in constructive and predicative univalent
foundations is accompanied by extensive formalisations that encompass, with very
few exceptions, all
of~\cref{chap:basic-domain-theory,chap:continuous-and-algebraic-dcpos,chap:applications}.

\subsection{Publications}\label{publications}

This thesis is based on the following papers, all of which have been published,
except for~\cite{deJongEscardo2022}, which has been accepted subject to minor
revisions.

\nocite{deJong2021a,deJongEscardo2021a,deJongEscardo2021b,deJongEscardo2022}
\AtNextBibliography{\renewbibmacro*{pageref}{}}

\defbibnote{awards}{The work~\cite{deJongEscardo2021b} won the \emph{Best
    Paper by a Junior Researcher} award and some of the above publications led
  to the award of the \emph{Homotopy Type Theory Dissertation Fellowship}.}
\printbibliography[keyword=introduction,heading=none,postnote=awards]


%% file: mainmatter/univalent-foundations.tex
\chapter{Univalent foundations}\label{chap:univalent-foundations}

Our foundational starting point is intensional Martin-L\"of Type
Theory~(MLTT)~\cite{MartinLof1975,MartinLof1984} with an empty type \(\Zero\),
unit type \(\One\), natural numbers type \(\Nat\), binary coproducts~(\(+\)),
dependent sums~(\(\Sigma\)), dependent products~(\(\Pi\)), intensional identity
types, and general inductive types (i.e.\ \(\WW\)-types as discussed in
\cref{sec:indexed-W-types}).
\index{Martin-Loef Type Theory@Martin-L\"of Type Theory}%
\index{MLTT|see {Martin-L\"of Type Theory}}%
\index{type!empty}%
\index{type!unit}%
\index{type!of natural numbers}%
\index{type!coproduct}%
\index{type!dependent sum}%
\index{type!dependent product}%
\index{type!identity}%
\index{type!inductive}%
\nomenclature[0B]{$\Zero$}{empty type}
\nomenclature[1B]{$\One$}{unit type}
\nomenclature[N]{$\Nat$}{type of natural numbers}
\nomenclature[plus]{$X + Y$}{binary coproduct type}
\nomenclature[Sigma]{$\Sigma_{x : X}Y(x)$}{dependent sum type}
\nomenclature[Pi]{$\Pi_{x : X}Y(x)$}{dependent product type}

In the upcoming sections, we will discuss additions to this type theory that
will define our foundational setup.  These additions will be: universes,
function extensionality, propositional extensionality, propositional
truncations, and (sometimes) univalence.

We also introduce some of the characteristic features of univalent foundations,
such as the stratification of types into (sub)singletons, sets, 1-groupoids,
\ldots, according to the complexity of their identity types
(\cref{sec:stratification}), the notions of embedding and equivalence
(\cref{sec:embeddings-equivalences-retracts}), propositional truncations
(\cref{sec:prop-trunc}) and univalence (\cref{sec:univalent-universes}).

The notion of subsingletons gives rise to a refinement of the Curry--Howard
paradigm to logical-propositions-as-subsingletons, as explained in
\cref{sec:logic}. This distinguishes univalent foundations from other type
theories such as the \emph{Calculus of (Inductive)
  Constructions}~\cite{CoquandHuet1988} and \Coq~\cite{Coq}, where logic is set
to take place in a special designated \verb|Prop| type. %
\index{Calculus of Constructions}%
\index{proof assistant!Coq@\Coq}%
The fact that our logic is constructive is discussed in
\cref{sec:constructivity}, while its predicativity is studied in
\cref{sec:impredicativity-resizing-axioms} after we introduce the theory of
(locally) small types in \cref{sec:small-and-locally-small-types}.
Finally, set quotients and set replacement are examined in a predicative
context in \cref{sec:quotients-replacement-prop-trunc-revisited}.

\paragraph{Notation}

If \(X\) is a type and \(Y(x)\) a dependent type over \(X\), then we denote its
type of dependent functions as \(\Pi_{x : X}Y(x)\) or sometimes just \(\Pi Y\),
and similarly for \(\Sigma\)-types.
If \(Y(x)\) does not depend on \(X\), then we respectively denote
\(\Pi_{x : X}Y(x)\) by \(X \to Y\) and \(\Sigma_{x : X}Y(x)\) by \(X \times Y\).
\nomenclature[arrow]{$X \to Y$}{type of functions from \(X\) to \(Y\)}%
\nomenclature[times]{$X \times Y$}{binary cartesian product of types}%
\index{type!function}%
\index{type!product}

Further, the first projection is denoted by
\(\fst : \pa{\Sigma_{x : X}Y(x)} \to X\), while the second projection is called
\(\snd : \Pi_{s : {\Sigma_{x : X}Y(x)}}Y(\fst(s))\).
\nomenclature[pr1]{$\fst$}{first projection}%
\nomenclature[pr2]{$\snd$}{second projection}%
The identity map on a type \(X\) is denoted by \(\id\) or \(\id_{X}\) %
\nomenclature[id]{$\id$}{identity map}%
and
function composition is denoted by \(g \circ f\) %
\nomenclature[circ]{$g \circ f$}{function composition}%
where the codomain of \(f\) is
definitionally equal to the domain of \(g\).

Definitional (or judgemental) equality is denoted by \({\equiv}\) and we use
\({\colonequiv}\) to signal that we are making a definition.
\nomenclature[equiv]{$x \equiv y$}{definitional (judgemental) equality}%
\nomenclature[equivcol]{$x \colonequiv y$}{making a definition}%
For a type \(X\) with elements \({x,y} : X\), the corresponding identity type is
denoted by \(x = y\), or sometimes \(x =_{X} y\) to highlight the type of \(x\)
and \(y\).
\nomenclature[equal]{$x = y$}{identity type}%
If \(p : x = y\), then we write \(p^{-1} : y = x\) for its inverse (up to
intensional equality) and if \(q : y = z\), then we write
\(p \pathcomp q : x = z\) for the composition of the identifications.
\nomenclature[inverse]{$p^{-1}$}{inverse of an identification}%
\nomenclature[bullet]{$p \pathcomp q$}{composition of identifications}%
Moreover, we sometimes find it convenient to write \(f \sim g\) for
\(\Pi_{x : X}f(x) = g(x)\) for (dependent) functions \(f,g : \Pi_{x : X}Y(x)\).
\nomenclature[sim]{$f \sim g$}{pointwise equality of functions}%

The unique element of the unit type will be written as \(\star : \One\).
\nomenclature[star]{$\star$}{unique element of the unit type}%
The constructors of the coproduct are denoted by \(\inl\) and \(\inr\). We
denote the particular coproduct \(\One + \One\) by \(\Two\) and write \(0\) and
\(1\) for its elements.
\nomenclature[inl]{$\inl$}{first coprojection into the binary coproduct}
\nomenclature[inr]{$\inr$}{second coprojection into the binary coproduct}
\nomenclature[2B]{$\Two$}{two-element type}
\nomenclature[0]{$0$}{one of the elements of the two point type \(\Two\)}
\nomenclature[1]{$1$}{one of the elements of the two point type \(\Two\)}

\paragraph{Terminology}

Following \cite{HoTTBook}, when we define a (dependent) function using the
elimination rule of the identity type, then we say we define it by
\emph{(path)~induction}. We will also refer to elements of the identity type as
``equalities'', ``identifications'' or ``paths'', and write \(\refl\) or
\(\refl_x\) for reflexivity at \(x\), the canonical element of \(x = x\).
\index{path}%
\index{path!induction}
\nomenclature[refl]{$\refl$}{reflexivity, constructor of the identity type}

In line with the Curry--Howard paradigm and \cite{HoTTBook}, we will work
informally in type theory and for example say that ``\(Y(x)\) holds for every
\(x : X\)'' if we have an element of \(\Pi_{x : X}Y(x)\). Similarly, ``\(P\) and
\(Q\) hold'' will mean that the type \(P \times Q\) has an element.

\section{Type universes}\label{sec:type-universes}%
\index{type!universe|(}%
\index{universe|see {type universe}}
Type universes will play a fundamental role in our development of domain theory
in a predicative context, because, by our definitions, they will keep track of
size for us.
But type universes are indispensable in MLTT anyway as they have many other
uses, e.g.\ defining type families by induction and collecting mathematical
structures into a single type (e.g.\ the type of all (small) groups).
Our setup of type universes follows that of \Agda~\cite{Agda} and
\cite{Escardo2019,Escardo2021}, but differs from that of \Coq~\cite{Coq} and
\cite{HoTTBook} because we do not assume cumulativity
(see~\cref{no-cumulativity}); any lifting of types to higher universes will be
annotated explicitly (see~\cref{lift-to-higher-universes}).

Intuitively a type universe is a type of types. In so-called Tarski-style
universes, the elements of a universe \(\U\) are codes for types and the
universe comes equipped with a decoding \(\tau\) such that if \(x : \U\), then
\(\tau(x)\) is a type. %
This is useful, because it maintains a clean separation between types and
elements of types, but cumbersome in practice. By contrast, in Russell-style
universes, the elements of a universe \(\U\) are actual types. This complicates
the meta-theory because now \(X : \U\) is both an element of \(\U\) and a type.
Our universes will be presented \`a la Russell, but one should read this as an
abbreviation for Tarski-style universes.

\subsection{Operations on universes}

First of all, we postulate that there is a universe \(\U_0\).
Secondly, we postulate two meta-operations on universes: a unary operation
\({(-)}^+\), called \emph{successor}, and a binary operation \((-) \sqcup (-)\)
satisfying the following conditions:
\nomenclature[U+]{$\U^+$}{successor universe}
\nomenclature[sqcup]{$\U \sqcup \V$}{least upper bound of two universes}
\begin{enumerate}[(i)]
\item for every universe \(\U\), we have \(\U_0 \sqcup \U \equiv \U\) and
  \(\U \sqcup \U^+ \equiv \U^+\);
\item the operation \((-) \sqcup (-)\) is definitionally idempotent, commutative
  and associative, i.e.\ for all universes \(\U\), \(\V\) and \(\W\), we assume
  \(\U \sqcup \U \equiv \U\) as well as
  \(\U \sqcup \V \equiv \V \sqcup \U\) and
  \((\U \sqcup \V) \sqcup \W \equiv \U \sqcup (\V \sqcup \W)\);
\item the successor operation \({(-)}^+\) distributes over \((-) \sqcup (-)\)
  definitionally, i.e.\ for every two universes \(\U\) and \(\V\), we have
  \(\pa*{\U \sqcup \V}^+ \equiv \U^+ \sqcup \V^+\).
\end{enumerate}
\nomenclature[U]{$\U$}{type universe}%
\nomenclature[V]{$\V$}{type universe}%
\nomenclature[W]{$\W$}{type universe}

In particular, we can iterate the successor operation starting with \(\U_0\) to
obtain an infinite tower of universes that we denote by
\(\U_0\), \(\U_1\), \(\U_2\), \ldots.
\nomenclature[Ui]{$\U_i$}{\(i^{\textsuperscript{th}}\) type universe}

\begin{remark}\label{no-cumulativity}%
  \index{cumulativity}%
  We do \emph{not} assume cumulativity of the universes, i.e.\ we do not require
  that \(A : \U\) implies \(A : \U \sqcup \V\) for every two universes \(\U\)
  and \(\V\).
  However, in \cref{lift-to-higher-universes} we describe how we can easily
  transport types to higher universes in a suitable sense.
\end{remark}

\subsection{Closure properties}
We assume the following closure properties regarding universes:
\begin{enumerate}[(i)]
\item if \(X : \U\), then the identity type \(\pa{x = y}\) lives in \(\U\) for
  every \(x,y : X\);
\item if \(X : \U\) and \(Y : \V\), then \(X + Y : \U \sqcup \V\);
\item if \(X : \U\) and \(Y : X \to \V\), then the types \(\Sigma_{x : X}Y(x)\)
  and \(\Pi_{x : X}Y(x)\) are both assumed to be in \(\U \sqcup \V\);
\item the universe \(\U_0\) contains the type of natural numbers \(\Nat\);
\item every universe \(\U\) contains copies \(\Zero_{\U}\) and \(\One_{\U}\) of
  respectively the empty and unit type.
  \nomenclature[0U]{$\Zero_{\U}$}{empty type in \(\U\)}
  \nomenclature[1U]{$\One_{\U}$}{unit type in \(\U\)}
\end{enumerate}

We write \(\Two_{\U} \colonequiv \One_{\U} + \One_{\U}\) and \(0\) and \(1\) for
its two inhabitants.
\nomenclature[2U]{$\Two_{\U}$}{two-element type in \(\U\)}

\begin{remark}\label{lift-to-higher-universes}
  To compensate for the fact that we do not assume cumulativity, we observe
  that, using the empty and unit types, it is easy to define a map
  \[
    \lift{U}{V} : \U \to \U \sqcup \V
  \]
  for every two universes such that every type \(X : \U\) is equivalent (a
  notion we define later) to \(\lift{U}{V}(X)\). %
  \nomenclature[liftUV]{$\lift{U}{V}$}{embedding into higher type universe}
  For instance, the map
  \(X \mapsto X + \Zero_{\V}\) does the job, as does
  \(X \mapsto X \times \One_{\V}\).
  But, in the absence of cumulativity, the types \(X\) and \(\lift{U}{V}(X)\)
  cannot be \emph{equal}, because they do not even live in the same universe.
\end{remark}
\index{type!universe|)}

\section{Identity types and function extensionality}%
\label{sec:id-and-fun-ext}%
\index{type!identity|(}

The identity type is defined uniformly for every type \(X : \U\) as the
inductive family \(X \to X \to \U\) generated by \(\refl : x = x\).
It is possible, however, to show that the identity type acts as expected for
specific types. For example, given \({(x,y),(x',y')} : X \times Y\), we would
expect \((x,y) = (x',y')\) to hold precisely when \(x = x'\) and \(y = y'\).
Similarly, we can show that if \(x : X\) and \(y : Y\), then
\(\inl(x) =_{X + Y} \inr(y)\) never holds, while \(\inl(x) =_{X + Y} \inl(x')\)
holds precisely when \(x =_{X} x'\).
The situation for \(\Sigma\)-types is slightly more involved and requires the
notion of \emph{transport}.

\begin{definition}[Transport]
  \index{transport}%
  \nomenclature[transport]{$\transport^Y(e,y)$}{transport of \(y : Y(x)\) along \(e : x = x'\)}%
  For every type \(X : \U\) and type family \(Y : X \to \V\), we have a function
  \(\transport^Y : \Pi_{x,x' : X}\pa*{x = x' \to Y(x) \to Y(x')}\) defined
  inductively as \(\transport^Y(\refl) \colonequiv \id\), where we have left the
  arguments \(x\) and \(x'\) implicit.
\end{definition}

We also take this opportunity to define the action of a map on paths.

\begin{definition}[Action on paths, \(\ap_f\)]
  \index{path!action on}%
  \nomenclature[ap]{$\ap_f(e)$}{application of \(f : X \to Y\) to \(e : x = y\)}%
  Every function \(f : X \to Y\) induces a map on identity types
  \(\ap_f : (x = y) \to (f(x) = f(y))\) for every \(x,y : X\) defined
  inductively by \(\ap_f(\refl) \colonequiv \refl\), and sometimes called the
  \emph{action (of \(f\)) on paths}.
\end{definition}

For characterising identity types, we introduce the notion of an invertible
map. In \cref{sec:embeddings-equivalences-retracts} we consider the more refined
notion of a map being an \emph{equivalence}.
Another useful notion is that of a left-cancellable map, which we will similarly
refine to the notion of an \emph{embedding} later.

\begin{definition}[Invertibility and left-cancellability]
  A map \(f : X \to Y\) is
  \begin{enumerate}[(i)]
  \item \emph{invertible}\index{invertibility|textbf} if we have a specified
    \(g : Y \to X\) with \(g(f(x)) = x\) for every \(x : X\) and \(f(g(y)) = y\)
    for every \(y : Y\), and
  \item \emph{left-cancellable}\index{left-cancellability} if for every
    \(x,x' : X\), we have a function
    \[
      (f(x) = f(x')) \to (x = x'). \qedhere
    \]
  \end{enumerate}
\end{definition}

\begin{lemma}\label{lc-if-invertible}
  Every invertible map is left-cancellable.
\end{lemma}
\begin{proof}
  If \(f : X \to Y\) is invertible with inverse \(g : Y \to X\), then for every
  \({x,x'} : X\), we have %
  \(f(x) = f(x') \xrightarrow{\ap_g} g(f(x)) = g(f(x')) \to x = x'\), where
  the final map is obtained using that \(g\) is the inverse of \(f\).
\end{proof}

\begin{lemma}\label{Id-of-Sigma}
  For every type family \(Y\) over a type \(X\), and every
  \({(x,y),(x',y')} : \Sigma Y\), we have invertible maps between the identity
  type \((x,y) =_{\Sigma Y} (x',y')\) and the \(\Sigma\)-type
  \(\Sigma_{p : x = x'}\transport^Y(p,y) = y'\).

  In particular, if \(Y\) is just a type, then we have invertible maps between
  the identity type \((x,y) =_{X \times Y} (x',y')\) and the product of identity
  types \((x=x') \times (y=y')\).
\end{lemma}

The need for transporting \(y\) arises from the fact that \(y : Y(x)\), while
\(y' : Y(x')\), so \(y\)~and~\(y'\) cannot be equal as they do not have the same
type.

\begin{proof}
  The invertible maps are inductively defined as
  \begin{align*}
    (x,y) &=_{\Sigma Y} (x',y') &&\phantom{\mapsto} & \Sigma_{p : x = x'}&\transport^Y(p,y) = y' \\
    & \refl_{(x,y)} &&\mapsto & &(\refl_{x},\refl_{y}) \\
    & \refl_{(x,y)} &&\mapsfrom & &(\refl_{x},\refl_{y}) \qedhere
  \end{align*}
\end{proof}

By contrast, it is not provable in intensional Martin-L\"of Type Theory that two
pointwise equal functions are equal, as shown
by~\cite[Theorem~3.17]{Streicher1993}. Therefore, we wish to add it as an
axiom. However, for reasons that we will explain later, the official formulation
of the axiom will have to wait. Even so, the official formulation will be
logically equivalent to the following unofficial axiom that we introduce now
under the name ``naive function extensionality''.

\begin{definition}[Naive function extensionality]\label{naive-fun-ext}
  \emph{Naive function extensionality}\index{extensionality!naive function}
  asserts that for every two functions \(f,g : X \to Y\), if \(f(x) = g(x)\) for
  every \(x : X\), then \(f = g\).
  In other words, naive function extensionality says that pointwise functions
  are equal.
\end{definition}

The reason that we introduce naive function extensionality this early is that it
allows us to present many useful results earlier. When using it, we typically
drop the word ``naive'' and simply say ``by function extensionality''.
The following lemma prepares us for the official formulation of function
extensionality later.

\begin{lemma}\label{fun-ext-alt}
  Naive function extensionality is logically equivalent to all of the below,
  seemingly stronger, statements:
  \begin{enumerate}[(i)]
  \item for every two \emph{dependent} functions \(f,g : \Pi_{x : X}Y(x)\), if
    \(f(x) = g(x)\) for every \(x : X\), then \(f = g\);
  \item for every two functions \(f,g : X \to Y\), the canonical function from
    \(f = g\) to \(\Pi_{x : X}f(x) = g(x)\) given by
    \(e \mapsto \lambdadot{x}{\ap_{\lambdadot{h}{h(x)}}(e)}\) is invertible;
  \item for every two \emph{dependent} functions \(f,g : \Pi_{x : X}Y(x)\), the
    canonical function from \(f = g\) to \(\Pi_{x : X}f(x) = g(x)\) is
    invertible.
  \end{enumerate}
\end{lemma}
\begin{proof}
  See \cite[Section~3.18]{Escardo2019}.
\end{proof}
\index{type!identity|)}

\section{Subsingletons, sets and (higher) groupoids}\label{sec:stratification}

A fundamental idea in univalent foundations is the stratification of types
according to the complexity of their identity types.

\begin{definition}[Subsingleton, proposition, truth value]
  A type \(X\) is a \emph{subsingleton} (or \emph{proposition} or \emph{truth
    value}) if it has at most one element, meaning we have an element of
  \(\Pi_{x,y : X}\, x = y\).%
  \index{subsingleton}%
  \index{proposition|see {subsingleton}}%
  \index{truth value|see {subsingleton}}
\end{definition}

\begin{remark}[Property and data]\label{property-and-data}
  If a type \(X\) is a subsingleton, then we like to say that \(X\) is
  \emph{property}. By contrast, if \(X\) can have more than one element, then we
  sometimes say that \(X\) is \emph{data}. %
  \index{property|textbf}\index{data}%
  For example, as we explain in~\cref{sec:embeddings-equivalences-retracts}
  the notion of being an equivalence is a property, while being invertible is
  data.
  Another example comes from our development of domain theory and is the
  distinction between structural continuity and continuity of a directed
  complete poset~(\cref{sec:continuous-dcpos}).
  The former equips the poset with additional structure in the form of a
  specified mapping, while the latter only requires some unspecified mapping to
  exist, in a sense to be made precise in~\cref{sec:prop-trunc}.
\end{remark}

The names ``proposition'' and ``truth value'' suggests that subsingletons are
related to logic and indeed we will use the subsingletons to encode logic in our
type theory in~\cref{sec:logic}. %

\begin{definition}[Type of subsingletons, \(\Omega_{\U}\)]%
  \label{def:Omega}%
  \nomenclature[Omega]{$\Omega_{\U}$}{type of subsingletons in \(\U\)}%
  \index{type!of subsingletons|textbf}%
  The \emph{type of subsingletons} in a universe~\(\U\) is defined as
  \(\Omega_{\U} \colonequiv \Sigma_{P : \U}\issubsingleton(P)\).
\end{definition}

\begin{definition}[Singleton, contractibility]
  A type \(X\) is a \emph{singleton} (or said to be \emph{contractible}) if it
  is a subsingleton and moreover we have an element of \(X\).%
  \index{contractibility|textbf}\index{singleton|see {contractibility}}
\end{definition}

\begin{theorem}\label{singleton-type-is-singleton}
  For every element \(x\) of a type \(X\), the type \(\Sigma_{y : X}\,x=y\) is a
  singleton with unique element \((x,\refl)\).
\end{theorem}
\begin{proof}
  We have to show that for every \(y : X\) and \(p : x = y\), the pair
  \((y , p) = (x,\refl)\), but by path induction we may assume that
  \(y \equiv x\) and \(p \equiv \refl\) in which case it is trivial.
\end{proof}

\begin{example}\label{examples-of-(sub)singletons}
  The empty type \(\Zero_{\U}\) and the unit type \(\One_{\U}\) in any universe
  \(\U\) are both subsingletons.
  A further example of a subsingleton is the type
  \[
    \Sigma_{n : \Nat}\pa*{n \text{ is the least number \(k\) for which
        \(\alpha_k = 0\)}}
  \] where \(\alpha : \Nat \to \Two\). By contrast, the type
  \(\Sigma_{n : \Nat}\,\alpha_n = 0\) is not necessarily a subsingleton, because
  \(\alpha\) could have multiple roots.
\end{example}

\begin{remark}
  The fact that the type
  \(\Sigma_{n : \Nat}\pa*{n \text{ is the least number \(k\) for which
      \(\alpha_k = 0\)}}\) from~\cref{examples-of-(sub)singletons} is a
  subsingleton shows us that subsingletons are not necessarily proof irrelevant,
  because an inhabitant of that type gives us an explicit natural number.
  Another example that we will discuss in some detail
  later~(\cref{sec:embeddings-equivalences-retracts}) is the notion of an
  equivalence: the type expressing that a map \(f\) is an equivalence is a
  subsingleton, but given an inhabitant of it, we can construct an inverse of
  \(f\).
\end{remark}

So a subsingleton is a type where any two elements are equal. Going up one
level, we consider types where any two identifications are equal.

\begin{definition}[Set]
  A type \(X\) is a \emph{set} if the type \(x = y\) is a subsingleton for every
  \(x,y : X\).%
  \index{set}
\end{definition}

In other words, in a set two elements are equal in at most one way.
\begin{example}
  Every type with decidable equality is a set. This classic result is known as
  Hedberg's Theorem~\cite{Hedberg1998} (\cref{Hedberg-Theorem} below). In
  particular, the type \(\Nat\) of natural numbers is a set.
\end{example}

We could iterate these definitions and arrive at higher \emph{groupoids}:
a~1\nobreakdash-groupoid is a type whose identity types are sets, a 2-groupoid
is a type whose identity types are 1-groupoids, etc.
\begin{example}\label{circle}%
  \nomenclature[S1]{$\Circ$}{circle (as a higher inductive type)}%
  \nomenclature[base]{$\operatorname{base}$}{basepoint of the circle}%
  \nomenclature[Z]{$\mathbb Z$}{type of integers}%
  \index{circle}\index{higher inductive type}%
  An example of a type that is not a set is the circle~\(\Circ\), a higher
  inductive type~\cite[Chapter~6]{HoTTBook} with a chosen basepoint
  \(\operatorname{base} : \Circ\) for which we can prove, assuming
  univalence~(see~\cref{def:univalence}), that
  \((\operatorname{base} = \operatorname{base})\) is equivalent
  to~\(\mathbb Z\), the type of the integers. In our work, we will not assume
  any higher inductive types other than propositional truncations.
\end{example}
As explained in~\cite[Example~3.1.9]{HoTTBook}, another example of a non-set is
given by any univalent universe \(\U\). By univalence of \(\U\), one can show
that \(\Two_{\U} = \Two_{\U}\) contains exactly two elements, which shows that
the universe \(\U\) is not a set.
A nice example of a 1-groupoid is the type of groups in a universe \(\U\): for
two groups \(G\) and \(H\), the type \(G =_{\operatorname{Grp}_{\U}} H\) is,
assuming univalence, equivalent to the type of group isomorphisms between \(G\)
and \(H\), which is a set.

In this thesis we do not need to develop the theory of higher groupoids and
besides universes we can often restrict our attention to sets and subsingletons,
like in the upcoming closure results.

\subsection{Hedberg's Lemma}

In proving various results about subsingletons and sets, the following lemma,
which we call Hedberg's Lemma, proves highly useful. While all the techniques
were already present in Hedberg's paper~\cite{Hedberg1998}, the precise
formulation presented below only appeared in \cite[Lemma~3.11]{KrausEtAl2017}.

\begin{definition}[Constant]
  A map \(f : X \to Y\) is \emph{constant} if \(f(x) = f(x')\) for every
  \(x,x' : X\).%
  \index{constant}
\end{definition}

\begin{remark}
  This is sometimes called \emph{weakly constant} or \emph{wildly constant},
  because if \(Y\) is not a set, then \(f\) can be constant in more than one
  way, but also in an incoherent way in a precise higher categorical
  sense~\cite{Kraus2015}.
  In other words, the above definition does not account for further
  coherence conditions.
  But we will only be interested in constant maps to sets, so we simply stick to
  ``constant''.
\end{remark}

\begin{lemma}[Hedberg's Lemma]%
  \label{Hedberg-Lemma}\index{Hedberg's!Lemma}%
  Let \(x\) be an arbitrary, but fixed element of a type \(X\). If we have
  a constant endofunction on \(x = y\) for every \(y : X\), then \(x = y\) is a
  proposition for every \(y : X\).
\end{lemma}
\begin{proof}
  Suppose that \(f_y : (x = y) \to (x = y)\) is constant for every \(y : X\).
  By induction on \(p : x = y\), we see that every \(p : x = y\) is equal to
  \({{f_{x} (\refl)}^{-1}} \pathcomp {f_y(p)}\).
  Hence, if \(y : X\) is arbitrary and \({p,q} : x = y\), then
  \(p = {{f_{x} (\refl)}^{-1}} \pathcomp {f_y(p)} =
  {{f_{x} (\refl)}^{-1}} \pathcomp {f_y(q)} = q\),
  as \(f_y\) is constant. Hence, each \(x = y\) is a proposition, as desired.
\end{proof}

\subsection{Closure properties}\index{subsingleton!closure properties|(}
The proofs of the following two lemmas illustrate how to apply Hedberg's Lemma.

\begin{lemma}\label{set-if-subsingleton}
  Every subsingleton is a set.
\end{lemma}
\begin{proof}
  If \(X\) is a subsingleton, then for every \(x,y : X\) we have a map
  \(\One \to (x=y)\). But the composite \((x=y) \to \One \to (x=y)\) is
  constant, because \(\One\) is a subsingleton, so \(X\) must be a set by
  Hedberg's Lemma.
\end{proof}

\begin{lemma}\label{lc-to-prop}
  If \(Y\) is a subsingleton (or set, respectively) and \(f : X \to Y\) is a
  left-cancellable map, then \(X\) is a subsingleton (or set, respectively) too.
  In particular, this holds if \(f\) is invertible.
\end{lemma}
\begin{proof}
  Assume that \(f\) is left-cancellable. Suppose first that \(Y\) is a
  subsingleton and let \(x,x' : X\) be arbitrary. Then \(f(x) = f(x')\) because
  \(Y\) is a subsingleton, but \(f\) is left-cancellable so we get the desired
  \(x = x'\), showing that \(X\) is a subsingleton.

  Now suppose that \(Y\) is a set. To show that \(X\) is a set, it suffices, by
  Hedberg's Lemma, to construct a constant endofunction on \(x = x'\) for every
  \(x,x' : X\). But because \(Y\) is a set, the second map, and hence the
  composite
  \[
    x = x' \xrightarrow{\ap_f} f(x) = f(x') \xrightarrow{\text{\(f\) is
        left-cancellable}} x = x'
  \]
  is constant.

  The final claim holds because every invertible map is left-cancellable as
  shown in \cref{lc-if-invertible}.
\end{proof}

\begin{theorem}\label{Sigma-is-prop}
  The subsingletons and sets are closed under \(\Sigma\), e.g.\ if \(X\) is a
  subsingleton and \(Y\) is a type family over \(X\) such that each \(Y(x)\) is
  a subsingleton, then \(\Sigma_{x : X}Y(x)\) is a subsingleton too.
  In particular, if \(Y\) is just a type, then \(X \times Y\) is a subsingleton
  (or set, respectively) if both \(X\)~and~\(Y\) are.
\end{theorem}
\begin{proof}
  Suppose first that \(X\) and each \(Y(x)\) are subsingletons and that we have
  two pairs \((x,y),(x',y') : \Sigma Y\). We wish to show that
  \((x,y) = (x',y')\). By \cref{Id-of-Sigma} it suffices to find an element
  \(p : x = x'\) and an element of \(\transport^Y(p,y) = y'\). But \(X\) is
  assumed to be a subsingleton, so we have such a \(p\) and moreover, \(Y(x')\)
  is assumed to be a proposition, so any two of its element are equal, in
  particular \(\transport^Y(p,y)\) and \(y'\) must be equal.

  Now suppose that \(X\) and each \(Y(x)\) are sets. To show that \(\Sigma Y\)
  is a set, we have to prove that \((x,y) = (x',y')\) is a subsingleton for
  every two pairs \((x,y),(x',y') : \Sigma Y\). By \cref{lc-to-prop,Id-of-Sigma}
  it is enough to show that \(\Sigma_{p : x = x'}\transport^Y(p,y) = y'\) is a
  subsingleton. But this is a \(\Sigma\)-type of subsingletons because \(X\) and
  \(Y(x')\) are assumed to be sets and we already proved that such
  \(\Sigma\)-types are subsingletons again.
\end{proof}

The proof of the following fundamental theorem features our first application of
function extensionality. That function extensionality is in fact necessary is
discussed in~\cite[Section~3.18]{Escardo2019}.

\begin{theorem}\label{Pi-is-prop}
  The subsingletons and sets form a (dependent) exponential ideal. That is, if
  \(Y\) is a type family over an \emph{arbitrary} type \(X\) such that each
  \(Y(x)\) is a subsingleton (or set, respectively), then \(\Pi_{x : X}Y(x)\) is
  a subsingleton (or set, respectively) too.

  In particular, if \(Y\) is just a type, then \(X \to Y\) is a subsingleton
  (or set, resp.) if \(Y\) is.
\end{theorem}

We stress that, unlike for \(\Sigma\)-types, \(X\) is \emph{not} required to be
a subsingleton or a set.

\begin{proof}
  Note that if each \(Y(x)\) is a subsingleton, then \(f(x) = g(x)\) for all
  \(x : X\) and all functions \(f,g : \Pi Y\). Hence, \(f = g\) for all
  \(f,g : \Pi_{x : X} Y(x)\) by function extensionality.

  Now assume that each \(Y(x)\) is a set and let \(f,g:\Pi_{x : X} Y(x)\) be
  arbitrary. We must show that \(f = g\) is a subsingleton.
  By function extensionality and \cref{fun-ext-alt} we have an invertible map
  from \(f = g\) to \(\Pi_{x : X} f(x) = g(x)\) for every two
  \(f,g : \Pi_{x : X} Y(x)\).
  Hence, by \cref{lc-to-prop} it suffices to prove that
  \(\Pi_{x : X}f(x) = g(x)\) is a subsingleton. But each \(Y(x)\) is a set, so
  this is a \(\Pi\)-type over a subsingleton-valued family and hence a
  subsingleton itself as we have just shown.
\end{proof}

\begin{lemma}\label{subsingleton-criterion}
  For every type \(X\), if we have a function \(X \to \issubsingleton(X)\), then
  \(X\) is a subsingleton.
\end{lemma}
\begin{proof}
  Suppose that we have a function \(f : X \to \issubsingleton(X)\). To show that
  \(X\) is a subsingleton, recall that
  \(\issubsingleton(X) \colonequiv \Pi_{x,y : X}(x = y)\). So let \(x : X\) be
  arbitrary and note that we must prove that \(x = y\) for every \(y : X\). But
  this is given by \(f_x(x)\).
\end{proof}

\begin{theorem}\label{being-prop-is-prop}
  Being a set or (sub)singleton is a property, i.e.\ for every type \(X\), the
  types \(\issubsingleton(X)\), \(\issingleton(X)\) and \(\isset(X)\) are
  themselves subsingletons.
\end{theorem}
\begin{proof}
  We first show that \(\issubsingleton(X)\) is a subsingleton. By
  \cref{subsingleton-criterion} we may assume that we have an element of
  \(\issubsingleton(X)\), i.e.\ that \(X\) is a subsingleton.
  Now recall that \(\issubsingleton(X) \colonequiv \Pi_{x,y : X}(x = y)\).
  By \cref{Pi-is-prop} it suffices to prove that \(x=y\) is a subsingleton for
  every \({x,y} : X\). But this is indeed the case, as \(X\) was assumed to be a
  subsingleton, and hence must be a set by \cref{set-if-subsingleton}.  Thus,
  for every type \(X\), the type \(\issubsingleton(X)\) is a subsingleton, as
  desired.
  In particular, if \(X\) is an arbitrary type, then \(\issubsingleton(x=y)\) is
  a subsingleton for every \(x,y : X\). Hence, \(\isset(X)\) is a subsingleton
  by \cref{Pi-is-prop}.
  Finally, to show that \(\issingleton(X)\) is a subsingleton, note that we can
  assume that \(X\) is a singleton by \cref{subsingleton-criterion}. Hence, in
  particular, it is a subsingleton, so
  \(\issingleton(X) \equiv X \times \issubsingleton(X)\) is seen to be a
  proposition by \cref{Sigma-is-prop}.
\end{proof}%
\index{subsingleton!closure properties|)}

\subsection{Propositional extensionality}

Having introduced subsingletons (or propositions), we ask: when should two
propositions be equal? Since they have at most one element, it seems natural to
want them to be equal exactly when one has an element if and only if the other
does.

\begin{definition}[\Axiom: Propositional extensionality]\label{def:prop-ext}%
  \index{extensionality!propositional}%
  \emph{Propositional extensionality} asserts that for every two propositions
  \(P\) and \(Q\), if \(P \to Q\) and \(Q \to P\), then \(P = Q\).
  In~other words, it says that logically equivalent propositions are equal.
\end{definition}

We take propositional extensionality as an axiom and often use it tacitly.

\begin{remark}
  This really is an axiom \emph{scheme}: we add propositional extensionality for
  propositions in every type universe \(\U\).
  One is forced to formulate the axiom for universes anyway, because writing
  \(P = Q\) only makes sense when \(P\) and \(Q\) are elements of the same type,
  which we are taking to be a universe \(\U\) here.
\end{remark}

\begin{theorem}\label{Omega-is-set}
  Assuming function extensionality and propositional extensionality for~\(\U\),
  the type \(\Omega_{\U}\) of propositions in \(\U\) is a set.
\end{theorem}
\begin{proof}
  The type \(\Omega_{\U}\) consists of pairs \((P,i)\) with \(P : \U\) a type
  and \(i : \issubsingleton(P)\). By~\cref{being-prop-is-prop}, two such pairs
  are equal if and only if their first components are equal, so this is what we
  set out to prove.

  By Hedberg's Lemma (\cref{Hedberg-Lemma}), it suffices to construct a constant
  endomap on \(P = Q\) for every two propositions \(P\) and \(Q\) in \(\U\). But
  the map
  \[
    (P = Q) \xrightarrow{\refl \mapsto (\id,\id)} (P \to Q) \times (Q \to P)
    \xrightarrow{\text{prop-ext}_{\U}} (P = Q)
  \]
  does the job, because the type \((P \to Q) \times (Q \to P)\) is a
  subsingleton by \cref{Sigma-is-prop,Pi-is-prop}.
\end{proof}

\begin{proposition}[Propositional extensionality is a property]
  Assuming only function extensionality, the type expressing propositional
  extensionality for a universe \(\U\) is a subsingleton.
\end{proposition}
\begin{proof}
  \cref{subsingleton-criterion} tells us that we may assume propositional
  extensionality for \(\U\) to prove the lemma.
  By a repeated application of \cref{Pi-is-prop}, it suffices to
  prove that \(P = Q\) is a proposition for every two propositions \(P\) and
  \(Q\) in \(\U\). But this is given by \cref{Omega-is-set}.
\end{proof}

\begin{remark}\label{axioms-should-be-props}
  It is important that propositional extensionality is a property, because it
  gives us a guarantee that a construction using propositional extensionality
  cannot depend on a specific witness of propositional extensionality as they
  are all the same up to intensional equality.
  Put differently, in adding it as an axiom to our type theory we are adding a
  property and not additional data.
\end{remark}

\section{Embeddings, equivalences and retracts}%
\label{sec:embeddings-equivalences-retracts}

A major application of having defined the notion of (sub)singleton is being able
to define what it means for a map to be an embedding or an equivalence.

\begin{definition}[Fibre, \(\fib_f\)]
  The \emph{fibre} of a map \(f : X \to Y\) at \(y : Y\) is the type
  \(\fib_f(y) \colonequiv \Sigma_{x : X}\,f(x) = y\).%
  \index{fibre|textbf}
  \nomenclature[fibre]{$\fib_f(y)$}{fibre of \(f : X \to Y\) at \(y : Y\)}
\end{definition}

\begin{definition}[Embedding, equivalence, \(X \hookrightarrow Y\), \(X \simeq Y\)]
  \label{def:embedding-equivalence}
  A map \(f : X \to Y\) is
  \begin{enumerate}[(i)]
  \item an \emph{embedding} if all of its fibres are subsingletons, and%
    \index{embedding}
  \item an \emph{equivalence} if all of its fibres are singletons.%
    \index{equivalence}%
    \qedhere
  \end{enumerate}
  We denote the type of embeddings from \(X\) to \(Y\) by
  \(X \hookrightarrow Y\) and the type of equivalences from \(X\) to \(Y\) by
  \(X \simeq Y\).
  \nomenclature[arrowhookright]{$X \hookrightarrow Y$}{type of embeddings from \(X\) to \(Y\)}%
  \nomenclature[arrowsimeq]{$X \simeq Y$}{type of equivalences from \(X\) to \(Y\)}
\end{definition}

We understand this definition as follows: a map \(f : X \to Y\) is an
equivalence if for every \(y : Y\) there is exactly one \(x : X\) with
\(f(x) = y\). Similarly, a map \(f : X \to Y\) is an embedding if for every
\(y : Y\) there is at most one \(x : X\) with \(f(x) = y\).

Observe how the notions of embedding and equivalence and defined in terms of the
fibres of the map. After we introduce the proposition truncation
in~\cref{sec:prop-trunc}, we will see that a map is a surjection if its fibres
are all inhabited and a split surjection if its fibres are all pointed. Thus,
the fibres of a map are of fundamental interest.

\begin{theorem}\label{being-equiv-is-prop}
  Being an embedding/equivalence is a property, i.e.\ for every map
  \(f : X \to Y\) the types expressing that \(f\) is an embedding/equivalence
  are subsingletons.
\end{theorem}
\begin{proof}
  Immediate consequence of \cref{Pi-is-prop,being-prop-is-prop}.
\end{proof}

It is natural to wonder if, in the definition of a map being an equivalence, the
uniqueness conditions can be expressed as: ``if we have another \(x' : X\) such
that \(f(x') = y\), then \(x' = x\)''. This is equivalent to the above
definition when \(Y\) is a set, but in general this fails to account for the
structure that the identity types of \(Y\) might carry in the sense that the
types expressing this may fail to be propositions.
The upshot of \cref{being-equiv-is-prop} is that the equivalence between \(X\)
and \(Y\) form a subtype of all functions from \(X\) to \(Y\), as we explain in
\cref{subtype-of-equivalences} below.

The following result is rather useful for proving that a map is an equivalence.
\index{invertibility|(}

\begin{proposition}\label{equiv-iff-invertible}
  A map \(f : X \to Y\) is an equivalence if and only if it is invertible.
\end{proposition}
\begin{proof}
  See \cite[Section~3.10]{Escardo2019} for a proof or \cite[Chapter 4]{HoTTBook}
  where this and related results and issues are discussed at length.
\end{proof}

One may ask why we did not define \(f\) to be an equivalence when it's
invertible. After all, \cref{equiv-iff-invertible} tells us that the two are
logically equivalent. However, they are \emph{not} equivalent in the sense of
\cref{def:embedding-equivalence}, as being invertible may fail to be
property~(recall~\cref{property-and-data}), cf.~\cite[Theorem~4.1.3]{HoTTBook},
while being an equivalence is.
Without going into the details of the proof too much, the circle \(\Circ\)
(recall \cref{circle}) provides an illustrative example where invertibility may
fail to be a property, as the following example makes clear.

\begin{example}[Invertibility is not necessarily a property]%
  \label{circle-invertibility}%
  \index{circle}%
  The type expressing that the identity on \(\Circ\) is invertible is
  \[
    \Sigma_{f : \Circ \to \Circ} \pa*{\Pi_{x : \Circ}{(f \circ \id)(x) =
        \id(x)}} \times \pa*{\Pi_{x : \Circ}{(\id \circ f)(x) = \id(x)}}.
  \]
  But by function extensionality, see \cref{fun-ext-alt}, this is equivalent to
  \[
    \Sigma_{f : \Circ \to \Circ}\pa*{f = \id} \times \pa*{f = \id},
  \]
  which, by \cref{singleton-type-is-singleton}, is equivalent to
  \(\id_{\Circ} = \id_{\Circ}\).
  By function extensionality, this is equivalent to \(\Pi_{x : \Circ}\,x=x\)
  which can be shown to be equivalent to
  \((\operatorname{base} = \operatorname{base})\) and hence to the type of
  integers \(\mathbb Z\), see \cref{circle}.
\end{example}%
\index{invertibility|)}

As mentioned above, the fact that being an equivalence is a property ensures
that the type of equivalences is a subtype (in the sense of~\cref{def:subtype}
below).
Moreover, the notion of an equivalence is crucial for formulating the univalence
axiom~(\cref{def:univalence} below), because the formulation using invertible
maps is provably false~\cite[Exercise~4.6(c)]{HoTTBook}.

\begin{proposition}\label{embedding-criterion}
  A map \(f : X \to Y\) is an embedding if and only if for every \(x,y : X\) we
  have an equivalence \((x = y) \simeq (f(x) = f(y))\).
\end{proposition}
\begin{proof}
  See \cite[Section~3.26]{Escardo2019}.
\end{proof}

\begin{definition}[Subtype]\label{def:subtype}%
  \index{subtype}%
  A \emph{subtype} of a type \(X\) is a type \(A\) together with an embedding
  \(A \hookrightarrow X\).
\end{definition}

\begin{lemma}\label{prop-subtype}
  If \(Y\) is a proposition-valued type family over \(X\), then
  \({\Sigma_{x : X}Y(x)}\) is a subtype of \(X\), as witnessed by the first
  projection.
  In particular, two elements \((x,y)\) and \((x',y')\) of the \(\Sigma\)-type
  are equal if and only if \(x = x'\).
\end{lemma}
\begin{proof}
  We have to show that all fibres of \(\fst : \pa*{\Sigma_{x : X}Y(x)} \to X\)
  are subsingletons. For arbitrary \(x : X\) we have
  \[
    \fib_{\fst}(x) \equiv \pa*{\Sigma_{(x',y) : \Sigma Y}\pa*{x' = x}}
    \simeq \pa*{\Sigma_{x' : X}\pa*{Y(x') \times (x' = x)}}
    \simeq Y(x)
  \]
  by reshuffling the \(\Sigma\)-types and the contractibility of the type
  \(\Sigma_{x' : X}(x'=x)\) at \((x,\refl)\). But \(Y(x)\) is a subsingleton by
  assumption, proving that the first projection is an embedding.
  The second claim follows from \cref{embedding-criterion}.
\end{proof}

\begin{remark}\label{subtype-omit}
  In the situation of \cref{prop-subtype} we often omit the second component of
  elements of \(\Sigma_{x : X}Y(x)\) which is justified because any two elements
  of the second component are equal anyway.
\end{remark}

\begin{example}\label{subtype-of-equivalences}
  For any types \(X\) and \(Y\), the types \(X \hookrightarrow Y\) and
  \(X \simeq Y\) are subtypes of \(X \to Y\) by
  \cref{being-equiv-is-prop,prop-subtype}. Hence, two embeddings or equivalences
  are equal precisely when they are equal as ordinary maps.
  In particular, following \cref{subtype-omit}, we simply write
  \(f : X \simeq Y\) for what is formally \((f,i) : X \simeq Y\) with \(i\)
  witnessing that \(f\) is an equivalence.
\end{example}

\begin{definition}[Section, retraction and retract]%
  \index{section|see {retract}}\index{retraction|see {retract}}\index{retract|textbf}%
  A \emph{section} is a map \(s : X \to Y\) together with a left inverse
  \(r : Y \to X\), i.e.\ the maps satisfy \(r(s(x)) = x\) for every \(x : X\).
  We call \(r\) the \emph{retraction} and say that \(X\) is a \emph{retract} of
  \(Y\).
\end{definition}

\begin{lemma}\label{sections-to-sets-are-embeddings}
  Sections to sets are embeddings.
\end{lemma}
\begin{proof}
  Let \(s : X \to Y\) be a section to a set. Since \(Y\) is a set, the type
  \(s(x) = y\) is a proposition for every \(x : X\) and \(y : Y\). Hence, by
  \cref{prop-subtype}, to prove that every fibre of \(s\) is a subsingleton, it
  suffices to prove that \(x = x'\) whenever \(s(x) = s(x')\). But if
  \(s(x) = s(x')\), then we can apply the retraction on both sides to get
  \(x = x'\).
\end{proof}

\begin{remark}
  The restriction to sets in \cref{sections-to-sets-are-embeddings} is a
  necessary one, because \cite[Remark~3.11(2)]{Shulman2016} tells us that if
  every section is an embedding then every type is a set.
\end{remark}

\section{Function extensionality revisited}\label{sec:fun-ext-revisited}

Armed with the notion of equivalence, we are now ready to give the official
definition of function extensionality that improves on the naive version
presented in \cref{naive-fun-ext}.

\begin{definition}[\Axiom: Function extensionality]\label{def:fun-ext}%
  \index{extensionality!function}%
  \emph{Function extensionality} asserts that for every two (dependent)
  functions \(f,g : \Pi_{x : X} Y(x)\), the canonical map from \(f = g\) to
  \(\Pi_{x : X}f(x) = g(x)\) is an equivalence.
\end{definition}

Naive function extensionality and function extensionality are logically
equivalent, but the advantage of the above official formulation of function
extensionality is that its type is a subsingleton, i.e.\ function extensionality
is a property. So the distinction is similar to that between
invertible maps and equivalences.

\begin{proposition}
  Assuming function extensionality, the type expressing function extensionality
  is a subsingleton.
\end{proposition}
\begin{proof}
  This follows from \cref{Pi-is-prop,being-equiv-is-prop}.
\end{proof}

\section{Propositional truncation, images and surjections}%
\label{sec:prop-trunc}%
\index{propositional truncation|(textbf}%
\index{truncation|see {propositional truncation}}

We turn to introducing propositional truncations, which we will motivate through
the problem of defining the image of a map. Intuitively, the image of
\(f : X \to Y\) should be the collection of elements \(y : Y\) such that there
exists \(x : X\) with \(f(x) = y\).
A naive attempt at defining this might lead us to the type
\(\Sigma_{y : Y}\Sigma_{x : X}\,f(x) = y\).

Notice that \(\Sigma\) plays a double role here: the first \(\Sigma\)
\emph{collects} elements \(y : Y\), while the second \(\Sigma\) supposedly
expresses the existence of \(x : X\) with \(f(x) = y\). This hints at a problem
and indeed there is one, because the type
\(\Sigma_{y : Y}\Sigma_{x : X}\,f(x)=y\) is obviously equivalent to
\(\Sigma_{x : X}\Sigma_{y : Y}\,f(x) = y\), which in turn is equivalent to just
\(X\), because for every \(x : X\), the type \(\Sigma_{y : Y}\,f(x) = y\) is a
singleton by \cref{singleton-type-is-singleton}.

The above illustrates our lack of expressing ``there exists'' or ``we have an
\emph{unspecified}''. Instead of collecting \(x : X\) for which \(f(x) = y\), we
only wish to record the knowledge that there is some \(x : X\) with
\(f(x) = y\). We will do so by means of propositional truncations and
subsequently use them to define images and surjections.

\subsection{Propositional truncation}\label{sec:propositional-truncation}

The propositional truncation is a higher inductive type; the only one we use in
this thesis. We postulate a constructor \(\squash*{-}\) that takes a type and
returns a proposition: its \emph{propositional truncation}.
\nomenclature[pipedouble]{$\squash{-}$}{propositional truncation}%
We require, unless explicitly stated otherwise, that our universes are closed
under propositional truncations, i.e.\ if \(X : \U\), then \(\squash{X} : \U\).
Moreover, as part of the propositional truncation, we postulate that we have a
map \(\tosquash*{-} : X \to \squash*{X}\) for every type \(X\).
\nomenclature[pipe]{$\tosquash{-}$}{unit of the propositional truncation}

We think of an element of \(\squash{X}\) as an \emph{unspecified} element of
\(X\). The map \(\tosquash{-}\) then says that every specified elements gives
rise to an unspecified one.

In the vocabulary of Martin-L\"of Type Theory, the above gives a formation and
an introduction rule for the propositional truncation, but we have not specified
an elimination and computation rule yet. The elimination principle expresses
that the propositional truncation is a reflector in the categorical sense: it is
a left adjoint to the inclusion of proposition into all types. Spelled out it
says that every map \(f : X \to P\) to a proposition factors through
\(\tosquash{-} : X \to \squash{X}\), i.e.\ we have a map
\(\bar{f} : \squash{X} \to P\) such that the diagram
\[
  \begin{tikzcd}
    X \ar[dr,"\tosquash{-}"'] \ar[rr,"f"] & & P \\
    & \squash{X} \ar[ur,dashed,"\bar{f}"']
  \end{tikzcd}
\]
commutes.

What is paramount here is that this universal property holds for all
propositions in arbitrary universes and not just for those in the same universe
as \(X\). We will return to this phenomenon
in~\cref{sec:set-quotients-from-propositional-truncations}.

Some sources, e.g.~\cite{HoTTBook}, also demand that the diagram above commutes
\emph{definitionally}: for every \(x : X\), we have
\(f(x) \equiv \bar{f}(\tosquash{x})\). Having definitional equalities has some
interesting consequences, such as being able to prove function
extensionality~\cite[Section~8]{KrausEtAl2017}.
We do not require definitional equalities, but notice that
we do have \(f(x) = \bar{f}(\tosquash{x})\) (up to an identification) for every
\(x : X\), as \(P\) is a subsingleton.
In particular it follows using function extensionality that \(\bar{f}\) is the
unique factorisation.

From the universal property, we can prove that \(\squash*{-}\) is a functor,
because any map \({f : X \to Y}\) between types gives rise to a necessarily
unique map \({\squash{f} : \squash{X} \to \squash{Y}}\) such that
\(\squash{f}(\tosquash{x}) = \tosquash{f(x)}\) for every \(x : X\).
Moreover, if a propositional truncation exists, then it is unique up to unique
equivalence.
\index{propositional truncation|)textbf}

\begin{definition}[Unspecified and specified existence, \(\exists_{x : X}Y(x)\)]
  \hfill
  \begin{enumerate}[(i)]
  \item We suggestively write \(\exists_{x : X}Y(x)\) for the propositional
    truncation of \(\Sigma_{x : X}Y(x)\).
    \nomenclature[exists]{$\exists_{x : X}Y(x)$}{existential quantifier,
      propositional truncation of the type \(\Sigma_{x : X}Y(x)\)}
  \item We say that ``there exists (some) \(x : X\) with \(Y(x)\)'' or that we
    ``have an unspecified \(x : X\) with \(Y(x)\)'' to mean that we have an
    element of \(\exists_{x : X}Y(x)\).%
    \index{existence!unspecified|textbf}
  \item By contrast, we say that we ``have \(x : X\) with \(Y(x)\)'' or
    sometimes for emphasis, that we ``have a \emph{specified} \(x : X\) with
    \(Y(x)\)'' to mean that we have an element of \(\Sigma_{x : X}Y(x)\).%
    \index{existence!specified|textbf}%
    \qedhere
  \end{enumerate}
\end{definition}

The following lemma is sometimes known as ``the type-theoretic axiom of
choice'', which is a misnomer, because, as emphasised in the above distinction
between specified and unspecified existence, there is no choice involved, since
elements of \(\Sigma\)-types are specified witnesses.
A correct formulation of the axiom of choice will be presented
in~\cref{def:axiom-of-choice}.

\begin{lemma}[Distributivity of \(\Pi\) over \(\Sigma\)]%
  \label{Pi-Sigma-distr}
  For every type family \(Y\) over a type \(X\) and every family
  \(P : \Pi_{x : X}\pa*{Y(x) \to \U}\) we have an equivalence
  \[
    \pa*{\Pi_{x : X}\Sigma_{y : Y(x)}P(x,y)}
    \simeq \pa*{\Sigma_{f : \Pi Y} \Pi_{x : X}P(x,f(x))}.
  \]
\end{lemma}
\begin{proof}
  In the left-to-right direction, assume we have
  \(\varphi : \Pi_{x : X}\Sigma_{y : Y(x)}P(x,y)\). Then, we define
  \(f : \Pi Y\) by \(f \colonequiv {{\fst} \circ {\varphi}}\) and we see that
  \(\varphi\) yields an element of \(P(x,f(x))\) for every \(x : X\).
  Conversely, if we have \(f : \Pi Y\) and \(\rho : \Pi_{x : X}P(x,f(x))\), then
  we define \(\varphi : \Pi_{x : X}\Sigma_{y : Y(x)}P(x,y)\) as
  \(\varphi(x) \colonequiv \pa*{f(x),\rho(x)}\).
  Finally, a direct computation shows that the maps above are indeed inverses.
\end{proof}

\begin{definition}[Inhabited]%
  \label{def:inhabited}
  We say that a type is \emph{inhabited} if we have an element of its
  propositional truncation.%
  \index{inhabited|textbf}
\end{definition}

Thus, a type is inhabited if we have an unspecified element of it. We do not use
the word \emph{nonempty} for this, because in our constructive setting this will
mean something weaker, as explained in~\cref{sec:logic}.

\subsection{Images and surjections}
As discussed, we now use the propositional truncation to define the image of
map.

\begin{definition}[Image, \(\image(f)\), corestriction]
  For a map \(f : X \to Y\) we define
  \begin{enumerate}[(i)]
  \item its \emph{image} as
    \(\image(f) \colonequiv \Sigma_{y : Y}\exists_{x : X}\,f(x) = y\), and%
    \index{image}
    \nomenclature[im]{$\image(f)$}{image of a map \(f\)}
  \item its \emph{corestriction} as the map \(f : X \to \image(f)\) given by
    \(x \mapsto \pa*{f(x),\tosquash{\pa{x,\refl}}}\).%
    \index{corestriction}%
    \qedhere
  \end{enumerate}
\end{definition}

Another take on the problem described in the introduction to this section is
that the type \(\Sigma_{y : Y}\Sigma_{x : X}\,f(x) = y\), being equivalent to
\(X\), is not a subtype of \(Y\). Reassuringly, with the above official
definition of image we do get a subtype.

\begin{definition}[Surjection, \(X \surj Y\)]
  A map \(f : X \to Y\) is a \emph{surjection} if all fibres are inhabited. In
  others words, for every \(y : Y\), there exists some \(x : X\) with
  \(f(x) = y\). The type of surjections from \(X\) to \(Y\) is denoted by
  \(X \surj Y\).%
  \index{surjection}%
  \nomenclature[arrowtwohead]{$X \surj Y$}{type of surjections from \(X\) to \(Y\)}
\end{definition}

\begin{lemma}\label{corestrictions-are-surjections}
  All corestrictions are surjections.
\end{lemma}
\begin{proof}
  By definition of the corestriction.
\end{proof}

\begin{lemma}[Surjection induction]\label{surjection-induction}%
  \index{surjection!induction}%
  If \(f : X \to Y\) is a surjection, then the following induction principle
  holds: for every subsingleton-valued \(P : Y \to \W\), with \(\W\) an
  \emph{arbitrary} universe, if \(P(f(x))\) holds for every \(x : X\), then
  \(P(y)\) holds for every \(y : Y\).

  In the other direction, for any map \(f : X \to Y\), if the above induction
  principle holds for the specific family
  \({P(y) \colonequiv \exists_{x : X}(f(x) = y)}\), then \(f\) is a surjection.
\end{lemma}
\begin{proof}
  Suppose that \(f : X \to Y\) is a surjection, let \(P : Y \to \W\) be
  subsingleton-valued and assume that \(P(f(x))\) holds for every \(x : X\).
  Now let \(y : Y\) be arbitrary. We are to prove that \(P(y)\) holds. Since
  \(f\) is a surjection, we have \(\exists_{x : X}(f(x) = y)\). But \(P(y)\) is
  a subsingleton, so, by the universal property of the propositional truncation,
  we may assume that we have a specific \(x : X\) with \(f(x) = y\). But then
  \(P(y)\) must hold, because \(P(f(x))\) does by assumption.

  For the other direction, notice that if
  \(P(y) \colonequiv \exists_{x : X}(f(x) = y)\), then \(P(f(x))\) clearly holds
  for every \(x : X\). So by assuming that the induction principle applies, we
  get that \(P(y)\) holds for every \(y : Y\), which says exactly that \(f\) is
  a surjection.
\end{proof}

\begin{lemma}\label{factor-map-through-corestriction}
  Every map \(f : X \to Y\) factors as a surjection followed by an embedding:
  \(X \surj \image(f) \hookrightarrow Y\), where the first map is the
  corestriction of \(f\) and the second map is the first projection.
\end{lemma}
\begin{proof}
  That \(f\) is equal to the composite \(X \surj \image(f) \hookrightarrow Y\)
  is immediate. Moreover, the corestriction \(X \to \image(f)\) is a surjection
  by \cref{corestrictions-are-surjections}, and the first projection
  \(\image(f) \to Y\) is an embedding because of \cref{prop-subtype}.
\end{proof}

\subsection{Mapping from propositional truncations into sets}%
\label{sec:theorem-by-Kraus-et-al}

We recall a result due to Kraus, Escard\'o, Coquand and
Altenkirch~\cite{KrausEtAl2017} which has several applications throughout this
thesis.

\begin{theorem}[{\cite[Theorem~5.4]{KrausEtAl2017}}]
  \label{constant-map-to-set-factors-through-truncation}%
  \index{constant}%
  Every constant map to a set factors through the truncation of its domain.
\end{theorem}
\begin{proof}
  Suppose that \(f : X \to Y\) is a constant map to a set. By
  \cref{factor-map-through-corestriction} and the universal property of the
  propositional truncation it suffices to prove that the image of \(f\) is a
  proposition, as this would yield the dashed map making the diagram
  \[
    \begin{tikzcd}
      & \squash{X} \ar[d,dashed] \\
      X
      \ar[ur,"\tosquash{-}"]
      \ar[r,twoheadrightarrow]
      \ar[rr,bend right,"f"'] &
      \image(f) \ar[r,hookrightarrow] & Y
    \end{tikzcd}
  \]
  commute.
  So suppose that we have \(y,y' : Y\) such that there exists some \(x : X\)
  with \(f(x) = y\) and some \(x' : X\) with \(f(x') = y'\). By
  \cref{prop-subtype}, we only have to prove that \(y = y'\). But \(Y\) is
  assumed to be a set, so this a proposition. Hence, we can assume that we have
  specified \(x : X\) and \(x' : X\) with \(f(x) = y\) and \(f(x') = y'\). But
  then \(y = f(x) = f(x') = y'\), as \(f\) is assumed to be constant.
\end{proof}

The theorem can be explained at a high level: since \(f\) is constant it does
not matter what element we have of \(X\) (at least when \(Y\) is a set and has
no higher dimensional structure). Thus, at least intuitively, as soon as we know
that there exists some \(x : X\) we should obtain a corresponding element of
\(Y\), because the choice of \(x\) is irrelevant.

\section{Logic, (semi)decidability and constructivity}\label{sec:logic}%
\index{logic|(}%
\index{Curry--Howard|(}

In univalent foundations and motivated by the discussion at the start of
\cref{sec:prop-trunc}, we refine the Curry--Howard paradigm of
propositions-as-types to propositions-as-subsingletons.
That is, logical statements will be interpreted as types that have at most one
element.
For example, we
interpret the existential quantifier as the propositional truncation of
\(\Sigma\).
Thus, the logic in traditional set-level mathematics is encoded according to the
following table.

\begin{table}[h]
  \centering
  \begin{tabular}{cc}\toprule
    Logical proposition & Subsingleton \\\midrule
    True  & \(\One\) \\
    False & \(\Zero\) \\
    \(P\) and \(Q\) & \(P \times Q\) \\
    \(P\) implies \(Q\) & \(P \to Q\) \\
    \(P\) or \(Q\) & \(\squash{P + Q}\) \\
    For every \(x : X\) we have \(P(x)\) & \(\Pi_{x : X}P(x)\) \\
    There exists \(x : X\) such that \(P(x)\) & \(\squash*{\Sigma_{x : X}P(x)}\) \\\bottomrule
  \end{tabular}
  \caption{Curry--Howard in univalent foundations.}
  \label{CH}
\end{table}

Note that \cref{Sigma-is-prop,Pi-is-prop} ensure that \(P \times Q\),
\(P \to Q\) and \(\Pi_{x : X}P(x)\) are propositions if \(P\), \(Q\) and each
\(P(x)\) are, while we use truncations to ensure that we get subsingletons for
\(\lor\) and the existential quantifier.

\begin{definition}[Logical or, \(\lor\)]
  We write \(X \lor Y \colonequiv \squash{X + Y}\) for any two types \(X\) and
  \(Y\).
  \nomenclature[vee]{$\lor$}{logical or}
  \nomenclature[not]{$\lnot$}{negation}
\end{definition}

\begin{definition}[Negation, \(\lnot\)]
  The negation of a type \(X\) is denoted by \(\lnot X\) and defined as
  \(X \to \Zero\).
\end{definition}

\begin{remark}
  Strictly speaking we should specify a universe for \(\Zero\), but the choice
  is immaterial because \(\Zero_\U\) and \(\Zero_\V\) are easily seen to be
  equivalent for any two universes \(\U\) and \(\V\). For the sake of
  definiteness, we take \(\Zero_{\U_0}\) in the definition of negation.
\end{remark}

Note that the negation of any type is a subsingleton by \cref{Pi-is-prop} and
the fact that \(\Zero\) is a subsingleton.

Finally, it is important to be mindful of the fact that our logic will be
constructive, as explained in \cref{sec:constructivity}.%
\index{logic|)}
\index{Curry--Howard|)}

\subsection{Subsets and powersets}\label{sec:subsets-and-powersets}

\begin{definition}[\powersetT{T}, \valuedT{T} subsets, \(\powerset_{\T}(X)\), \(x \in A\)]%
  \label{def:powerset}
  For any type \(X\) and a type universe \(\T\), the \emph{\powersetT{T}}
  \(\powerset_{\T}(X)\) of \(X\) is defined as \(X \to \Omega_{\T}\). %
  \index{powerset|textbf}\index{subset|textbf}%
  \index{type!of subsets|see {powerset}}%
  \nomenclature[PT]{$\powerset_{\T}(X)$}{type of \(\T\)-valued subsets of \(X\)}%
  \nomenclature[T]{$\T$}{type universe}%
  We refer to its elements as \emph{\(\T\)-valued subsets} of \(X\).
  Given a \valuedT{T} subset \(A\) of \(X\) and \(x : X\), we write \(x \in A\)
  for \(A(x)\).%
  \index{subset!membership}
  \nomenclature[in]{$x \in A$}{membership of a subset}
\end{definition}

\begin{lemma}\label{powerset-is-set}
  The \powersetT{T} of any type is a set.
\end{lemma}
\begin{proof}
  By \cref{Pi-is-prop,Omega-is-set}.
\end{proof}

\begin{definition}%
  [Union, intersection, empty subset, singleton, \(A \cup B\), \(A \cap B\), \(\emptyset\), \(\set{x}\)]%
  \index{subset!union}%
  \index{subset!intersection}%
  \index{subset!empty}%
  \index{subset!singleton}%
  For \valuedT{T} subsets \(A\) and \(B\) of a type \(X\), we respectively
  write \(A \cup B\) and \(A \cap B\) for union and intersection that are
  formally defined by the maps
  \(\lambdadot{x}\pa{x \in A} \lor \pa{x \in B}\) and
  \(\lambdadot{x}\pa{x \in A} \times \pa{x \in B}\).
  \nomenclature[cup]{$A \cup B$}{union of two subsets}%
  \nomenclature[cap]{$A \cap B$}{intersection of two subsets}%

  The empty subset of \(X\) is denoted by \(\emptyset\) (leaving \(X\) and
  \(\T\) implicit) and formally defined as \(\lambdadot{x}\Zero_{\T}\).
  \nomenclature[emptyset]{$\emptyset$}{empty subset}%

  If \(X\) is a set in \(\T\), then we have singleton subsets \(\set{x}\) for
  every \(x : X\), formally defined by \(\lambdadot{y}(x=y)\).
  \nomenclature[singleton]{$\set{x}$}{singleton subset with \(x\) as its only member}%
  Note that the
  requirement that \(X\) is a set in \(\T\) is used to ensure that \(x=y\) is
  indeed an element of \(\Omega_{\T}\).
\end{definition}

Of course, besides binary unions and intersections, we could use \(\exists\) and
\(\Pi\) to construct unions and intersections of families of subsets, but this
matter is tightly connected to predicativity issues, so we will revisit it in
some detail later.

\begin{definition}[Subset inclusion, \(\subseteq\)]
  Given a \(\T\)-valued subset \(A\) and a \(\W\)-valued subset \(B\) of a type
  \(X : \U\), we write \(A \subseteq B\) for the notion of \emph{subset
    inclusion} that is formally defined as having an element of
  \(\Pi_{x : X}\pa*{x \in A \to x \in B}\).%
  \index{subset!inclusion}%
  \nomenclature[subseteq]{$A \subseteq B$}{inclusion of subsets}
\end{definition}

\subsection{Decidability}

In discussing constructive logic the notion of decidability is fundamental.

\begin{definition}[(Weak) decidability of a type/equality]
  A type \(X\)
  \begin{enumerate}[(i)]
  \item\label{item:decidable} is \emph{decidable} if we have an element of
    \(X + \lnot X\),%
    \index{decidability|textbf}
  \item\label{item:weakly-decidable} is \emph{weakly decidable} if we have an
    element of \(\lnot\lnot X + \lnot X\), and%
    \index{decidability!weak}
  \item\label{item:decidable-equality} has \emph{decidable equality} if
    \(x = y\) is decidable for every \(x,y : X\).%
    \index{decidability!of equality|textbf}%
    \qedhere
  \end{enumerate}
\end{definition}

Note that~\ref{item:decidable} is data, while~\ref{item:weakly-decidable}
and~\ref{item:decidable-equality} are property. For~\ref{item:weakly-decidable},
this holds because negated types are subsingletons and because \(\lnot\lnot X\)
and \(\lnot X\) are mutually exclusive. For~\ref{item:decidable-equality}, this
is a consequence of Hedberg's Theorem, which is~\cref{Hedberg-Theorem} below.

\begin{example}
  The types \(\Zero\), \(\One\) and \(\Nat\) all have decidable equality.
\end{example}

\begin{lemma}\label{decidable-iff}
  If we have maps back and forth between two types \(X\) and \(Y\) and one of
  the types is decidable, then so is the other.
\end{lemma}
\begin{proof}
  Suppose that \(f \colon X \to Y\) and \(g \colon Y \to X\) and that \(X\) is
  decidable. Then we have \(x : X\) or \(\lnot X\). In the first case we get
  \(f(x) : Y\) and in the second case, we get \(\lnot Y\) by contraposition and
  \(g : Y \to X\).
\end{proof}

\begin{theorem}[Hedberg's Theorem~\cite{Hedberg1998}]\label{Hedberg-Theorem}%
  \index{Hedberg's!Theorem}%
  If a type has decidable equality, then it is a set.
\end{theorem}
\begin{proof}
  Suppose that \(X\) is a type with decidable equality. By Hedberg's~Lemma
  (\cref{Hedberg-Lemma}), it suffices to construct a constant endofunction on
  \(x = y\) for every \({x,y} : X\). By assumption, we either have \(p : x = y\)
  or \(x \neq y\). In the latter case, \(x = y\) is equivalent to \(\Zero\) and
  so it certainly has a constant endomap. And if we have \(p : x = y\), then the
  function \((x = y) \to (x = y)\) mapping everything to \(p\) is constant.
\end{proof}

When studying the Scott model of PCF 
in~\cref{sec:Scott-model-of-PCF}, the notion of \emph{semi}decidability also makes
an appearance. Intuitively, a proposition is semidecidable if it can be affirmed
through some finite procedure. Notice how we do not impose such a restriction on
it being refuted, so this notion is characteristically asymmetric.

\begin{definition}[Semidecidability]
  A proposition \(P\) is \emph{semidecidable} if there exists some binary
  sequence \(\alpha : \Nat \to \Two\) such that \(P\) holds if and only if there
  exists some \(n : \Nat\) for which \(\alpha(n) = 1\).%
  \index{semidecidability|textbf}
\end{definition}

Indeed, if a semidecidable proposition \(P\) holds, then we will eventually find
\(n : \Nat\) for which \(\alpha(n)\). But inspecting the sequence \(\alpha\) for
any finite number of values will never allow us to conclude that the negation of
\(P\) holds as this would require knowing that the sequence \emph{never} attains
the value \(1\).
Also notice that every decidable proposition is semidecidable.

\begin{lemma}\label{semidecidable-criterion}
  A proposition \(P\) is semidecidable if and only if there exists a natural
  number \(k : \Nat\) and a family \(Q : \Nat^k \to \U_0\) such that
  \(Q\pa*{\vec n}\) is decidable for all inputs \(\vec n : \Nat^k\) and
  \(P\) holds exactly when \(\exists_{\vec n : \Nat^k}Q(\vec n)\) does.
\end{lemma}
\begin{proof}
  Note that the type of decidable propositions is equivalent to \(\Two\) and
  that for every natural number \(k\) we have a bijection
  \(\Nat^k \simeq \Nat\), so that we can always turn a family with \(k\)
  parameters into a corresponding one with a single parameter.
\end{proof}

\subsection{Constructivity}\label{sec:constructivity}%
\index{constructivity|(}

Our foundational setup will be constructive in the sense that we do not assume
any additional logical axioms beyond propositional extensionality and function
extensionality. In particular, we do not assume (weak) excluded
middle~(\cref{def:(w)em}), the limited principle of omniscience (LPO;
\cref{def:LPO}) or the axiom of (countable) choice~(\cref{def:axiom-of-choice}),
as these are constructively unacceptable~\cite[p.~9]{Bishop1967}, and even
provably false in some varieties of constructive
mathematics~\cite[pp.~3--4]{BridgesRichman1987}. In MLTT, and univalent
foundations, they are simply independent: they cannot be proved, but neither can
their negations.

This does \emph{not} mean that these logical principles will have no use for
us. In fact, they will feature as \emph{constructive taboos}. That is, sometimes
we wish to argue that something is inherently nonconstructive and we can do so
as follows: if we can show that \(X\) implies excluded middle, then this tells
us that \(X\) is constructively unacceptable, because excluded middle is.

\begin{definition}[Nonemptiness]\label{Nonemptiness}%
  \index{nonemptiness|textbf}
  A type \(X\) is \emph{nonempty} if \(\lnot\lnot X\) has an element.
\end{definition}

\begin{lemma}\index{inhabited}
  Every inhabited type is nonempty.
\end{lemma}
\begin{proof}
  Let \(X\) be an inhabited type, i.e.\ we have an element \(t\) of
  \(\squash{X}\). Since being nonempty is a proposition by~\cref{Pi-is-prop}, we
  can eliminate \(t\) and assume to have \(x : X\). But then
  \(\lambdadot{\pa*{f : \lnot X}}{f(x)}\) is an element of \(\lnot\lnot X\).
\end{proof}

\begin{definition}[\(\lnot\lnot\)-stability, type of \(\lnot\lnot\)-stable
  propositions, \(\Omeganotnot{\U}\)] \hfill
  \begin{enumerate}[(i)]
  \item A type \(X\) is said to be \emph{\(\lnot\lnot\)-stable} if we have an
    element of \(\lnot\lnot X \to X\).%
    \index{not-not-stability@\(\lnot\lnot\)-stability|textbf}
  \item We denote the \emph{type of \(\lnot\lnot\)-stable propositions} in a
    universe \(\U\) by \(\Omeganotnot{\U}\).%
    \index{type!of \(\lnot\lnot\)-stable propositions|textbf}
    \qedhere
    \nomenclature[Omeganotnot]{$\Omeganotnot{\U}$}{type of \(\lnot\lnot\)-stable
      propositions in \(\U\)}
  \end{enumerate}
\end{definition}

\begin{remark}
  Using the terminology of \cref{def:inhabited} and the observation that a
  proposition is equivalent to its propositional truncation, we see that a
  proposition is \(\lnot\lnot\)-stable precisely when it is inhabited as soon as
  it is nonempty.
\end{remark}

\begin{definition}[(Weak) excluded middle]\label{def:(w)em}
  For a type universe \(\U\), \emph{(weak) excluded middle} in \(\U\) asserts
  that every proposition in \(\U\) is (weakly) decidable.%
  \index{excluded middle|textbf}\index{excluded middle!weak|textbf}
\end{definition}

\begin{remark}\label{em-only-for-props}
  The restriction to propositions in the formulation of (weak) excluded middle
  can be explained in two ways. Firstly, given our interpretation of
  (logical)-propositions-as-subsingletons, it seems appropriate to restrict
  (weak) excluded middle to the logical fragment of our framework.
  In fact, the statement \(X + \lnot X\) for all types \(X\) is \emph{global
    choice}: it says that we can choose a specified element of every nonempty
  type, and is incompatible with
  univalence~\cite[Section~3.35.6]{Escardo2019}.\index{choice!global}
  Secondly, the unrestricted formulation is provably false in the presence of
  the univalence axiom, while the restricted formulation is consistent with
  univalent type theory, as shown by Voevodsky's simplicial sets
  model~\cite{KapulkinLumsdaine2021}.
\end{remark}

\begin{lemma}\label{em-and-wem-equivalent-formulations}%
  \index{not-not-stability@\(\lnot\lnot\)-stability}%
  \index{type!of \(\lnot\lnot\)-stable propositions}%
  \index{nonemptiness}%
  \index{inhabited}%
  The following logical statements are equivalent for a universe \(\U\):
  \begin{enumerate}[(i)]
  \item\label{item-em} excluded middle in \(\U\);
  \item\label{item-em'} every proposition in \(\U\) is equal to either
    \(\Zero_{\U}\) or \(\One_{\U}\);
  \item\label{item-Omega-dec-eq} the type \(\Omega_{\U}\) has decidable
    equality;
  \item\label{item-em-Bool} the map \(\Two_{\V} \to \Omega_{\U}\) given by
    \(0 \mapsto \Zero_{\U}\) and \(1 \mapsto \One_{\U}\) is an equivalence for
    any type universe \(\V\);
  \item\label{item-dns} all elements of \(\Omega_{\U}\) are
    \(\lnot\lnot\)-stable \emph{(double negation elimination)};
  \item\label{item-inhabited-if-nonempty} every nonempty type is inhabited.
  \end{enumerate}
  Similarly, the following logical statements are equivalent for every universe \(\U\):
  \begin{enumerate}[(i')]
  \item\label{item-wem} weak excluded middle in \(\U\);
  \item\label{item-wem'} every \(\lnot\lnot\)-stable proposition in \(\U\) is
    equal to either \(\Zero_{\U}\) or \(\One_{\U}\);
  \item\label{item-Omeganotnot-dec-eq} the type \(\Omeganotnot{\U}\) has
    decidable equality;
  \item\label{item-wem-Bool} the map \(\Two_{\V} \to \Omeganotnot{\U}\) given
    by \(0 \mapsto \Zero_{\U}\) and \(1 \mapsto \One_{\U}\) is an equivalence
    for any type universe \(\V\).
  \end{enumerate}
\end{lemma}
\begin{proof}
  For the equivalence of \ref{item-em}--\ref{item-Omega-dec-eq}, observe that,
  by propositional extensionality, a proposition holds if and only if it is
  equal to the unit type and does not hold if and only if it is equal to the
  empty type.
  It is also clear that \ref{item-em-Bool} implies \ref{item-em'}. For the
  converse, assume \ref{item-em'} and note that it implies that
  \[
    \Omega_{\U} \simeq
    \Sigma_{P : \U}\pa*{\pa*{P = \One_{\U}} + \pa{P = \Zero_{\U}}}
    \simeq \pa*{\Sigma_{P : \U}\,{P = \One_{\U}}} + %
           \pa*{\Sigma_{P : \U}\,{P = \Zero_{\U}}}
    \simeq \Two_{\V},
  \]
  where the final equivalence holds by \cref{singleton-type-is-singleton} and
  the fact that every contractible type is equivalent to \(\One_{\V}\).

  The equivalence of \ref{item-em} and \ref{item-dns} is well-known
  in constructive logic. If we assume excluded middle in \(\U\) and \(P\) is an
  arbitrary proposition in \(\U\), then we have \(P\) or \(\lnot P\). In the
  first case, obviously \(\lnot\lnot P \to P\) and in the second case this also
  holds, because the antecedent of the implication contradicts \(\lnot P\).
  Conversely, if double negation elimination for \(\U\) holds and \(P\) is an
  arbitrary proposition in \(\U\), then in particular the proposition
  \(P + \lnot P\) is \(\lnot\lnot\)-stable. But \(\lnot\lnot (P + \lnot P)\) is
  a tautology: for if we assume \(\lnot (P + \lnot P)\), then assuming either
  \(P\) or \(\lnot P\) would yield a contradiction, hence we have
  \(\lnot P \times \lnot\lnot P\), which is false. Thus, by our double negation
  elimination assumption, we get decidability of \(P\), completing the proof
  that items \ref{item-em}--\ref{item-dns} are equivalent.

  The equivalence of \ref{item-dns} and \ref{item-inhabited-if-nonempty} follows
  from the observations that every proposition is equivalent to its
  propositional truncation, and that \(\lnot\lnot X\) is equivalent to
  \(\lnot\lnot\squash{X}\).

  The final claim of the lemma, concerning
  \ref{item-wem}--\ref{item-wem-Bool}, follows from the fact that decidability
  and weak decidability coincide for \(\lnot\lnot\)-stable propositions.
\end{proof}

\begin{definition}[Limited principle of omniscience (LPO)]\label{def:LPO}
  Bishop's \emph{limited principle of omniscience (LPO)} asserts that for every
  binary sequence \(\alpha : \Nat \to \Two\) the proposition
  \(\exists_{n : \Nat}\,\alpha(n) = 1\) is decidable.%
  \index{limited principle of omniscience|textbf}%
  \index{LPO|see {limited principle of omniscience}}
\end{definition}

\begin{remark}\label{LPO-iff-semidecidable-implies-decidable}
  Unfolding the definitions, we see that LPO says exactly that every
  semidecidable proposition is decidable.
\end{remark}

\begin{definition}[Axiom of (countable) choice]\label{def:axiom-of-choice}
  The \emph{axiom of choice} with respect to universes \(\U\) and \(\V\) says
  that for every set \(X : \U\) and set-valued type-family \(Y : X \to \V\), if
  every \(Y(x)\) is inhabited, then \(\Pi_{x : X}Y(x)\) is inhabited as well. %
  \index{choice!axiom of|textbf}%
  \index{choice!axiom of countable|textbf}%
  Symbolically, this reads
  \begin{equation}\label{axiom-of-choice}
    \pa*{\Pi_{x : X}\squash*{Y(x)}} \to \squash*{\Pi_{x : X}Y(x)}
  \end{equation}
  The special case where \(X \equiv \Nat\) is called \emph{countable choice}.
\end{definition}

There are multiple equivalent ways of phrasing the axiom of choice in univalent
foundations, see~\cite[Section~3.35]{Escardo2019}, but the above is the most
convenient formulation for us.

\begin{remark}\label{semidecidability-and-choice}
  Semidecidability and countable choice are closely linked as investigated
  in~\cite{Knapp2018,EscardoKnapp2017} and further in~\cite{deJong2022}, but it
  is somewhat beyond the scope of this thesis to go into this here.
\end{remark}
\index{constructivity|)}

\section{Univalent universes}\label{sec:univalent-universes}

By analogy to propositional extensionality (\cref{def:prop-ext}) and function
extensionality (\cref{def:fun-ext}), we define an extensionality axiom for
\emph{types} and say that a universe is univalent if its types satisfy it.

\begin{definition}[Univalence]\label{def:univalence}
  A type universe \(\U\) is \emph{univalent} if for every \(X,Y : \U\) the map
  \((X =_{\U} Y) \to (X \simeq Y)\) defined by path induction as
  \(\refl \mapsto \id\) is an equivalence.%
  \index{univalence|textbf}
\end{definition}

In other words, two types are equal precisely when they are equivalent, although
the formulation above is carefully chosen to ensure (through
\cref{Pi-is-prop,being-equiv-is-prop}) that being univalent is a property of a
universe.

The \emph{univalence axiom}\index{univalence} asserts that all universes
are univalent. Unlike propositional extensionality and function extensionality,
we do not assume this globally, but rather add the univalence of a universe as
an explicit hypothesis to our theorems when needed.

The consistency of the univalence axiom was established by Voevodsky through the
simplicial sets model~\cite{KapulkinLumsdaine2021}.

\begin{theorem}
  If \(\U\) is univalent, then we have propositional extensionality in \(\U\)
  and function extensionality for functions between types in \(\U\).
\end{theorem}
\begin{proof}
  That univalence implies propositional extensionality is straightforward as two
  propositions are equivalent precisely when they imply each other.  That
  function extensionality can be derived from univalence is due to Voevodsky,
  see~\cite[Section~3.17]{Escardo2019} or~\cite[Section~4.9]{HoTTBook} for
  proofs.
\end{proof}

The following result should be compared to \cref{singleton-type-is-singleton}.

\begin{theorem}\label{univalent-iff-contractibility}%
  \index{contractibility}%
  A universe \(\U\) is univalent if and only if for every \(X : \U\) the type
  \(\Sigma_{Y : \U}\,X \simeq Y\) is contractible.
\end{theorem}
\begin{proof}
  See \cite[Section~3.14]{Escardo2019}.
\end{proof}

\begin{definition}[Universe embedding; {\cite[Section~3.30]{Escardo2019}}]
  A function \(f\) between universes \(\U\) and \(\V\) is a \emph{universe
    embedding} if \(f(X) \simeq X\) for every \(X : \U\).%
  \index{universe!embedding}
\end{definition}

\begin{proposition}[{\cite[Section~3.30]{Escardo2019}}]\label{universe-embeddings-are-embeddings}
  If \(\U\) and \(\V\) are univalent universes, then every universe embedding
  between \(\U\) and \(\V\) is an embedding.
\end{proposition}
\begin{proof}
  Suppose that \(\U\) and \(\V\) are univalent and that \(f\) is a universe
  embedding. Then for every \(X,Y : \U\), we have
  \begin{align*}
    (X = Y) &\simeq (X \simeq Y) &&\text{(since \(\U\) is univalent)} \\
    &\simeq (f(X) \simeq f(Y)) &&\text{(since \(f\) is a universe embedding)} \\
    &\simeq (f(X) = f(Y)) &&\text{(since \(\V\) is univalent)}
  \end{align*}
  so that \(f\) is an embedding by \cref{embedding-criterion}.
\end{proof}

\begin{corollary}\label{lift-is-embedding}
  If the universes \(\U\) and \(\U \sqcup \V\) are univalent, then the universe
  embedding \(\lift{U}{V} : \U \to \U\sqcup\V\) defined by
  \(X \mapsto X \times \One_{\V}\) is an embedding.
\end{corollary}

\section{Small and locally small types}\label{sec:small-and-locally-small-types}

A fundamental theme of this work will be the concept of a (locally) small type,
as (im)predicativity is all about whether (types of) propositions are small or
not.

\begin{definition}[(Local) \(\U\)-smallness; Rijke~\cite{Rijke2017}]%
  \label{def:smallness}
  A type \(X\) in any universe is
  \begin{enumerate}[(i)]
  \item \emph{\smallt{U}} if it is equivalent to a type in the universe \(\U\),
    i.e.%
    \index{smallness|textbf}
    \[
      {X \issmalltype{\U}} \colonequiv \textstyle\sum_{Y : \U} \pa*{Y \simeq X},
    \]
  \item \emph{locally \smallt{U}} if the type \((x = y)\) is \smallt{U} for
    every \(x,y : X\).%
    \index{smallness!local}\qedhere
  \end{enumerate}
\end{definition}

From a categorical perspective, \(\U\)-small really means
\emph{essentially} \(\U\)\nobreakdash-small, because we are considering types up to
equivalence. We simply call it \(\U\)-smallness, because the corresponding
strict notions where \(Y\) is equal (definitional, or up to the intensional
identity type) to \(X\) can only make sense if \(X\) is already in the same
universe as \(Y\).

A fact that we will often use tacitly is the useful but simple observation that
if \(X\) is \smallt{U} and \(X \simeq Y\) then \(Y\) is \smallt{U} too.

\begin{example}\hfill
  \begin{enumerate}[(i)]
  \item Every \smallt{U} type is locally \smallt{U}.
  \item The type \(\Omega_{\U}\) of propositions in a universe \(\U\) lives in
    \(\U^+\), but is locally \smallt{U} by propositional extensionality.
    \qedhere
  \end{enumerate}
\end{example}

Even though we phrased \(\U\)-smallness using equivalences, a type can be
\(\U\)-small in at most one way, provided that the universes involved are
univalent.

\begin{proposition}\label{is-small-is-prop}
  If \(\V\) and \(\U \sqcup \V\) are univalent universes, then the type
  expressing that \(X\) is \(\V\)-small is a proposition for every \(X : \U\).
\end{proposition}
\begin{proof}
  If \(\U \sqcup \V\) is univalent, then
  \begin{align*}
    X \issmalltype{\V} &\equiv {\Sigma_{Y : \V}(Y \simeq X)} \\
    &\simeq {\Sigma_{Y : \V}(\lift{V}{U}(Y) \simeq \lift{U}{V}(X))}
    &&\text{(as the lifts are universe embeddings)} \\
    &\simeq {\Sigma_{Y : \V}(\lift{V}{U}(Y) = \lift{U}{V}(X))}
    &&\text{(as \(\U\sqcup\V\) is univalent)} \\
    &\equiv \fib_{\lift{V}{U}}\pa*{\lift{U}{V}(X)}.
  \end{align*}
  But the latter is a proposition because \(\lift{V}{U}\) is an embedding by
  \cref{lift-is-embedding} using our assumptions that \(\V\) and
  \(\U \sqcup \V\) are univalent.
\end{proof}
The converse also holds in the following form.
\begin{proposition}\label{is-small-univalence}
  The type expressing that \(X\) is \(\U\)-small is a proposition for every
  \(X : \U\) if and only if the universe \(\U\) is univalent.
\end{proposition}
\begin{proof}
  This follows from \cref{univalent-iff-contractibility}.
\end{proof}

\begin{lemma}\label{being-small-prop-is-prop}
  Propositional extensionality suffices to prove that being \(\U\)-small is a
  property for propositions.
\end{lemma}
\begin{proof}
  For propositions we see that the argument of \cref{is-small-is-prop} and its
  dependency \cref{universe-embeddings-are-embeddings} only require
  propositional extensionality and not full univalence.
\end{proof}

We end this section by showing our main technical result on small types here,
namely that being small is closed under retracts.
The following original definition extends the notion of a small type to
functions.

\begin{definition}[\(\U\)-smallness for maps]
  A map \(f : X \to Y\) is said be \(\U\)-small if all its fibres are.%
  \index{smallness!for maps}
\end{definition}

\begin{lemma}\label{small-maps-lemmas}\hfill
  \begin{enumerate}[(i)]
  \item\label{small-type-in-terms-of-small-map} A type \(X\) is \(\U\)-small if
    and only if the unique map \(X \to \One_{\U_0}\) is \(\U\)-small.
  \item\label{small-domain-iff-small-map} If \(Y\) is \(\U\)-small, then a map
    \(f : X \to Y\) is \(\U\)-small if and only if \(X\) is.
  \end{enumerate}
\end{lemma}
\begin{proof}
  \ref{small-type-in-terms-of-small-map} Writing \(!_X\) for the map
  \(X \to \One_{\U_0}\) we have \(\fib_{!_X}(\star) \simeq X\).
  \ref{small-domain-iff-small-map} If \(X\) and \(Y\) are both
  \(\U\)\nobreakdash-small, witnessed respectively by \(\varphi : X' \simeq X\)
  and \(\psi : Y' \simeq Y\), then \(\fib_f(y)\) is \(\U\)\nobreakdash-small for
  every \(y : Y\), because
  \(\fib_f(y) \equiv \Sigma_{x : X}\pa*{f(x) = y} \simeq \Sigma_{x' :
    X'}\pa*{\psi^{-1}\pa*{f(\varphi(x'))} = \psi^{-1}(y)}\).
  Conversely, if \(f\) and \(Y\) are \(\U\)-small, then so is \(X\), because
  \cite[Lemma~4.8.2]{HoTTBook} tells us that
  \(X \simeq \Sigma_{y : Y}\fib_f(y)\).
\end{proof}

\begin{theorem}\label{is-small-retract}
  Every section into a \(\U\)-small type is \(\U\)-small. In particular, its
  domain is \(\U\)-small.
\end{theorem}
\begin{proof}
  We show that the domain is \(\U\)-small from which it follows that the section
  is \(\U\)-small by~\cref{small-maps-lemmas}\ref{small-domain-iff-small-map}.
  So suppose we have a section \(s : X \to Y\) with retraction \(r : Y \to X\)
  and that \(Y\) is \(\U\)-small.
  By~\cite[Lemma~3.6]{Shulman2016}, the endomap \(f \colonequiv r \circ s\) on
  \(Y\) is a quasi-idempotent~\cite[Definition~3.5]{Shulman2016}. Hence,
  \cite[Theorem~5.3]{Shulman2016} tells us that \(f\) can be split as
  \(Y \xrightarrow{r'} A \xrightarrow{s'} Y\) for some maps \(s'\) and \(r'\)
  and some type \(A\) recalled below.
  Now \(X\) and \(A\) are equivalent as witnessed by the maps
  \(x \mapsto r'(s(x))\) and \(a \mapsto r(s'(a))\).
  Finally, we recall from the proof of~\cite[Theorem~5.3]{Shulman2016} that
  \(A \colonequiv \Sigma_{\sigma : \Nat \to Y}\Pi_{n : \Nat}\pa*{f(\sigma_{n+1})
    = \sigma_n}\) which is \(\U\)-small because \(Y\) is assumed to be.
\end{proof}

\begin{remark}\label{small-retract-improvement}
  In~\cite{deJongEscardo2021b} we had a weaker version of~\cref{is-small-retract} where
  we included the additional assumption that the section was an embedding. (Note
  that if every section is an embedding, then every type is a
  set~\cite[Remark~3.11(2)]{Shulman2016}, but that all sections into \emph{sets}
  are embeddings~\cite[\texttt{lc-maps-into-sets-are-embeddings}]{Escardo2019}.)
  We are grateful to the anonymous reviewer of~\cite{deJongEscardo2022} who
  proposed the above strengthening.
\end{remark}

\section{Impredicativity: resizing axioms}\label{sec:impredicativity-resizing-axioms}

We have already explained in \cref{sec:constructivity} that our setup is
constructive because we do not assume excluded middle or the axiom of choice.
Similarly, our setup is \emph{predicative} because we do not assume certain
resizing principles concerning propositions which we define below.
Recall that the type of all propositions in a universe \(\U\) is denoted by
\(\Omega_{\U}\) and lives in \(\U^+\). Similarly, the type of all
\(\lnot\lnot\)-stable propositions in \(\U\) is denoted by \(\Omeganotnot{\U}\)
and also lives in \(\U^+\).

\begin{definition}[Propositional resizing]%
  \index{resizing!propositional|textbf}%
  \index{resizing!of the type of propositions|textbf}%
  \hfill
  \begin{enumerate}[(i)]
  \item\label{item-prop-res}%
    By \emph{Propositional-\(\text{Resizing}_{\U,\V}\)} we mean the assertion
    that every proposition \(P\) in a universe \(\U\) is \smallt{V}.
    \nomenclature[PropositionalResizing]{$\Propresizing{\U}{\V}$}{propositional
      resizing of propositions in \(\U\) to \(\V\)}
  \item\label{item-Omega-res}%
    We write
    \emph{\(\Omega\)-\(\text{Resizing}_{\U,\V}\)} for the assertion that the type
    \(\Omega_{\U}\) is \smallt{V}.
    \nomenclature[OmegaResizing]{$\Omegaresizing{\U}{\V}$}{assertion that
      \(\Omega_{\U}\) is \(\V\)-small}
  \item%
    By \emph{\(\Omega_{\lnot\lnot}\)-\(\text{Resizing}_{\U,\V}\)} we mean the
    assertion that the type \(\Omeganotnot{\U}\) is \smallt{V}.
    \nomenclature[OmeganotnotResizing]{$\Omeganotnotresizing{\U}{\V}$}{assertion that
      \(\Omeganotnot{\U}\) is \(\V\)-small}
  \item For the particular case of a single universe, we respectively write
    \emph{\(\Omega\)-\(\text{Resizing}_{\U}\)} and
    \emph{\(\Omega_{\lnot\lnot}\)-\(\text{Resizing}_{\U}\)} for the assertions
    that \(\Omega_{\U}\) is \smallt{U} and \(\Omeganotnot{\U}\) is
    \smallt{U}. \qedhere
    \nomenclature[OmegaResizingalt]{$\Omegaresizingalt{\U}$}{assertion that
      \(\Omega_{\U}\) is \(\U\)-small}
    \nomenclature[OmeganotnotResizingalt]{$\Omeganotnotresizingalt{\U}$}{assertion
      that \(\Omeganotnot{\U}\) is \(\U\)-small}
  \end{enumerate}
\end{definition}

The resizing of the type of propositions in a universe~\ref{item-Omega-res} is
closely related to the resizing of the propositions
themselves~\ref{item-prop-res}, as we show now.

\begin{proposition}[cf.~{\cite[Section~3.36.4]{Escardo2019}}]\label{Omega-prop-resizing-iff}%
  For every two type universes \(\U\)~and~\(\V\) we have that
  \begin{enumerate}[(i)]
  \item\label{Omegaresizing-implies-Propresizing} \(\Omegaresizing{\U}{\V}\)
    implies \(\Propresizing{\U}{\V}\), and
  \item\label{Propresizing-implies-Omegaresizing} the conjunction of
    \(\Propresizing{\U}{\V}\) and \(\Propresizing{\V}{\U}\) implies
    \(\Omegaresizing{\U}{\V^+}\).
  \end{enumerate}
\end{proposition}
\begin{proof}
  \ref{Omegaresizing-implies-Propresizing} Assuming \(\Omegaresizing{\U}{\V}\)
  we have an equivalence \(\varphi : \Omega_{\U} \simeq T : \V\), so if \(P\) is
  an arbitrary proposition in \(\U\), then \(P\) holds if and only if
  \(P = \One_{\U}\) if and only if \(\varphi(P) = \varphi\pa*{\One_{\U}}\), but
  the latter is a type in \(\V\).
  \ref{Propresizing-implies-Omegaresizing} Assuming \(\Propresizing{\U}{\V}\)
  and \(\Propresizing{\V}{\U}\) yields maps
  \(\phi : \Omega_{\U} \to \Omega_{\V}\) and
  \(\psi : \Omega_{\V} \to \Omega_{\U}\) such that for every \(P : \Omega_{\U}\)
  and \(Q : \Omega_{\V}\) we have \(P \simeq \phi(P)\) and \(Q \simeq \psi(Q)\).
  Hence, for every \(P : \Omega_{\U}\) we have
  \(P \simeq \phi(P) \simeq \psi(\phi(P))\) and similarly
  \(Q \simeq \phi(\psi(Q))\) for every \(Q : \Omega_{\V}\).
  Thus, by propositional extensionality, we get an equivalence
  \(\Omega_{\U} \simeq \Omega_{\V}\), but \(\Omega_{\V} : \V^+\), so
  \(\Omegaresizing{\U}{\V^+}\) must hold.
\end{proof}

\begin{remark}
  In light of the occurrence of \(\V^+\)
  in~\cref{Omega-prop-resizing-iff}\ref{Propresizing-implies-Omegaresizing}, it
  is worth observing that if we assume the strong principle
  \(\Propresizing{\U}{\U_0}\) then \(\Omegaresizing{\U}{\U_1}\) holds. Hence, if
  we can resize every proposition to one in the lowest universe \(\U_0\), then
  the type of propositions in an arbitrary universe is equivalent to a type in
  \(\U_1\), which is a significant resizing for all universes except \(\U_0\).
\end{remark}

With the classical axiom of excluded middle, impredicativity becomes a
theorem. Thus, if we wish to explore predicativity (in the form of propositional
resizing axioms) in univalent foundations then we must work constructively.

\begin{proposition}[cf.~{\cite[Section~3.36.2]{Escardo2019}}]
  Excluded middle implies impredicativity. Specifically,%
  \index{excluded middle}\index{excluded middle!weak}
  \begin{enumerate}[(i)]
  \item\label{em-implies-Omegaresizing} excluded middle in \(\U\) implies
    \(\Omegaresizing{\U}{\U_0}\), and
  \item\label{wem-implies-Omeganotnotresizing} weak excluded middle in \(\U\)
    implies \(\Omeganotnotresizing{\U}{\U_0}\).
  \end{enumerate}
\end{proposition}
\begin{proof}
  \ref{em-implies-Omegaresizing} By
  \cref{em-and-wem-equivalent-formulations}\ref{item-em-Bool} we know that
  excluded middle in \(\U\) implies having an equivalence
  \(\Omega_{\U} \simeq \Two_{\U_0}\).
  \ref{wem-implies-Omeganotnotresizing} By
  \cref{em-and-wem-equivalent-formulations}\ref{item-wem-Bool}.
\end{proof}

\section{Quotients, replacement, and propositional truncations revisited}
\label{sec:quotients-replacement-prop-trunc-revisited}

We investigate the inter-definability and interaction of type universe levels of
propositional truncations and set quotients in the absence of propositional
resizing axioms. In particular, we will see that it is not so important if the
set quotient or propositional truncation lives in a higher universe. What is
paramount instead is whether the universal property applies to types in
arbitrary universes.
However, in some cases, like in~\cref{sec:small-suprema-of-ordinals}, it is
relevant whether set quotients are small and we show this to be equivalent to a
set replacement principle in~\cref{sec:set-replacement}.

\begin{remark}\label{prop-trunc-universes}
  Recall from~\cref{sec:propositional-truncation} that in this thesis we
  typically assume our universes to be closed under propositional
  truncations. However, in this section, we will be more general and assume that
  \(\squash{X} : F(\U)\) where \(F\) is a metafunction on universes, so that the
  above case is obtained by taking \(F\) to be the identity. We will also
  consider \(F(\U) = \U_1 \sqcup \U\) in the final subsection.
  \nomenclature[F]{$F(\U)$}{the propositional truncation of a type in \(\U\)
    lives in \(F(\U)\)}
\end{remark}

\subsection{Propositional truncations and propositional resizing}%
\label{sec:prop-trunc-resizing}%
\index{resizing!propositional|(}%
\index{propositional truncation|(}%
Voevodsky~\cite{Voevodsky2011} introduced propositional resizing rules in order
to construct propositional
truncations~\cite[Section~2.4]{PelayoVoevodskyWarren2015}. Here we review
Voevodsky's construction, paying special attention to the universes involved.

\emph{NB.\ We do not assume the availability of propositional truncations in
this section.}

\begin{definition}[Voevodsky propositional truncation, \(\squashVV{X}\)]
  The \emph{Voevodsky propositional truncation} \(\squashVV{X}\) of a type
  \(X : \U\) is defined as%
  \index{propositional truncation!Voevodsky}%
  \nomenclature[pipevdouble]{$\squashVV{-}$}{Voevodsky propositional truncation}
  \[
    \squashVV{X} \colonequiv \prod_{P : \U}
    \pa*{\issubsingleton(P) \to (X \to P) \to P}. \qedhere
  \]
\end{definition}

Observe that this is a
System-F~\cite{Girard1971,Reynolds1974,AwodeyFreySpeight2018} style definition
where we use the desired universal property and a (large) quantification to
encode a type.

Also notice that the part \((X \to P) \to P\) in the Voevodsky propositional
truncation generalises the double negation (which is given by taking
\(P \equiv \Zero\)). The double negation of a type is a proposition, but it
enjoys the universal property of the propositional truncation if and only if
excluded middle holds~\cite[Section~7]{KrausEtAl2017}.

Because of \cref{Pi-is-prop}, one can show that \(\squashVV{X}\) is indeed a
proposition for every type \(X\). Moreover, we have a map
\(\tosquashVV{-} : X \to \squashVV{X}\) given by
\(\tosquashVV{x} \colonequiv (P,i,f) \mapsto f(x)\).
\nomenclature[pipev]{$\tosquashVV{-}$}{unit of the Voevodsky propositional truncation}

Observe that \(\squashVV{X} : \U^+\), so using the notation
from~\cref{prop-trunc-universes}, we have \({F(\U) \colonequiv \U^+}\).
However, as we will argue for set quotients, it
does not matter so much where the truncated proposition lives; it is much more
important that we can eliminate into subsingletons in arbitrary universes, i.e.\
that \(\squashVV{-}\) satisfies the right universal property.
Given \(X : \U\) and a map \(f : X \to P\) to a proposition \(P : \U\) with
\(i : \issubsingleton(P)\), we have a map \(\squashVV{X} \to P\) given as
\(\Phi \mapsto \Phi(P,i,f)\).
However, if the proposition \(P\) lives in some other universe \(\V\), then we
seem to be completely stuck. To clarify this, we consider the example of
functoriality.

\begin{example}\label{Voevodsky-prop-trunc-func}
  If we have a map \(f : X \to Y\) with \(X : \U\) and \(Y : \U\), then we get a
  lifting simply by precomposition, i.e.\ we define
  \(\tosquashVV{f} : \squashVV{X} \to \squashVV{Y}\) by
  \(\tosquashVV{f}(\Phi) \colonequiv (P,i,g) \mapsto \Phi(P,i,g \circ f)\).
  But obviously, we also want functoriality for maps \(f : X \to Y\) with
  \(X : \U\) and \(Y : \V\), but this is impossible with the above definition of
  \(\tosquashVV{f}\), because for \(\squashVV{X}\) we are considering
  propositions in \(\U\), while for \(\squashVV{Y}\) we are considering
  propositions in \(\V\).

  In particular, even if the types \(X : \U\) and \(Y : \V\) are equivalent,
  then it does not seem possible to construct an equivalence between
  \(\squashVV{X}\) and \(\squashVV{Y}\). This issue also comes up if one tries
  to prove that the map \(\tosquashVV{-} : X \to \squashVV{X}\) is a
  surjection~\cite[Section~3.34.1]{Escardo2019}.
\end{example}

\begin{proposition}[cf.\ {\cite[Theorem~3.8]{KrausEtAl2017}}]%
  \label{prop-trunc-implies-resizing-of-VV-trunc}
  If our type theory has propositional truncations with \(\squash{X} : \U\)
  whenever \(X : \U\), then \(\squashVV{X}\) is \(\U\)-small.
\end{proposition}
\begin{proof}
  We will show that \(\squash{X}\) and \(\squashVV{X}\) are logically equivalent
  (i.e.\ we have maps in both directions), which suffices, because both types
  are subsingletons. We obtain a map \(\squash{X} \to \squashVV{X}\) by applying
  the universal property of \(\squash{X}\) to the map
  \(\tosquashVV{-} : X \to \squashVV{X}\). Observe that it is essential that the
  universal property allows for elimination into subsingletons in universes
  other than \(\U\), as \(\squashVV{X} : \U^+\). For the function in the other
  direction, simply note that \(\squash{X} : \U\), so that we can construct
  \(\squashVV{X} \to \squash{X}\) as
  \(\Phi \mapsto \Phi(\squash{X},i,\tosquash{-})\) where
  \(i\) witnesses that \(\squash{X}\) is a subsingleton.
\end{proof}
Thus, as is folklore in the univalent foundations community, we can view
higher inductive types as specific resizing axioms. But note that the converse
to the above proposition does not appear to hold, because even if
\(\squashVV{X}\) is \(\U\)-small, then it still wouldn't have the appropriate
universal property. This is because the definition of \(\squashVV{X}\) is a
dependent product over propositions in \(\U\) only, which now includes
\(\squashVV{X}\), but still misses propositions in other universes.
In the presence of resizing axioms, we could obtain the full universal property,
because we would have (equivalent copies of) all propositions in a single
universe:

\begin{proposition}[cf.~{\cite[Section~36.5]{Escardo2019}}]
  If \(\Propresizing{\U}{\U_0}\) holds for every universe \(\U\), then the
  Voevodsky propositional truncation satisfies the full universal property with
  respect to all types in all universes.
\end{proposition}%
\index{resizing!propositional|)}%
\index{propositional truncation|)}

\subsection{Set quotients from propositional truncations}
\label{sec:set-quotients-from-propositional-truncations}

In this section we assume to have propositional truncations with
\(\squash{X} : F(\U)\) when \(X : \U\) for some metafunction~\(F\) on
universes. We will be mainly interested in \(F(\U) = \U\) and
\(F(\U) = \U_1 \sqcup \U\) for the reasons explained below.

We prove that we can construct set quotients using propositional
truncations. The construction is due to Voevodsky and also appears
in~\cite[Section~6.10]{HoTTBook} and \cite[Section~3.4]{RijkeSpitters2015}. %
While Voevodsky assumed propositional resizing rules in his construction, we
show, following~\cite[Section~3.37]{Escardo2019}, that resizing is not needed to
prove the universal property of the set quotient, provided propositional
truncations are available.

\begin{definition}[Equivalence relation]
  An \emph{equivalence relation}\index{equivalence relation} on a type \(X\) is
  a binary type family \({\approx} : X \to X \to \V\) that is
  \begin{enumerate}[(i)]
  \item subsingleton-valued, i.e.\ \(x \approx y\) is a subsingleton for
    every \(x,y : X\),
  \item reflexive, i.e.\ \(x \approx x\) for every \(x : X\),
  \item symmetric, i.e.\ \(x \approx y\) implies \(y \approx x\) for
    every \(x,y : X\), and
  \item transitive, i.e.\ the conjunction of \(x \approx y\) and \(y \approx z\)
    implies \(x \approx z\) for every \(x,y,z : X\). \qedhere
  \end{enumerate}
\end{definition}

\begin{definition}[Set quotient, \(X/{\approx}\)]%
  \index{set quotient}\index{quotient|see {set quotient}}%
  We define the \emph{set quotient} \(X/{\approx}\) of \(X\) by \(\approx\) to
  be the image of \(e_\approx\) where
  \nomenclature[quotient]{$X/{\approx}$}{set quotient of \(X\) by the
    equivalence relation \({\approx}\)}
  \nomenclature[eapprox]{$e_{\approx}(x)$}{subset of elements related to \(x\)
    via the equivalence relation \({\approx}\)}
  \begin{align*}
    e_\approx : X &\to (X \to \Omega_{\V}) \\
    x &\mapsto \pa*{y \mapsto (x \approx y,p(x,y))}
  \end{align*}
  and \(p\) is the witness that \(\approx\) is subsingleton-valued.
\end{definition}

Of course, we should prove that \(X/{\approx}\) really is the quotient of \(X\)
by \(\approx\) by proving a suitable universal property. The following
definition and lemmas indeed build up to this. For the remainder of this
section, we will fix a type \(X : \U\) with an equivalence relation
\({\approx} : X \to X \to \V\).

\begin{remark}\label{quotient-universes}
  Since
  \(X/{\approx} \equiv \image(e_\approx) \equiv \Sigma_{\varphi : X \to
    \Omega_{\V}} \exists_{x : X} \pa*{\lambdadot{y}{x \approx y}} = \varphi\),
  we see, recalling \cref{prop-trunc-universes} and the fact that
  \(\Omega_{\V}\) is a type in \(\V^+\), that we have
  \(X/{\approx} : \T \sqcup F(\T)\) with \(\T \colonequiv \V^+ \sqcup \U\).
  In the particular case that \(F\) is the identity, we obtain the simpler
  \(X/{\approx} : \V^+ \sqcup \U\).
\end{remark}

\begin{lemma}
  The quotient \(X/{\approx}\) is a set.
\end{lemma}
\begin{proof}
  Recall that \({X/{\approx}}\) is defined as the image of \(e_\approx\) and
  that this is a subtype of the powerset \(\powerset_{\V}(X)\) which is a set
  by~\cref{powerset-is-set}. Since it holds generally that subtypes of sets are
  sets, this proves the lemma.
\end{proof}

\begin{definition}[Universal map, \(\eta\)]
  The \emph{universal map} \(\eta : X \to X/{\approx}\) is defined to be the
  corestriction of \(e_{\approx}\).%
  \nomenclature[eta]{$\eta(x)$}{equivalence class of \(x\)}%
  \index{set quotient!universal map}
\end{definition}

Although, in general, the type \(X/{\approx}\) lives in another universe than
\(X\) (see~\cref{quotient-universes}), we can still prove the following
induction principle for subsingleton-valued families into \emph{arbitrary}
universes.

\begin{lemma}[Set quotient induction]\label{set-quotient-induction}%
  \index{set quotient!induction}%
  For every subsingleton-valued family \({P : X/{\approx} \to \W}\), with \(\W\)
  \emph{any} universe, if \(P(\eta(x))\) holds for every \(x : X\), then
  \(P(x')\) holds for every \(x' : X/{\approx}\).
\end{lemma}
\begin{proof}
  The map \(\eta\) is surjective by~\cref{corestrictions-are-surjections}, so
  that~\cref{surjection-induction} yields the desired result.
\end{proof}

\begin{definition}[Respect equivalence relation]%
  \index{equivalence relation!respect}%
  A map \(f : X \to A\) \emph{respects the equivalence relation} \(\approx\) if
  \(x \approx y\) implies \(f(x) = f(y)\) for every \(x,y : X\).
\end{definition}

Observe that respecting an equivalence relation is property rather than data,
when the codomain \(A\) of the map \(f : X \to A\) is a set.

\begin{lemma}\label{eta-respects-and-effective}
  The map \(\eta : X \to X/{\approx}\) respects the equivalence relation
  \({\approx}\) and the set quotient is \emph{effective}, i.e.\ for every
  \(x,y : X\), we have \(x \approx y\) if and only if \(\eta(x) = \eta(y)\).
\end{lemma}
\begin{proof}
  By definition of the image and function extensionality, we have for every
  \(x,y : X\) that \(\eta(x) = \eta(y)\) holds if and only if
  \begin{equation}
    \forall_{z : X}\pa*{x \approx z \iff y \approx z}
    \tag{\(\ast\)}\label{eta-equality}
  \end{equation}
  holds. If~\eqref{eta-equality} holds, then so does \(x \approx y\) by
  reflexivity and symmetry of the equivalence relation. Conversely, if
  \(x \approx y\) and \(z : X\) is such that \(x \approx z\), then
  \(y \approx z\) by symmetry and transitivity; and similarly if \(z : X\) is
  such that \(y \approx z\). Hence, \eqref{eta-equality} holds if and only if
  \(x \approx y\) holds.  Thus, \(\eta(x) = \eta(y)\) if and only if
  \(x \approx y\), as desired.
\end{proof}

The universal property of the set quotient states that the map
\(\eta : X \to X/{\approx}\) is the universal function to a set preserving the
equivalence relation. We can prove it using only~\cref{set-quotient-induction}
and~\cref{eta-respects-and-effective}, without needing to inspect the definition
of the quotient.

\begin{theorem}[Universal property of the set quotient]
  \label{set-quotient-universal-property}%
  \index{set quotient!universal property}%
  For every \emph{set} \(A : \W\) in \emph{any} universe~\(\W\) and function
  \(f : X \to A\) respecting the equivalence relation, there is a unique
  function \(\bar{f} : X/{\approx} \to A\) such that the diagram
  \[
    \begin{tikzcd}
      X \ar["f"', dr] \ar["\eta", rr] & & X/{\approx} \ar["\bar{f}",dl,dashed] \\
      & A
    \end{tikzcd}
  \]
  commutes.
\end{theorem}
\begin{proof}
  Let \(A : \W\) be a set and \(f : X \to A\) respect the equivalence relation.
  The following auxiliary type family over \(X/{\approx}\) will be at the heart
  of our proof:
  \[
    B(x') \colonequiv \Sigma_{a : A}\exists_{x : X}\pa{(\eta(x) = x') \times
      (f(x) = a)}.
  \]
  \vspace{-3ex} 
  \begin{claim}
    The type \(B(x')\) is a subsingleton for every \(x' : X/{\approx}\).
  \end{claim}
  \begin{proof}[Proof of claim]
    By function extensionality, the type expressing that \(B(x')\) is a
    subsingleton for every \(x' : X/{\approx}\) is itself a subsingleton. So by
    set quotient induction, it suffices to prove that \(B(\eta(x))\) is a
    subsingleton for every \(x : X\). So assume that we have
    \((a,p) , (b,q) : B(\eta(x))\). We only need to show that \(a = b\). The
    elements \(p\)~and~\(q\) witness
    \[
      \exists_{x_1 : X}\pa{(\eta(x_1) = \eta(x)) \times (f(x_1) = a)}
    \]
    and
    \[
      \exists_{x_2 : X}\pa{(\eta(x_2) = \eta(x)) \times (f(x_2) = b)},
    \]
    respectively. By~\cref{eta-respects-and-effective} and the fact that \(f\)
    respects the equivalence relation, we obtain \(f(x) = a\) and \(f(x) = b\)
    and hence the desired \(a = b\).
  \end{proof}
  Next, we define \(k : \Pi_{x : X} B(\eta(x))\) by
  \(k(x) = \pa*{f(x),\tosquash*{(x,\refl,\refl)}}\). By set quotient induction and
  the claim, the function \(k\) induces a dependent map
  \(\bar{k} : \Pi_{\pa*{x' : X/{\approx}}} B(x')\).

  We then define the (nondependent) function \(\bar{f} : X/{\approx} \to A\) as
  \({\fst} \circ {\bar{k}}\). We proceed by showing that
  \(\bar{f} \circ \eta = f\). By function extensionality, it suffices to prove
  that we have \(\bar{f}(\eta(x)) = f(x)\) for every \(x : X\). But notice that:
  \begin{align*}
    \bar{f}(\eta(x)) &\equiv \fst(\bar{k}(\eta(x))) \\
                     &= \fst(k(x)) &&\text{(since \(\bar{k}(\eta(x)) = k(x)\)
                                           because of the claim)} \\
    &\equiv f(x).
  \end{align*}

  Finally, we wish to show that \(\bar{f}\) is the unique such function, so
  suppose that \({g : X/{\approx} \to A}\) is another function such that
  \(g \circ \eta = f\). By function extensionality, it suffices to prove that
  \(g(x') = \bar{f}(x')\) for every \(x' : X/{\approx}\), which is a
  subsingleton as \(A\) is a set. Hence, set quotient induction tells us that it
  is enough to show that \(g(\eta(x)) = \bar{f}(\eta(x))\) for every \(x : X\),
  but this holds as both sides of the equation are equal to \(f(x)\).
\end{proof}

\begin{remark}[cf.\ Section~3.21 of~\cite{Escardo2019}]
  In univalent foundations, some attention is needed in phrasing unique
  existence, so we pause to discuss the phrasing
  of~\cref{set-quotient-universal-property} here.
  Typically, if we wish to express unique existence of an element \(x : X\)
  satisfying \(P(x)\) for some type family \(P : \U \to \V\), then we should
  phrase it as \(\issingleton(\Sigma_{x : X}P(x))\). That is, we require that
  there is a unique \emph{pair} \((x,p) : \Sigma_{x : X}P(x)\).
  However, if \(P\) is subsingleton-valued, then it is
  equivalent to the traditional formulation of unique existence: i.e.\ that
  there is an \(x : X\) with \(P(x)\) such that every \(y : X\) with \(P(y)\) is
  equal to \(x\).
  This happens to be the situation in~\cref{set-quotient-universal-property},
  because of function extensionality and the fact that \(A\) is a set.
\end{remark}

We stress that although the set quotient increases universe levels,
see~\cref{quotient-universes}, it does satisfy the appropriate universal
property, so that resizing is not needed.

Having small set quotients is closely related to propositional resizing, as we
show now.

\begin{proposition}
  Suppose that \(\squash{-}\) does not increase universe levels, i.e.\ in the
  notation of~\cref{prop-trunc-universes}, the function \(F\) is the identity.
  \begin{enumerate}[(i)]
  \item\label{quotient-resizing-1} %
    If \(\Omegaresizing{\V}{\U}\) holds for universes \(\U\) and \(\V\), then
    the set quotient \(X/{\approx}\) is \(\U\)\nobreakdash-small for any type
    \(X : \U\) and any \(\V\)-valued equivalence relation.
  \item\label{quotient-resizing-2} %
    Conversely, if the set quotient \(\Two/{\approx}\) is \(\U\)-small for
    every \(\V\)-valued equivalence relation on \(\Two\), then
    \(\Propresizing{\V}{\U}\) holds.
  \end{enumerate}
\end{proposition}
\begin{proof}
  \ref{quotient-resizing-1}: If we have \(\Omegaresizing{\V}{\U}\), then
  \(\Omega_{\V}\) is \(\U\)-small, so that
  \(X/{\approx} \equiv \image(e_{\approx})\) is \(\U\)-small too when \(X : \U\)
  and \({\approx}\) is \(\V\)-valued.
  \ref{quotient-resizing-2}: Let \(P : \V\) be any proposition and consider the
  \(\V\)-valued equivalence relation
  \( x \approx_P y \colonequiv (x = y) \vee P \) on \(\Two\). Notice that
  \[
    \pa*{\Two/{\approx_P}} \text{ is a subsingleton} \iff P \text{ holds},
  \]
  so if \(\Two/{\approx_P}\) is \(\U\)-small, then so is the type
  \(\issubsingleton\pa*{\Two/{\approx_P}}\) and therefore \(P\).
\end{proof}

\subsection{Propositional truncations from set quotients}
\label{sec:propositional-truncations-from-set-quotients}

Conversely, the propositional truncation arises as a particular set quotient,
namely by identifying all elements of a type. However, in order to get an exact
match in terms of back-and-forth constructions, we must pay some attention to
the universes involved as in \cref{comparison-of-universes} below.

\emph{NB.\ We do not assume the availability of propositional truncations in
  this section.}

\begin{definition}[Existence of set quotients]\label{existence-of-set-quotients}%
  \index{set quotient!existence of}%
  We say that \emph{set quotients exist} if for every type \(X\) and equivalence
  relation \({\approx}\) on \(X\), we have a set \(X/{\approx}\) with a
  universal map \(\eta : X \to X/{\approx}\) that respects the equivalence
  relation such that
  the universal property set out in~\cref{set-quotient-universal-property} is
  satisfied.
\end{definition}

\begin{theorem}\index{set quotient!induction}%
  Any set quotient satisfies the induction principle
  of~\cref{set-quotient-induction}, i.e.\ it is implied by the universal
  property of the set quotient.
\end{theorem}
\begin{proof}
  Suppose that \(P : X/{\approx} \to \W\) is a proposition-valued type-family
  over the set quotient \(X/{\approx}\) and that we have
  \(\rho : \Pi_{x : X}P(\eta(x))\). We write
  \(S \colonequiv \Sigma_{x' : X/{\approx}}\,P(x')\) and define the map
  \(f : X \to S\) by \(f(x) \colonequiv \pa*{\eta(x) , \rho(x)}\). Note that
  \(f\) respects the equivalence relation since \(\eta\)~does and \(P\) is
  proposition-valued.
  Moreover, \(S\) is a set, because subtypes of sets are sets and the quotient
  \(X/{\approx}\) is a set by assumption.
  Hence, by the universal property, \(f\) induces a map
  \(\bar{f} : X/{\approx} \to S\) such that \(\bar{f} \circ \eta = f\).
  We claim that \(\bar{f}\) is a section of \(\fst : S \to X/{\approx}\).
  Note that this would finish the proof, because if we have
  \(e : \Pi_{x' : X/{\approx}}\,\fst\pa*{\bar{f}(x')} = x'\), then we obtain
  \(P(x')\) for every \(x'\) by transporting \(\snd\pa*{\bar{f}(x')}\) along
  \(e(x')\).
  But \(\bar{f}\) must be a section of~\(\fst\), because we can take both
  \({\fst} \circ {\bar{f}}\) and \(\id\) for the dashed map in the commutative
  diagram
  \begin{center}
    \begin{tikzcd}
      X \ar["\eta"', dr] \ar["\eta", rr] & & X/{\approx} \ar[dl,dashed] \\
      & X/{\approx}
    \end{tikzcd}
  \end{center}
  since \({\fst} \circ {\bar{f}} \circ {\eta} = {\fst} \circ {f} = \eta\), so
  \({\fst} \circ {\bar{f}}\) and \(\id\) must be equal by the universal property
  of the set quotient.
\end{proof}

\begin{theorem}\label{set-quotients-give-propositional-truncations}%
  \index{propositional truncation}%
  If set quotients exist, then every type has a propositional truncation.
\end{theorem}
\begin{proof}
  Let \(X : \U\) be any type and consider the \(\U_0\)-valued equivalence
  relation that identifies all elements:
  \({x \approx_{\One} y \colonequiv \One_{\U_0}}\).
  To see that \(X/{\approx_{\One}}\) is a subsingleton, note that by set
  quotient induction it suffices to prove \(\eta(x) = \eta(y)\) for every
  \(x,y : X\).
  But \(x \approx_{\One} y\) for every \(x,y : X\), and \(\eta\) respects the
  equivalence relation, so this is indeed the case.
  Now if \(P : \V\) is any subsingleton and \(f : X \to P\) is any map, then
  \(f\) respects the equivalence relation \({\approx_{\One}}\) on~\(X\), simply
  because \(P\) is a subsingleton. Thus, by the universal property of the
  quotient, we obtain the desired map \(\bar{f} : X/{\approx_{\One}} \to P\) and
  hence, \(X/{\approx_{\One}}\) has the universal property of the propositional
  truncation.
\end{proof}

\begin{remark}\label{comparison-of-universes}
  Because the set quotients constructed using the propositional truncation live
  in higher universes, we embark on a careful comparison of universes. Suppose
  that propositional truncations of types \(X : \U\) exist and that
  \(\squash{X} : F(\U)\). Then by~\cref{quotient-universes}, the set quotient
  \(X/{\approx_{\One}}\) in the proof above lives in the type universe
  \((\U_1 \sqcup \U) \sqcup F(\U_1 \sqcup \U)\).

  In particular, if \(F\) is the identity and the propositional truncation of
  \(X : \U\) lives in \(\U\), then the quotient \(X/{\approx_{\One}}\) lives in
  \(\U_1 \sqcup \U\), which simplifies to \(\U\) whenever \(\U\) is at least
  \(\U_1\). In other words, the universes of \(\squash{X}\) and
  \(X/{\approx_{\One}}\) match up for types \(X\) in every universe,
  \emph{except} the first universe \(\U_0\).

  If we always wish to have \(X/{\approx_{\One}}\) in the same universe as
  \(\squash{X}\), then we can achieve this by assuming
  \(F(\V) \colonequiv \U_1 \sqcup \V\), which says that the propositional
  truncations stay in the same universe, \emph{except} when the type is in the
  first universe \(\U_0\) in which case the truncation will be in the second
  universe \(\U_1\).
\end{remark}

\begin{theorem}\index{set quotient!effective}%
  All set quotients are effective, i.e.\ \(\eta(x) = \eta(y)\) implies
  \(x \approx y\).
\end{theorem}
\begin{proof}
  If we have set quotients, then we have propositional truncations
  by~\cref{set-quotients-give-propositional-truncations} which we can use to
  construct effective set quotients
  following~\cref{sec:set-quotients-from-propositional-truncations}.
  But any two set quotients of a type by an equivalence relation must be
  equivalent, so the original set quotients are effective too.
\end{proof}

\subsection{Set replacement}\label{sec:set-replacement}
In this section, we return to our running assumption that universes are closed
under propositional truncations, i.e.\ the metafunction \(F\) above is assumed
to be the identity. We study the equivalence of a set replacement principle and
the smallness of set quotients using our construction of
\cref{sec:set-quotients-from-propositional-truncations}. These principles will
find application in~\cref{sec:small-suprema-of-ordinals}, but are quite relevant
to us in any case, as smallness of types is a central theme in this thesis.

\begin{definition}[Set replacement]%
  \index{set replacement|textbf}%
  The \emph{set replacement} principle asserts that the image of a map
  \(f : X \to Y\) is \(\pa*{\U\sqcup\V}\)-small if \(X\) is a \(\U\)-small type
  and \(Y\) is a locally \(\V\)-small set.
\end{definition}

In particular, if \(\U\) and \(\V\) are the same, then the image is
\(\U\)-small.
The name ``set replacement'' is inspired by~\cite[Section~2.19]{BezemEtAl2022},
but the principle presented here differs from the one in \cite{BezemEtAl2022} in
two ways: In \cite{BezemEtAl2022}, replacement is not restricted to maps into
sets, and the universe parameters \(\U\)~and~\(\V\) are taken to be the same.
Rijke~\cite{Rijke2017} shows that the replacement of~\cite{BezemEtAl2022} is
provable in the presence of a univalent universe closed under pushouts.

We are going to show that set replacement is logically equivalent to having
small set quotients, where the latter means that the quotient of a type
\(X : \U\) by a \(\V\)-valued equivalence relation lives in \(\U \sqcup \V\).

\begin{definition}[Existence of small set quotients]%
  \label{existence-of-small-set-quotients}%
  \index{set quotient!existence of!small}%
  We say that \emph{small set quotients exist} if set quotients exists in the
  sense of~\cref{existence-of-set-quotients}, and moreover, the quotient
  \(X/{\approx}\) of a type \(X : \U\) by a \(\V\)-valued equivalence relation
  lives in \(\U \sqcup \V\).
\end{definition}

Note that we would get small set quotients if we added set quotients as a
primitive higher inductive type. Also, if one assumes \(\Omegaresizingalt{\V}\),
then the construction of set quotients
in~\cref{sec:set-quotients-from-propositional-truncations} yields a quotient
\({X/{\approx}}\) in \({\U \sqcup \V}\) when \(X : \U\) and \({\approx}\) is a
\(\V\)-valued equivalence relation on \(X\).

\begin{theorem}
  Set replacement is logically equivalent to the existence of small set
  quotients.
\end{theorem}
\begin{proof}
  Suppose set replacement is true and that a type \(X : \U\) and a \(\V\)-valued
  equivalence relation~\({\approx}\) are given. Using the construction laid out
  in~\cref{sec:set-quotients-from-propositional-truncations}, we construct a set
  quotient \(X/{\approx}\) in \(\U \sqcup \V^+\) as the image of a map
  \(X \to (X \to \Omega_{\V})\). But by propositional extensionality
  \(\Omega_{\V}\) is locally \(\V\)-small and by function extensionality so is
  \(X \to \Omega_{\V}\). Hence, \(X/{\approx}\) is \((\U \sqcup \V)\)-small by
  set replacement, so \(X/{\approx}\) is equivalent to a type
  \(Y : \U\sqcup\V\).  It is then straightforward to show that \(Y\) satisfies
  the properties of the set quotient as well, finishing the proof of one
  implication.

  Conversely, let \(f : X \to Y\) be a map from a \(\U\)-small type to a locally
  \(\V\)-small set. Since \(X\) is \(\U\)-small, we have \(X' : \U\) such that
  \(X' \simeq X\). And because \(Y\) is locally \(\V\)-small, we have a
  \(\V\)-valued binary relation \({=_{\V}}\) on \(Y\) such that
  \((y =_{\V} y') \simeq (y = y')\) for every \(y,y' : Y\).
  We now define the \(\V\)-valued equivalence relation \({\approx}\) on \(X'\)
  by \((x \approx x') \colonequiv \pa*{f'(x) =_{\V} f'(x')}\), where
  \(f'\) is the composite \(X' \simeq X \xrightarrow{f} Y\).
  By assumption, the quotient \(X'/{\approx}\) lives in \(\U \sqcup \V\). But it
  is straightforward to work out that \(\image(f)\) is equivalent to this
  quotient. Hence, \(\image(f)\) is \(\pa*{\U \sqcup \V}\)-small, as desired.
\end{proof}

The left-to-right implication of the theorem above is similar
to~\cite[Corollary~5.1]{Rijke2017}, but our theorem generalises the universe
parameters and restricts to maps into sets. The latter is the reason why the
converse also holds.

\section{Indexed \texorpdfstring{\(\WW\)}{W}-types}\label{sec:indexed-W-types}

This final section discusses a general encoding of inductive types known as
\(\WW\)-types, which are due to Martin-L\"of~\cite{MartinLof1984}. The purpose
of this encoding is that it allows us to prove results about general inductive
types by proving them for \(\WW\)-types. This is exactly what we do in this
section. Specifically, we present a criterion for having decidable equality for
a further generalisation of \(\WW\)-types known as indexed
\(\WW\)-types~\cite{GambinoHyland2004,AbbottAltenkirchGhani2004,AltenkirchEtAl2015,Sattler2015},
which will find application in~\cref{sec:Scott-model-of-PCF}. The further
generalisation allows for inductive types with many-sorted constructors as we
explain below. But we start by defining and illustrating (nonindexed)
\(\WW\)-types.

\subsection{Basic definitions and examples}
\begin{definition}[\(\WW\)-type, \(\WW_{A,B}\), \(\sup\)]%
  \index{W-type@\(\WW\)-type}%
  \index{type!\(\WW\)-|see {\(\WW\)-type}}%
  The \emph{\(\WW\)-type \(\WW_{A,B}\)} specified by a type \(A : \U\) and type
  family \({B : A \to \V}\) is the inductive type with a single constructor
  \[\sup : \Pi_{a : A}\pa*{\pa*{B(a) \to \WW_{A,B}} \to \WW_{A,B}}.\]
  We postulate that \(\WW_{A,B}\) lives in the universe \(\U \sqcup \V\).
  \nomenclature[W]{$\WW_{A,B}$}{\(\WW\)-type with parameters \(A : \U\) and
    \(B : A \to \V\)}
  \nomenclature[sup]{$\sup$}{constructor for (indexed) \(\WW\)-types (see
    also p.~\pageref{def:indexed-W-type})}
\end{definition}

\begin{remark}[The induction principle of a \(\WW\)-type]
  Spelling out the induction principle of \(\WW_{A,B}\), it reads: for every
  \(Y : \WW_{A,B} \to \T\), then to prove \(Y(w)\) for every \(w : \WW_{A,B}\),
  it suffices to prove that for any \(a : A\) and \(f : B(a) \to \WW_{A,B}\)
  satisfying \(Y(f(b))\) for every \(b : B(a)\) (the ``induction hypothesis''),
  we have \(Y(\sup(a,f))\).
\end{remark}

The elements of a \(\WW\)-type \(\WW_{A,B}\) can be thought of as some kind of
well-founded trees, hence the name \(\WW\)-type, but we prefer another
viewpoint. We think of the type \(A\) in the definition above as a type of
labels for constructors, while \(B(a)\) encodes the arity of the constructor
labelled by \(a\).\index{well-founded!tree}\footnote{The name \(\sup\) does not
  make much sense from this point of view, but it is the traditional name in the
  existing literature.}
It is instructive to see how \(\WW\)-types can encode the
type of natural numbers.
\begin{example}[The type of natural numbers as a \(\WW\)-type]
  Following the above description, we define \(A \colonequiv \Two\), since
  \(\Nat\) has two constructors. Furthermore, we put \(B(0) = \Zero\), since the
  \(\zeroo\) constructor takes no arguments, while \(B(1) = \One\), because
  \(\succc\) takes one recursive argument.

  We recursively define functions back and forth between the types as follows:
  \begin{equation*}
    \begin{minipage}{.45\textwidth}
    \vspace{-\baselineskip}
    \begin{align*}
      \phi : \Nat &\to \WW_{A,B} \\
      \zeroo &\mapsto \sup(0,\uniquefromZero), \\
      \succc(n) &\mapsto \sup(1,\lambdadot{\star}{\phi(n)}),
    \end{align*}
  \end{minipage}
  \begin{minipage}{.45\textwidth}
    \vspace{-\baselineskip}
    \begin{align*}
      \psi : \WW_{A,B} &\to \Nat \\
      \sup(0,f) &\mapsto \zeroo, \\
      \sup(1,f) &\mapsto \succc(\psi(f(\star))),
    \end{align*}
  \end{minipage}
  \end{equation*}
  where \(\uniquefromZero\) is the unique map from \(\Zero\) to \(\WW_{A,B}\).

  Using the induction principles of \(\Nat\) and \(\WW_{A,B}\) it is then
  straightforward to prove that \(\phi\)~and~\(\psi\) are inverses. Thus, the
  types \(\Nat\) and \(\WW_{A,B}\) are equivalent.
\end{example}

Another example that will come in useful later when studying the programming
language PCF is the encoding of PCF types.

\begin{example}[The PCF types as a \(\WW\)-type]\label{W-PCF-types}%
  \index{W-type@\(\WW\)-type!encoding PCF types}%
  \index{PCF|seealso{\(\WW\)-type}}%
  The PCF types are inductively generated: \(\iota\) is a PCF type, known as the
  \emph{base type} and if \(\sigma\) and \(\tau\) are PCF types, then we have
  another PCF type, called the \emph{(PCF) function type} and denoted by
  \(\sigma \Rightarrow \tau\).

  We can encode this inductive type as a \(\WW\)-type as follows. Take
  \(A \colonequiv \Two\) since we have two constructors and put
  \(B(0) = \Zero\), since the base type needs no arguments, while
  \(B(1) = \Two\), because to construct a PCF function type we need to be given
  two PCF types. The maps back and forth the inductive types are given by
  \begin{equation*}
    \begin{minipage}{.45\textwidth}
    \vspace{-\baselineskip}
    \begin{align*}
      \phi : \PCFTypes &\to \WW_{A,B}, \\
      \iota &\mapsto \sup(0,\uniquefromZero), \\
      \sigma \Rightarrow \tau &\mapsto \sup(1,g),
    \end{align*}
  \end{minipage}
  \begin{minipage}{.45\textwidth}
    \vspace{-\baselineskip}
    \begin{align*}
      \psi : \WW_{A,B} &\to \PCFTypes \\
      \sup(0,f) &\mapsto \iota, \\
      \sup(1,f) &\mapsto \psi(f(0)) \Rightarrow \psi(f(1)),
    \end{align*}
  \end{minipage}
\end{equation*}
with \(g : \Two \to \WW_{A,B}\) given by \(g(0) \colonequiv \phi(\sigma)\) and
\(g(1) \colonequiv \phi(\tau)\).
Using the induction principles of \(\PCFTypes\) and \(\WW_{A,B}\) it is then
straightforward to prove that \(\phi\)~and~\(\psi\) are inverses. Thus, the
types \(\PCFTypes\) and \(\WW_{A,B}\) are equivalent.
\end{example}

We now generalise \(\WW\)-types to indexed \(\WW\)-types which encode
inductively defined \emph{families} over some type \(I\).
\begin{definition}[Indexed \(\WW\)-type, \(\WW_{s,t}\), \(\isup\)]%
  \label{def:indexed-W-type}%
  \index{W-type@\(\WW\)-type!indexed|textbf}%
  Let $A$ and $I$ be types and let $B$ be a type family over $A$. Suppose we
  have $t : A \to I$ and $s : \pa*{\Sigma_{a : A} B(a)} \to I$. The
  \emph{indexed \(\WW\)-type \(\WW_{s,t}\)} specified by $s$~and~$t$ is the
  inductive type family over $I$ generated by the constructor
  \[
    \isup : \prod_{a : A}\pa*{\pa*{\Pi_{b : B(a)}\,\WW_{s,t}(s(a,b))} \to
      \WW_{s,t}(t(a))}.
  \]
  \nomenclature[W']{$\WW_{s,t}$}{indexed \(\WW\)-type over \(I\) with parameters
    \({t : A \to I}\) and \({s : \pa*{\Sigma_{a : A}B(a)} \to I}\)}
  If we have \(I : \U\), \(A : \V\) and \(B : A \to \W\), then we assume
  \(\WW_{s,t}(i) : \V \sqcup \W\) for all \(i : I\).
\end{definition}
\begin{remark}[The induction principle of an indexed \(\WW\)-type]
  We spell out the induction principle for indexed \(\WW\)-types. If we have
  $Y : \Pi_{i : I}\pa*{\WW_{s,t}(i) \to \mathcal T}$, then to prove
  $\Pi_{i : I}\Pi_{w : \WW_{s,t}(i)} Y(i,w)$, it suffices to show that for
  any $a : A$ and $f : \Pi_{b:B(a)}\WW_{s,t}(s(a,b))$ satisfying
  $Y(s(a,b),f(b))$ for every $b : B(a)$ (the \emph{induction hypothesis}), we
  have a term of type $Y(t(a),\isup(a,f))$.
\end{remark}

That this is indeed a generalisation of \(\WW\)-types is witnessed by the
following result.

\begin{proposition}\label{nonindexed-as-indexed-over-One}
  Every \(\WW\)-type is equivalent to an indexed \(\WW\)-type over \(\One\).
\end{proposition}
\begin{proof}
  Given a \(\WW\)-type with parameters \(A\) and \(B\), we define the functions
  \(t : A \to \One\) and \({s : \pa*{\Sigma_{a : A}B(a)} \to \One}\) to be the
  unique maps to \(\One\). It is then not hard to see that
  \(\WW_{A,B} \simeq \WW_{s,t}(\star) \simeq \Sigma_{u : \One}\WW_{s,t}(u)\).
\end{proof}

As with ordinary \(\WW\)-types, we think of \(A\) as a type of (labels of)
constructors and of \(B\) as encoding the arity of the constructors. But in the
indexed case each constructor \(a : A\) has a sort given by \(t(a) : I\) and
the arguments \(b : B(a)\) to a constructor \(a : A\) also have sorts given by
\(s(a,b)\).%
\footnote{The names \(s\) and \(t\) stand for source and target respectively.}
Again, it is illuminating to look at an example.

\begin{example}[A subset of PCF as an indexed \(\WW\)-type]\label{W-PCF-terms}%
  \index{W-type@\(\WW\)-type!indexed!encoding PCF terms|textbf}%
  In this example, we define the terms of a very basic typed programming
  language, using the PCF types of \cref{W-PCF-types}. This will be a subset of
  the PCF programming language studied in~\cref{sec:PCF}.

  The type family $\mathsf{T} : \PCFTypes \to \U_0$ is inductively defined as:
  \begin{enumerate}[(i)]
  \item $\mathsf{zero}$ is a term of type $\iota$ (i.e.\ \(\zeroo : \mathsf{T}(\iota)\));
  \item $\mathsf{succ}$ is a term of type $\iota \Rightarrow \iota$;
  \item for every two PCF types $\sigma$ and $\tau$, given a term \(s\) of type
    \(\sigma \Rightarrow \tau\) and a term \(t\) of type \(\sigma\), we have a
    term of type \(\tau\) denoted by juxtaposition \(st\) and called \emph{\(s\)
      applied to \(t\)}.
  \end{enumerate}
  Besides application our only other terms are \(\mathsf{zero}\) and
  \(\mathsf{succ}\), so the fact that our application defined for general types
  doesn't get us very much, but it helps to illustrate indexed
  \(\WW\)-types. For the indexed \(\WW\)-type, take $I$ to be the type of PCF
  types and put $A \colonequiv \Two + (I \times I)$. Define
  $B : A \to \mathcal U_0$ by
  \begin{align*}
    B(\inl(0)) &\colonequiv B(\inl(1)) \colonequiv \Zero, \text{ and} \\
    B(\inr(\sigma,\tau)) &\colonequiv \Two.
  \end{align*}
  Finally, define $t$ by
  \begin{align*}
    t(\inl(0)) &\colonequiv \iota, \\
    t(\inl(1)) &\colonequiv \iota \Rightarrow \iota, \text{ and} \\
    t(\inr(\sigma,\tau)) &\colonequiv \tau;
  \end{align*}
  and $s$ by
  \begin{align*}
    s(\inr(\sigma,\tau),0) &\colonequiv \sigma \Rightarrow \tau, \text{ and} \\
    s(\inr(\sigma,\tau),1) &\colonequiv \sigma;
  \end{align*}
  on the other elements $s$ is defined as the unique function from $\Zero$.

  It is then tedious, but straightforward to define, by induction on PCF types,
  maps between \(\mathsf{T}(\sigma)\) and \(\WW_{s,t}(\sigma)\) for every PCF
  type \(\sigma\) and prove that they are inverses, establishing that
  \(\mathsf{T}(\sigma)\) and \(\WW_{s,t}(\sigma)\) are equivalent for every PCF
  type \(\sigma\).
\end{example}

\subsection{Indexed \texorpdfstring{$\WW$}{W}-types with decidable equality}
We wish to isolate some conditions on the parameters of an indexed
$\WW$-type that are sufficient to conclude that an indexed
$\WW$-type has decidable equality. We first need a few definitions before
we can state the theorem.

\begin{definition}[\(\Pi\)-compactness; {\cite[\mkTTurl{TypeTopology.CompactTypes}]{TypeTopology}}]\label{picompact}%
  \index{Pi-compactness@\(\Pi\)-compactness}%
  A type $X$ is \emph{\(\Pi\)-compact} when every type family $Y$ over $X$
  satisfies: if $Y(x)$ is decidable for every $x : X$, then so is the dependent
  product $\Pi_{x : X} Y(x)$.
\end{definition}

\begin{example}\label{emptyunitpicompact}
  The empty type $\Zero$ is vacuously $\Pi$-compact. The unit type $\One$ is
  also easily seen to be $\Pi$-compact. There are also interesting examples of
  infinite types that are $\Pi$-compact, such as $\Nat_{\infty}$, the one-point
  compactification of the natural numbers~\cite{Escardo2019Types}.
\end{example}

\begin{lemma}[cf.~{\cite[\mkTTurl{TypeTopology.CompactTypes}]{TypeTopology}}]%
  \label{coproductpicompact}
  The $\Pi$-compact types are closed under binary coproducts.
\end{lemma}
\begin{proof}
  Let $X$ and $Y$ be $\Pi$-compact types. Suppose $F$ is a type family over
  $X + Y$ such that $F(z)$ is decidable for every $z : X + Y$. We must show that
  $\Pi_{z : X + Y}F(z)$ is decidable.

  Define the type family \(F_X\) over \(X\) by \(F_X(x) \colonequiv F(\inl(x))\)
  and the type family \(F_Y\) over \(Y\) by \(F_Y(y) \colonequiv F(\inr(y))\). %
  By our assumption on $F$, the types $F_X(x)$ and $F_Y(y)$ are decidable for
  every $x : X$ and $y : Y$. Hence, since $X$ and $Y$ are assumed to be
  $\Pi$-compact, the dependent products $\Pi_{x : X}F_X(x)$ and
  $\Pi_{y : Y}F_Y(y)$ are decidable.

  Finally, $\Pi_{z : X + Y}F(z)$ is logically equivalent to
  $\pa*{\Pi_{x : X}F_X(x)} \times \pa*{\Pi_{y : Y}F_Y(y)}$. Since the product of
  two decidable types is again decidable, an application of~\cref{decidable-iff}
  now finishes the proof.
\end{proof}

We are now in position to state the general theorem about decidable equality on
indexed $\WW$-types. We will use this theorem in~\cref{sec:Scott-model-of-PCF}
to prove that the syntax of the typed programming language PCF has decidable
equality.

\begin{theorem}
  \label{indexedWtypedeceq'}%
  \label{decidability!of equality}%
  An indexed \(\WW\)-type \(\WW_{s,t}\) specified by parameters \(s : A \to I\)
  and \({t : \pa*{\Sigma_{a : A}B(a)} \to I}\) has decidable equality at every
  \(i : I\) if
  \begin{enumerate}[(i)]
  \item \(I\) is a set,
  \item \(A\) has decidable equality, and
  \item \(B(a)\) is \(\Pi\)-compact for every \(a : A\).
  \end{enumerate}
\end{theorem}

\begin{corollary}\label{Wtypedeceq}
  A \(\WW\)-type \(\WW_{A,B}\) specified by a type \(A\) and a type family \(B\)
  over \(A\) has decidable equality if \(A\) has decidable equality and \(B(a)\)
  is \(\Pi\)-compact for every \(a : A\).
\end{corollary}
\begin{proof}
  By \cref{nonindexed-as-indexed-over-One} the type \(\WW_{A,B}\) is an indexed
  \(\WW\)-type over \(I \colonequiv \One\) which is a set.
\end{proof}

The proof of~\cref{indexedWtypedeceq'} is quite technical, so we postpone it
until~\cref{proofofindexedWtypedeceq}. Instead, we next describe how to apply
the theorem to prove that the PCF types and the type family
from~\cref{W-PCF-terms} have decidable equality.

\begin{proposition}\label{PCF-terms-have-decidable-equality}
  The type family from~\cref{W-PCF-terms} has decidable equality.
\end{proposition}
\begin{proof}
  We apply \cref{indexedWtypedeceq'} with parameters \(t : A \to I\) and
  \({s : \pa*{\Sigma_{a : A}B(a)} \to I}\) for the indexed \(\WW\)-type defined
  in~\cref{W-PCF-terms}. By using \cref{emptyunitpicompact} and
  \cref{coproductpicompact} we see that $B(a)$ is $\Pi$-compact for every
  $a : A$. Further, note that $A \equiv I$ has decidable equality if $I$
  does. So it remains to prove that $I$, the type of PCF types, has decidable
  equality.
  But this follows from \cref{Wtypedeceq}, because by \cref{W-PCF-types} we
  can encode PCF types as a (nonindexed) \(\WW\)-type with parameters
  \(A \colonequiv \Two\) and \(B : A \to \mathcal U_0\) given by
  $B(0) \colonequiv \Zero$ and $B(1) \colonequiv \Two$,
  and \(A\) has decidable equality while \(B(a)\) is \(\Pi\)-compact for every
  \(a : A\) because of \cref{emptyunitpicompact} and \cref{coproductpicompact}.
\end{proof}

\subsection{Proving indexed \texorpdfstring{\(\WW\)}{W}-types to have decidable equality}
\label{proofofindexedWtypedeceq}

In this section we prove \cref{indexedWtypedeceq'} by deriving it as a corollary
of another result, namely \cref{indexedWtypedeceq} below. This result seems to
have been first established by Jasper Hugunin, who reported on it in a post on
the \emph{Homotopy Type Theory} mailing list~\cite[]{HuguninMail2017}. Our proof
of \cref{indexedWtypedeceq} is a simplified written-up account of Hugunin's
\Coq~code \cite[\texttt{FiberProperties.v}]{Hugunin2017}.

\begin{theorem}
  [Jasper Hugunin]
  \label{indexedWtypedeceq}
  An indexed \(\WW\)-type \(\WW_{s,t}\) specified by \(s : A \to I\) and
  \(t : \pa*{\Sigma_{a : A}B(a)} \to I\) has decidable equality at every
  \(i : I\) if
  \begin{enumerate}[(i)]
  \item \(B(a)\) is \(\Pi\)-compact for every \(a : A\), and
  \item the fibres of \(t\) at every \(i : I\) all have decidable equality.
  \end{enumerate}
\end{theorem}

Let us see how to obtain \cref{indexedWtypedeceq'} from \cref{indexedWtypedeceq}.
\begin{proof}[Proof of \cref{indexedWtypedeceq'} (using \cref{indexedWtypedeceq})]
  Suppose that $A$ has decidable equality and $I$ is a set. We are to show that
  the fibre of $t$ over $i$ has decidable equality for every $i : I$. So suppose
  we have $(a,p)$ and $(a',p')$ in the fibre of $t$ over $i$. Since $A$ has
  decidable equality, we can decide whether $a$ and $a'$ are equal or not. If
  they are not, then certainly $(a,p) \neq (a',p')$. If they are, then we claim
  that the dependent pairs $(a,p)$ and $(a',p')$ are also equal. If $e : a = a'$
  is the supposed equality, then it suffices to show that
  $\transport^{\lambda x : A.t(x) = i}(e,p) = p'$, but both these elements are
  identifications in $I$ and $I$ is a set, so they must be equal.
\end{proof}

We now embark on a proof of \cref{indexedWtypedeceq}. For the remainder of this
section, let us fix types $A$ and $I$, a type family $B$ over $A$ and maps
$t : A \to I$ and $s : \pa*{\Sigma_{a : A} B(a)} \to I$.

We do not prove the theorem directly. The statement makes it impossible to
assume two elements $u,v : \WW_{s,t}(i)$ and proceed by induction on \emph{both}
$u$ and $v$. Instead, we will state and prove a more general result that is
amenable to a proof by induction. But first, we need additional general lemmas
and some definitions.

\begin{lemma}
  \label{rightpairinj}
  If $X$ is a set, then the right pair function of any type family \(Y\) over
  \(X\) is left-cancellable, in the following sense: if $(x,y) = (x,y')$ as
  elements of $\Sigma_{a : X}Y(a)$, then $y = y'$.
\end{lemma}
\begin{proof}
  Suppose $X$ is a set, $x : X$ and $y,y' : Y(x)$ with $e : (x,y) =
  (x,y')$. From $e$, we obtain $e_1 : x = x$ and
  $e_2 : \transport^{Y}(e_1,y) = y'$. Since $X$ is a set, we must have that
  $e_1 = \refl_{x}$, so that \(e_2\) yields
  $y \equiv \transport^Y(\refl_{x},y) = y'$, as desired.
\end{proof}

\begin{definition}[Subtrees, \(\sub_i\)]%
  \index{subtree}%
  For each $i : I$, define
  \[
    \sub_i : \WW_{s,t} (i) \to \sum_{p:\fib_t(i)}
    \prod_{b : B(\fst(p))}\WW_{s,t}( s(\fst(p),b))
  \]
  \nomenclature[sub]{$\sub_i$}{subtree of an indexed \(\WW\)-type at index \(i\)}
  by induction as
  \[
    \sub_{t(a)}(\isup(a,f)) \colonequiv
    \pa*{\pa*{a,\refl_{t(a)}},f}.
  \]
  For notational convenience, we will omit the subscript of $\sub$.
\end{definition}

The name \(\sub\) comes from \emph{sub}trees, thinking of the elements of a
\(\WW\)-type as a well-founded trees.

\begin{lemma}
  \label{eqtosubtreeeq}
  If the fibre of $t$ over $i$ has decidable equality for every $i : I$, then
  $\isup(a,f) = \isup(a,g)$ implies $f = g$ for every $a : A$ and
  $f,g : \Pi_{b : B(a)} \WW_{s,t}(s(a,b))$.
\end{lemma}
\begin{proof}
  Suppose $\isup(a,f) = \isup(a,g)$. Then
  \[
    \pa*{\pa*{a,\refl_{t(a)}},f} \equiv
    \sub(\isup(a,f)) =
    \sub(\isup(a,g)) \equiv
    \pa*{\pa*{a,\refl_{t(a)}},g}.
  \]
  As \(\fib_t(i)\) is assumed to be decidable, it is a set by
  Hedberg's~theorem~(\cref{Hedberg-Theorem}). Hence \(f = g\) by
  \cref{rightpairinj}.
\end{proof}

\begin{definition}[\(\tofib_i\)]\label{tofib}%
  \index{fibre}%
  For every $i : I$, we define $\tofib_i : \WW_{s,t}(i) \to \fib_t(i)$
  inductively by
  \[
    \tofib_{t(a)}(\isup(a,f)) \colonequiv
    (a,\refl_{t(a)}).
  \]
  \nomenclature[tofib]{$\tofib_i$}{map from \(\WW_{s,t}(i)\) to the fibre of \(t\) at \(i\)}
  In future use, we omit the subscript of $\tofib$.
\end{definition}
\begin{lemma}\label{tofibtransport}
  For $i,j : I$ with a path $p : i = j$ and $w : \WW_{s,t}(i)$, we have
  the following equality:
  \[
    \tofib(\transport^{\WW_{s,t}}(p,w)) =
    (\fst(\tofib(w)),\snd(\tofib(w)) \pathcomp p).
  \]
\end{lemma}
\begin{proof}
  By path induction on $p$.
\end{proof}

We are now in position to state and prove the lemma from which we will derive
\cref{indexedWtypedeceq}.
\begin{lemma}\label{indexedWtypedeceqtransport}
  Suppose that $B(a)$ is $\Pi$-compact for every $a :A $ and that the fibre of
  $t$ over each $i : I$ has decidable equality. For any $i : I$,
  $u : \WW_{s,t}(i)$, $j : I$, path $p : i = j$ and
  $v : \WW_{s,t}(j)$, the type
  \[
    \transport^{\WW_{s,t}}(p,u) = v
  \]
  is decidable.
\end{lemma}
\begin{proof}
  Suppose $i : I$ and $u : \WW_{s,t}(i)$. We proceed by induction on $u$
  and so we assume that $u \equiv \isup(a,f)$. The induction
  hypothesis reads:
  \begin{equation*}\label{IH}
    \prod_{b : B(a)}\prod_{j' : I}\prod_{p' : s(a,b) = j'}\prod_{v' :
      \WW_{s,t}(j')} (\transport^{\WW_{s,t}}(p',f(b))=v')
    \text{ is decidable }. \tag{$\ast$}
  \end{equation*}
  Suppose we have $j : I$ with path $p : t(a) = j$ and $v : \WW_{s,t}(j)$.
  By induction, we may assume that $v \equiv \isup(a',f')$. We are
  tasked to show that
  \begin{equation*}\label{indexedsupdec}
    \transport^{\WW_{s,t}}(p, \isup(a,f)) =
    \isup(a',f') \tag{$\dagger$}
  \end{equation*}
  is decidable, where $p : t(a) = t(a')$.

  By assumption the fibre of $t$ over $t(a')$ has decidable equality. Hence, we
  can decide if $\pa*{a',\refl_{t(a')}}$ and $(a,p)$ are equal or
  not. Suppose first that the pairs are not equal. We claim that in this case
  $\lnot \eqref{indexedsupdec}$. For suppose we had $e : \eqref{indexedsupdec}$,
  then
  \[
    \ap_{\tofib}(e) :
    \tofib(\transport^{\WW_{s,t}}(p,
    \isup(a,f))) = \tofib(\isup(a',f')).
  \]
  By definition, the right hand side is $(a',\refl_{t(a')})$. By
  \cref{tofibtransport}, the left hand side is equal to
  $(a, \refl_{t(a)} \pathcomp p)$ which is in turn equal to $(a,p)$,
  contradicting our assumption that $\pa*{a',\refl_{t(a')}}$ and $(a,p)$
  were not equal.

  Now suppose that $\pa*{a',\refl_{t(a')}} = (a,p)$. From this, we
  obtain paths $e_1 : a' = a$ and
  $e_2~:~\transport^{\lambda x : A.t(x) =
    t(a')}\pa*{e_1,\refl_{t(a')}} = p$. By path induction, we may assume
  $e_1 \equiv \refl_{a'}$, so that from $e_2$ we obtain an identification
  \[
    \rho : \refl_{t(a')} = p.
  \]
  Using this identification, we see that the left hand side of
  \eqref{indexedsupdec} is equal to $\isup(a',f)$, so we are left to show that
  \[
    \isup(a',f) = \isup(a',f')
  \]
  is decidable.
  By induction hypothesis \eqref{IH} and the fact that $a \equiv a'$, the type
  $ f(b) = f'(b) $ is decidable for every $b : B(a')$. Since $B(a')$ is
  $\Pi$-compact, this implies that the type
  $\pa*{\Pi_{b : B(a')}\,f(b) = f'(b)}$ is decidable.
  Suppose first that $\Pi_{b : B(a')}\,f(b) = f'(b)$. Function extensionality
  then yields $f = f'$, so that
  $\isup(a',f) = \isup(a',f')$.
  On the other hand, suppose $\lnot\pa*{\Pi_{b : B(a')}\,f(b) = f'(b)}$. We
  claim that the elements $\isup(a',f)$ and $\isup(a',f')$ cannot be equal. For
  if they were equal, then \cref{eqtosubtreeeq} would yield $f = f'$,
  contradicting our assumption that $\lnot\pa*{\Pi_{b : B(a)}\,f(b) = f'(b)}$,
  which finishes the proof.
\end{proof}

The proof of~\cref{indexedWtypedeceq} now follows readily.
\begin{proof}[Proof of \cref{indexedWtypedeceq}]
  Let $i : I$ and $u,v : \WW_{s,t}(i)$. Taking $j \colonequiv i$ and
  $p \colonequiv \refl_i$ in \cref{indexedWtypedeceqtransport}, we see
  that $u = v$ is decidable, as desired.
\end{proof}

\section{Notes}

Our discussion of type universes in \cref{sec:type-universes} closely follows
that of \cite[Section~2.1]{Escardo2021}.
Our treatment of univalent foundations in Sections
\ref{sec:id-and-fun-ext}--\ref{sec:univalent-universes} is our own, but based on
the expositions in \cite{HoTTBook} and \cite{Escardo2019}.
The notion of a (locally) \(\U\)-small type appears in~\cite{Rijke2017}, but the
lemmas in \cref{sec:small-and-locally-small-types} involving retracts are our
original results and were included in our paper~\cite{deJongEscardo2021b} and
its extended version~\cite{deJongEscardo2022}.

\cref{sec:prop-trunc-resizing} closely follows the exposition in
\cite[Section~3.34.1]{Escardo2019}, while
\cref{sec:set-quotients-from-propositional-truncations} is based on
\cite[Section~3.37]{Escardo2019}.
Both \cref{sec:propositional-truncations-from-set-quotients,sec:set-replacement}
are original
contributions. \cref{sec:quotients-replacement-prop-trunc-revisited} as a whole
was included in our work \cite{deJongEscardo2021b,deJongEscardo2022}.

Finally, \(\WW\)-types, studied in \cref{sec:indexed-W-types}, were
introduced by Per Martin-L\"of~\cite{MartinLof1984} and the main theorem
presented in that section is, as mentioned before, due to Jasper
Hugunin~\cite{Hugunin2017,HuguninMail2017}.
This result was also included in our paper~\cite{deJong2021a} with an
application to semidecidability questions pertaining to the Scott model of
PCF~(as discussed in~\cref{sec:Scott-model-of-PCF} of this thesis).


%% file: mainmatter/basic-domain-theory.tex
\chapter{Basic domain theory}\label{chap:basic-domain-theory}

Domain theory~\cite{AbramskyJung1994} is a well-established subject in
mathematics and theoretical computer science with applications to programming
language semantics~\cite{Scott1972,Scott1993,Plotkin1977}, higher-type
computability~\cite{LongleyNormann2015}, topology, and
more~\cite{GierzEtAl2003}.
In this chapter we introduce basic domain theory within the context of
constructive and predicative univalent foundations. Specifically, we discuss
\begin{description}
\item[\cref{sec:directed-complete-posets}:] (pointed) dcpos: the basic objects
  of domain theory,
\item[\cref{sec:Scott-continuous-maps}:] Scott continuous maps: morphism between
  dcpos,
\item[\cref{sec:lifting}:] the lifting of a set and of a dcpo to get the free
  pointed dcpo,
\item[\cref{sec:products-and-exponentials}:] products and exponentials of dcpos, and
\item[\cref{sec:bilimits}:] bilimits of dcpos.
\end{description}
The basic theory will find application in the semantics of programming
languages, as laid out in \cref{chap:applications}.
We offer the following overture in preparation of our development, especially if
the reader is familiar with domain theory in a classical, set-theoretic setting.

\section{Introduction}\label{sec:basic-domain-theory-introduction}

The basic object of study in domain theory is that of a \emph{directed complete
  poset} (dcpo).
In (impredicative) set-theoretic foundations, a dcpo can be defined to be a
poset that has least upper bounds of all directed subsets.
A naive translation of this to our foundation would be to proceed as
follows. Define a poset in a universe \(\U\) to be a type \(P:\U\) with a
reflexive, transitive and antisymmetric relation
\(-\below- : P \times P \to \U\).
Since we wish to consider posets and not categories we require that the values
\(p \below q\) of the order relation are \emph{subsingletons}.
Then we could say that the poset \((P,\below)\) is \emph{directed complete} if
every directed family \(I \to P\) with indexing type \(I : \U\) has a least
upper bound (supremum). The problem with this definition is that there are no
interesting examples in our constructive and predicative setting.
For instance, assume that the poset~$\Two$ with two elements \(0\below 1\) is
directed complete, and consider a proposition~\(A:\U\) and the directed family
\(A + \One \to \Two\) that maps the left component to~\(0\) and the right
component to~\(1\). By case analysis on its hypothetical
supremum~(\cref{def:supremum}), we conclude that the negation of \(A\) is
decidable. This amounts to weak excluded middle~(\cref{def:(w)em}) which is
constructively unacceptable.

To try to get an example, we may move to the poset \(\Omega_{\U_0}\) of
propositions in the universe \(\U_0\), ordered by implication. This poset does
have all suprema of families \(I \to \Omega_{\U_0}\) indexed by types \(I\) in
the \emph{first universe} \(\U_0\), given by existential quantification. But if
we consider a directed family \(I \to \Omega_{\U_0}\) with \(I\) in the
\emph{same universe} as \(\Omega_{\U_0}\) lives, namely the \emph{second
  universe} \(\U_1\), existential quantification gives a proposition in the
\emph{second universe} \(\U_1\) and so doesn't give its supremum. In this
example, we get a poset such that
\begin{enumerate}[(i)]
\item the carrier lives in the universe \(\U_1\),
\item the order has truth values in the universe \(\U_0\), and
\item suprema of directed families indexed by types in \(\U_0\) exist.
\end{enumerate}

Regarding a poset as a category in the usual way, we have a large, but locally
small, category with small filtered colimits (directed suprema). This is typical
of all the concrete examples that we will consider, such as the dcpos in the
Scott model of PCF (\cref{sec:Scott-model-of-PCF}) and Scott's \(D_\infty\)
model of the untyped \(\lambda\)-calculus (\cref{sec:Scott-D-infty}).
We may say that the predicativity restriction increases
the universe usage by one.  However, for the sake of generality, we formulate
our definition of dcpo with the following universe conventions:
\begin{enumerate}[(i)]
\item the carrier lives in a universe \(\U\),
\item the order has truth values in a universe \(\T\), and
\item suprema of directed families indexed by types in a universe \(\V\) exist.
\end{enumerate}
So our notion of dcpo has three universe parameters \(\U,\V\) and \(\T\). We
will say that the dcpo is \emph{locally small} when \(\T\) is not necessarily
the same as \(\V\), but the order has \(\V\)-small (recall~\cref{def:smallness})
truth values. Most of the time we mention \(\V\) explicitly and leave \(\U\) and
\(\T\) to be understood from the context.

\section{Directed complete posets}\label{sec:directed-complete-posets}

We introduce the basic object of domain theory: a directed complete poset. We
carefully explain our use of the propositional truncation in our definitions
and, as mentioned above, the type universes involved.

\begin{definition}[Preorder, reflexivity, transitivity]
  A \emph{preorder} \((P,\sqsubseteq)\) is a type \(P : \U \) together with a
  proposition-valued binary relation \({\sqsubseteq} : {P \to P \to \Omega_\T}\)
  satisfying%
  \index{preorder|textbf}%
  \nomenclature[sqsubseteq]{$x \below y$}{order relation of a preorder}
  \begin{enumerate}[(i)]
  \item \emph{reflexivity}: for every \(p : P\), we have \(p \below p\), and%
    \index{reflexivity|textbf}%
  \item \emph{transitivity}: for every \(p,q,r : P\), if \(p \below q\) and
    \(q \below r\), then \(p \below r\).%
    \index{transitivity}%
    \qedhere
  \end{enumerate}
\end{definition}

\begin{definition}[Poset, antisymmetry]
  A \emph{poset} is a preorder \((P,\below)\) that is \emph{antisymmetric}: if
  \(p \below q\) and \(q \below p\), then \(p = q\) for every \(p,q : P\).%
  \index{poset}\index{antisymmetry}
\end{definition}

\begin{lemma}\label{posets-are-sets}
  If \((P,\below)\) is a poset, then \(P\) is a set.
\end{lemma}
\begin{proof}
  For every \(p,q : P\), the composite
  \[
    \pa{p = q} \xrightarrow{\text{by reflexivity}}
    {\pa{p \below q} \times \pa{q \below p}} \xrightarrow{\text{by antisymmetry}}
    \pa{p = q}
  \]
  is constant, since \({\pa{p \below q} \times \pa{q \below p}}\) is a
  proposition. Hence, by \cref{Hedberg-Lemma} it follows that \(P\) must be a
  set.
\end{proof}

From now on, we will simply write ``let \(P\) be a poset'' leaving the partial
order \(\below\) implicit. We will often use the symbol \({\below}\) for partial
orders on different carriers when it is clear from the context which one it
refers to.

\begin{definition}[(Semi)directed family]
  A family \(\alpha : I \to P\) of elements of a poset \(P\) is
  \emph{semidirected} if whenever we have \(i,j : I\), there exists some
  \(k : I\) such that \(\alpha_i \below \alpha_k\) and
  \(\alpha_j \below \alpha_k\).\index{semidirectedness}
  We frequently use the shorthand \({\alpha_i,\alpha_j} \below \alpha_k\) to
  denote the latter requirement.
  Such a family is \emph{directed} if it is semidirected and its domain \(I\) is
  inhabited.%
  \index{directedness}
\end{definition}

The name ``semidirected'' matches Taylor's terminology~\cite[Definition~3.4.1]{Taylor1999}.

\begin{remark}\index{propositional truncation}%
  Note our use of the propositional truncation in defining when a family is
  \emph{directed}. To make this explicit, we write out the definition in
  type-theoretic syntax: a family \(\alpha : I \to P\) is directed if
  \begin{enumerate}[(i)]
  \item\label{dir-inh} we have an element of \(\squash{I}\), and
  \item\label{dir-semidir}
    \(\Pi_{i,j : I} \squash*{\Sigma_{k : I}\pa*{\alpha_i \below \alpha_k} \times
      \pa*{\alpha_j \below \alpha_k}}\).
  \end{enumerate}
  The use of the propositional truncation ensures \cref{dir-inh} and
  \cref{dir-semidir} are propositions and hence that being (semi)directed is a
  property of a family.
  \cref{dir-semidir} without truncating would be asking us to assign a chosen
  \(k : I\) for every \(i,j : I\) instead.
\end{remark}

Following~Scott~\cite{Scott1970}, we sometimes think of the elements of \(P\) as
pieces of information and \(p \below q\) as expressing that \(q\) contains
more information or refines \(p\). With this viewpoint, a directed family
is a collection of pieces of information that are consistent in the sense
that any two pieces of information can be refined to a third one that is a
member of the collection.
In the next definition we ask for such families to have least upper bounds,
which is like saying that such consistent collections of information can be
patched together to a piece of information that refines everything in the
family.

\begin{definition}[(Least) upper bound, supremum]\label{def:supremum}
  An element \(x\) of a poset \(P\) is an \emph{upper bound} of a family
  \(\alpha : I \to P\) if \(\alpha_i \below x\) for every \(i : I\).%
  \index{upper bound}
  It is a \emph{least upper bound} of~\(\alpha\) if it is an upper bound, and
  whenever \(y : P\) is an upper bound of \(\alpha\), then \(x \below y\).
  \index{upper bound!least}
  By antisymmetry, a least upper bound is unique if it exists, so in this case
  we will speak of \emph{the} least upper bound of \(\alpha\), or sometimes the
  \emph{supremum} of \(\alpha\).%
  \index{supremum!directed}
\end{definition}

\begin{definition}[\(\V \)-directed complete poset, \(\V\)-dcpo, %
  \(\bigsqcup \alpha\), \(\bigsqcup_{i : I}\alpha_i\)]
  For a universe \(\V\), a \emph{\(\V \)-directed complete poset} (or
  \emph{\(\V \)-dcpo}, for short) is a poset \(D\) such that every directed
  family \(\alpha : I \to D\) with \(I : \V \) has a supremum in \(D\) that we
  denote by \(\bigsqcup \alpha\) or \(\bigsqcup_{i : I} \alpha_i\).%
  \index{directed complete poset|see {dcpo}}%
  \index{poset!directed complete|see {dcpo}}%
  \index{dcpo|textbf}
  \nomenclature[sqcup]{$\bigsqcup \alpha$}{supremum of a directed family \(\alpha\)}
\end{definition}

\begin{remark}\label{directed-completeness-is-prop}
  Explicitly, we ask for an element of the type
  \[
    \Pi_{I : \V}\Pi_{\alpha : I \to D}\pa*{\operatorname{is-directed} \alpha \to
      \Sigma_{x : D}\pa*{x \mathrel{\operatorname{is-sup-of}} \alpha}},
  \]
  where \(\pa*{x \mathrel{\operatorname{is-sup-of}} \alpha}\) is the type expressing
  that \(x\) is the supremum of \(\alpha\).
  Even though we used \(\Sigma\) and not \(\exists\) in this expression, this
  type is still a proposition: By \cref{Pi-is-prop}, it suffices to prove that
  the type \(\Sigma_{x : D}(x \mathrel{\operatorname{is-sup-of}} \alpha)\) is a
  proposition. So suppose that we have \(x,y : D\) with
  \(p : x \mathrel{\operatorname{is-sup-of}} \alpha\) and
  \(q : y \mathrel{\operatorname{is-sup-of}} \alpha\). Being the supremum of a
  family is a property because the partial order is proposition-valued. Hence,
  by \cref{prop-subtype}, to prove that \((x,p) = (y,q)\), it suffices to prove
  that \(x = y\). But this follows from antisymmetry and the fact that \(x\) and
  \(y\) are both suprema of \(\alpha\).
\end{remark}

We will sometimes leave the universe \(\V \) implicit, and simply speak of a
dcpo. On other occasions, we need to carefully keep track of universe levels. To
this end, we make the following definition.
\begin{definition}[\(\DCPO{V}{U}{T}\)]
  Let \(\V\), \(\U\) and \(\T \) be universes. We write \(\DCPO{V}{U}{T}\) for
  the type of \(\V \)-dcpos with carrier in \(\U \) and order taking values in
  \(\T \).
  \nomenclature[DCPO]{$\DCPO{V}{U}{T}$}{type of \(\V\)-dcpos with carriers in
    \(\U\) and orders taking values in \(\T\)}
\end{definition}

\begin{remark}\label{universe-levels-of-lifting-and-exponentials}%
  \index{universe!parameters}%
  In particular, it is very important to keep track of the universe parameters
  of the lifting~(\cref{sec:lifting}) and of
  exponentials~(\cref{sec:products-and-exponentials}) in order to ensure that it
  is possible to construct the Scott model of PCF and Scott's \(D_\infty\) in
  our predicative setting, as we do in~\cref{chap:applications}.
\end{remark}

In many examples and applications, we require our dcpos to have a least element.

\begin{definition}[Pointed dcpo]
  A dcpo \(D\) is \emph{pointed} if it has a least element which we will denote
  by \(\bot_{D}\), or simply \(\bot\).%
  \index{dcpo!pointed}%
  \nomenclature[bot]{$\bot$}{least element of a poset}
\end{definition}

\begin{definition}[Local smallness]\label{def:local-smallness}
  A \(\V\)-dcpo \(D\) is \emph{locally small} if \(x \below y\) is \(\V\)-small
  (recall~\cref{def:smallness}) for every \(x,y : D\).%
  \index{dcpo!locally small|textbf}%
\end{definition}

\begin{lemma}\label{local-smallness-alt}
  A \(\V\)-dcpo \(D\) is locally small if and only if we have
  \({\below_{\V}} : D \to D \to \V\) such that \(x \below y\) holds
  precisely when \(x \below_{\V} y\) does.
\end{lemma}
\begin{proof}
  The \(\V\)-dcpo \(D\) is locally small exactly when we have an element of
  \[
    \Pi_{x,y : D}\Sigma_{T : \V}\pa*{T \simeq {x \below y}}.
  \]
  But this type is equivalent to
  \[
    \Sigma_{R : {D \to D \to \V}}\Pi_{x,y : D}\pa*{R(x,y) \simeq {x \below y}}
  \]
  by~\cref{Pi-Sigma-distr}.
\end{proof}

Nearly all examples of \(\V\)-dcpos in this thesis will be locally small. We now
introduce two fundamental examples of dcpos: the type of subsingletons and
powersets.

\begin{example}[The type of subsingletons as a pointed dcpo]\label{Omega-as-pointed-dcpo}%
  \index{dcpo!of subsingletons}%
  For any type universe~\(\V\), the type \(\Omega_{\V}\) of subsingletons in
  \(\V\) is a poset if we order the propositions by implication.
  Note that antisymmetry holds precisely because of propositional extensionality
  (\cref{def:prop-ext}).
  Moreover, \(\Omega_{\V}\) has a least element, namely \(\Zero_{\V}\), the
  empty type in \(\V\).
  We also claim that \(\Omega_{\V}\) has suprema for all (not necessarily
  directed) families \(\alpha : I \to \Omega_{\V}\) with \(I : \V\).
  Given such a family \(\alpha\), its least upper bound is given by
  \(\exists_{i : I}\,\alpha_i\). It is clear that this is indeed an upper bound
  for \(\alpha\). And if \(P\) is a subsingleton such that \(\alpha_i \below P\)
  for every \(i : I\), then to show that
  \(\pa*{\exists_{i : I}\,\alpha_i} \to P\) it suffices to construct to
  construct a map \(\pa*{\Sigma_{i : I}\,\alpha_i} \to P\) as \(P\) is a
  subsingleton. But this is easy because we assumed that \(\alpha_i \below P\)
  for every \(i : I\).
  Finally, paying attention to the universe levels we observe that
  \(\Omega_{\V} : \DCPO{V}{V^+}{V}\).
\end{example}

\begin{example}[Powersets as pointed dcpos]\label{powersets-as-pointed-dcpos}%
  \index{dcpo!of subsets}%
  Recalling our treatment of subset and powersets from
  \cref{sec:subsets-and-powersets}, we show that powersets give examples of
  pointed dcpos.
  Specifically, for every type \(X : \U\) and every type universe \(\V\), the
  subset inclusion \(\subseteq\) makes \(\powerset_{\V}(X)\) into a poset, where
  antisymmetry holds by function extensionality and propositional
  extensionality.
  Moreover, \(\powerset_{\V}(X)\) has a least element of course: the empty set
  \(\emptyset\).
  We also claim that \(\powerset_{\V}(X)\) has suprema for all (not necessarily
  directed) families \(\alpha : I \to \powerset_{\V}(X)\) with \(I : \V\).
  Given such a family \(\alpha\), its least upper bound is given by
  \(\bigcup \alpha \colonequiv \lambdadot{x}{\exists_{i : I}\,x\in\alpha_i}\),
  the set-theoretic union, which is well-defined as
  \(\pa*{\exists_{i : I}\,x\in\alpha_i} : \V\).
  It is clear that this is indeed an upper bound for \(\alpha\). And if \(A\) is
  a \(\V\)-subset of \(X\) such that \(\alpha_i \subseteq A\) for every
  \(i : I\), then to show that \(\bigcup \alpha \subseteq A\) it suffices to
  construct for every \(x : X\), a map
  \(\pa*{\Sigma_{i : I}\,{x \in \alpha_i}} \to \pa*{x\in A}\) as \(x \in A\) is a
  subsingleton. But this is easy because we assumed that
  \(\alpha_i \subseteq A\) for every \(i : I\).
  Finally, paying attention to the universe levels we observe that
  \(\powerset_{\V}(X) : \DCPO{V}{V^+ \sqcup \U}{V \sqcup \U}\).
  In the case that \(X : \U \equiv \V\), we obtain the simpler, locally small
  \(\powerset_{\V}(X) : \DCPO{V}{V^+}{V}\).
\end{example}

Of course, \(\Omega_{\V}\) is easily seen to be equivalent to
\(\powerset_{\V}(\One_\V)\), so \cref{powersets-as-pointed-dcpos} subsumes
\cref{Omega-as-pointed-dcpo}, but it is instructive to understand
\cref{Omega-as-pointed-dcpo} first.

\begin{proposition}[\(\omega\)-completeness]\label{dcpo-has-sups-of-chains}%
  \index{omega-completeness@\(\omega\)-completeness|textbf}
  Every \(\V\)-dcpo \(D\) is \emph{\(\omega\)-complete}, viz.\ if we
  have elements \(x_0 \below x_1 \below x_2 \below \dots\) of \(D\), then the
  supremum of \(\pa*{x_n}_{n : \Nat}\) exists in \(D\).
\end{proposition}
\begin{proof}
  Recalling \cref{lift-to-higher-universes} and using the fact that
  \(\Nat : \U_0\), the type \(\operatorname{lift}_{\U_0,\V}(\Nat) \) is in the
  universe \(\V\) and is equivalent to \(\Nat\).
  Now
  \(\operatorname{lift}_{\U_0,\V}(\Nat) \simeq \Nat \xrightarrow{x_{(-)}} D\) is
  a directed family as \(x_n \below x_{n+1}\) for every natural number \(n\),
  and it is indexed by a type in \(\V\).
  Hence, it has a least upper bound in \(D\) which is the supremum of
  \(\pa*{x_n}_{n : \Nat}\).
\end{proof}

\section{Scott continuous maps}\label{sec:Scott-continuous-maps}

We discuss an appropriate notion of morphism between \(\V\)-dcpos, namely one
that requires preservation of directed suprema and the order
(\cref{continuous-implies-monotone}).

\begin{definition}[Scott continuity]
  A function \(f : D \to E\) between two \(\V\)-dcpos is \emph{(Scott)
    continuous} if it preserves directed suprema, i.e.\ if \(I : \V \) and
  \(\alpha : I \to D\) is directed, then \(f\pa*{\bigsqcup \alpha}\) is the
  supremum in \(E\) of the family \(f \circ \alpha\).%
  \index{Scott continuity}\index{continuity|see {Scott continuity}}
\end{definition}

\begin{remark}
  When we speak of a Scott continuous function between \(D\) and \(E\), then we
  will always assume that \(D\) and \(E\) are both \(\V\)-dcpos for some
  arbitrary but fixed type universe \(\V\).
\end{remark}

\begin{remark}\label{continuous-wrt-Scott-topology}
  The name ``Scott continuous'' is due to the fact that such maps are continuous
  with respect to the so-called Scott topology. We will not discuss the Scott
  topology in this thesis, but see~\cref{chap:conclusion} for a brief
  discussion of our work on apartness and the Scott topology in a constructive
  setting.%
  \index{topology}
\end{remark}

\begin{lemma}
  Being Scott continuous is a property. In particular, two Scott continuous maps
  are equal if and only if they are equal as functions.
\end{lemma}
\begin{proof}
  By \cref{Pi-is-prop} and the fact that being the supremum of a family is a
  property, cf.\ \cref{directed-completeness-is-prop}.
\end{proof}

\begin{lemma}\label{continuous-implies-monotone}
  If \(f : D \to E\) is Scott continuous, then it is \emph{monotone}, i.e.\
  \(x \below_{D} y\) implies \(f(x) \below_{E} f(y)\).%
  \index{monotonicity}
\end{lemma}
\begin{proof}
  Given \(x,y : D\) with \(x \below y\), consider the directed family
  \(\Two_{\V} \xrightarrow{\alpha} D\) defined by \(\alpha(0) \colonequiv x\)
  and \(\alpha(1) \colonequiv y\). Its supremum is \(y\) and \(f\) must preserve
  it. Hence, \(f(y)\) is an upper bound of \(f(\alpha(0)) \equiv f(x)\), so
  \(f(x) \below f(y)\), as we wished to show.
\end{proof}

\begin{lemma}\label{image-is-directed}
  If \(f : D \to E\) is monotone and \(\alpha : I \to D\) is directed, then so
  is \(f \circ \alpha\).
\end{lemma}
\begin{proof}
  Since \(\alpha\) is directed, \(I\) is inhabited, so it remains to prove that
  \(f \circ \alpha\) is semidirected. If we have \(i,j : I\), then by
  directedness of \(\alpha\), there exists \(k : I\) such that
  \({\alpha_i,\alpha_j}\below\alpha_k\). By monotonicity,
  we obtain \({f(\alpha_i) , f(\alpha_j)} \below f(\alpha_k)\) as desired.
\end{proof}

\begin{lemma}\label{continuity-criterion}
  A monotone map \(f : D \to E\) between \(\V\)-dcpos is Scott continuous if and
  only if \(f\pa*{\bigsqcup \alpha} \below \bigsqcup {f \circ \alpha}\).
\end{lemma}

Note that we are justified in writing \(\bigsqcup {f \circ \alpha}\) because
\cref{image-is-directed} tells us that \(f \circ \alpha\) is directed by the
assumed monotonicity of \(f\).
\begin{proof}
  The left-to-right implication is immediate. For the converse, note that it
  only remains to show that
  \(f\pa*{\bigsqcup \alpha} \aboveorder \bigsqcup{f \circ \alpha}\).
  But for this it suffices that
  \(f\pa*{\alpha_i} \below f\pa*{\bigsqcup \alpha}\) for every \(i : I\), which
  holds as \(\bigsqcup \alpha\) is an upper bound of \(\alpha\) and \(f\) is
  monotone.
\end{proof}

\begin{remark}
  In constructive mathematics it is not possible to exhibit a discontinuous
  function from \(\Nat^\Nat\) to \(\Nat\), because
  sheaf~\cite[Chapter~15]{TroelstraVanDalen1988} and realizability
  models~\cite[e.g.~Proposition~3.1.6]{vanOosten2008} imply that it is
  consistent to assume that all such functions are continuous.
  This does not mean, however, that we cannot exhibit a discontinuous function
  between dcpos. In fact, the negation map \({\lnot} : \Omega \to \Omega\) is
  not monotone and hence not continuous.
  If we were to preclude such examples, then we can no longer work with the full
  type \(\Omega\) of all propositions, but instead we must restrict to a subtype
  of propositions, for example by using dominances~\cite{Rosolini1986}.
  Indeed, this approach is investigated in the context of topos theory
  in~\cite{Phao1991,Longley1995} and for computability instead of continuity in
  univalent foundations in~\cite{EscardoKnapp2017}.
\end{remark}

\begin{definition}[Strictness]
  A Scott continuous function \(f : D \to E\) between pointed dcpos is \emph{strict}
  if \(f\pa*{\bot_{D}} = \bot_{E}\).%
  \index{strictness}
\end{definition}

\begin{lemma}\label{pointed-dcpos-sups}
  A poset \(D\) is a pointed \(\V\)-dcpo if and only if it has suprema for all
  semidirected families indexed by types in \(\V\) that we will denote using the
  \(\bigvee\) symbol.%
  \index{supremum!semidirected}%
  \nomenclature[veeb]{$\bigvee \alpha$}{supremum of a semidirected family \(\alpha\)}
  In~particular, a pointed \(\V\)-dcpo has suprema of all families indexed by
  propositions in \(\V\).%
  \index{supremum!subsingleton}

  Moreover, if \(f\) is a Scott continuous and strict map between pointed
  \(\V\)-dcpos, then \(f\) preserves suprema of semidirected families.
\end{lemma}
\begin{proof}
  If \(D\) is complete with respect to semidirected families indexed by types
  in~\(\V\), then it is clearly a \(\V\)-dcpo and it is pointed because the
  supremum of the family indexed by the empty type is the least element.
  Conversely, if \(D\) is a pointed \(\V\)-dcpo and \(\alpha : I \to D\) is a
  semidirected family with \(I : \V\), then the family
  \begin{align*}
    \hat\alpha : I + \One_{\V} &\to D \\
    \inl(i) &\mapsto \alpha_i \\
    \inr(\star) &\mapsto \bot
  \end{align*}
  is directed and hence has a supremum in \(D\) which is also the least
  upper bound of \(\alpha\).

  A pointed \(\V\)-dcpo must have suprema for all families indexed by
  propositions in~\(\V\), because any such family is semidirected.
  Finally, suppose that \(\alpha : I \to D\) is semidirected and that
  \(f : D \to E\) is Scott continuous and strict. Using the
  \(\widehat{(-)}\)-construction from above, we see that
  \begin{align*}
    f\pa*{\textstyle\bigvee \alpha}
    &\equiv f\pa*{\textstyle\bigsqcup\hat\alpha} \\
    &= \textstyle\bigsqcup f \circ \hat\alpha
    &&\text{(by Scott continuity of \(f\))} \\
    &= \textstyle\bigsqcup \widehat{f \circ \alpha}
    &&\text{(since \(f\) is strict)} \\
    &\equiv \textstyle\bigvee {f \circ \alpha},
  \end{align*}
  finishing the proof.
\end{proof}

\begin{proposition}\label{continuity-closure}%
  \hfill
  \begin{enumerate}[(i)]
  \item\label{id-is-continuous} The identity on any dcpo is Scott continuous.
  \item\label{const-is-continuous} For dcpos \(D\) and \(E\) and \(y : E\), the
    constant map \(x \mapsto y : D \to E\) is Scott continuous.
  \item\label{comp-is-continuous} If \(f : D \to E\) and \(g : E \to E'\) are
    Scott continuous, then so is \(g \circ f\).
  \end{enumerate}
  Moreover, if \(D\) is pointed, then the identity on \(D\) is strict, and if
  \(f\)~and~\(g\) are strict in \ref{comp-is-continuous}, then so is
  \(g \circ f\).
\end{proposition}
\begin{proof}
  The proofs of~\ref{id-is-continuous}~and~\ref{const-is-continuous} are
  obvious. For~\ref{comp-is-continuous}, let \(\alpha : I \to D\) be directed
  and notice that
  \(g\pa*{f\pa*{\textstyle\bigsqcup \alpha}} = g\pa*{\textstyle\bigsqcup f \circ
    \alpha} = \textstyle\bigsqcup {g \circ f \circ \alpha}\) by respectively
  continuity of \(f\)~and~\(g\). The claims about strictness are also clear.
\end{proof}

\begin{definition}[Isomorphism]
  A Scott continuous map \(f : D \to E\) is an \emph{isomorphism} if we have a
  Scott continuous inverse \(g : E \to D\).%
  \index{dcpo!isomorphism|textbf}
\end{definition}

\begin{lemma}\label{isomorphism-is-strict}
  Every \(f : D \to E\) isomorphism between pointed dcpos is strict.
\end{lemma}
\begin{proof}
  Let \(y : E\) be arbitrary and notice that \(\bot_D \below g(y)\) because
  \(\bot_D\) is the least element of \(D\). By monotonicity of \(f\), we get
  \(f(\bot_D) \below f(g(y)) = y\) which shows that \(f(\bot_D)\) is the least
  element of \(E\).
\end{proof}

\begin{definition}[Scott continuous retract]\label{def:continuous-retract}%
  A dcpo \(D\) is a \emph{(Scott) continuous retract} of \(E\) if we have Scott
  continuous maps \(s : D \to E\) and \(r : E \to D\) such that \(s\)~is a
  section of \(r\). We denote this situation by \(\retract{D}{E}\).%
  \index{retract!Scott continuous|textbf}
  \nomenclature[arrowz]{$\retract{D}{E}$}{Scott continuous retract}
\end{definition}

\begin{lemma}\label{locally-small-retract}
  If \(D\) is a continuous retract of \(E\) and \(E\) is locally small,
  then so is~\(D\).
\end{lemma}
\begin{proof}
  We claim that \(x \below_D y\) and \(s(x) \below_E s(y)\) are equivalent,
  which proves the lemma as \(E\) is assumed to be locally small.
  One direction of the equivalence is given by the fact that \(s\) is monotone.
  In the other direction, assume that \(s(x) \below s(y)\) and note that
  \(x = r(s(x)) \below r(s(y)) = y\), as \(r\) is monotone and \(s\) is a
  section of \(r\).
\end{proof}

\section{Lifting}\label{sec:lifting}

We now turn to constructing pointed \(\V\)-dcpos from sets.
First of all, every discretely ordered set is a \(\V\)-dcpo, where discretely
ordered means that we have \(x \below y\) exactly when \(x = y\). This is
because if \(\alpha : I \to X\) is a directed family into a discretely ordered
set~\(X\), then \(\alpha\) has to be constant (by semidirectedness), so
\(\alpha_i\) is its supremum for any \(i : I\). And since directedness includes
the condition that the domain is inhabited, it follows that \(\alpha\) must have
a supremum in \(X\).
In fact, ordering \(X\) discretely yields the free \(\V\)-dcpo on the set \(X\)
in the categorical sense.

With excluded middle, the situation for \emph{pointed} \(\V\)-dcpos is also very
straightforward. Simply order the set \(X\) discretely and add a least element,
as depicted for \(X \equiv \Nat\) in the Hasse diagram of~\cref{flat-gives-LPO}.%
\index{excluded middle}
The point of that proposition is to show, by a reduction to the constructive
taboo LPO~(\cref{def:LPO}), that this approach is constructively unsatisfactory.%
\index{constructivity}
In~\cref{chap:predicativity-in-order-theory} we will prove a general
constructive no-go theorem~(\cref{nontrivial-weak-em}) showing that there is a
nontrivial dcpo with decidable equality if and only if weak excluded middle
holds.%
\index{decidability!of equality}\index{excluded middle!weak}

\begin{proposition}\label{flat-gives-LPO}%
  \index{limited principle of omniscience}%
  \index{omega-completeness@\(\omega\)-completeness}%
  If the poset \(\Nat_\bot = (\Nat + \One,{\below})\) with order depicted by the Hasse diagram
  \nomenclature[Nbot]{$\Nat_{\bot}$}{the set \(\Nat + \One\) with the flat order}
  \[
    \begin{tikzcd}
      0 \ar[drr,dash] & 1 \ar[dr,dash] & 2 \ar[d,dash] & 3 \ar[dl,dash] & \cdots \ar[dll,dash] \\
      & & \bot
    \end{tikzcd}
  \]
  is \(\omega\)-complete, then LPO holds.
  In particular, by \cref{dcpo-has-sups-of-chains}, if it is
  \(\U_0\)-directed complete, then LPO holds.
\end{proposition}
\begin{proof}
  Let \(\alpha : \Nat \to \Two\) be an arbitrary binary sequence. We show that
  \(\exists_{n : \Nat}\,\alpha_n = 1\) is decidable.
  Define the family \(\beta : \Nat \to \Nat_\bot\) by
  \[
    \beta_n \colonequiv
    \begin{cases}
      \inl(k) &\text{if \(k\) is the least integer below \(n\) for which \(\alpha_k = 1\), and} \\
      \inr(\star) &\text{else}.
    \end{cases}
  \]
  Then \(\beta\) is a chain, so by assumption it has a supremum \(s\) in \(\Nat_\bot\).
  By the induction principle of coproducts, we have \(s = \inr(\star)\) or we
  have \(k : \Nat\) such that \(s = \inl(k)\).
  If the latter holds, then \(\alpha_k = 1\), so
  \(\exists_{n : \Nat}\,\alpha_n = 1\) is decidable.
  We claim that \(s = \inr(\star)\) implies that
  \(\lnot\pa*{\exists_{n : \Nat}\,\alpha_n = 1}\).
  Indeed, assume for a contradiction that \(\exists_{n : \Nat}\,\alpha_n =
  1\). Since we are proving a proposition, we may assume to have \(n : \Nat\)
  with \(\alpha_n = 1\).
  Then, \(\beta_n = \inl(k)\) for a natural number \(k \leq n\). Since \(s\) is
  the supremum of \(\beta\) we have
  \(\inl(k) = \beta_n \below s = \inr(\star)\), but \(\inr(\star)\) is the least
  element of \(\Nat_\bot\), so by antisymmetry \(\inl(k) = \inr(\star)\), which is
  impossible.
\end{proof}

Our solution to the above will be to work with the lifting monad, sometimes
known as the partial map classifier monad from topos
theory~\cite{Johnstone1977,Rosolini1986,Rosolini1987,Kock1991}, which has been
extended to constructive type theory by
\citeauthor{ReusStreicher1999}~\cite{ReusStreicher1999} and recently to
univalent foundations by
\citeauthor{EscardoKnapp2017}~\cite{EscardoKnapp2017,Knapp2018}; see the
\nameref{sec:basic-domain-theory-notes} of this chapter for a discussion of
additional related work.%
\index{partial!map classifier}%

\begin{definition}[Lifting, partial element, \(\lifting_{\V}(X)\); %
  {\cite[Section~2.2]{EscardoKnapp2017}}]%
  \index{partial!element|see {lifting}}%
  \index{lifting|textbf}%
  We define the type of \emph{partial elements} of a type \(X : \U\) with
  respect to a universe \(\V\) as
  \[
    \lifting_{\V}(X) \colonequiv \Sigma_{P : \Omega_{\V}}(P \to X)
  \]
  \nomenclature[L]{$\lifting_{\V}(X)$}{lifting of a type \(X\) with respect to a
    universe \(\V\)}
  and we also call it the \emph{lifting} of \(X\) with respect to \(\V\).
\end{definition}

Every (total) element of \(X\) gives rise to a partial element of \(X\) through
the following map, which will be shown to be the unit of the monad later.

\begin{definition}[\(\eta_X\)]
  The map \(\eta_X : X \to \lifting_{\V}(X)\) is defined by mapping \(x\) to the
  tuple \(\pa*{\One_{\V},\lambdadot{u}{x}}\), where, following
  \cref{subtype-omit}, we have omitted the witness that \(\One_{\V}\) is a
  subsingleton.
  \nomenclature[eta]{$\eta$}{unit of the lifting monad}
  We sometimes omit the subscript in \(\eta_X\).
\end{definition}

Besides these total elements, the lifting has another distinguished element that
will be the least with respect to the order with which we shall equip the
lifting.

\begin{definition}[\(\bot\)]\label{def:lifting-bot}
  For every type \(X : \U\) and universe \(\V\), we denote the element
  \(\pa*{\Zero_{\V},\varphi} : \lifting_{\V}(X)\) by \(\bot\). (Here \(\varphi\)
  is the unique map from \(\Zero_{\V}\) to \(X\).)
\end{definition}

Next we introduce appropriate names for the projections from the type of partial
elements.
\begin{definition}[\(\isdefined\) and \(\liftvalue\)]
  We write \(\isdefined : \lifting_{\V}(X) \to \Omega_{\V}\) for the first
  projection and
  \(\liftvalue : \Pi_{l : \lifting_{\V}(X)}\pa*{\isdefined(l) \to X}\) for the
  second projection.
  \nomenclature[is-defined]{$\isdefined$}{first projection from the lifting of a type}
  \nomenclature[value]{$\liftvalue$}{second projection from the lifting of a type}
\end{definition}

Thus, with this terminology, the element \(\star\) witnesses that \(\eta(x)\) is
defined with value \(x\) for every \(x : X\), while \(\bot\) is not defined
because \(\isdefined(\bot)\) is the empty type.

Excluded middle says exactly that such elements are the only elements of the
lifting of a type \(X\), as the following proposition shows.
Thus, the lifting generalises the classical construction of adding a new
element.

\begin{proposition}[{\cite[Section~2.2]{EscardoKnapp2017}}]%
  \label{lifting-is-plus-one-iff-em}
  The map \(X + \One \xrightarrow{[\eta,\const_{\bot}]} \lifting_{\V}(X)\) is an
  equivalence for every type \(X : \U\) if and only if excluded middle in \(\V\)
  holds.
\end{proposition}
\begin{proof}
  By the proof of \cref{em-and-wem-equivalent-formulations}\ref{item-em-Bool},
  excluded middle in \(\V\) is equivalent to the map
  \([\const_{\Zero},\const_{\One}] : \Two_{\V} \to \Omega_{\V}\) being an
  equivalence.
  But if that map is an equivalence, then it follows that the map
  \([\eta,\const_{\bot}]: X + \One \to \lifting_{\V}(X)\) is also an equivalence
  for every type \(X\).
  Conversely, we can take \(X \colonequiv \One_{\V}\) to see that
  \([\const_{\Zero},\const_{\One}] : \Two_{\V} \to \Omega_{\V}\) must be an
  equivalence.
\end{proof}

\begin{lemma}\label{lifting-equality}
  Two partial elements \(l,m : \lifting_{\V}(X)\) of a type \(X\) are equal if
  and only if we have \(\isdefined(l) \iff \isdefined(m)\) and the diagram
  \[
    \begin{tikzcd}
      \isdefined(m) \ar[dr] \ar[rr,"\liftvalue(m)"] & & X \\
      & \isdefined(l) \ar[ur,"\liftvalue(l)"']
    \end{tikzcd}
  \]
  commutes.
\end{lemma}
\begin{proof}
  By \cref{Id-of-Sigma} we have
  \[
    (l = m) \simeq \pa*{\Sigma_{e : \isdefined(l) = \isdefined(m)}\,
      \transport^{\lambdadot{P}{P \to X}}(e,\liftvalue(l)) = \liftvalue(m)}
  \]
  By path induction on \(e\) we can prove that
  \[
    \transport^{\lambdadot{P}{P \to X}}(e,\liftvalue(l)) = \liftvalue(l) \circ \tilde{e}^{-1},
  \]
  where \(\tilde{e}\) is the equivalence \(\isdefined(l) \simeq \isdefined(m)\)
  induced by \(e\).
  Hence, using function extensionality and propositional extensionality, the
  right hand side of the equivalence given above is logically equivalent to
  \[
    \Sigma_{(e_1,e_2) :
      \isdefined(l)\leftrightarrow\isdefined(m)}\,\liftvalue(l) \circ e_2 \sim
    \liftvalue(m),
  \]
  as desired.
\end{proof}

\begin{remark}
  It is possible to promote the logical equivalence of \cref{lifting-equality}
  to an equivalence of types using univalence and a generalised \emph{structure
    identity principle}~\cite[Section~3.33]{Escardo2019}, as done
  in~\cite[Lemma~44]{Escardo2021} and~\cite[\mkTTurl{Lifting.IdentityViaSIP}]{TypeTopology}.
  But the above logical equivalence will suffice.
\end{remark}

\begin{theorem}[Lifting monad, Kleisli extension, \(f^\#\);
  {\cite[Section~2.2]{EscardoKnapp2017}}]
  \label{lifting-is-monad}%
  \index{lifting!monad|textbf}%
  \index{lifting!Kleisli extension|textbf}%
  \index{Kleisli extension|see {lifting}}%
  The lifting is a monad with unit \(\eta\). That~is, for every map
  \(f : X \to \lifting_{\V}(Y)\) we have a map
  \(f^\# : \lifting_{\V}(X) \to \lifting_{\V}(Y)\), the \emph{Kleisli extension}
  of \(f\), such that
  \nomenclature[f#]{$f^\#$}{Kleisli extension of a map
    \(f : X \to \lifting_{\V}(Y)\) with respect to the lifting monad}
  \begin{enumerate}[(i)]
  \item\label{eta-ext} \(\eta_X^\# \sim \id_{\lifting_{\V}(X)}\) for every type
    \(X\),
  \item\label{ext-eta} \(f^\# \circ \eta_X \sim f\) for every map
    \(f : X \to \lifting_{\V}(Y)\), and
  \item\label{comp-ext} \((g^\# \circ f)^\# \sim g^\# \circ f^\#\) for every two
    maps \(f : X \to \lifting_{\V}(Y)\) and \(g : Y \to \lifting_{\V}(Z)\).
  \end{enumerate}
\end{theorem}
\begin{proof}
  Given \(f : X \to \lifting_{\V}(Y)\), we define
  \begin{align*}
    f^\# : \lifting_{\V}(X) &\to \lifting_{\V}(Y)\\
    (P,\varphi) &\mapsto \pa*{\Sigma_{p : P}\isdefined(f(\varphi(p))),\psi},
  \end{align*}
  where \(\psi(p,q) \colonequiv \liftvalue\pa*{f(\varphi(p)),q}\).

  Now for the proof of \ref{eta-ext}: Let \((P,\varphi) : \lifting_{\V}(X)\) be
  arbitrary and we calculate that
  \begin{align*}
    \eta^\#(P,\varphi) &\equiv
    \pa*{\Sigma_{p : P}\isdefined(\eta(\varphi(p))),
      \lambdadot{(p,q)}{\liftvalue(\eta(\varphi(p)),q)}} \\
    &\equiv
    \pa*{P \times \One,
      \lambdadot{(p,q)}{\varphi(p)}} \\
    &= (P,\varphi),
  \end{align*}
  where the final equality is seen to hold using \cref{lifting-equality}.  For
  \ref{ext-eta}, let \(x : X\) and \(f : X \to \lifting_{\V}(Y)\) be arbitrary
  and observe that
  \begin{align*}
    f^\#(\eta(x))
    &\equiv f^\#(\One,\lambdadot{u}{x}) \\
    &\equiv \pa*{\One \times \isdefined(f(x)),
      \lambdadot{(u,p)}{\liftvalue(f(x),p)}} \\
    &= \pa*{\isdefined(f(x)),
      \lambdadot{p}{\liftvalue(f(x),p)}} \\
    &\equiv f(x)
  \end{align*}
  where the penultimate equality is another easy application of
  \cref{lifting-equality}.
  We see that these proofs amount to the fact that \(\One\) is the unit for
  taking the product of types.
  For \ref{comp-ext} the proof amounts to the associativity of \(\Sigma\).
\end{proof}

\begin{remark}
  It should be noted that if \(X : \U\), then
  \(\lifting_{\V}(X) : \V^+ \sqcup \U\), so in general the lifting is a monad
  ``across universes''.
  But this increase in universes does not hinder us in stating and proving the
  monad laws and using them in later proofs.
\end{remark}
\begin{remark}
  The equalities of \cref{lifting-is-monad} do not include any coherence
  conditions which may be needed when \(X\) is not a set but a higher type. We
  will restrict to the lifting of sets, but the more general case is considered
  in~\cite{Escardo2021} where the coherence conditions are not needed for its
  goals either.
\end{remark}

\begin{definition}[Lifting functor, \(\lifting_{\V}(f)\)]\label{def:lifting-functor}%
  \index{lifting!functor}%
  The functorial action of the lifting could be defined from the unit and
  Kleisli extension as \(\lifting_{\V}(f) \colonequiv \pa*{\eta_Y \circ f}^\#\)
  for \(f : X \to Y\).
  \nomenclature[L']{$\lifting_{\V}(f)$}{functorial action of the lifting on a map \(f : X \to Y\)}
  But it is equivalent and easier to define \(\lifting_{\V}(f)\) directly by
  post-composition:
  \[
    \lifting_{\V}(f)(P,\varphi) \colonequiv \pa{P,f\circ\varphi}.
    \qedhere
  \]
\end{definition}

We now work towards showing that \(\lifting_{\V}(X)\) is the free pointed
\(\V\)-dcpo on a set \(X\).

\begin{proposition}\label{lifting-order}\index{lifting!partial order}%
  The relation
  \({\below} : \lifting_{\V}(X) \to \lifting_{\V}(X) \to \V^+\sqcup \U\)
  given by
  \[
    l \below m \colonequiv {\isdefined(l) \to l = m}
  \]
  is a partial order on \(\lifting_{\V}(X)\) for every set \(X : \U\).
  Moreover, it is equivalent to the more verbose relation
  \[
    (P,\varphi) \below' (Q,\psi) \colonequiv \Sigma_{f : P \to Q}%
    \pa*{\varphi \sim \psi \circ f}
  \]
  that is valued in \(\V \sqcup \U\).
\end{proposition}
\begin{proof}
  Note that \({\below}\) is subsingleton-valued because \(X\) is assumed to be a
  set. The other properties follow using \cref{lifting-equality}.
\end{proof}

In light of~\cref{universe-levels-of-lifting-and-exponentials}, we carefully
keep track of the universe parameters of the lifting in the following
proposition.

\begin{proposition}[cf.~{\cite[Theorem~1]{EscardoKnapp2017}}]%
  \label{lifting-is-pointed-dcpo}%
  \index{lifting!as a pointed dcpo}%
  For a set \(X : \U\), the lifting \(\lifting_{\V}(X)\) ordered as in
  \cref{lifting-order} is a pointed \(\V\)-dcpo.
  In general, \(\lifting_{\V}(X) : \DCPO{V}{V^+ \sqcup U}{V^+ \sqcup U}\), but
  if \(X : \V\), then \(\lifting_{\V}(X)\) is locally small.
\end{proposition}
\begin{proof}
  By \cref{lifting-order} we have a poset and it is clear that \(\bot\) from
  \cref{def:lifting-bot} is its least element.
  Now let \(\pa*{Q_{(-)},\varphi_{(-)}} : I \to \lifting_{\V}(X)\) be a directed
  family with \(I : \V\). We claim that the map
  \({\Sigma_{i : I}Q_i} \xrightarrow{(i,q) \mapsto \varphi_i(q)} X\) is
  constant.
  Indeed, given \(i,j : I\) with \(p : Q_i\) and \(q : Q_j\), there exists
  \(k : I\) such that
  \(\pa*{Q_i,\varphi_i},\pa*{Q_j,\varphi_j} \below \pa*{Q_k,\varphi_k}\) by
  directedness of the family.
  But by definition of the order and the elements \(p : Q_i\) and \(q : Q_j\),
  this implies that
  \(\pa*{Q_i,\varphi_i} = \pa*{Q_j,\varphi_j} = \pa*{Q_k,\varphi_k}\) which in
  particular tells us that \(\varphi_i(p) = \varphi_j(q)\).
  Hence, by \cref{constant-map-to-set-factors-through-truncation}, we have a
  (dashed) map \(\psi\) making the diagram%
  \index{propositional truncation}%
  \[
    \begin{tikzcd}
      \Sigma_{i : I}Q_i \ar[dr,"{\tosquash{-}}"']
      \ar[rr,"{(i,q) \mapsto \varphi_i(q)}"] & & X \\
      & \exists_{i : I}Q_i \ar[ur,dashed,"\psi"']
    \end{tikzcd}
  \]
  commute.
  We claim that \(\pa*{\exists_{i : I}Q_i,\psi}\) is the least upper bound of
  the family.
  To see that it is an upper bound, let \(j : I\) be arbitrary. By the
  commutative diagram and \cref{lifting-order} we see that
  \(\pa*{Q_j,\varphi_j} \below \pa*{\exists_{i : I}Q_i,\psi}\), as desired.
  Moreover, if \((P,\rho)\) is an upper bound for the family, then
  \(\pa*{Q_i,\varphi_i} = (P,\rho)\) for all \(i : I\) such that \(Q_i\) holds.
  Hence, \(\pa*{\exists_{i : I}Q_i,\psi} \below (P,\rho)\), as desired.
  Finally, local smallness in the case that \(X\) is a type in \(\V\) follows
  from \cref{lifting-order}.
\end{proof}

\begin{proposition}\label{lifting-extension-is-continuous}
  \index{lifting!Kleisli extension}%
  The Kleisli extension \(f^\# : \lifting_{\V}(X) \to \lifting_{\V}(Y)\) is
  Scott continuous for any map \(f : X \to \lifting_{\V}(Y)\).
\end{proposition}
\begin{proof}
  It is straightforward to prove that \(f^\#\) is monotone. Hence, it remains to
  prove that \(f^\#\pa*{\bigsqcup \alpha} \below \bigsqcup f^\# \circ \alpha\)
  for every directed family \(\alpha : I \to \lifting_{\V}(X)\).
  So suppose that \(f^\#\pa*{\bigsqcup \alpha}\) is defined. Then we have to
  show that it equals \(\bigsqcup f^\# \circ \alpha\).
  By our assumption and definition of \(f^\#\) we get that \(\bigsqcup \alpha\)
  is defined too. By the definition of suprema in the lifting and because we are
  proving a proposition, we may assume to have \(i : I\) such that \(\alpha_i\)
  is defined.
  But since \(\alpha_i \below \bigsqcup \alpha\), we get
  \(\alpha_i = \bigsqcup \alpha\) and hence,
  \(f^\#\pa*{\alpha_i} = f^\#\pa*{\bigsqcup \alpha}\).
  Finally, \(f^\#\pa*{\alpha_i} \below \bigsqcup f^\# \circ \alpha\), but by
  assumption \(f^\#\pa*{\bigsqcup \alpha}\) is defined and hence so is
  \(f^\#\pa*{\alpha_i}\) which implies
  \(f^\#\pa*{\bigsqcup \alpha} = f^\#\pa*{\alpha_i} = \bigsqcup f^\# \circ
  \alpha\), as desired.
\end{proof}

Recall from \cref{pointed-dcpos-sups} that pointed \(\V\)-dcpos have suprema of
families indexed by propositions in \(\V\). We make use of this fact in the
following lemma.

\begin{lemma}\label{lifting-element-as-sup}
  For a set \(X\), every partial element \((P,\varphi) : \lifting_{\V}(X)\) is
  equal to supremum \(\bigvee_{p : P}\eta_X\pa*{\varphi(p)}\).
\end{lemma}
\begin{proof}
  Note that if \(p : P\), then \((P,\varphi) = \eta_X(\varphi(p))\), so that the
  lemma follows from antisymmetry.
\end{proof}

The lifting \(\lifting_{\V }(X)\) gives the free pointed \(\V \)-dcpo on a set
\(X\). Keeping track of universes, it holds in the following generality:
\begin{theorem}\label{lifting-is-free}%
  \index{lifting!free pointed dcpo}%
  \index{dcpo!free pointed|see {lifting}}%
  \index{dcpo!of partial elements|see {lifting}}%
  If \(X : \U \) is a set, then for every pointed \(\V \)-dcpo
  \(D : \DCPO{V}{U'}{T'}\) and function \(f : X \to D\), there is a unique
  strict and continuous function \(\bar{f} : \lifting_{\V }(X) \to D\) making
  the diagram
  \[
    \begin{tikzcd}
      X \ar[dr, "\eta_X"'] \ar[rr, "f"] & & D  \\
      & \lifting_{\V }(X) \ar[ur, dashed, "\bar{f}"']
    \end{tikzcd}
  \]
  commute.
\end{theorem}
\begin{proof}
  We define \(\bar f : \lifting_{\V }(X) \to D\) by
  \((P,\varphi) \mapsto \bigvee_{p : P} f(\varphi(p))\), which is well-defined
  by \cref{pointed-dcpos-sups} and easily seen to be strict
  and continuous. For uniqueness, suppose that we have
  \({g : \lifting_{\V }(X) \to D}\) strict and continuous such that
  \(g \circ \eta_X = f\) and let \((P,\varphi)\) be an arbitrary element of
  \(\lifting_{\V }(X)\). Then,
  \begin{align*}
    g\pa*{P,\varphi}
    &= g\pa*{\textstyle\bigvee_{p : P}\eta_X\pa*{\varphi(p)}}
    &&\text{(by \cref{lifting-element-as-sup})} \\
    &= \textstyle\bigvee_{p : P}g\pa*{\eta_X\pa*{\varphi(p)}}
    &&\text{(by \cref{pointed-dcpos-sups} and strictness and continuity of \(g\))} \\
    &= \textstyle\bigvee_{p : P} f\pa*{\varphi(p)}
    &&\text{(by assumption on \(g\))} \\
    &\equiv \bar{f}(P,\varphi),
  \end{align*}
  as desired.
\end{proof}

The proof tells us that there is yet another way in which the lifting is a free
construction, namely as the free subsingleton-complete poset. What is noteworthy
about this is that freely adding subsingleton suprema automatically gives all
directed suprema.

\begin{definition}[Subsingleton completeness]
  A poset \(P\) is \emph{subsingleton complete} with respect to a type universe
  \(\V\) if it has suprema for all families indexed by a subsingleton in \(\V\).%
  \index{subsingleton!completeness}
\end{definition}

The lifting \(\lifting_{\V }(X)\) gives the free \(\V\)-subsingleton complete
poset on a set \(X\). Keeping track of universes, it holds in the following
generality:
\begin{theorem}\label{lifting-is-free2}%
  \index{lifting!free subsingleton complete poset}%
  \index{subsingleton!completeness!free|see {lifting}}%
  If \(X : \U \) is a set, then for every \(\V\)-subsingleton complete poset
  \(P\) (with carrier and order taking values in arbitrary, possibly distinct,
  universes) and function \(f : X \to P\), there exists a unique monotone
  \(\bar{f} : \lifting_{\V }(X) \to P\) preserving all suprema indexed by
  propositions in \(\V \) making the diagram
  \[
    \begin{tikzcd}
      X \ar[dr, "\eta_X"'] \ar[rr, "f"] & & P \\
      & \lifting_{\V }(X) \ar[ur,dashed,"\bar{f}"']
    \end{tikzcd}
  \]
  commute.
\end{theorem}
\begin{proof}
  Similar to the proof of \cref{lifting-is-free}.
\end{proof}

Finally, we consider a variation of \cref{lifting-order} which allows us to
freely add a least element to a \(\V\)-dcpo instead of just a set.

\begin{proposition}\label{lifting-order-alt}\index{lifting!partial order}
  For a poset \(D\) whose order takes values in \(\T\), the binary relation
  \({\below} : \lifting_{\V}(D) \to \lifting_{\V}(D) \to \V \sqcup \T\) %
  given by
  \[
    (P,\varphi) \below (Q,\psi) \colonequiv \Sigma_{f : P \to Q}%
    \pa*{\Pi_{p : P}\pa*{\varphi(p) \below_D \psi(f(p))}}
  \]
  is a partial order on \(\lifting_{\V}(D)\).
\end{proposition}
\begin{proof}
  Similar to \cref{lifting-order}, but using that \({\below_D}\) is reflexive,
  transitive and antisymmetric.
\end{proof}

\begin{proposition}\label{lifting-is-pointed-dcpo-alt}\index{lifting!as a pointed dcpo}
  For a dcpo \(D : \DCPO{V}{U}{T}\), the lifting \(\lifting_{\V}(D)\) ordered as
  in \cref{lifting-order-alt} is a pointed \(\V\)-dcpo.
  In general, \(\lifting_{\V}(D) : \DCPO{V}{V^+ \sqcup U}{V \sqcup T}\), but if
  \(D\) is locally small, then so is \(\lifting_{\V}(D)\).
\end{proposition}
\begin{proof}
  The element \(\bot\) from \cref{def:lifting-bot} is still the least element
  with respect to the new order.
  If \(\alpha : I \to \lifting_{\V}(D)\) is directed, then, writing
  \(\pa*{Q_i,\varphi_i} \colonequiv \alpha_i\), we consider
  \(\Phi : \pa*{\Sigma_{i : I}Q_i} \to D\) given by
  \((i,q) \mapsto \varphi_i(q)\).
  This family is semidirected, for if we have \(i,j : I\) with \(p : Q_i\) and
  \(q : Q_j\), then there exists \(k : I\) such that
  \({\alpha_i,\alpha_j} \below \alpha_k\) in \(\lifting_{\V}(D)\) by
  directedness of \(\alpha\), which implies that \(\Phi(i,p) \below \Phi(j,q)\)
  in \(D\).
  Thus, if we know that \(\exists_{i : I}Q_i\), then the family \(\Phi\) is
  directed and must have a supremum in \(D\).
  Hence we have a partial element
  \(\pa*{\exists_{i : I}Q_i,\psi} : \lifting_{\V}(D)\) where \(\psi\) takes the
  witness that the domain of \(\Phi\) is inhabited to the directed supremum
  \(\bigsqcup \Phi\) in \(D\).
  It is not hard to verify that this partial element is the least upper bound of
  \(\alpha\) in \(\lifting_{\V}(D)\), completing the proof.
\end{proof}

The lifting \(\lifting_{\V }(D)\) with the partial order of
\cref{lifting-order-alt} gives the free \emph{pointed} \(\V\)-dcpo on a
\(\V\)-dcpo \(D\). Keeping track of universes, it holds in the following
generality:
\begin{theorem}\label{lifting-is-free3}\index{lifting!free pointed dcpo}
  If \(D : \DCPO{V}{U}{T}\) is a \(\V\)-dcpo, then for every pointed \(\V\)-dcpo
  \(E : \DCPO{V}{U'}{T'}\) and continuous function \(f : D \to E\), there is a
  unique strict continuous function \(\bar{f} : \lifting_{\V }(D) \to E\) making
  the diagram
  \[
    \begin{tikzcd}
      D \ar[dr, "\eta_{D}"'] \ar[rr, "f"] & & E  \\
      & \lifting_{\V }(D) \ar[ur, dashed, "\bar{f}"']
    \end{tikzcd}
  \]
  commute.
\end{theorem}
\begin{proof}
  Similar to the proof of \cref{lifting-is-free}.
\end{proof}

Notice how \cref{lifting-is-free3} generalises \cref{lifting-is-free} as any set
can be viewed as a discretely ordered \(\V\)-dcpo.

\section{Products and exponentials}\label{sec:products-and-exponentials}

We describe two constructions of \(\V\)-dcpos, namely products and exponentials.
Exponentials will be crucial in the Scott model of PCF, as discussed
in~\cref{sec:Scott-model-of-PCF}.
Products are not needed for this purpose as we will work with the combinatory
version of PCF.
However, product allows us to state the universal property of the exponential
(\cref{exponential-universal-property}).
Moreover, products are also needed when extending the Scott model to account for
a version of PCF with variables and \(\lambda\)-abstraction, see the
\nameref{sec:basic-domain-theory-notes} for this chapter.

\begin{definition}[Product of (pointed) dcpos, \(D_1 \times D_2\)]
  The \emph{product} of two \(\V\)-dcpos \(D_1\) and \(D_2\) is given by the
  \(\V\)-dcpo \(D_1 \times D_2\) defined as follows. Its carrier is the
  cartesian product of the carriers of \(D_1\) and \(D_2\). The order is given
  componentwise, i.e.\ \((x,y) \below_{D_1 \times D_2} (x',y')\) if
  \(x \below_{D_1} y\) and \(y \below_{D_2} y'\).%
  \index{dcpo!product|textbf}%
  \nomenclature[times']{$D \times E$}{binary cartesian product of (pointed) dcpos}
  Accordingly directed suprema are also given componentwise. That is, given a
  directed family \(\alpha : I \to {D_1 \times D_2}\), one quickly verifies that
  the families \(\fst \circ {\alpha}\) and \(\snd \circ \alpha\) are also
  directed. We then define the supremum of \(\alpha\) as
  \((\bigsqcup {\fst \circ \alpha} , \bigsqcup {\snd \circ \alpha})\).
  Moreover, if \(D\) and \(E\) are pointed, then so is \(D \times E\) by taking
  the least elements in both components.
\end{definition}

\begin{remark}
  Notice that if \(D_1 : \DCPO{V}{U}{T}\) and \(D_2 : \DCPO{V}{U'}{T'}\), then
  for their product we have
  \(D_1 \times D_2 : \DCPO{V}{U \sqcup U'}{T \sqcup T'}\), which simplifies to
  \(\DCPO{V}{U}{T}\) when \(\U' \equiv \U\) and \(\T' \equiv \T\).
\end{remark}

\begin{proposition}\label{binary-product-universal-property}
  The product defined above satisfies the appropriate universal property: the
  projections \(\fst : D_1 \times D_2 \to D_1\) and
  \(\snd : D_1 \times D_2 \to D_2\) are Scott continuous and if
  \(f : E \to D_1\) and \(g : E \to D_2\) are Scott continuous functions from a
  \(\V\)-dcpo~\(E\), then there is a unique Scott continuous map
  \(k \colon E \to D_1 \times D_2\) such that the diagram
  \[
    \begin{tikzcd}
      & D_1 \times D_2 \ar[dl,"\fst"'] \ar[dr,"\snd"] \\
      D_1 \ & & D_2 \\
      & E \ar[ul,"f"] \ar[ur,"g"'] \ar[uu,dashed,"k"]
    \end{tikzcd}
  \]
  commutes.
\end{proposition}
\begin{proof}
  The projections are Scott continuous by definition of directed suprema in
  \(D_1 \times D_2\). Moreover, if \(f \colon E \to D_1\) and
  \(g \colon E \to D_2\) are Scott continuous maps, then we see that we have no
  choice but to define \(k \colon E \to {D_1 \times D_2}\) by
  \(e \mapsto (f(e),g(e))\). Moreover, this assignment is Scott continuous,
  because for a directed family \(\alpha : I \to E\), we have
  \(k\pa*{\bigsqcup \alpha} \equiv \pa*{f(\bigsqcup \alpha),g(\bigsqcup\alpha)}
  = \pa*{\bigsqcup f \circ \alpha,\bigsqcup g \circ \alpha} \equiv \bigsqcup {k
    \circ \alpha}\) by Scott continuity of \(f\)~and~\(g\) and the definition of
  directed suprema in \(D_1 \times D_2\).
\end{proof}

\begin{lemma}\label{continuous-in-both-arguments}
  A map \(f \colon D_1 \times D_2 \to E\) is Scott continuous if and only if the
  maps \(f(x,-) : D_2 \to E\) and \(f(-,y) : D_1 \to E\) are Scott continuous
  for every \(x : D_1\) and \(y : D_2\).
\end{lemma}
\begin{proof}
  Suppose first that \(f \colon D_1 \times D_2 \to E\) is Scott continuous and
  let \(x : D_1\) be arbitrary. If \(\alpha : I \to D_2\) is a directed family,
  then \(f\pa*{x,\bigsqcup \alpha} = f\pa*{\bigsqcup \alpha_x}\), where
  \(\alpha_x : I \to D_1 \times D_2\) is the directed family given by
  \(i \mapsto (x,\alpha(i))\).
  But \(f\) is Scott continuous, so
  \(\bigsqcup_{i : I} f(x,\alpha(i)) = \bigsqcup_{i : I}f(\alpha_x(i)) =
  f(\bigsqcup\alpha_x(i)) = f(x,\bigsqcup \alpha)\), as desired.
  Continuity of \(f(-,y)\) is proved similarly of course.

  Conversely, suppose that the conditions in the lemma hold and let
  \(\alpha : I \to {D_1 \times D_2}\) be directed. We need to show that
  \(f \pa{\bigsqcup \alpha} \equiv f \pa*{\bigsqcup \alpha_1,\bigsqcup
    \alpha_2}\) is the least upper bound of \(\bigsqcup f \circ \alpha\), where
  \(\alpha_1 \colonequiv \fst \circ \alpha\) and
  \(\alpha_2 \colonequiv \snd \circ \alpha\).
  To see that it is indeed an upper bound, assume that \(i : I\) and observe
  that
  \[
    f(\alpha(i)) \equiv f\pa*{\alpha_1(i),\alpha_2(i)} \below
    f\pa*{\alpha_1(i),\textstyle\bigsqcup\alpha_2} \below
    f\pa*{\textstyle\bigsqcup \alpha_1,\textstyle\bigsqcup\alpha_2},
  \]
  by monotonicity of \(f\pa*{\alpha_1(i),-}\)
  and \(f\pa*{-,\bigsqcup\alpha_2}\).
  To see that it is least, suppose that \(y \aboveorder f(\alpha(i))\) for every
  \(i : I\). By Scott continuity of \(f\pa*{-,\bigsqcup \alpha_2}\) it is
  sufficient to prove that \(f\pa*{\alpha_1(i),\bigsqcup \alpha_2}\) for every
  \(i : I\). So let \(i : I\) be arbitrary. By Scott continuity of
  \(f\pa*{\alpha_1(i),-}\) it suffices to prove
  \(f\pa*{\alpha_1(i),\alpha_2(j)} \below y\) for every \(j : I\).
  So let \(j : I\) be arbitrary. By directedness of \(\alpha\), there exists
  \(k : I\) such that \(\alpha(i),\alpha(j) \below \alpha(k)\).
  Hence,
  \(f\pa*{\alpha_1(i),\alpha_2(j)} \below f(\alpha(k)) \below y\), as desired.
\end{proof}

\begin{definition}[Exponential of (pointed) dcpos, \(E^D\)]
  The \emph{exponential} of two \(\V\)\nobreakdash-dcpos \(D\) and \(E\) is
  given by the poset \(E^D\) defined as follows. Its carrier is the type of
  Scott continuous functions from \(D\) to \(E\).
  \index{exponential|textbf}%
  \nomenclature[ED]{$E^D$}{exponential of (pointed) dcpos}%
  The order is given pointwise, i.e.\ \(f \below_{E^D} g\) holds if
  \(f(x) \below_{E} g(x)\) for every \(x : D\).
  Notice that if \(E\) is pointed, then so is \(E^D\) with least element given
  the constant function \(\lambdadot{x : D}\bot_E\) which is Scott continuous by
  \cref{continuity-closure}\ref{const-is-continuous}.
\end{definition}

Note that the exponential \(E^D\) is a priori not locally small even if \(E\) is
because the partial order quantifies over all elements of \(D\). But if \(D\) is
continuous (a notion that we will study in detail in
\cref{chap:continuous-and-algebraic-dcpos}) then \(E^D\) will be locally small
when \(E\) is (\cref{exponential-is-locally-small}).

\begin{proposition}
  The exponential \(E^D\) of two \(\V\)-dcpos \(D\) and \(E\) is \(\V\)-directed
  complete.
\end{proposition}
\begin{proof}
  Since the partial order is given pointwise, we expect directed suprema to be
  calculated pointwise too.
  Explicitly, given a directed family \(\alpha : I \to {E^D}\), we verify that
  for every \(x : D\), the family \(\alpha_x : I \to E\) defined by
  \(i \mapsto \alpha_i(x)\) is also directed.
  Indeed, if we have \(i,j : I\), then there exists \(k : I\) such that
  \({\alpha_i,\alpha_j} \below \alpha_k\). Hence, for arbitrary \(x : D\), we
  have \({\alpha_i(x),\alpha_j(x)} \below \alpha_k(x)\), which shows that
  \(\alpha_x\) is directed.
  Because the order is pointwise, it is clear that the function
  \(\lambdadot{x : D}{\bigsqcup \alpha_x}\) is the least upper bound of
  \(\alpha\), but we must also check that this function is Scott continuous.
  We employ \cref{continuity-criterion} for this, so we first check that the
  function is monotone.
  Indeed if \(x \below y\) in \(D\), then \(\alpha_i(x) \below \alpha_i(y)\) for
  every \(i : I\) as Scott continuous functions are monotone. Hence,
  \(\bigsqcup \alpha_x \below \bigsqcup \alpha_y\) in this case.
  Now let \(\beta : J \to D\) be directed. We have to prove that
  \(\bigsqcup_{i : I}\alpha_i\pa*{\bigsqcup \beta}
  \below \bigsqcup_{j : J}\bigsqcup_{i : I}\alpha_i\pa*{\beta_j}\), %
  for which it is enough to know that \(\alpha_i\pa*{\bigsqcup \beta}
  \below \bigsqcup_{j : J}\bigsqcup_{i : I}\alpha_i\pa*{\beta_j}\) %
  for every \(i : I\).
  But this is clear as
  \(\alpha_i\pa*{\bigsqcup \beta} = \bigsqcup_{j : J}\alpha_i\pa*{\beta_j}\) by
  Scott continuity of each \(\alpha_i\).
\end{proof}

\begin{remark}\label{exponential-universe-parameters}%
  \index{universe!parameters}%
  Recall from~\cref{universe-levels-of-lifting-and-exponentials} that it is
  necessary to carefully keep track of the universe parameters of the
  exponential.
  In general, the universe levels of \(E^{D}\) can be quite large and
  complicated. For~if \(D : \DCPO{V}{U}{T}\) and \(E : \DCPO{V}{U'}{T'}\), then
  \[
    E^{D} : \DCPO{V}{V^+\sqcup U\sqcup T\sqcup U'\sqcup T'}{U\sqcup T'}.
  \]
  Even~if \(\V = \U \equiv \T \equiv \U' \equiv \T'\), the carrier of \(E^{D}\)
  still lives in the larger universe~\(\V ^+\), because the type expressing
  Scott continuity for \(\V\)-dcpos quantifies over all types in~\(\V\).
  Actually, the scenario where \(\U = \U' = \V\) cannot happen in a predicative
  setting unless \(D\) and \(E\) are trivial, in a sense made precise in
  \cref{chap:predicativity-in-order-theory}.

  Even so, in many applications such as those in \cref{chap:applications}, if we
  take \(\V \equiv \U_0\) and all other parameters to be
  \(\U \equiv \T \equiv \U' \equiv \T' \equiv \U_1\), then the situation is much
  simpler and \(D\),~\(E\)~and the exponential \(E^D\) are all elements of
  \(\DCPOnum{0}{1}{1}\) with all of them being locally small (remember that this
  is defined up to equivalence).
  This turns out to be a very favourable situation for both the Scott model of
  PCF and Scott's \(D_\infty\) model of the untyped \(\lambda\)-calculus.
\end{remark}

In the proposition below we can have \(D : \DCPO{\V}{\U}{\T}\) and
\(E : \DCPO{\V}{\U'}{\T'}\) for arbitrary universes \(\U\), \(\T\), \(\U'\) and
\(\T'\). In particular, the universe parameters of \(D\) and \(E\), apart from
the universe of indexing types, need to be the same.

\begin{proposition}\label{exponential-universal-property}%
  The exponential defined above satisfies the appropriate universal property:
  the \emph{evaluation map} \(\ev: E^D \times D \to E, (g,x) \mapsto g(x)\) is
  Scott continuous and if \(f : {D'\times D} \to E\) is a Scott continuous
  function, then there is a unique Scott continuous map
  \(\bar{f} \colon D' \to E^D\) such that the diagram
  \[
    \begin{tikzcd}
      D' \times D \ar[dr,"f"] \ar[d,dashed,"{\bar{f}}\,\times\,{\id_D}"'] \\
      E^D \times D \ar[r,"\ev"'] & E
    \end{tikzcd}
  \]
  commutes.
\end{proposition}
\begin{proof}
  We use \cref{continuous-in-both-arguments} to prove that \(\ev\) is Scott
  continuous: It is continuous in the second argument, because the first
  argument is a Scott continuous function, and it is continuous in the first
  argument, because suprema in the exponential are calculated pointwise.
  From the diagram we see that we have no choice but to define \(\bar{f}\) as
  \(y \mapsto \lambdadot{x}{f(y,x)}\).
  It remains to prove that \(\bar{f}(y)\) is Scott continuous for every
  \(y : D'\) and that this assignment itself defines a Scott continuous function
  \(D' \to E^D\).
  For the former, note that \(\bar{f}(y) \colonequiv f(y,-)\) is indeed Scott
  continuous by \cref{continuous-in-both-arguments}.
  For the latter, note that if \(\alpha : I \to D'\) is directed, then
  \[
    \bar{f}\pa*{\textstyle\bigsqcup \alpha} \equiv
    \lambdadot{x}{f\pa*{\textstyle\bigsqcup\alpha,x}} =
    \lambdadot{x}\textstyle\bigsqcup_{i : I}f\pa*{\alpha_i,x} \equiv
    \textstyle\bigsqcup_{i : I}\pa*{\lambdadot{x}{f\pa*{\alpha_i,x}}}
  \]
  by Scott continuity of \(f\) and the fact that suprema are calculated
  pointwise in the exponential.
  Thus, \(\bar{f}\) is Scott continuous, completing the proof.
\end{proof}

The following theorem lies at the heart of the Scott model of PCF that we will
study in \cref{sec:Scott-model-of-PCF}.

\begin{theorem}[Least fixed point, \(\mu\)]\label{least-fixed-point}%
  \index{least fixed point}%
  \index{fixed point|see {least fixed point}}%
  Every Scott continuous endomap \(f\) on a pointed \(\V\)-dcpo \(D\) has a
  least fixed point given by
  \[
    \mu(f) \colonequiv \textstyle\bigsqcup_{n : \Nat} f^n(\bot).
  \]
  \nomenclature[mu]{$\mu(f)$}{least fixed point of a Scott continuous endomap
    \(f\) on a pointed dpco}
  Specifically, the following two conditions hold:
  \begin{enumerate}[(i)]
  \item\label{is-fixed-point} \(f(\mu(f)) = \mu(f)\), and
  \item\label{is-least-fixed-point} for every \(x : D\), if \(f(x) \below x\), then
    \(\mu(f) \below x\).
  \end{enumerate}
  Moreover, the assignment \(f \mapsto \mu(f)\) defines a Scott continuous map
  \(D^D \to D\).
\end{theorem}
\begin{proof}
  We follow the proof given in~\cite[Theorem~2.1.19]{AbramskyJung1994} and first
  establish~\ref{is-least-fixed-point}. Suppose that \(f(x) \below x\). To show
  that \(\mu(f) \below x\), it suffices to prove that \(f^n(\bot) \below x\) for
  every \(n : \Nat\). But this follows easily by induction on \(n\) and the fact
  that \(f\) is monotone.
  For~\ref{is-fixed-point}, first notice that
  \begin{equation}\label{towards-fixed-point}\tag{\(\dagger\)}
    f(\mu(f)) \equiv f\pa*{\textstyle\bigsqcup_{n : \Nat}f^n(\bot)}
    = \textstyle\bigsqcup_{n : \Nat}f^{n+1}(\bot)
    \below \mu(f)
  \end{equation}
  by Scott continuity of \(f\), proving one of the inequalities.
  But \(f\) is monotone, so \eqref{towards-fixed-point} yields
  \(f(f(\mu(f))) \below f(\mu(f))\), which by~\ref{is-least-fixed-point} implies
  \(\mu(f) \below f(\mu(f))\), so that \(f(\mu(f)) = \mu(f)\) by antisymmetry as
  we set out to prove.
  To see that the assignment \(f \mapsto \mu(f)\) is continuous we will
  reconstruct it as the least upper bound of a family in the exponential
  \(D^{\pa*{D^D}}\).
  First define for every natural number \(n : \Nat\), the function
  \begin{align*}
    \iter_{n} : D^D &\to D\\
    f &\mapsto f^n(\bot)
  \end{align*}
  Observe that \(\iter_n\) can be factored as
  \( D^D \xrightarrow{f \mapsto f^n} D^D \xrightarrow{\text{evaluate at
      \(\bot\)}} D \). %
  By induction on \(n\) and \cref{continuity-closure} the first map is seen to
  be continuous, while the second is continuous by
  \cref{exponential-universal-property}. Hence, the composite, \(\iter_n\) is
  continuous for every \(n : \Nat\) by \cref{continuity-closure}.
  Thus, each \(\iter_n\) is an element of \(D^{\pa*{D^D}}\). Moreover, the
  assignment \(n \mapsto \iter_n\) is directed in \(D^{\pa*{D^D}}\) because if
  \(n \leq m\), then
  \(\iter_n(f) \equiv f^n(\bot) \below f^m(\bot) \equiv \iter_m(f)\) for every
  Scott continuous \(f : D \to D\).
  Hence, we can take the supremum of \(\pa*{\iter}_{n : \Nat}\) in
  \(D^{\pa*{D^D}}\) which is Scott continuous by construction. But suprema are
  calculated pointwise, so we can compute that
  \(\pa*{\bigsqcup \iter}(f) \equiv \bigsqcup_{n : \Nat}f^n(\bot)\),
  establishing the continuity of \(f \mapsto \mu(f)\) and completing the proof.
\end{proof}

In the Scott model of~PCF (\cref{sec:Scott-model-of-PCF}), the \(\mu\) operation
is used to model general recursion in the programming language~PCF. The equation
\(f(\mu(f)) = \mu(f)\) may be regarded as the unfolding of a recursive
definition, while the least element \(\bot\) represents nontermination.

\section{Bilimits}\label{sec:bilimits}

Recall that in a \(\V\)-dcpo \(D\) we can take suprema of directed families
\(\alpha : I \to D\). It is a striking feature of directed complete posets that
this act is reflected in the \emph{category} of dcpos, although it does require
us to specify an appropriate notion of one dcpo being ``below'' another one.
This notion will be exactly that of an embedding projection pair.
The technical results developed in this section will find application in the
construction of Scott's \(D_\infty\), a model of the untyped
\(\lambda\)-calculus, as discussed in \cref{sec:Scott-D-infty}.

A priori it is not clear that \(D_\infty\) should exist in predicative univalent
foundations and it is one of the contributions of this work that this is indeed
possible.
Our construction largely follows the classical development of Scott's original
paper~\cite{Scott1972}, but with some crucial differences.
First of all, we carefully keep track of the universe parameters and try to be
as general as possible. In the particular case of Scott's \(D_\infty\) model of
the untyped \(\lambda\)-calculus, we obtain a \(\U_0\)-dcpo whose carriers lives
in the second universe \(\U_1\).
Secondly, difference arises from proof relevance and these complications are
tackled with techniques in univalent foundations
and~\cref{constant-map-to-set-factors-through-truncation} in particular, as
discussed right before~\cref{kappa-is-constant}.
Finally, we we generalise Scott's treatment from sequential bilimits to directed
bilimits.

\begin{definition}[Deflation]
  An endofunction \(f : D \to D\) on a poset \(D\) is a \emph{deflation} if
  \(f(x) \below x\) for all \(x : D\).%
  \index{deflation}
\end{definition}
\begin{definition}[Embedding-projection pair]%
  \index{embedding-projection pair|textbf}%
  An \emph{embedding-projection pair} from a \(\V\)-dcpo \(D\) to a \(\V\)-dcpo
  \(E\) consists of two Scott continuous functions \(\varepsilon : D \to E\)
  (the~\emph{embedding}) and \(\pi : E \to D\) (the~\emph{projection})
  such~that:
  \begin{enumerate}[(i)]
  \item \(\varepsilon\) is a section of \(\pi\), and
  \item \(\varepsilon \circ \pi\) is a deflation.%
    \qedhere
  \end{enumerate}
\end{definition}

For the remainder of this section, fix the following setup, where we try to be
as general regarding universe levels as we can be.
We fix a directed preorder \((I,\below)\) with \(I : \V\) and \(\below\) takes
values in some universe \(\W\). Now suppose that \((I,\below)\) indexes a family
of \(\V\)-dcpos with embedding-projection pairs between them, i.e.\ we have
\begin{itemize}
\item for every \(i : I\), a \(\V \)-dcpo \(D_i : \DCPO{V}{U}{T}\), and
\item for every \(i,j : I\) with \(i \sqsubseteq j\), an embedding-projection
  pair \(\pa*{\varepsilon_{i,j},\pi_{i,j}}\) from \(D_i\) to \(D_j\).
\end{itemize}
Moreover, we require that the following compatibility conditions hold:
\begin{align}
  &\text{for every \(i : I\), we have \(\varepsilon_{i,i} = \pi_{i,i} = \id\)}; %
  \label{epsilon-pi-id} \\
  &\text{for every \(i \sqsubseteq j \sqsubseteq k\) in \(I\), we have
  \(\varepsilon_{i,k} \sim \varepsilon_{j,k} \circ \varepsilon_{i,j}\) and
    \(\pi_{i,k} \sim \pi_{i,j} \circ \pi_{j,k}\).} %
  \label{epsilon-pi-comms}
\end{align}

\begin{example}
  If \(I \colonequiv \Nat\) with the usual ordering, then we are looking at a
  diagram of \(\V\)-dcpos like this
  \[
    \begin{tikzcd}
      D_0 \ar[r,"\varepsilon_{0,1}", shift left, hookrightarrow] &
      D_1 \ar[l,"\pi_{0,1}", shift left, two heads]
      \ar[r,"\varepsilon_{1,2}", shift left, hookrightarrow ]
      &
      D_2 \ar[l,"\pi_{1,2}", shift left, two heads]
      \ar[r,"\varepsilon_{2,3}", shift left, hookrightarrow ]
      &
      D_3 \ar[l,"\pi_{2,3}", shift left, two heads]
      \ar[r,"\varepsilon_{3,4}", shift left, hookrightarrow ]
      &
      \cdots
      \ar[l,"\pi_{3,4}", shift left, two heads]
    \end{tikzcd}
  \]
  where, for example, we have not pictured
  \(\varepsilon_{1,1} : D_1 \hookrightarrow D_1\) and
  \(\varepsilon_{0,2} : D_0 \hookrightarrow D_2\) explicitly, as they are equal
  to \(\id_{D_1} : D_1 \to D_1\) and the composition of
  \(D_0 \xhookrightarrow{\varepsilon_{0,1}} D_1\) and
  \(D_1\xhookrightarrow{\varepsilon_{1,2}} D_2\), respectively.
\end{example}

The goal is now to construct another \(\V\)-dcpo \(D_\infty\) with
embedding-projections pairs
\(\pa*{\varepsilon_{i,\infty} : D_1 \hookrightarrow D_\infty, {\pi_{i,\infty} :
    D_\infty \to D_i}}\) for every \(i : I\), such that
\(\pa*{D_\infty,\pa*{\varepsilon_{i,\infty}}_{i : I}}\) is the colimit of the
diagram given by \(\pa*{\varepsilon_{i,j}}_{i \below j \text{ in } I}\) and
\(\pa*{D_\infty,\pa*{\pi_{i,\infty}}_{i : I}}\) is the limit of the
diagram given by \(\pa*{\pi_{i,j}}_{i \below j \text{ in } I}\).
In other words,
\(\pa*{D_\infty,\pa*{\varepsilon_{i,\infty}}_{i : I},\pa*{\pi_{i,\infty}}_{i :
    I}}\) is both the colimit and the limit in the category of \(\V\)-dcpos with
embedding-projections pairs between them. We say that it is the \emph{bilimit}.

\begin{definition}[\(D_\infty\)]\label{def:D-infty}%
  \index{bilimit|textbf}\index{colimit}\index{limit}%
  We define a poset \(D_\infty\) as follows. Its carrier is given by dependent
  functions \(\sigma : \Pi_{i : I}D_i\) satisfying
  \(\pi_{i,j}(\sigma_j) = \sigma_i\) whenever \(i \below j\).
  \nomenclature[Dinfty]{$D_\infty$}{bilimit of a directed diagram of dcpos with
    embedding-projection pairs}
  That is, the carrier is the type
  \[
    \sum_{\sigma : \Pi_{i : I} D_i} \prod_{{i,j} : I , i \below j}
    \pi_{i,j}\pa*{\sigma_j} = \sigma_i.
  \]
  Note that this defines a subtype of \(\Pi_{i : I}D_i\) as the condition
  \(\prod_{{i,j} : I , i \below j} \pi_{i,j}\pa*{\sigma_j} = \sigma_i\) is a
  property by \cref{Pi-is-prop} and the fact that each \(D_i\) is a set.

  These functions are ordered pointwise, i.e.\ if
  \(\sigma,\tau : \Pi_{i : I} D_i\), then \(\sigma \below_{D_\infty} \tau\)
  exactly when \(\sigma_i \below_{D_i} \tau_i\) for every \(i : I\).
\end{definition}

\begin{lemma}
  The poset \(D_\infty\) is \(\V\)-directed complete with suprema calculated
  pointwise.
  Paying attention to the universe levels involved, we have
  \(D_\infty : \DCPO{V}{U \sqcup V \sqcup W}{U \sqcup T}\).
\end{lemma}
\begin{proof}
  If \(\alpha : A \to D_\infty\) is a directed family, then the family
  \(\alpha_i : A \to D_i\) given by
  \(\alpha_i(a) \colonequiv \pa*{\alpha(a)}_i\) is directed again, and we define
  the supremum of \(\alpha\) in \(D_\infty\) as the function
  \(i \mapsto \bigsqcup \alpha_i\).
  To see that this indeed defines an element of \(D_\infty\), observe that for
  every \(i,j : I\) with \(i \below j\) we have
  \begin{align*}
    \pi_{i,j}\big(\pa*{\textstyle\bigsqcup \alpha}_j\big)
    &\equiv \pi_{i,j}\pa*{\textstyle\bigsqcup\alpha_j}
    \\ &= \textstyle\bigsqcup \pi_{i,j} \circ \alpha_j
       &&\text{(by Scott continuity of \(\pi_{i,j}\))}
    \\ &\equiv \textstyle\bigsqcup_{a : A}\big(\pi_{i,j}\big(\pa*{\alpha(a)}_j\big)\big)
    \\ &= \textstyle\bigsqcup \alpha_i
       &&\text{(as \(\alpha(a)\) is an element of \(D_\infty\))},
  \end{align*}
  as desired.
\end{proof}

\begin{remark}\label{bilimit-universe-parameters}\index{universe!parameters}%
  We allow for general universe levels here, which is why \(D_\infty\) lives in
  the relatively complicated universe \(\U \sqcup \V \sqcup \W\). In concrete
  examples, the situation often simplifies. E.g., in \cref{sec:Scott-D-infty} we
  find ourselves in the favourable situation described
  in~\cref{exponential-universe-parameters} where \(\V \equiv \W \equiv \U_0\)
  and \(\U \equiv \T \equiv \U_1\), so that we get
  \(D_\infty : \DCPOnum{0}{1}{1}\), as the bilimit of a diagram of dcpos
  \(D_n : \DCPOnum{0}{1}{1}\) indexed by natural numbers.
\end{remark}

\begin{definition}[\(\pi_{i,\infty}\)]\label{pi-infty}
  For every \(i : I\), we define the Scott continuous function
  \(\pi_{i,\infty} : {D_\infty \to D_i}\) by \(\sigma \mapsto \sigma_i\).%
  \index{bilimit!projection}%
  \nomenclature[piiinfty]{$\pi_{i,\infty}$}{projection from \(D_\infty\) to
    \(D_i\)}
\end{definition}

\begin{lemma}\label{pi-infty-is-continuous}
  The map \(\pi_{i,\infty} : D_\infty \to D_i\) is Scott continuous for every
  \(i : I\).
\end{lemma}
\begin{proof}
  This holds because suprema in \(D_\infty\) are calculated pointwise and
  \(\pi_{i,\infty}\) selects the \(i\)-th component.
\end{proof}

While we could closely follow~\cite{Scott1972} up until this point, we will now
need a new idea to proceed.
Our goal is to define maps \(\varepsilon_{i,\infty} : D_i \to D_\infty\) for
every \(i : I\) so that \(\varepsilon_{i,\infty}\) and \(\pi_{i,\infty}\) form
an embedding-projection pair.
We give an outline of the idea for defining this map
\(\varepsilon_{i,\infty}\). For~an arbitrary element \(x : D_i\), we need to
construct \(\sigma : D_\infty\) at component \(j : I\), say. If we had \(k : I\)
such that \(i,j \below k\), then we could define \(\sigma_j : D_j\) by
\(\pi_{j,k}\pa*{\varepsilon_{i,k}(x)}\).
Now semidirectedness of \(I\) tells us that there exists such a \(k : I\), so
the point is to somehow make use of this propositionally truncated fact. This is
where~\cref{constant-map-to-set-factors-through-truncation} comes
in.\index{propositional truncation} %
We define a map
\(\kappa_{i,j}^x : \pa*{\Sigma_{k : I}\,\pa{i \below k} \times \pa{j \below k}}
\to D_j\) by sending \(k\) to \(\pi_{j,k}\pa*{\varepsilon_{i,k}(x)}\) and show
it to be constant, so that it factors through the truncation of its domain.
In the special case that \(I \equiv \Nat\), as in~\cite{Scott1972}, we could
simply take \(k\) to be the sum of the natural numbers \(i\) and \(j\), but this
does not work in the more general directed case, of course.

\begin{definition}[\(\kappa_{i,j}^x\)]\label{def:kappa}
  For every \(i,j : I\) and \(x : D_i\) we define the function
  \[
    \kappa_{i,j}^x : \pa*{\Sigma_{k : I}\,\pa{i \below k} \times \pa{j \below
        k}} \to D_j
  \]
  by mapping \(k : I\) with \(i,j \below k\) to
  \(\pi_{j,k}\pa*{\varepsilon_{i,k}(x)}\).
\end{definition}

\begin{lemma}\label{kappa-is-constant}
  The function \(\kappa_{i,j}^x\) is constant for ever \(i,j : I\) and
  \(x : D_i\).
  \nomenclature[kappaijx]{$\kappa_{i,j}^x$}{auxiliary function for defining
    \(\rho_{i,j} : D_i \to D_j\)}
  Hence, \(\kappa_{i,j}^x\) factors through
  \(\exists_{k : I}\,\pa{i \below k}\times\pa{j \below k}\) by
  \cref{constant-map-to-set-factors-through-truncation}.
\end{lemma}
\begin{proof}
  If we have \(k_1,k_2 : I\) with \(i \below k_1,k_2\) and \(j \below k_1,k_2\),
  then by semidirectedness of \(I\), there exists some \(k : K\) with
  \(k_1,k_2 \below k\) and hence,
  \begin{align*}
    \pa*{\pi_{j,k_1} \circ \varepsilon_{i,k_1}} (x)
    &= \pa*{\pi_{j,k_1} \circ \pi_{k_1,k} \circ \varepsilon_{k_1,k} \circ \varepsilon_{i,k_1}} (x)
      &&\text{(since \(\varepsilon_{k_1,k}\) is a section of \(\pi_{k_1,k}\))}
    \\
    &= \pa*{\pi_{j,k} \circ \varepsilon_{i,k}}(x)
      &&\text{(by \cref{epsilon-pi-comms})}
    \\
    &= \pa*{\pi_{j,k} \circ \pi_{k_2,k} \circ \varepsilon_{k_2,k} \circ \varepsilon_{i,k_2}} (x)
      &&\text{(since \(\varepsilon_{k_2,k}\) is a section of \(\pi_{k_2,k}\))}
    \\
    &= \pa*{\pi_{j,k_2} \circ \varepsilon_{i,k_2}}(x)
      &&\text{(by \cref{epsilon-pi-comms})},
  \end{align*}
  proving that \(\kappa_{i,j}^x\) is constant.
\end{proof}

\begin{definition}[\(\rho_{i,j}\)]
  For every \(i,j : I\), the type
  \(\exists_{k : I}\,\pa{i \below k} \times \pa{j \below k}\) has an element
  since \((I,\below)\) is directed. Thus, \cref{kappa-is-constant} tells us that
  we have a function \(\rho_{i,j} : D_i \to D_j\) such that if \(i,j \below k\),
  then the equation
  \begin{equation}\label{rho-eq}
    \rho_{i,j}(x) = \kappa_{i,j}^x(k) \equiv \pi_{j,k}\pa*{\varepsilon_{i,k}(x)}
  \end{equation}
  holds for every \(x : D_i\).
  \nomenclature[rhoij]{\(\rho_{i,j}\)}{auxiliary map from \(D_i\) to \(D_j\) for
    defining the embedding \(\varepsilon_{i,\infty} : D_i \to D_\infty\)}
\end{definition}

\begin{definition}[\(\varepsilon_{i,\infty}\)]\label{epsilon-infty}
  The map \(\rho\) induces a map \(\varepsilon_{i,\infty} : D_i \to D_\infty\)
  by sending \(x : D_i\) to the function \(\lambdadot{j : I}{\rho_{i,j}(x)}\).%
  \index{bilimit!embedding}%
  \nomenclature[epsiloniinfty]{\(\varepsilon_{i,\infty}\)}{embedding from \(D_i\)
    to \(D_\infty\)}
  To see that this is well-defined, assume that we have \(j_1 \below j_2\) in
  \(J\) and \(x : D_i\). We have to show that
  \(\pi_{j_1,j_2}\pa*{\pa*{\varepsilon_{i,\infty}(x)}_{j_2}} =
  \pa*{\varepsilon_{i,\infty}(x)}_{j_1}\).
  By semidirectedness of \(I\) and the fact that are looking to prove a
  proposition, we may assume to have \(k : I\) with \(i \below k\) and
  \(j_1 \below j_2 \below k\). Then,
  \begin{align*}
    \pi_{j_1,j_2}\pa*{\pa*{\varepsilon_{i,\infty}(x)}_{j_2}}
    &\equiv \pi_{j_1,j_2}\pa*{\rho_{i,j_2}(x)} \\
    &= \pi_{j_1,j_2}\pa*{\pi_{j_2,k}\pa*{\varepsilon_{i,k}(x)}}
    &&\text{(by \cref{rho-eq})} \\
    &= \pi_{j_1,k}\pa*{\varepsilon_{i,k}(x)}
    &&\text{(by \cref{epsilon-pi-comms})}
    \\
    &= \rho_{i,j_1}(x)
    &&\text{(by \cref{rho-eq})} \\
    &\equiv \pa*{\varepsilon_{i,\infty}(x)}_{j_1}
  \end{align*}
  as desired.
\end{definition}

This completes the definition of \(\varepsilon_{i,\infty}\). From this point on,
we can typically work with it by using~\cref{rho-eq} and the fact that
\(\pa*{\varepsilon_{i,\infty}(x)}_j\) is defined as \(\rho_{i,j}(x)\).

\begin{lemma}\label{rho-is-continuous}
  The map \(\rho_{i,j} : D_i \to D_j\) is Scott continuous for every \(i,j : I\).
\end{lemma}
\begin{proof}
  Since we are proving a property, we may use semidirectedness of \(I\) to get
  \(k : I\) with \(i,j \below k\). Then,
  \(\rho_{i,j} \sim \pi_{j,k} \circ \varepsilon_{i,k}\) by \cref{rho-eq}. But
  the functions \(\pi_{j,k}\) and \(\varepsilon_{i,k}\) are continuous and
  continuity is preserved by function composition, so \(\rho_{i,j}\) is
  continuous, as we wished to show.
\end{proof}

\begin{lemma}\label{epsilon-infty-is-continuous}
  The map \(\varepsilon_{i,\infty} : D_i \to D_\infty\) is Scott continuous for
  every \(i : I\).
\end{lemma}
\begin{proof}
  If \(\alpha : A \to D_i\) is directed, then for every \(j : I\) we have
  \begin{align*}
    \pa*{\varepsilon_{i,\infty}\pa*{\textstyle\bigsqcup \alpha}}_j
    &\equiv \rho_{i,j} \pa*{\textstyle\bigsqcup \alpha} \\
    &= \textstyle\bigsqcup \rho_{i,j} \circ \alpha
    &&\text{(by \cref{rho-is-continuous})} \\
    &\equiv \textstyle\bigsqcup_{a : A} \pa*{\varepsilon_{i,\infty}\pa*{\alpha(a)}}_j \\
    &\equiv \pa*{\textstyle\bigsqcup \pa*{\varepsilon_{i,\infty} \circ \alpha}}_j
      &&\text{(as suprema in \(D_\infty\) are calculated pointwise)}.
  \end{align*}
  Hence,
  \(\varepsilon_{i,\infty}\pa*{\bigsqcup \alpha} = \bigsqcup
  \pa*{\varepsilon_{i,\infty} \circ \alpha}\) and \(\varepsilon_{i,\infty}\) is
  seen to be Scott continuous.
\end{proof}

\begin{theorem}\label{epsilon-pi-infty-ep-pair}%
  \index{bilimit!embedding-projection pair}%
  For every \(i : I\), the pair
  \(\pa*{\varepsilon_{i,\infty},\pi_{i,\infty}}\) is an embedding-projection
  pair from \(D_i\) to \(D_\infty\).
\end{theorem}
\begin{proof}
  Scott continuity of both maps is given by
  \cref{pi-infty-is-continuous,epsilon-infty-is-continuous}. To see that
  \(\varepsilon_{i,\infty}\) is a section of \(\pi_{i,\infty}\), observe that
  for every \(x : D_i\), we have
  \begin{align*}
    \pi_{i,\infty}\pa*{\varepsilon_{i,\infty}(x)}
    &\equiv \pa*{\varepsilon_{i,\infty}(x)}_i \\
    &\equiv \rho_{i,i}(x) \\
    &= \pi_{i,i}\pa*{\varepsilon_{i,i}(x)}
      &&\text{(by \cref{rho-eq})}
    \\
    &\equiv x &&\text{(by \cref{epsilon-pi-id})},
  \end{align*}
  so that \(\varepsilon_{i,\infty}\) is indeed a section of \(\pi_{i,\infty}\).
  It remains to prove that
  \(\varepsilon_{i,\infty}\pa*{\pi_{i,\infty}\pa{\sigma}} \below \sigma\) for
  every \(\sigma : D_\infty\). The order is given pointwise, so let \(j : I\) be
  arbitrary and since we are proving a proposition, assume that we have
  \(k : I\) with \(i,j \below k\). Then,
  \begin{align*}
    \pa*{\varepsilon_{i,\infty}\pa*{\pi_{i,\infty}(\sigma)}}_j
    &\equiv \pa*{\varepsilon_{i,\infty}\pa{\sigma_i}}_j \\
    &\equiv \rho_{i,j}\pa*{\sigma_i} \\
    &= \pi_{j,k}\pa*{\varepsilon_{i,k}\pa*{\sigma_i}}
    &&\text{(by \cref{rho-eq})}
    \\
    &= \pi_{j,k}\pa*{\varepsilon_{i,k}\pa*{\pi_{i,k}\pa*{\sigma_k}}}
    &&\text{(since \(\sigma\) is an element of \(D_\infty\))}
    \\
    \shortintertext{But \(\pi_{i,k} \circ \varepsilon_{i,k}\) is deflationary and \(\pi_{j,k}\) is monotone, so}
    &\below \pi_{j,k}\pa*{\sigma_k}
    \\
    &= \sigma_j
    &&\text{(since \(\sigma\) is an element of \(D_\infty\))},
  \end{align*}
  finishing the proof.
\end{proof}

\begin{lemma}\label{epsilon-pi-infty-commutes}
  The maps \(\pi_{i,\infty}\) and \(\varepsilon_{i,\infty}\) respectively
  commute with \(\pi_{i,j}\) and \(\varepsilon_{i,j}\) whenever \(i \below j\),
  viz.\ the diagrams
  \[
    \begin{tikzcd}
      D_\infty \ar[dr,"\pi_{j,\infty}"'] \ar[rr,"\pi_{i,\infty}"] & & D_i \\
        & D_j \ar[ur, "\pi_{i,j}"']
    \end{tikzcd}
    \quad
    \quad
    \quad
    \begin{tikzcd}
      D_i \ar[dr,"\varepsilon_{i,j}"'] \ar[rr,"\varepsilon_{i,\infty}"] & & D_\infty \\
        & D_j \ar[ur, "\varepsilon_{j,\infty}"']
    \end{tikzcd}
  \]
  commute for all \(i,j : I\) with \(i \below j\).
\end{lemma}
\begin{proof}
  If \(i \below j\) and \(\sigma : D_\infty\) is arbitrary, then
  \[
    \pi_{i,j}\pa*{\pi_{j,\infty}\pa{\sigma}} \equiv \pi_{i,j}\pa*{\sigma_j} =
    \sigma_i
  \]
  precisely because \(\sigma\) is an element of \(D_\infty\), which proves the
  commutativity of the first diagram.
  For the second, let \(x : D_i\) be arbitrary and we compare
  \(\pa*{\varepsilon_{j,\infty}\pa*{\varepsilon_{i,j}(x)}}\) and
  \(\varepsilon_{i,\infty}(x)\) componentwise. So let \(j' : I\) be
  arbitrary. Since we are proving a proposition, we may assume to have \(k : I\)
  with \(j,j' \below k\) by semidirectedness of \(I\). We now calculate that
  \begingroup
  \allowdisplaybreaks
  \begin{align*}
    \pa*{\varepsilon_{j,\infty}\pa*{\varepsilon_{i,j}(x)}}_{j'}
    &\equiv \rho_{j,j'}\pa*{\varepsilon_{i,j}(x)} \\
    &= \pi_{j',k}\pa*{\varepsilon_{j,k}\pa*{\varepsilon_{i,j}(x)}}
    &&\text{(by \cref{rho-eq})} \\
    &= \pi_{j',k}\pa*{\varepsilon_{i,k}(x)}
    &&\text{(by \cref{epsilon-pi-comms})} \\
    &= \rho_{i,j'}(x)
    &&\text{(by \cref{rho-eq})} \\
    &\equiv \pa*{\varepsilon_{i,\infty}(x)}_{j'}
  \end{align*}
  \endgroup
  as desired.
\end{proof}

\begin{theorem}\label{limit}\index{limit}%
  The \(\V\)-dcpo \(D_\infty\) with the maps \(\pa*{\pi_{i,\infty}}_{i : I}\) is
  the limit of the diagram
  \(\pa*{\pa*{D_i}_{i : I} , \pa*{\pi_{i,j}}_{i \below j}}\).
  That is, given a \(\V \)-dcpo \(E : \DCPO{V}{U'}{T'}\) and Scott continuous
  functions \(f_i : E \to D_i\) for every \(i : I\) such that the diagram
  \begin{equation}\label{fs-are-cone}
    \begin{tikzcd}
      E \ar[dr,"f_j"'] \ar[rr,"f_i"] & & D_i \\
      & D_j \ar[ur,"\pi_{i,j}"']
    \end{tikzcd}
  \end{equation}
  commutes for every \(i \below j\),
  we have a unique Scott continuous function \(f_\infty : E \to D_\infty\) making
  the diagram
  \begin{equation}\label{f-infty}
    \begin{tikzcd}
      E \ar[dr,dashed,"f_\infty"'] \ar[rr,"f_i"] & & D_i \\
      & D_\infty \ar[ur,"\pi_{i,\infty}"']
    \end{tikzcd}
  \end{equation}
  commute for every \(i : I\).
\end{theorem}
\begin{proof}
  Note that \cref{f-infty} dictates that we must have
  \(\pa*{f_\infty(y)}_i = f_i(y)\) for every \(i : I\). Hence, we define
  \(f_\infty : E \to D_\infty\) as
  \(f_\infty(y) \colonequiv \lambdadot{i : I}{f_i(y)}\), which is Scott
  continuous because each \(f_i\) is and suprema are calculated pointwise in
  \(D_\infty\). To see that \(f_\infty\) is well-defined, i.e.\ that
  \(f_\infty(y)\) is indeed an element of \(D_\infty\), observe that for every
  \(i \below j\), the equation
  \(\pi_{i,j}\pa*{\pa*{f_\infty(y)}_j} \equiv \pi_{i,j}\pa*{f_j(y)} = f_i(y)\)
  holds because of \cref{fs-are-cone}.
\end{proof}

It should be noted that in the above universal property \(E\) can have its
carrier in any universe \(\U'\) and its order taking values in any universe
\(\T'\), even though we required all \(D_i\) to have their carriers and orders
in two fixed universes \(\U \) and \(\T \), respectively.

\begin{lemma}\label{epsilon-infty-fam-monotone}
  If \(i \below j\) in \(I\), then
  \(\varepsilon_{i,\infty}\pa{\sigma_i} \below
  \varepsilon_{j,\infty}\pa{\sigma_j}\) for every \(\sigma : D_\infty\).
\end{lemma}
\begin{proof}
  The order of \(D_\infty\) is pointwise, so we compare
  \(\varepsilon_{i,\infty}\pa{\sigma_i}\) and
  \(\varepsilon_{j,\infty}\pa{\sigma_j}\) at an arbitrary component \(k : I\).
  We may assume to have \(m : I\) such that \(j,k \below m\) by semidirectedness
  of \(I\). We then calculate that
  \begingroup
  \allowdisplaybreaks
  \begin{align*}
    \pa*{\varepsilon_{i,\infty}(\sigma_i)}_k
    &\equiv \rho_{i,k}(\sigma_i) \\
    &= \pa*{\pi_{k,m} \circ \varepsilon_{i,m}}(\sigma_i)
    &&\text{(by \cref{rho-eq})} \\
    &= \pa*{\pi_{k,m} \circ \varepsilon_{i,m} \circ \pi_{i,j}}(\sigma_j)
    &&\text{(since \(\sigma\) is an element of \(D_\infty\))} \\
    &= \pa*{\pi_{k,m} \circ \varepsilon_{j,m} \circ \varepsilon_{i,j} \circ\pi_{i,j}}(\sigma_j)
    &&\text{(by \cref{epsilon-pi-comms})} \\
    \shortintertext{But \(\varepsilon_{i,j} \circ\pi_{i,j}\) is deflationary
    and \(\pi_{k,m} \circ \varepsilon_{j,m}\) is monotone, so}
    &\below \pa*{\pi_{k,m} \circ \varepsilon_{j,m}}(\sigma_j) \\
    &= \rho_{j,k}(\sigma_j)
    &&\text{(by \cref{rho-eq})} \\
    &\equiv \pa*{\varepsilon_{j,\infty}(\sigma_j)}_k,
  \end{align*}
  \endgroup
  as we wished to show.
\end{proof}

\begin{lemma}\label{sigma-sup-of-epsilon-pis}
  Every element \(\sigma : D_\infty\) is equal to the directed supremum
  \(\bigsqcup_{i : I} \varepsilon_{i,\infty}\pa*{\sigma_i}\).
\end{lemma}
\begin{proof}
  The domain of the family is inhabited, because \((I,\below)\) is assumed to be
  directed.
  Moreover, if we have \(i,j : I\), then there exists \(k : I\) with
  \(i,j \below k\), which implies
  \(\varepsilon_{i,\infty}(\sigma_i),\varepsilon_{j,\infty}(\sigma_j) \below
  \varepsilon_{k,\infty}(\sigma_k)\) by \cref{epsilon-infty-fam-monotone}.
  Thus, the family \(i \mapsto \varepsilon_{i,\infty}(\sigma_i)\) is indeed
  directed.
  To see that its supremum is indeed \(\sigma\) we use antisymmetry at an
  arbitrary component \(j : I\).
  Firstly, observe that
  \begin{align*}
    \sigma_j
    &= \pi_{j,j}\pa*{\varepsilon_{j,j}(\sigma_j)} &&\text{(by \cref{epsilon-pi-id})} \\
    &= \rho_{j,j}(\sigma_j) &&\text{(by \cref{rho-eq})} \\
    &\equiv \pa*{\varepsilon_{j,\infty}(\sigma_j)}_j \\
    &\below \pa*{\textstyle\bigsqcup_{i : I}\varepsilon_{i,\infty}(\sigma_i)}_j
    &&\text{(since suprema are computed pointwise in \(D_\infty\))}.
  \end{align*}
  Secondly, to prove that
  \(\pa*{\bigsqcup_{i : I}\varepsilon_{i,\infty}(\sigma_i)}_j \below \sigma_j\)
  it suffices to show that
  \(\pa*{\varepsilon_{i,\infty}(\sigma_i)}_j \below \sigma_j\) for every
  \(i : I\). But this just says that
  \(\varepsilon_{i,\infty} \circ \pi_{i,\infty}\) is a deflation, which was
  proved in \cref{epsilon-pi-infty-ep-pair}.
\end{proof}

Although the composites \(\varepsilon_{i,\infty} \circ \pi_{i,\infty}\) are
deflations for each \(i : I\), the supremum of all of them is the identity. This
fact will come in useful in~\cref{sec:Scott-D-infty}.
\begin{lemma}\label{epsilon-pi-sup}
  The family \(i \mapsto \varepsilon_{i,\infty} \circ \pi_{i,\infty}\) is
  directed in the exponential \(D_\infty^{D_\infty}\) and its supremum is the
  identity on \(D_\infty\).
\end{lemma}
\begin{proof}
  The order and suprema are given pointwise in exponentials, so this follows
  from~\cref{sigma-sup-of-epsilon-pis}.
\end{proof}

\begin{theorem}\label{colimit}\index{colimit}
  The \(\V\)-dcpo \(D_\infty\) with the maps
  \(\pa*{\varepsilon_{i,\infty}}_{i : I}\) is the colimit of the diagram
  \(\pa*{\pa*{D_i}_{i : I} , \pa*{\varepsilon_{i,j}}_{i \below j}}\).
  That is, given a \(\V\)-dcpo \(E : \DCPO{V}{U'}{T'}\) and Scott continuous
  functions \(g_i : D_i \to E\) for every \(i : I\) such that the diagram
  \begin{equation}\label{gs-are-cocone}
    \begin{tikzcd}
      D_i \ar[dr,"\varepsilon_{i,j}"'] \ar[rr,"g_i"] & & E \\
      & D_j \ar[ur,"g_j"']
    \end{tikzcd}
  \end{equation}
  commutes for every \(i \below j\),
  we have a unique Scott continuous function \(g_\infty : D_\infty \to E\) making
  the diagram
  \begin{equation}\label{g-infty}
    \begin{tikzcd}
      D_i \ar[dr,"\varepsilon_{i,\infty}"'] \ar[rr,"g_i"] & & E \\
      & D_\infty \ar[ur,dashed,"g_\infty"']
    \end{tikzcd}
  \end{equation}
  commute for every \(i : I\).
\end{theorem}
\begin{proof}
  Note that any such Scott continuous function \(g_\infty\) must satisfy
  \begin{align*}
    g_\infty(\sigma)
    &= g_\infty\pa*{\textstyle\bigsqcup_{i : I} \varepsilon_{i,\infty}(\sigma_i)}
    &&\text{(by \cref{sigma-sup-of-epsilon-pis})} \\
    &= \textstyle\bigsqcup_{i : I} g_\infty\pa*{\varepsilon_{i,\infty}(\sigma_I)}
    &&\text{(as \(g_\infty\) is assumed to be Scott continuous)} \\
    &= \textstyle\bigsqcup_{i : I}g_i(\sigma_i)
    &&\text{(by \cref{g-infty})}
  \end{align*}
  for every \(\sigma : D_\infty\).
  Accordingly, we define \(g_\infty\) by
  \(g_\infty(\sigma) \colonequiv \bigsqcup_{i : I}g_i(\sigma_i)\), where we
  verify that the family is indeed directed:
  If we have \(i,j : I\), then there exists \(k : I\) with \(i,j \below k\), and
  we have
  \begin{align*}
    g_i(\sigma_i)
    &= g_i\pa*{\pi_{i,k}(\sigma_k)}
    &&\text{(since \(\sigma\) is an element of \(D_\infty\))} \\
    &= g_k\pa*{\varepsilon_{i,k}\pa*{\pi_{i,k}(\sigma_k)}}
    &&\text{(by \cref{gs-are-cocone})} \\
    &\below g_k(\sigma_k)
    &&\text{(since \(\varepsilon_{i,k} \circ \pi_{i,k}\) is deflationary and \(g_k\) is monotone)},
  \end{align*}
  and similarly, \(g_j(\sigma_j) \below g_k(\sigma_k)\).
  To see that \(g_\infty\) satisfies \cref{g-infty}, let \(x : D_i\) be
  arbitrary and first observe that
  \[
    g_\infty\pa*{\varepsilon_{i,\infty}(x)} \equiv \textstyle\bigsqcup_{j :
      I}g_j\pa*{\pa*{\varepsilon_{i,\infty}(x)}_j} \equiv \textstyle\bigsqcup_{j
      : I}g_j\pa*{\rho_{i,j}(x)}.
  \]
  We now use antisymmetry to prove that this is equal to \(g_i(x)\).
  In one direction this is easy as
  \( g_i(x) = \pa*{g_i \circ \pi_{i,i} \circ \varepsilon_{i,i}}(x) \equiv
  g_i\pa*{\rho_{i,i}(x)} \below \textstyle\bigsqcup_{j :
    I}g_j\pa*{\rho_{i,j}(x)} \).  In the other direction, it suffices to prove
  that \(g_j(\rho_{i,j}(x)) \below g_i(x)\) for every \(j : I\). By directedness
  of \(I\) there exists \(k : I\) with \(i,j \below k\) so that
  \begin{align*}
    g_j(\rho_{i,j}(x))
    &= \pa*{g_j \circ \pi_{j,k} \circ \varepsilon_{i,k}}(x)
    &&\text{(by \cref{rho-eq})} \\
    &= \pa*{g_k \circ \varepsilon_{j,k} \circ \pi_{j,k} \circ \varepsilon_{i,k}}(x)
    &&\text{(by \cref{gs-are-cocone})} \\
    \shortintertext{But \(\varepsilon_{j,k} \circ \pi_{j,k}\) is deflationary and \(g_k\) is monotone, so}
    &\below \pa*{g_k \circ \varepsilon_{i,k}}(x) \\
    &= g_i(x)
    &&\text{(by \cref{gs-are-cocone})},
  \end{align*}
  as we wished to show.

  Finally, we verify that \(g_\infty\) is Scott continuous. We first check that
  \(g_\infty\) is monotone. If \(\sigma \below \tau\) in \(D_\infty\), then
  \(g_\infty(\sigma) \equiv \bigsqcup_{i : I} g_i(\sigma_i) \below \bigsqcup_{i
    : I} g_i(\tau_i) \equiv g_\infty(\tau)\), as each \(g_i\) is monotone.
  It remains to show that
  \(g_\infty\pa*{\bigsqcup \alpha} \below \bigsqcup \pa*{g_\infty \circ
    \alpha}\) for every directed family \(\alpha : A \to D_\infty\).
  By definition of \(g_\infty\), it suffices to show that
  \(g_i\pa*{\pa*{\bigsqcup\alpha}_i} \below \bigsqcup \pa*{g_\infty \circ
    \alpha}\) for every \(i : I\).
  By continuity of \(g_i\) it is enough to establish that
  \(g_i\pa*{\pa*{\alpha(a)}_i} \below \bigsqcup \pa*{g_\infty \circ \alpha}\)
  for every \(a : A\). But this holds as
  \( g_i\pa*{\pa*{\alpha(a)}_i} \below g_\infty(\alpha(a)) \below \bigsqcup
  \pa*{g_\infty \circ \alpha} \), completing our proof.
\end{proof}

\begin{proposition}\label{locally-small-bilimit}\index{bilimit!locally small}
  The bilimit of locally small dcpos is locally small, i.e.\ if every \(\V\)-dcpo
  \(D_i\) is locally small for all \(i : I\), then so is \(D_\infty\).
\end{proposition}
\begin{proof}
  If every \(D_i\) is locally small, then for every \(i : I\), we have a
  \emph{specified} \(\V\)-valued partial order \(\below_{\V}^i\) on \(D_i\) such
  that for every \(i : I\) and every \(x,y : D_i\), we have an equivalence
  \(\pa{x \below_{D_i} y} \simeq \pa{x \below_{\V}^i y}\).
  Hence,
  \(\pa{\sigma \below_{D_\infty} \tau} \equiv \pa{\Pi_{i : I}\pa{\sigma_i
      \below_{D_i} \tau_i}} \simeq \pa{\Pi_{i : I}\pa{\sigma_i \below_{\V}^i
      \tau_i}}\), but the latter is small, because \(I : \V\) and
  \(\below_{\V}^i\) is \(\V\)-valued.
\end{proof}

\section{Notes}\label{sec:basic-domain-theory-notes}

This chapter is based on our two publications~\cite{deJongEscardo2021a} and
\cite{deJong2021a}, but with an improved exposition and the inclusion of more
proofs.
More precisely,
\cref{sec:directed-complete-posets,sec:Scott-continuous-maps,sec:lifting} and
the exponentials of \cref{sec:products-and-exponentials} feature in both of
these works, whereas the material of \cref{sec:basic-domain-theory-introduction}
and \cref{sec:bilimits} only appears in \cite{deJongEscardo2021a}, while
\cref{least-fixed-point} is only included in~\cite{deJong2021a}.
In~\cite{deJongEscardo2021a,deJong2021a} the definition of a poset included the
requirement that the carrier is a set, because we only realised later that this
was redundant~(\cref{posets-are-sets}).
Products of dcpos were not discussed in these works, but were, building on our
previous work, formalised in \Agda\ by Brendan Hart~\cite{Hart2020} for a
final year MSci project supervised by Mart\'in Escard\'o and myself.%
\index{dcpo!product}\index{formalisation}\index{proof assistant!Agda@\Agda}

\cref{sec:directed-complete-posets,sec:Scott-continuous-maps,sec:products-and-exponentials}
are predicative universe-aware adaptations of classical domain theory as
expounded in~\cite{AbramskyJung1994,GierzEtAl2003}, while \cref{sec:bilimits} is
a predicative account of Scott's original paper~\cite{Scott1972}, but with
important differences, including an application
of~\cref{constant-map-to-set-factors-through-truncation}, as discussed at the
start of~\cref{sec:bilimits}.

\cref{sec:lifting} uses the lifting monad in univalent type theory, which
originated with the works~\cite{EscardoKnapp2017,Knapp2018} that deal with
partiality in univalent foundations and aim to avoid (weak) countable choice,
which is not provable in constructive univalent
foundations~\cite{CoquandMannaaRuch2017,Coquand2018,Swan2019a,Swan2019b}.%
\index{choice!axiom of countable}\index{lifting!monad}\index{constructivity}%
\index{partial}%
This is to be contrasted to other approaches to partiality in Martin-L\"of Type
Theory. The first is Capretta's delay monad~\cite{Capretta2005}, which uses
coinduction. %
\index{delay monad}\index{coinduction}
Arguably, the correct notion of equality of Capretta's delay monad
is that of weak bisimilarity where two partial elements are considered equal
when they are both undefined, or, when one of them is, so is the other and they
have equal values in this case.%
\index{bisimilarity}
This prompted the authors of~\cite{ChapmanUustaluVeltri2019} to consider its
quotient by weak bisimilarity, but they use countable choice to show that
quotient is again a monad. Again using countable choice, they show that their
quotient yields free pointed \(\omega\)-complete posets (\(\omega\)-cpos).%
\index{omega-completeness@\(\omega\)-completeness}
In~\cite{AltenkirchDanielssonKraus2017} the authors use a so-called quotient
inductive-inductive type (QIIT) to construct the free pointed \(\omega\)-cpo,
essentially by definition of the QIIT. It~was shown
in~\cite{ChapmanUustaluVeltri2019} that a simpler higher inductive type actually
suffices.
Regardless, we stress that our approach yields free dcpos as opposed to
\(\omega\)-cpos and does not use countable choice or higher inductive types
other than the propositional truncation.
But the notion of \(\omega\)-completeness and countable choice will resurface in
our discussion of the Scott model of PCF in~\cref{sec:Scott-model-of-PCF}.


%% file: mainmatter/continuous-and-algebraic.tex
\chapter{Continuous and algebraic dcpos}
\label{chap:continuous-and-algebraic-dcpos}

In the previous chapter we developed sufficient domain theory for the
applications considered in \cref{chap:applications}, but we have not yet
discussed a fundamental topic in domain theory: \emph{algebraic} and
\emph{continuous} dcpos.
The study of continuous and algebraic dcpo is a rich and deep
subject~\cite{GierzEtAl2003}. We present a treatment of the basic theory and
examples in our constructive and predicative approach, where we deal with size
issues by taking direct inspiration from category theory and the
work of Johnstone and Joyal~\cite{JohnstoneJoyal1982} in particular.

\section{Introduction}\label{sec:continuous-and-algebraic-introduction}

Classically, a dcpo \(D\) is said to be \emph{continuous} if for every element
\(x\) of \(D\) the set of elements \emph{way below} it is directed and has
supremum \(x\).
The problem with this definition in our foundational setup is that the type of
elements way below \(x\) is not necessarily small. Although this does not stop
us from asking it to be directed and having supremum \(x\), this still poses a
problem: for example, there would be no guarantee that its supremum is preserved
by a Scott continuous function, as it is only required to preserve suprema of
directed families indexed by small types.

Our solution is to take inspiration from category
theory~\cite{JohnstoneJoyal1982} and to use the ind-completion to give a
predicatively suitable definition of continuity of a dcpo.
Some care is needed to ensure that the resulting definition expresses a property
of a dcpo, rather than an equipment with additional structure. This is of course
where the propositional truncation comes in useful, but there are two natural
ways of using the truncation. We show that one of them yields a well-behaved
notion that serves as our definition of continuity, while the other, which we
call \emph{pseudocontinuity}, is problematic in a constructive context. In a
classical setting where the axiom of choice is assumed, the two notions
(continuity and pseudocontinuity) are equivalent.

Another approach is to turn to the notion of a
basis~\cite[Section~2.2.2]{AbramskyJung1994}, but to include smallness
conditions. While we cannot expect the type of elements way below an element
\(x\) to be small, in many examples it is the case that the type of \emph{basic}
elements way below \(x\) is small.
We show that if a dcpo has a small basis, then it is continuous. In fact, all
our running examples of continuous dcpos are actually examples of dcpos with
small bases. Moreover, dcpos with small bases are better behaved. For example,
they are locally small and so are their exponentials.
Furthermore, we show that having a small basis is equivalent to being presented
by ideals.

Once we have carefully set up predicatively suitable notions of continuity and
small bases, the theory can be developed quite smoothly. Specifically, we
discuss
\begin{description}[leftmargin=!,labelwidth=\widthof{Section 7.77:}]
\item[\cref{sec:way-below}:] the way-below relation and compact elements;
\item[\cref{sec:ind-completion}:] the ind-completion: a tool used in discussing
  (pseudo) continuity;
\item[\cref{sec:continuous-dcpos}:] continuity of a dcpo and the interpolation
  property of the way-below relation;
\item[\cref{sec:pseudocontinuity}:] pseudocontinuity of a dcpo and issues
  concerning the axiom of choice;
\item[\cref{sec:algebraic-dcpos}:] algebraicity of a dcpo;
\item[\cref{sec:small-bases}:] the notion of a small basis: a strengthening of
  continuity;
\item[\cref{sec:small-compact-bases}:] the notion of a small compact basis: a
  strengthening of algebraicity;
\item[\cref{sec:algebraic-examples}:] examples of dcpos with small compact
  bases: the type of subsingletons, the lifting of a set, and the powerset; and
  an example of an algebraic dpco that does not necessarily have a small basis;
\item[\cref{sec:rounded-ideal-completion}:] the (rounded) ideal completion of an
  abstract basis, including an example of a dcpo with a small basis that is not
  algebraic: the ideal completion of inductively defined dyadic rationals;
\item[\cref{sec:ideal-completions-of-small-bases}:] the ideal completion of a
  small (compact) basis and its relation to the original dcpo;
\item[\cref{sec:structurally-continuous-and-algebraic-bilimits}:] bilimits of
  structurally continuous or algebraic dcpos (with small (compact) bases); and
\item[\cref{sec:exponentials-with-small-bases}:] exponentials of sup-complete
  dcpos with small (compact) bases.
\end{description}

\section{The way-below relation and compactness}\label{sec:way-below}

The way-below relation is the fundamental ingredient in the development of
continuous dcpos. Following Scott~\cite{Scott1970}, we intuitively think of
\(x \ll y\) as saying that every computation of \(y\) has to print \(x\), or
something better than \(x\), at some stage.

\begin{definition}[Way-below relation, \(x \ll y\)]\label{def:way-below}%
  \index{way-below relation|textbf}%
  An element \(x\) of a \(\V\)-dcpo \(D\) is \emph{way below} an element \(y\)
  of \(D\) if whenever we have every directed family \(\alpha : I \to D\)
  indexed by \(I : \V\) such that \(y \below \bigsqcup \alpha\), then there
  exists \(i : I\) such that \(x \below \alpha_i\) already.
  We denote this situation by \(x \ll y\).
  \nomenclature[ll]{$x \ll y$}{way-below relation}
\end{definition}

\begin{lemma}\label{way-below-properties}
  The way-below relation enjoys the following properties:
  \begin{enumerate}[(i)]
  \item\label{way-below-prop-valued} it is proposition-valued;
  \item\label{below-if-way-below} if \(x \ll y\), then \(x \below y\);
  \item\label{below-way-below-way-below} if \(x \below y \ll v \below w\), then
    \(x \ll w\);
  \item\label{way-below-antisymmetric} it is antisymmetric;
  \item\label{way-below-transitive} it is transitive.
  \end{enumerate}
\end{lemma}
\begin{proof}
  \ref{way-below-prop-valued} By \cref{Pi-is-prop} and the fact that we
  propositionally truncated the existence of \(i : I\) in the definition.
  \ref{below-if-way-below} Simply take \(\alpha : \One_{\V} \to D\) to be
  \(u \mapsto y\).
  \ref{below-way-below-way-below} Suppose that \(\alpha : I \to D\) is directed
  with \(w \below \bigsqcup \alpha\). Then \(v \below \bigsqcup \alpha\), so by
  assumption that \(y \ll v\) there exists \(i : I\) with \(y \below \alpha_i\)
  already. But then \(x \below \alpha_i\).
  \ref{way-below-antisymmetric} Follows from \ref{below-if-way-below}.
  \ref{way-below-transitive} Follows from \ref{below-if-way-below} and
  \ref{below-way-below-way-below}.
\end{proof}

In general, the way below relation is not reflexive. The elements for which it
is have a special status and are called compact. We illustrate this notion by a
series of examples.

\begin{definition}[Compactness]\index{compactness|textbf}%
  An element %
  is \emph{compact} if it is way below itself.
\end{definition}

\begin{example}\label{least-element-is-compact}\index{compactness!least element}
  The least element of a pointed dcpo is always compact.
\end{example}

\begin{example}[Compact elements in \(\Omega_{\V}\)]%
  \label{compact-elements-in-Omega}%
  \index{compactness!in the type of subsingletons}\index{type!of subsingletons}%
  The compact elements of \(\Omega_{\V}\) are exactly \(\Zero_{\V}\) and
  \(\One_{\V}\). In other words, the compact elements of \(\Omega_{\V}\) are
  precisely the decidable propositions.
\end{example}
\begin{proof}
  By \cref{least-element-is-compact} we know that \(\Zero_{\V}\) must be
  compact. For \(\One_{\V}\), suppose that we have
  \(Q_{(-)} : I \to \Omega_{\V}\) directed such that
  \(\One_{\V} \below \exists_{i : I}Q_i\). Then there exists \(i : I\) such that
  \(Q_i\) holds, and hence, \(\One_{\V} \below Q_i\).
  Now suppose that \(P : \Omega_{\V}\) is compact. We show that \(P\) is
  decidable. The family \(\alpha : \pa*{P + \One_{\V}} \to \Omega_{\V}\) given
  by \(\inl(p) \mapsto \One_{\V}\) and \(\inr(\star) \mapsto \Zero_{\V}\) is
  directed and \(P \below \bigsqcup \alpha\). Hence, by compactness, there
  exists \(i : P + \One_{\V}\) such that \(P \below \alpha_i\) already. %
  Since being decidable is a property of a proposition, we actually get such an
  \(i\) and by case distinction on it we get decidability of \(P\).
\end{proof}

\begin{example}[Compact elements in the lifting]
  \label{compact-elements-in-lifting}%
  \index{compactness!in the lifting}\index{lifting}%
  An element \((P,\varphi)\) of the lifting \(\lifting_{\V}(X)\) of a set
  \(X : \U\) is compact if and only if \(P\) is decidable.
  Hence, the compact elements of \(\lifting_{\V}(X)\) are exactly \(\bot\) and
  \(\eta(x)\) for \(x : X\).
\end{example}
\begin{proof}
  To see that compactness implies decidability of the domain of the partial
  element, we proceed as in the proof of \cref{compact-elements-in-Omega}, but
  for a partial element \((P,\varphi)\), we consider the family
  \(\alpha : \pa*{P + \One_{\V}} \to \lifting_{\V }(X)\) given by
  \(\inl(p) \mapsto \eta(\varphi(p))\) and \(\inr(\star) \mapsto \bot\).
  Conversely, if we have a partial element \((P,\varphi)\) with \(P\) decidable,
  then either \(P\) is false in which case \((P,\varphi) = \bot\) which is
  compact by \cref{least-element-is-compact}, or \(P\) holds. %
  So suppose that \(P\) holds and let \(\alpha : I \to \lifting_{\V}(X)\) be
  directed with \(P \below \bigsqcup \alpha\). %
  Since \(P\) holds, the element \(\bigsqcup \alpha\) must be defined, which
  means that there exists \(i : I\) such that \(\alpha_i\) is defined. %
  But for this \(i : I\) we also have \(\bigsqcup \alpha = \alpha_i\) by
  construction of the supremum, and hence, \(P \below \alpha_i\), proving
  compactness of \((P,\varphi)\).
\end{proof}

For characterising the compact elements of the powerset, we introduce a lemma,
as well as the notion of Kuratowski finiteness and the induction principle for
Kuratowski finite subsets.

\begin{lemma}\label{binary-join-is-compact}%
  \index{compactness!of binary joins}%
  \index{join!binary}%
  The compact elements of a dcpo are closed under (existing) binary joins.
\end{lemma}
\begin{proof}
  Suppose that \(x\) and \(y\) are compact elements of a \(\V\)-dcpo \(D\), let
  \(z\) be their least upper bound and suppose that we have \(\alpha : I \to D\)
  directed with \(z \below \bigsqcup \alpha\). Then
  \(x \below \bigsqcup \alpha\) and \(y \below \bigsqcup \alpha\), so by
  compactness there exist \(i : I\) and \(j : J\) such that
  \(x \below \alpha_i\) and \(y \below \alpha_j\). By semidirectedness of
  \(\alpha\), there exists \(k : I\) with
  \({\alpha_i,\alpha_j} \below \alpha_k\), so that \({x,y} \below
  \alpha_k\). But \(z\) is the join of \(x\) and \(y\), so
  \(z \below \alpha_k\), as desired.
\end{proof}

\begin{definition}[Total space of a subset, \(\totalspace\)]%
  \label{def:total-space}%
  \index{subset!total space}\index{total space|see {subset}}
  The \emph{total space} of a \(\T\)-valued subset \(S\) of a type \(X\) is
  defined as \(\totalspace(S) \colonequiv \Sigma_{x : X} (x \in S)\).
  \nomenclature[T]{$\totalspace(S)$}{total space of a subset \(S\)}
\end{definition}

\begin{definition}[Kuratowski finiteness]\hfill%
  \index{finiteness!Kuratowski|textbf}%
  \begin{enumerate}[(i)]
  \item A type \(X\) is \emph{Kuratowski finite} if there exists some natural
    number \(n : \Nat\) and a surjection \(e : \Fin(n) \surj X\), where
    \(\Fin(n)\) is the inductively defined type with exactly \(n\) elements.%
    \index{type!standard finite}%
    \nomenclature[Fin]{$\Fin(n)$}{standard finite type with exactly \(n\) elements}
  \item A subset is \emph{Kuratowski finite} if its total space is a Kuratowski
    finite type.%
    \qedhere
  \end{enumerate}
\end{definition}

Thus, a type \(X\) is \emph{Kuratowski finite} if its elements can be finitely
enumerated, possibly with repetitions, although the repetitions can be removed
when \(X\) has decidable equality.

\begin{lemma}\label{Kuratowski-finite-closure-properties}
  The Kuratowski finite subsets of a set are closed under finite unions and
  contain all singletons.%
  \index{join!finite}
\end{lemma}
\begin{proof}
  The empty set and any singleton are clearly Kuratowski finite. Moreover, if
  \(A\) and \(B\) are Kuratowski finite subsets, then we may assume to have
  natural numbers \(n\) and \(m\) and surjections
  \(\sigma : \Fin(n) \surj \totalspace(A)\) and
  \(\tau : \Fin(m) \surj \totalspace(B)\). We can then patch these together to
  obtain a surjection \(\Fin(n + m) \surj \totalspace(A \cup B)\), as desired.
\end{proof}

\begin{lemma}[Induction for Kuratowski finite subsets]%
  \label{Kuratowski-finite-subsets-induction}%
  \index{finiteness!Kuratowski!induction}%
  \index{join!binary}%
  A property of subsets of a type~\(X\) holds for all Kuratowski finite subsets
  of \(X\) as soon as
  \begin{enumerate}[(i)]
  \item\label{empty-set-case} it holds for the empty set,
  \item\label{singleton-case} it holds for any singleton subset, and
  \item\label{binary-union-case} it holds for \(A \cup B\), whenever it holds
    for \(A\) and \(B\).
  \end{enumerate}
\end{lemma}
\begin{proof}
  Let \(Q\) be a such a property and let \(A\) be an arbitrary Kuratowski finite
  subset of \(X\). Since \(Q\) is proposition-valued, we may assume to have a
  natural number \(n\) and a surjection
  \(\sigma : \Fin(n) \surj \totalspace(A)\). Then the subset \(A\) must be equal
  to the finite join of singletons
  \(\{\sigma_0\} \cup \{\sigma_1\} \cup \dots \cup \{\sigma_{n-1}\}\), which can
  be shown to satisfy \(Q\) by induction on~\(n\), and hence, so must \(A\).
\end{proof}

\begin{definition}[\(\beta\)]\label{def:list-to-powerset}%
  \index{type!of lists}%
  For a set \(X : \U\), we write \(\beta : \List(X) \to \powerset_{\U}(X)\) for
  the map inductively defined by \([] \mapsto \emptyset\) and
  \(x :: l \mapsto \{x\} \cup \beta(l)\).
  \nomenclature[{[]}]{$[]$}{empty list}%
  \nomenclature[colons]{$x :: l$}{adding an element to the start of a list}%
  \nomenclature[beta]{$\beta$}{canonical map from lists to subsets}
\end{definition}

\begin{lemma}\label{Kuratowski-finite-iff-list}
  A subset \(A : X \to \Omega_{\U}\) of a set \(X : \U\) is Kuratowski finite if
  and only if it is in the image of \(\beta\).
\end{lemma}
\begin{proof}
  The left to right direction follows from
  \cref{Kuratowski-finite-closure-properties}, while the converse follows easily
  from the induction principle for Kuratowski finite subsets where we use list
  concatenation in case~\ref{binary-union-case}.
\end{proof}

\begin{example}[Compact elements in \(\powerset_{\U }(X)\)]%
  \label{compact-elements-in-powerset}%
  \index{compactness!in the powerset}%
  \index{finiteness!Kuratowski}%
  \index{powerset}%
  The compact elements of \(\powerset_{\U}(X)\) for a set \(X : \U\) are exactly
  the Kuratowski finite subsets of \(X\).
\end{example}
\begin{proof}
  Suppose first that \(A : \powerset_{\U}(X)\) is a compact element. The family
  \[
    \pa*{\Sigma_{l : \List(X)}\,\beta(l) \subseteq A}
    \xrightarrow{\beta \circ \fst}
    \powerset_{\U}(X)
  \]
  is directed, as it contains \(\emptyset\) and we can concatenate lists to
  establish semidirectedness. Moreover,
  \(\pa*{\Sigma_{l : \List(X)}\,\beta(l) \subseteq A}\) lives in \(\U\) and we
  clearly have \(A \subseteq \bigsqcup {\beta \circ \fst}\). So by compactness,
  there exists \(l : \List(X)\) with \(\beta(l) \subseteq A\) such that
  \(A \subseteq \beta(l)\) already. But this says exactly that \(A\) is
  Kuratowski finite by \cref{Kuratowski-finite-iff-list}.

  For the converse we use the induction principle for Kuratowski finite subsets:
  the empty set is compact by \cref{least-element-is-compact}, singletons are
  easily shown to be compact, and binary unions are compact by
  \cref{binary-join-is-compact}.
\end{proof}

We end this section by presenting a few lemmas connecting the way-below relation
and compactness to Scott continuous sections.

\begin{lemma}\label{continuous-retraction-way-below-criterion}%
  \index{retract!Scott continuous}\index{way-below relation}
  If we have a Scott continuous retract \(\retract{D}{E}\), then \(y \ll s(x)\)
  implies \(r(y) \ll x\) for every \(x : D\) and \(y : E\).
\end{lemma}
\begin{proof}
  Suppose that \(y \ll s(x)\) and that \(x \below \bigsqcup \alpha\) for a
  directed family \(\alpha : I \to D\). Then
  \(s(x) \below s\pa*{\bigsqcup \alpha} = \bigsqcup s \circ \alpha\) by Scott
  continuity of \(s\), so there exists \(i : I\) such that
  \(y \below s(\alpha_i)\) already. Now monotonicity of \(r\) implies
  \(r(y) \below r(s(\alpha_i)) = \alpha_i\) which completes the proof that
  \(r(y) \ll x\).
\end{proof}

\begin{lemma}\label{embedding-preserves-and-reflects-way-below}%
  \index{embedding-projection pair}\index{way-below relation}%
  The embedding in an embedding-projection pair
  \(\retractalt{D}{E}{\varepsilon}{\pi}\) preserves and reflects the way-below
  relation, i.e.\ \(x \ll y \iff \varepsilon(x) \ll \varepsilon(y)\).
  In particular, an element \(x\) is compact if and only if \(\varepsilon(x)\)
  is.
\end{lemma}
\begin{proof}
  Suppose that \(x \ll y\) in \(D\) and let \(\alpha : I \to E\) be directed
  with \(\varepsilon(y) \below \bigsqcup \alpha\). Then
  \(y = \pi(\varepsilon(y)) \below \bigsqcup \pi \circ \alpha\) by Scott
  continuity of \(\pi\). Hence, there exists \(i : I\) such that
  \(x \below \pi(\alpha_i)\). But then
  \(\varepsilon(x) \below \varepsilon(\pi(\alpha_i)) \below \alpha_i\) by
  monotonicity of \(\varepsilon\) and the fact that \(\varepsilon \circ \pi\) is a
  deflation. This proves that \(x \ll y\).
  Conversely, if \(\varepsilon(x) \ll \varepsilon(y)\), then
  \(x = \pi(\varepsilon(x)) \ll y\) by
  \cref{continuous-retraction-way-below-criterion}.
\end{proof}

\section{The ind-completion}\label{sec:ind-completion}

The ind-completion will be a useful tool for phrasing and proving results about
directed complete \emph{posets} and is itself a directed complete
\emph{preorder}, cf.~\cref{ind-completion-is-directed-complete}.
It was introduced by \citeauthor{SGA41} in \cite[Section~8]{SGA41} in the
context of category theory, but its role in order theory is discussed in
\cite[Section~1]{JohnstoneJoyal1982}.
We will also use it in the context of order theory, but our treatment will
involve a careful consideration of the universes involved, very similar to the
original treatment in~\cite{SGA41}.

\begin{definition}[\(\V\)-ind-completion \(\Ind{V}(X)\), cofinality, \({\cof}\)]
  The \emph{\(\V\)-ind-completion} \(\Ind{V}(X)\) of a preorder \(X\) is the
  type of directed families in \(X\) indexed by types in \(\V\), ordered by
  cofinality.%
  \index{Ind-completion}%
  \nomenclature[Ind]{$\Ind{V}(X)$}{\(\V\)-ind-completion of a preorder \(X\)}
  A directed family \(\alpha : I \to X\) is \emph{cofinal} in
  \(\beta : J \to X\) if for every \(i : I\), there exists \(j : J\) such that
  \(\alpha_i \below \beta_j\), and we denote this by \(\alpha \cof \beta\).%
  \index{cofinality}%
  \nomenclature[leqsim]{$\alpha \cof \beta$}{cofinality relation}
\end{definition}

\begin{lemma}
  Cofinality defines a preorder on the ind-completion.
\end{lemma}
\begin{proof}
  Straightforward.
\end{proof}

\begin{lemma}\label{ind-completion-is-directed-complete}
  The \(\V\)-ind-completion \(\Ind{V}(X)\) of a preorder \(X\) is \(\V\)-directed
  complete.
\end{lemma}
\begin{proof}
  Suppose that we have a directed family \(\alpha : I \to \Ind{V}(X)\) with
  \(I : \V\). Then each \(\alpha_i\) is a directed family in \(X\) indexed by a
  type \(J_i : \V\). We define the family
  \(\hat\alpha : \pa*{\Sigma_{i : I}J_i} \to X\) by \((i,j) \mapsto
  \alpha_i(j)\). It is clear that each \(\alpha_i\) is cofinal in
  \(\hat\alpha\), and that \(\hat\alpha\) is cofinal in \(\beta\) if every
  \(\alpha_i\) is cofinal in \(\beta\). So it remains to show that
  \(\hat\alpha\) is in fact an element of \(\Ind{V}(X)\), i.e.\ that it is
  directed. Because \(\alpha\) and each \(\alpha_i\) are directed, \(I\) and
  each \(J_i\) are inhabited. Hence, so is the domain of \(\hat\alpha\).
  Thus, we show that \(\hat\alpha\) is semidirected. Suppose we have
  \((i_1,j_1) , (i_2,j_2)\) in the domain of \(\hat\alpha\). By directedness of
  \(\alpha\), there exists \(i : I\) such that \(\alpha_{i_1}\) and
  \(\alpha_{i_2}\) are cofinal in \(\alpha_i\). Hence, there exist
  \(j_1',j_2' : J_i\) with \(\alpha_{i_1}(j_1) \below \alpha_{i}(j_1')\) and
  \(\alpha_{i_2}(j_2) \below \alpha_{i}(j_2')\).
  Because the family \(\alpha_i\) is directed in \(X\), there exists \(j : J_i\)
  such that \(\alpha_{i}(j_1'),\alpha_{i}(j_2') \below \alpha_i(j)\).  Hence, we
  conclude that
  \(\hat\alpha(i_1,j_1) \equiv \alpha_{i_1}(j_1) \below \alpha_i(j_1') \below
  \alpha_i(j) \equiv \hat\alpha(i,j)\), and similarly for \((i_2,j_2)\), which
  proves semidirectedness of \(\hat\alpha\).
\end{proof}

\begin{lemma}\label{sup-map-is-monotone}%
  \index{supremum!as a map from the ind-completion to a dcpo}%
  Taking directed suprema defines a monotone map from a \(\V\)-dcpo to its
  \(\V\)-ind-completion, denoted by \({\bigsqcup} : \Ind{V}(D) \to D\).
\end{lemma}
\begin{proof}
  We have to show that \(\bigsqcup \alpha \below \bigsqcup \beta\) for directed
  families \(\alpha\) and \(\beta\) such that \(\alpha\) is cofinal in
  \(\beta\). Note that the inequality holds as soon as
  \(\alpha_i \below \bigsqcup \beta\) for every \(i\) in the domain of
  \(\alpha\). For this, it suffices that for every such \(i\), there exists a
  \(j\) in the domain of \(\beta\) such that \(\alpha_i \below \beta_j\). But
  the latter says exactly that \(\alpha\) is cofinal in \(\beta\).
\end{proof}

\citeauthor{JohnstoneJoyal1982}~\cite{JohnstoneJoyal1982} generalise the notion
of continuity from posets to categories and do so by phrasing it in terms of
\(\bigsqcup : \Ind{V}(D) \to D\) having a left adjoint.
We follow their approach and build towards this.
It turns out to be convenient to use the following two notions, which are in
fact equivalent by \cref{approximate-adjunct-coincidence}.

\begin{definition}[Approximate, left adjunct]%
  \index{approximation}\index{adjunct}
  For an element \(x\) of a dcpo \(D\) and a directed family
  \(\alpha : I \to D\), we say that
  \begin{enumerate}[(i)]
  \item \(\alpha\) \emph{approximates} \(x\) if the supremum of \(\alpha\) is
    \(x\) and each \(\alpha_i\) is way below \(x\), and
  \item \(\alpha\) is \emph{left adjunct} to \(x\) if
    \(\alpha \cof \beta \iff x \below \bigsqcup \beta\) for every directed
      family \(\beta\).
      \qedhere
  \end{enumerate}
\end{definition}

\begin{remark}\label{left-adjoint-in-terms-of-left-adjunct-to}
  For a \(\V\)-dcpo \(D\), a function \(L : D \to \Ind{V}(D)\) is a left adjoint
  to \({{\bigsqcup} : \Ind{V}(D) \to D}\) precisely when \(L(x)\) is left
  adjunct to \(x\) for every \(x : D\).

  One may object at this point and argue that a function
  \(L : D \to \Ind{V}(D)\) needs to be monotone in order to truly be a left
  adjoint. But monotonicity actually follows from the condition that \(L(x)\) is
  left adjunct to \(x\) for every \(x : D\), as we show in
  \cref{left-adjoint-monotone}.
\end{remark}

\begin{lemma}\label{approximate-adjunct-coincidence}
  A directed family \(\alpha\) approximates an element \(x\) if and only if it
  is left adjunct to it.
\end{lemma}
\begin{proof}
  Suppose first \(\alpha\) approximates \(x\). If \(\alpha \cof \beta\), then
  \(x = \bigsqcup \alpha \below \bigsqcup \beta\), by
  \cref{sup-map-is-monotone}. Conversely, if \(x \below \bigsqcup \beta\), then
  \(\alpha\) is cofinal in \(\beta\): for if \(i\) is in the domain of
  \(\alpha\), then \(\alpha_i \ll x\), so there exists \(j\) such that
  \(\alpha_i \below \beta_j\) already.

  In the other direction, suppose that \(\alpha\) is left adjunct to \(x\). We
  show that each \(\alpha_i\) is way below \(x\). If \(\beta\) is a directed
  family with \(x \below \bigsqcup \beta\), then \(\alpha\) is cofinal in
  \(\beta\) as \(\alpha\) is assumed to be left adjunct to \(x\). Hence, for
  every \(i\), there exists \(j\) with \(\alpha_i \below \beta_j\), proving that
  \(\alpha_i \ll x\).
  Since \(\alpha\) is cofinal in itself, we get \(x \below \bigsqcup \alpha\) by
  assumption. For the other inequality, we note that
  \(x \below \bigsqcup \hat{x}\), where \(\hat{x} : \One \to D\) is the directed
  family that maps to \(x\). Hence, as \(\alpha\) is left adjunct to
  \(x\), we must have that \(\alpha\) is cofinal in \(\hat{x}\), which means
  that each \(\alpha_i\) is below \(x\). Thus, we conclude
  \(\bigsqcup \alpha \below x\) and \(x = \bigsqcup \alpha\), as desired.
\end{proof}

\begin{proposition}\label{left-adjoint-characterization}%
  \index{adjoint}%
  For a \(\V\)-dcpo \(D\), a function \(L : D \to \Ind{V}(D)\) is a left
  adjoint to \(\bigsqcup : \Ind{V}(D) \to D\) if and only if \(L(x)\)
  approximates \(x\) for every \(x : D\).
\end{proposition}
\begin{proof}
  Immediate from \cref{approximate-adjunct-coincidence} and
  \cref{left-adjoint-in-terms-of-left-adjunct-to}.
\end{proof}

\begin{corollary}\label{left-adjoint-monotone}
  If a function \(L : D \to \Ind{V}(D)\) is a left adjoint to the function
  \(\bigsqcup : {\Ind{V}(D) \to D}\), then it is monotone. In fact, it is also
  order-reflecting in this case.
\end{corollary}
\begin{proof}
  Suppose that \(L\) is a left adjoint to \(\bigsqcup\) and that we have
  elements \(x,y : D\). Since \(L\) is a left adjoint, we have
  \(L(x) \cof L(y) \iff x \below \bigsqcup L(y)\), but \(L(y)\) approximates
  \(y\), so \(\bigsqcup L(y) = y\) and hence
  \(L(x) \cof L(y) \iff x \below y\), so \(L\) preserves and reflects the
  order.
\end{proof}

\section{Continuous dcpos}\label{sec:continuous-dcpos}

We define what it means for a \(\V\)-dcpo to be (structurally) continuous and
prove the interpolation properties for the way-below relation. We postpone
giving examples (see~\cref{sec:algebraic-examples,sec:dyadics}) until we have
developed the theory further.

\begin{definition}[Structural continuity]%
  \index{continuity!structural|textbf}%
  \index{approximating family}%
  A \(\V\)-dcpo \(D\) is \emph{structurally continuous} if for every \(x : D\)
  we have a specified \(I : \V\) and directed \emph{approximating family}
  \(\alpha : I \to D\) such that \(\alpha\) has supremum \(x\) and each element
  \(\alpha(i)\) is way below \(x\).
\end{definition}

\begin{remark}[\(I_x\), \(\alpha_x\)]
  Note how structural continuity equips a dcpo with a function assigning an
  approximating family to every element of the dcpo. If we have such an
  equipment, we will write \(\alpha_x : I_x \to D\) for the approximating family
  of an element \(x\).
  \nomenclature[alphax]{$\alpha_x$}{specified approximating family of \(x\) in a
    structurally continuous dcpo}
  \nomenclature[Ix]{$I_x$}{type indexing the approximating (resp.\ compact)
    family of \(x\) in a structurally continuous (resp.\ algebraic) dcpo (see
    also p.~\pageref{spec-compact-family-notation})}
\end{remark}

The somewhat verbose definition of structural continuity can be succinctly
phrased as follows.

\begin{proposition}\label{structural-continuity-in-terms-of-ladj}%
  \index{adjoint}%
  Structural continuity of a \(\V\)-dcpo \(D\) is equivalent to having a
  specified left adjoint to \({\bigsqcup} : \Ind{V}(D) \to D\).
\end{proposition}
\begin{proof}
  Immediate from \cref{left-adjoint-characterization}.
\end{proof}

\begin{remark}\index{property}\index{cofinality!bi-}%
  It should be noted that structural continuity is \emph{not} a property because
  an element \(x : D\) can have several different but bicofinal approximating
  families, e.g.\ if \(\alpha : I \to D\) approximates \(x\), then so does
  \([\alpha,\alpha] : (I + I) \to D\).
  In other words, the left adjoint to \(\bigsqcup : \Ind{V}(D) \to D\) is not
  unique, although it is unique up to isomorphism.
  In category theory this is typically sufficient (and often the best one can
  do). \citeauthor{JohnstoneJoyal1982} follow this philosophy
  in~\cite{JohnstoneJoyal1982}, but we want the type of continuous
  \(\V\)\nobreakdash-dcpos to be a subtype of the \(\V\)-dcpos.
  One reason that a property is preferred is that it is considered good
  mathematical practice to only consider morphisms that preserve imposed
  structure. In the case of structural continuity, this would imply preservation
  of the way-below relation, but this excludes many Scott continuous maps,
  e.g.\ if the output of a constant map is not compact, then it does not preserve
  the way-below relation.

  One may ask why the univalence axiom cannot be used to identify these
  isomorphic objects. The point is that the ind-completion \(\Ind{V}(D)\) is
  \emph{not} a univalent category in the sense
  of~\cite{AhrensKapulkinShulman2015}, because it is a preorder and not a poset.
  One way to obtain a subtype is to propositionally truncate the notion of
  structural continuity and this is indeed the approach that we will
  take. However, another choice that would yield a property is to identify
  bicofinal elements of \(\Ind{V}(D)\) by quotienting. This approach is
  discussed at length in \cref{sec:pseudocontinuity} and in particular it is
  explained to be inadequate in a constructive setting.
\end{remark}

\begin{definition}[Continuity of a dcpo]\index{continuity!of a dcpo|textbf}
  A \(\V\)-dcpo is \emph{continuous} if the propositional truncation of its
  structural continuity holds.
\end{definition}

Thus, a dcpo is continuous if we have an \emph{unspecified} function assigning
an approximating family to every element of the dcpo.

\begin{proposition}\index{adjoint}
  Continuity of a \(\V\)-dcpo \(D\) is equivalent to having an unspecified left
  adjoint to \({\bigsqcup} : \Ind{V}(D) \to D\).
\end{proposition}
\begin{proof}
  By \cref{structural-continuity-in-terms-of-ladj} and functoriality of the
  propositional truncation.
\end{proof}

\begin{lemma}\label{structurally-continuous-below-characterization}
  For elements \(x\) and \(y\) of a structurally continuous dcpo, the following
  are equivalent:
  \begin{enumerate}[(i)]
  \item\label{item-below} \(x \below y\);
  \item\label{item-approx-below} \(\alpha_x(i) \below y\) for every \(i : I_x\);
  \item\label{item-approx-way-below} \(\alpha_x(i) \ll y\) for every
    \(i : I_x\).
  \end{enumerate}
\end{lemma}
\begin{proof}
  Note that \ref{item-approx-way-below} implies \ref{item-approx-below} and
  \ref{item-approx-below} implies \ref{item-below}, because if
  \(\alpha_x(i) \below y\) for every \(i : I_x\), then
  \(x = \bigsqcup \alpha_x \below y\), as desired. So it remains to prove that
  \ref{item-below} implies \ref{item-approx-way-below}, but this holds, because
  \(\alpha_x(i) \ll x\) for every \(i : I_x\).
\end{proof}

\begin{lemma}\label{structurally-continuous-way-below-characterization}
  For elements \(x\) and \(y\) of a structurally continuous dcpo, \(x\) is way
  below~\(y\) if and only if there exists \(i : I_y\) such that
  \(x \below \alpha_y(i)\).
\end{lemma}
\begin{proof}
  The left-to-right implication holds, because \(\alpha_y\) is a directed family
  with supremum \(y\), while the converse holds because \(\alpha_y(i) \ll y\)
  for every \(i : I_y\).
\end{proof}

We now prove three interpolation lemmas for structurally continuous
dcpos. Because the conclusions of the lemmas are propositions, the results will
follow for continuous dcpo immediately.

\index{interpolation|(}%
\begin{lemma}[Nullary interpolation for the way-below relation]
  \label{nullary-interpolation}
  For every \(x : D\) of a (structurally) continuous dcpo \(D\), there exists
  \(y : D\) such that \(y \ll x\).
\end{lemma}
\begin{proof}
  The approximating family \(\alpha_x\) is directed, so there exists \(i : I_x\)
  and hence we can take \(y \colonequiv \alpha_x(i)\) since
  \(\alpha_x(i) \ll x\).
\end{proof}

Our proof of the following lemma is inspired by
\cite[Proposition~2.12]{JohnstoneJoyal1982}.
\begin{lemma}[Unary interpolation for the way-below relation]\label{unary-interpolation}
  If \(x \ll y\) in a (structurally) continuous dcpo \(D\), then there exists an
  \emph{interpolant} \(d : D\) such that \(x \ll d \ll y\).
\end{lemma}
\begin{proof}
  By structural continuity, we can approximate every approximant \(\alpha_y(i)\)
  of \(y\) by an approximating family \(\beta_{i} : J_i \to D\). This defines a
  map \(\hat{\beta}\) from \(I_y\) to \(\Ind{V}(D)\), the ind-completion of the
  \(\V\)-dcpo \(D\), by sending \(i : I_y\) to the directed family \(\beta_i\).
  We claim that \(\hat\beta\) is directed in \(\Ind{V}(D)\). Since \(\alpha_y\)
  is directed, \(I_y\) is inhabited, so it remains to prove that \(\hat\beta\)
  is semidirected. So suppose we have \(i_1,i_2 : I_y\). Because \(\alpha_y\) is
  semidirected, there exists \(i : I_y\) such that
  \(\alpha_y(i_1),\alpha_y(i_2) \below \alpha_y(i)\). We claim that
  \(\beta_{i_1}\) and \(\beta_{i_2}\) are cofinal in \(\beta_i\), which would
  prove semidirectedness of \(\hat\beta\). We give the argument for \(i_1\) only
  as the case for \(i_2\) is completely analogous.
  For the cofinality, we have to show that for every \(j : J_{i_1}\), there
  exists \(j' : J_i\) such that \(\beta_{i_1}(j) \below \beta_{i}(j')\).
  But this holds because \(\beta_{i_1}(j) \ll \bigsqcup \beta_i\) for every such
  \(j\), as we have
  \(\beta_{i_1}(j) \ll \alpha_y(i_1) \below \alpha_y(i) \below \bigsqcup
  \beta_i\).

  Thus, \(\hat\beta\) is directed in \(\Ind{V}(D)\) and hence we can calculate
  its supremum in \(\Ind{V}(D)\) to obtain the \emph{directed} family
  \(\gamma : \pa{\Sigma_{i : I}J_i} \to D\) given by
  \((i,j) \mapsto \beta_i(j)\).

  We now show that \(y\) is below the supremum of \(\gamma\). By
  \cref{structurally-continuous-below-characterization}, it suffices to prove
  that \(\alpha_y(i) \below \bigsqcup \gamma\) for every \(i : I_y\), and, in
  turn, to prove this for an \(i : I_y\) it suffices to prove that
  \(\beta_i(j) \below \bigsqcup \gamma\) for every \(j : J_i\). But this is
  immediate from the definition of \(\gamma\). Thus,
  \(y \below \bigsqcup \gamma\).
  Because \(x \ll y\), there exists \((i,j) : \Sigma_{i : I}J_i\) such that
  \(x \below \gamma(i,j) \equiv \beta_i(j)\).

  Finally, for our interpolant, we take \(d \colonequiv \alpha_y(i)\). Then,
  indeed, \(x \ll d \ll y\), because
  \(x \below \beta_i(j) \ll \alpha_y(i) \equiv d\) and
  \(d \equiv \alpha_y(i) \ll y\), completing the proof.
\end{proof}

\begin{lemma}[Binary interpolation for the way-below relation]%
  \label{binary-interpolation}
  If \(x \ll z\) and \(y \ll z\) in a (structurally) continuous dcpo \(D\), then
  there exists an interpolant \(d : D\) such that \(x,y \ll d\) and
  \(d \ll z\).
\end{lemma}
The proof is a straightforward application of unary interpolation.
\begin{proof}
  Using that \(x \ll z\) and \(y \ll z\), there exist interpolants
  \({d_1,d_2} : D\) such that \(x \ll d_1 \ll z\) and \(y \ll d_2 \ll z\).
  Hence, there exist \(i_1,i_2 : I_z\) such that \(d_1 \below \alpha_z(i_1)\)
  and \(d_2 \below \alpha_z(i_2)\). By semidirectedness of \(\alpha_z\), there
  then exists \(i : I_z\) for which \(d_1,d_2 \below \alpha_z(i)\). Our final
  interpolant is defined as \(d \colonequiv \alpha_z(i)\), which works because
  \(x \ll d_1 \below d\), \(y \ll d_2 \below d\) and
  \(d \equiv \alpha_z(i) \ll z\).
\end{proof}%
\index{interpolation|)}

Both continuity and structural continuity are closed under Scott continuous
retracts. Keeping track of universes, it holds in the following generality:
\begin{theorem}\label{continuity-closed-under-continuous-retracts}%
  \index{retract!Scott continuous}%
  If we have dcpos \(D : \DCPO{V}{U}{T}\) and \(E : \DCPO{V}{U'}{T'}\) such that
  \(D\)~is a Scott continuous retract of \(E\), then \(D\) is (structurally)
  continuous if \(E\) is.
\end{theorem}
\begin{proof}
  We prove the result for structural continuity, as the other will follow from
  that and the fact that the propositional truncation is functorial. So suppose
  that we have a Scott continuous section \(s : D \to E\) and retraction
  \(r : E \to D\) and structural continuity of \(E\). We claim that for every
  \(x : D\), the family \(r \circ \alpha_{s(x)}\) is approximating for \(x\).
  Firstly, it is directed, because \(\alpha_{s(x)}\) is and \(r\) is Scott
  continuous.
  Secondly,
  \begin{align*}
    \textstyle\bigsqcup r \circ \alpha_{s(x)}
    &= r\pa*{\textstyle\bigsqcup\alpha_{s(x)}}
      &&\text{(by Scott continuity of \(r\))}
    \\
    &= r(s(x))
      &&\text{(as \(\alpha_{s(x)}\) is the approximating family of \(s(x)\))}
    \\
    &= x &&\text{(because \(s\) is a section of \(r\))},
  \end{align*}
  so the supremum of \(r \circ \alpha_{s(x)}\) is \(x\).
  Finally, we must prove that \(r\pa*{\alpha_{s(x)}(i)} \ll x\) for every
  \(i : I_x\). By \cref{continuous-retraction-way-below-criterion}, this is
  implied by \(\alpha_{s(x)}(i) \ll s(x)\), which holds as \(\alpha_{s(x)}\) is
  the approximating family of \(s(x)\).
\end{proof}

\begin{proposition}\label{cont-loc-small-iff-way-below-small}%
  \index{dcpo!locally small}
  A (structurally) continuous dcpo is locally small if and only if its way-below
  relation has small values.
\end{proposition}
\begin{proof}
  By \cref{structurally-continuous-below-characterization,%
    structurally-continuous-way-below-characterization}, we have
  \[
    x \below y \iff \forall_{i : I_x}\pa*{\alpha_x(i) \ll y} %
    \quad\text{and}\quad %
    x \ll y \iff \exists_{i : I_y}\pa*{x \below \alpha_y(i)},
  \]
  for every two elements \(x\) and \(y\) of a structurally continuous dcpo.
  But the types \(I_x\)~and~\(I_y\) are small, finishing the proof.
  The result also holds for continuous dcpos, because what we are proving is a
  proposition.
\end{proof}

\cref{cont-loc-small-iff-way-below-small} is significant because the definition
of the way-below relation for a \(\V\)-dcpo \(D\) quantifies over all families
into \(D\) indexed by types in \(\V\).

\section{Pseudocontinuity}\label{sec:pseudocontinuity}
In light of~\cref{structural-continuity-in-terms-of-ladj}, we see that a
\(\V\)\nobreakdash-dcpo \(D\) can be structurally continuous in more than one
way: the map \({\bigsqcup} : \Ind{V}(D) \to D\) can have two left adjoints
\(L_1,L_2\) such that for some \(x : D\), the directed families \(L_1(x)\) and
\(L_2(x)\) are bicofinal, yet unequal.
In order for the left adjoint to be truly unique (and not just up to
isomorphism), the preorder \(\Ind{\V}(D)\) would have to identify bicofinal
families.\index{cofinality!bi-}
Of course, we could enforce this identification by passing to the poset
reflection \(\Ind{V}(D)/{\approx}\) of \(\Ind{V}(D)\) and this section studies
exactly that.%
\index{poset!reflection}%
\nomenclature[Ind']{$\Ind{V}(D)/{\approx}$}{poset reflection of the preorder
  \(\Ind{V}(D)\)}

Another perspective on the situation is the following: The type-theoretic
definition of structural continuity of a \(\V\)-dcpo \(D\) is of the following form
\(\Pi_{x : D}\Sigma_{I : \V}\Sigma_{\alpha : I \to D}\dots\), while continuity
is defined as its propositional truncation
\(\squash*{\Pi_{x : D}\Sigma_{I : \V}\Sigma_{\alpha : I \to D}\dots}\).%
\index{property}
Yet another way to obtain a property is by putting the propositional truncation on
the \emph{inside} instead:
\(\Pi_{x : D}\squash*{\Sigma_{I : \V}\Sigma_{\alpha : I \to D}\dots}\).
We study what this amounts to and how it relates to (structural) continuity and
the poset reflection. Our results are summarised in~\cref{continuity-table}
below.

\begin{definition}[Pseudocontinuity]\index{continuity!pseudo-|textbf}%
  \index{pseudocontinuity|see {continuity}}
  A \(\V\)-dcpo \(D\) is \emph{pseudocontinuous} if for every \(x : D\) there
  exists an unspecified directed family that approximates \(x\).
\end{definition}

Note that structural continuity \(\Rightarrow\) continuity \(\Rightarrow\)
pseudocontinuity, but reversing the first implication is an instance of global
choice, while reversing the second amounts to an instance of the axiom of
choice~(\cref{axiom-of-choice}) that we do not expect to be provable in our
constructive setting. We further discuss this point in
\cref{pseudocontinuous-choice-issues}.%
\index{choice!axiom of}

For a \(\V\)-dcpo \(D\), the map \({\bigsqcup} : \Ind{V}(D) \to D\) is monotone, so
it induces a unique monotone map
\({\bigsqcup_{\approx}} : \Ind{V}(D)/{\approx} \to D\) such that the diagram
\begin{equation}\label{supremum-map-quotient-comm}
  \begin{tikzcd}
    \Ind{V}(D)/{\approx}\ar[rr,"{\bigsqcup_{\approx}}"]
    \ar[dr,"\toquotient{-}"']
    & & D \\
    & \Ind{V}(D) \ar[ur,"{\bigsqcup}"']
  \end{tikzcd}
\end{equation}
\nomenclature[sqcup']{$\bigsqcup_{\approx}$}{map from
  \(\Ind{V}(D)/{\approx}\) to \(D\) induced by taking directed suprema}
commutes.

\begin{proposition}\index{adjoint}
  A \(\V\)-dcpo \(D\) is pseudocontinuous if and only if the map of posets
  \({\bigsqcup}_{\approx} : \Ind{V}(D)/{\approx} \to D\) has a (specified) left
  adjoint.
\end{proposition}

Observe that the type of left adjoints to
\({\bigsqcup}_{\approx} : \Ind{V}(D)/{\approx} \to D\) is a proposition,
precisely because \(\Ind{V}(D)/{\approx}\) is a poset,
cf.~\cite[Lemma~5.2]{AhrensKapulkinShulman2015}.

\begin{proof}
  Suppose that \({\bigsqcup}_{\approx} : \Ind{V}(D)/{\approx} \to D\) has a left
  adjoint \(L\) and let \(x : D\) be arbitrary. We have to prove that there
  exists a directed family \(\alpha : I \to D\) that approximates \(x\). By
  surjectivity of the universal map \(\toquotient{-}\), there exists a directed
  family \(\alpha : I \to D\) such that \(L(x) = \toquotient{\alpha}\).
  Moreover, \(\alpha\) approximates \(x\) by virtue of
  \cref{approximate-adjunct-coincidence}, since for every
  \(\beta : \Ind{V}(D)\), we have
  \begin{align*}
    \alpha \cof \beta
    &\iff L(x) \leq \toquotient{\beta}
      &&\text{(since \(L(x) = \toquotient{\alpha}\))} \\
    &\iff x \below \textstyle\bigsqcup_{\approx}{\toquotient{\beta}}
      &&\text{(since \(L\) is a left adjoint to \(\textstyle{\bigsqcup}_{\approx}\))}\\
    &\iff x \below \textstyle\bigsqcup \beta
      &&\text{(by \cref{supremum-map-quotient-comm})}.
  \end{align*}

  The converse is more involved and features another application
  of~\cref{constant-map-to-set-factors-through-truncation}, similar to the proof
  of~\cref{lifting-is-pointed-dcpo}.
  Assume that \(D\) is pseudocontinuous. We start by constructing the left
  adjoint, so let \(x : D\) be arbitrary. Writing \(\mathcal A_x\) for the type
  of directed families that approximate \(x\), we have an obvious map
  \(\varphi_x : \mathcal A_x \to \Ind{V}(D)\) that forgets that the directed
  family approximates \(x\).

  We claim that all elements in the image of \(\varphi_x\) are bicofinal. For if
  \(\alpha\) and \(\beta\) are directed families both approximating \(x\), then
  for every \(i\) in the domain of \(\alpha\) we know that
  \(\alpha_i \ll x = \bigsqcup \beta\), so that there exists \(j\) with
  \(\alpha_i \below \beta_j\).
  Hence, passing to the poset reflection, the composite
  \(\toquotient{-} \circ \varphi_x\) is constant.
  Thus, by~\cref{constant-map-to-set-factors-through-truncation} we have a
  (necessarily unique) map \(\psi_x\) making the diagram
  \[
    \begin{tikzcd}
      \mathcal A_x \ar[rr,"\toquotient{-} \circ \varphi_x"]
      \ar[dr,"\tosquash{-}"']
      & & \Ind{V}(D)/{\approx} \\
      & \squash*{\mathcal A_x}
      \ar[ur,"\psi_x"',dashed]
    \end{tikzcd}
  \]
  commute.
  Since \(D\) is assumed to be pseudocontinuous, we have exactly
  \(\squash*{\mathcal A_x}\) for every \(x : D\), so together with \(\psi_x\)
  this defines a map \(L : D \to \Ind{V}(D)/{\approx}\) by
  \(L(x) \colonequiv \psi_x(p)\), where \(p\) witnesses pseudocontinuity at
  \(x\).

  Lastly, we prove that \(L\) is indeed a left adjoint to
  \({\bigsqcup}_{\approx}\). So let \(x : D\) be arbitrary. Since we're proving
  a property, we can use pseudocontinuity at \(x\) to specify a directed family
  \(\alpha\) that approximates \(x\). We now have to prove
  \(\toquotient{\alpha} \leq \beta' \iff x \below \bigsqcup_{\approx} \beta'\)
  for every \(\beta' : \Ind{V}(D)/{\approx}\). This is a proposition, so using
  quotient induction once more, it suffices to prove
  \(\toquotient{\alpha} \leq \toquotient{\beta} \iff x \below
    \bigsqcup_{\approx} \toquotient{\beta}\) for every \(\beta : \Ind{V}(D)\).
  Indeed, for such \(\beta\) we have
  \begin{align*}
    \toquotient{\alpha} \leq \toquotient{\beta}
    &\iff \alpha \cof \beta \\
    &\iff x \below \textstyle\bigsqcup \beta
    &&\text{(by \cref{approximate-adjunct-coincidence} and the fact
      that \(\alpha\) approximates \(x\))} \\
    &\iff x \below \textstyle\bigsqcup_{\approx} \toquotient{\beta}
    &&\text{(by \cref{supremum-map-quotient-comm})},
  \end{align*}
  finishing the proof.
\end{proof}

Thus, the explicit type-theoretic formulation and the formulation using left
adjoints in each row of~\cref{continuity-table} (which summarises our findings)
are equivalent.

\begin{table}[h]%
  \index{property}\index{continuity!structural}%
  \index{continuity!of a dcpo}\index{continuity!pseudo-}%
  \index{propositional truncation}%
  \index{existence!specified}\index{existence!unspecified}%
  \centering
  \begin{tabular}{lcp{4.4cm}c}\toprule 
    & Type-theoretic formulation
    & Formulation with adjoints & Property \\\midrule
    Struc.\ cont.
    & \(\Pi_{x : D} \Sigma_{I : \V} \Sigma_{\alpha : I \to D}\,\delta(\alpha,x)\)
    & Specified left adjoint to \({\bigsqcup : \Ind{V}(D) \to D}\) 
    & \ding{53} \\
    Cont.
    & \(\squash*{\Pi_{x : D} \Sigma_{I : \V} \Sigma_{\alpha : I \to D}\,\delta(\alpha,x)}\)
    & Unspecified left adjoint to \({\bigsqcup : \Ind{V}(D) \to D}\) 
    & \checkmark \\
    Pseudocont.
    & \(\Pi_{x : D}\squash*{\Sigma_{I : \V} \Sigma_{\alpha : I \to D}\,\delta(\alpha,x)}\)
    & Specified left adjoint to \({\bigsqcup_{\approx} : \Ind{V}(D)/{\approx} \to D}\)
    & \checkmark \\
    \bottomrule
  \end{tabular}
  \caption{(Structural) continuity and pseudocontinuity of a dcpo \(D\), where
    \(\delta(\alpha,x)\) abbreviates that \(\alpha\) is directed and
    approximates \(x\).}
  \label{continuity-table}
\end{table}

\begin{remark}\label{pseudocontinuous-choice-issues}
  \index{choice!axiom of}\index{continuity!pseudo-}\index{cofinality!bi-}%
  The issue with pseudocontinuity is that taking the quotient by bicofinality
  introduces a dependence on instances of the axiom of choice when it comes to
  proving properties of pseudocontinuous dcpos. An illustrative example is the
  proof of unary interpolation~(\cref{unary-interpolation}), where we used
  structural continuity to first approximate an element \(y\) by \(\alpha_y\)
  and then, in turn, approximate every approximant \(\alpha_y(i)\). With
  pseudocontinuity this argument would require \emph{choosing} an approximating
  family for every \(i\).
  Another example is that while the preorder \(\Ind{V}(D)\) is \(\V\)-directed
  complete, a direct lifting of the proof of this fact to the poset reflection
  \(\Ind{V}(D)/{\approx}\) requires the axiom of choice.\index{poset!reflection}
  Hence, the Rezk completion~\cite{AhrensKapulkinShulman2015}, of which the
  poset reflection is a special case, does not necessarily preserve (filtered)
  colimits.\index{Rezk completion}
  Related issues concerning the axiom of choice are also discussed
  in~\cite[pp.~260--261]{JohnstoneJoyal1982}.
\end{remark}

\section{Algebraic dcpos}\label{sec:algebraic-dcpos}

Many of our examples of dcpos are not just continuous, but satisfy the stronger
condition of being algebraic, meaning their elements can be approximated by
compact elements only.

\begin{definition}[Structural algebraicity]%
  \index{algebraicity!structural|textbf}%
  \index{compact family}%
  A \(\V\)-dcpo \(D\) is \emph{structurally algebraic} if for every \(x : D\)
  we have a specified \(I : \V\) and directed \emph{compact family}
  \(\kappa : I \to D\) such that \(\kappa\) has supremum \(x\) and each element
  \(\kappa(i)\) is compact.
\end{definition}

\begin{remark}[\(I_x\), \(\kappa_x\)]\label{spec-compact-family-notation}
  Note how structural algebraicity equips a dcpo with a function assigning a
  compact family to every element of the dcpo. If we have such an equipment, we
  will write \(\kappa_x : I_x \to D\) for the compact family of an element
  \(x\).
  \nomenclature[kappa]{$\kappa_x$}{specified compact family of \(x\) in a
    structurally algebraic dcpo}
\end{remark}

\begin{definition}[Algebraicity]\index{algebraicity|textbf}
  A \(\V\)-dcpo is \emph{algebraic} if the propositional truncation of its
  structural algebraicity holds.
\end{definition}

Thus, a dcpo is continuous if we have an \emph{unspecified} function assigning
a compact family to every element of the dcpo.

\begin{lemma}
  Every (structurally) algebraic dcpo is (structurally) continuous.
\end{lemma}
\begin{proof}
  We prove that structurally algebraic dcpos are structurally continuous. The
  claim for algebraic and continuous then follows by functoriality of the
  propositional truncation.
  It suffices to prove that \(\kappa_x(i) \ll x\) for every \(i : I_x\).  By
  assumption, \(\kappa_x(i)\) is compact and has supremum \(x\). Hence,
  \(\kappa_x(i) \ll \kappa_x(i) \below \bigsqcup\kappa_x = x\), so
  \(\kappa_x(i) \ll x\).
\end{proof}

\section{Small bases}\label{sec:small-bases}

Recall that the traditional, set-theoretic definition of a dcpo \(D\) being
continuous says that for every element \(x \in D\), the subset
\(\set{y \in D \mid y \ll x}\) is directed with supremum~\(x\).
As explained in the \nameref{sec:continuous-and-algebraic-introduction} of this
chapter, the problem with this definition in a predicative context is that the
subset \(\set{y \in D \mid y \ll x}\) is not small.
But, as is well-known in domain theory, it is sufficient (and in fact
equivalent) to instead ask that \(D\) has a subset \(B\), known as a
\emph{basis}, such that the subset \(\set{b \in B \mid b \ll x} \subseteq B\) is
directed with supremum \(x\), see~\cite[Section~2.2.2]{AbramskyJung1994} and
\cite[Definition~III-4.1]{GierzEtAl2003}.
The idea developed in this section is that in many examples we can find a
\emph{small} basis giving us a predicative handle on the situation.%
\index{smallness}

If a dcpo has a small basis, then it is continuous. In fact, all our running
examples of continuous dcpos are actually examples of dcpos with small
bases. Moreover, dcpos with small bases are better behaved. For example, they
are locally small and so are their exponentials, which also have small bases
(\cref{sec:exponentials-with-small-bases}). Moreover, in
\cref{sec:ideal-completions-of-small-bases} we show that having a small basis is
equivalent to being presented by ideals.

\begin{definition}[Small basis]
  \index{basis|textbf}%
  For a \(\V\)-dcpo \(D\), a map \(\beta : B \to D\) with \(B : \V\) is a
  \emph{small basis} for \(D\) if the following conditions hold:
  \begin{enumerate}[(i)]
  \item\label{basis-approximating} for every \(x : D\), the family
    \(\pa*{\Sigma_{b : B}\pa*{\beta(b) \ll x}} \xrightarrow{\beta \circ \fst}
    D\) is directed and has supremum \(x\);
  \item\label{basis-small-way-below} for every \(x : D\) and \(b : B\), the
    proposition \(\beta(b) \ll x\) is \(\V\)-small.
  \end{enumerate}
  We will write \(\ddset_\beta x\) for the type
  \(\Sigma_{b : B}\pa*{\beta(b) \ll x}\) and conflate this type with the
  canonical map \(\ddset_\beta x \xrightarrow {\beta \circ \fst} D\).
  \nomenclature[ddsetbetax]{$\ddset_\beta x$}{type of basis elements way below
    \(x\) (or its associated family)}
\end{definition}

\cref{basis-small-way-below} ensures not only that the type
\(\Sigma_{b : B}\pa*{\beta(b) \ll x}\) is \(\V\)-small, but also that a dcpo
with a small basis is locally small (\cref{locally-small-if-small-basis}).

\begin{remark}\label{tacitly-use-small-basis}
  If \(\beta : B \to D\) is a small basis for a \(\V\)-dcpo \(D\), then the type
  \(\ddset_\beta x\) is small. Hence, we have a type \(I : \V\) and an equivalence
  \(\varphi : I \simeq \ddset_\beta x\) and we see that the family
  \(I \xrightarrow{\varphi} \ddset_\beta x \xrightarrow{\beta \circ \fst} D\) is
  directed and has the same supremum as \(\ddset_\beta x \to D\).
  We will use this tacitly and write as if the type \(\ddset_\beta x\) is actually a
  type in \(\V\).
\end{remark}

\begin{lemma}\label{structural-continuity-if-small-basis}%
  \index{continuity!of a dcpo}\index{continuity!structural}%
  If a dcpo comes equipped with a small basis, then it is structurally
  continuous. Hence, if a dpco has an unspecified small basis, then it is
  continuous.
\end{lemma}
\begin{proof}
  For every element \(x\) of a dcpo \(D\), the family \(\ddset_\beta x \to D\)
  approximates \(x\), so the assignment \(x \mapsto \ddset_\beta x\) makes \(D\)
  structurally continuous.
\end{proof}

\begin{lemma}\label{below-in-terms-of-way-below-basis}
  In a dcpo \(D\) with a small basis \(\beta : B \to D\), we have
  \( x \below y\) if and only if \(\forall_{b : B}\pa*{\beta(b) \ll x \to \beta(b) \ll y} \)
    for every \(x,y : D\).
\end{lemma}
\begin{proof}
  If \(x \below y\) and \(\beta(b) \ll x\), then \(\beta(b) \ll y\), so the
  left-to-right implication is clear.
  For the converse, suppose that the condition of the lemma holds. Because
  \(x = \bigsqcup \ddset_\beta x\), the inequality \(x \below y\) holds as soon as
  \(\beta(b) \below y\) for every \(b : B\) with \(\beta(b) \ll x\), but this is
  implied by the condition.
\end{proof}

\begin{proposition}\label{locally-small-if-small-basis}\index{dcpo!locally small}%
  A dcpo with a small basis is locally small. Moreover, the way-below relation
  on all of the dcpo has small values.
\end{proposition}
\begin{proof}
  The first claim follows from \cref{below-in-terms-of-way-below-basis} and the
  second follows from the first and \cref{cont-loc-small-iff-way-below-small}.
\end{proof}

A notable feature of dcpos with a small basis is that interpolants for the
way-below relation,
cf.\ \cref{nullary-interpolation,unary-interpolation,binary-interpolation}, can
be found in the basis, as we show now.

\index{interpolation!in the basis|(}
\begin{lemma}[Nullary interpolation in the basis for the way-below relation]
  \label{nullary-interpolation-basis}
  In a dcpo \(D\) with a small basis \(\beta : B \to D\), there exists \(b : B\)
  with \(\beta(b) \ll x\) for every \(x : D\).
\end{lemma}
\begin{proof}
  For every \(x : D\), the approximating family \(\ddset x\) is directed, so
  there exists \(b : B\) with \(\beta(b) \ll x\).
\end{proof}

\begin{lemma}[Unary interpolation in the basis for the way-below relation]
  \label{unary-interpolation-basis}
  If \(x \ll y\) in a dcpo \(D\) with a small basis \(\beta : B \to D\), then
  there exists an interpolant \(b : B\) such that \(x \ll \beta(b) \ll y\).
\end{lemma}
\begin{proof}
  The small basis ensures that \(D\) is structurally continuous by
  \cref{structural-continuity-if-small-basis}. Hence, if \(x \ll y\), then there
  exists an interpolant \(d : D\) with \(x \ll d \ll y\).
  Now \(d \ll y \below \bigsqcup \ddset_\beta y\), so there exists \(b : B\) such
  that \(d \below \beta(b) \ll y\). Moreover, \(x \ll d \below \beta(b)\),
  completing the proof.
\end{proof}

\begin{lemma}[Binary interpolation in the basis for the way-below relation]
  \label{binary-interpolation-basis}
  If \(x \ll z\) and \(y \ll z\) in a dcpo \(D\) with a small basis
  \(\beta : B \to D\), then there exists an interpolant \(b : B\) such that
  \(x,y \ll \beta(b)\) and \(\beta(b) \ll z\).
\end{lemma}
\begin{proof}
  The small basis ensures that \(D\) is structurally continuous by
  \cref{structural-continuity-if-small-basis}. Hence, if \(x \ll y\) and
  \(y \ll z\), then there exists an interpolant \(d : D\) with
  \(x,y \ll d \ll z\).
  Now \(d \ll z \below \bigsqcup \ddset_\beta z\), so there exists \(b : B\) such
  that \(d \below \beta(b) \ll z\). Moreover, \(x,y \ll d \below \beta(b)\),
  completing the proof.
\end{proof}%
\index{interpolation!in the basis|)}

Before proving the analogue of
\cref{continuity-closed-under-continuous-retracts} (closure under Scott
continuous retracts) for small bases, we need a type-theoretic analogue of
\cite[Proposition~2.2.4]{AbramskyJung1994} and
\cite[Proposition~III-4.2]{GierzEtAl2003}, which essentially says that it is
sufficient for a ``subset'' of \(\ddset_\beta x\) (given by \(\sigma\) in the
lemma) to be directed and have suprema \(x\).

\begin{lemma}\label{subbasis-lemma}
  Suppose that we have an element \(x\) of a \(\V\)-dcpo \(D\) together with two
  maps \(\beta : B \to D\) and
  \(\sigma : I \to \Sigma_{b : B}\pa*{\beta(b) \ll x}\) with \(I : \V\).
  If \(\ddset_\beta x \circ \sigma\) is directed and has supremum \(x\), then
  \(\ddset_\beta x\) is directed with supremum \(x\) too.
\end{lemma}
\begin{proof}
  Suppose that \({\ddset_\beta x} \circ {\sigma}\) is directed and has supremum
  \(x\). Obviously, \(x\) is an upper bound for \(\ddset_\beta x\), so we are to
  prove that it is the least. If \(y\) is an upper bound for \(\ddset_\beta x\),
  then it is also an upper bound for \(\ddset_\beta x \circ \sigma\) which has
  supremum \(x\), so that \(x \below y\) follows.  So the point is directedness
  of \(\ddset_\beta x\). Its domain is inhabited, because \(\sigma\) is
  directed.
  Now suppose that we have \(b_1,b_2 : B\) with \(\beta(b_1),\beta(b_2) \ll
  x\). Since \(x = \bigsqcup \pa*{\ddset_\beta x \circ \sigma}\), there exist
  \(i_1,i_2 : I\) such that \(\beta(b_1) \below \beta(\fst(\sigma(i_1)))\) and
  \(\beta(b_2) \below \beta(\fst(\sigma(i_2)))\).
  Since \(\ddset_\beta x \circ \sigma\) is directed, there exists \(i : I\) with
  \(\beta(\fst(\sigma(i_1))),\beta(\fst(\sigma(i_2))) \below
  \beta(\fst(\sigma(i)))\). Hence, writing \(b \colonequiv \fst(\sigma(i))\), we
  have \(\beta(b) \ll x\) and \(\beta(b_1),\beta(b_2) \below \beta(b)\).
  Thus, \(\ddset_\beta x\) is directed, as desired.
\end{proof}

\begin{theorem}\label{small-basis-closed-under-continuous-retracts}%
  \index{retract!Scott continuous}%
  If we have a Scott continuous retract \(\retract{D}{E}\) and
  \(\beta : B \to E\) is a small basis for \(E\), then \(r\circ \beta\) is a
  small basis for \(D\).
\end{theorem}
\begin{proof}
  First of all, note that \(E\) is locally small by
  \cref{locally-small-if-small-basis}. But being locally small is closed under
  Scott continuous retracts by~\cref{locally-small-retract}, so \(D\) is locally
  small too. Moreover, \(D\) is structurally continuous by virtue of
  \cref{continuity-closed-under-continuous-retracts} and
  \cref{structural-continuity-if-small-basis}. Hence, the way-below relation is
  small-valued by \cref{cont-loc-small-iff-way-below-small}.
  In particular, the type \(r(\beta(b)) \ll x\) is small for every \(b : B\) and
  \(x : D\).

  We are going to use \cref{subbasis-lemma} to show that
  \(\ddset_{r \circ \beta} x\) is directed and has supremum \(x\) for every
  \(x : D\). By \cref{continuous-retraction-way-below-criterion}, the identity
  on \(B\) induces a well-defined map
  \(\sigma : \pa*{\Sigma_{b : B}\pa*{\beta(b) \ll s(x)}} \to \pa*{\Sigma_{b :
      B}\pa*{r(\beta(b)) \ll y}}\).
  Now \cref{subbasis-lemma} tells us that it suffices to prove that
  \(r \circ \ddset_\beta s(x)\) is directed with supremum \(x\).
  But \(\ddset_\beta s(x)\) is directed with supremum \(x\), so by Scott
  continuity of \(r\), the family \(r \circ \ddset_\beta s(x)\) is directed with
  supremum \(r(s(x)) = x\), as desired.
\end{proof}

Finally, a nice use of dcpos with small bases is that they yield locally small
exponentials, as we can restrict the quantification in the pointwise order to
elements of the small basis.

\begin{proposition}\label{exponential-is-locally-small}%
  \index{dcpo!locally small}\index{exponential}%
  If \(D\) is a dcpo with an unspecified small basis and \(E\) is a locally
  small dcpo, then the exponential \(E^D\) is locally small too.
\end{proposition}
\begin{proof}
  Being locally small is a proposition, so in proving the result we may assume
  that \(D\) comes equipped with a small basis \(\beta : B \to D\). For
  arbitrary Scott continuous functions \(f,g : D \to E\), we claim that
  \(f \below g\) precisely when \(\forall_{b : B}\pa*{f(\beta(b)) \below
    g(\beta(b))}\), which is a small type using that \(E\) is locally small.
  The left-to-right implication is obvious, so suppose that
  \(f(\beta(b)) \below g(\beta(b))\) for every \(b : B\) and let \(x : D\) be
  arbitrary. We are to show that \(f(x) \below g(x)\). Since
  \(x = \bigsqcup \ddset_\beta x\), it suffices to prove
  \(f\pa*{\bigsqcup \ddset_\beta x} \below g\pa*{\bigsqcup \ddset_\beta x}\) and in
  turn, that \(f(\beta(b)) \below g\pa*{\bigsqcup \ddset_\beta x}\) for every
  \(b : B\). But is easily seen to hold, because
  \(f(\beta(b)) \below g(\beta(b))\) for every \(b : B\) by assumption.
\end{proof}

\section{Small compact bases}\label{sec:small-compact-bases}

Similarly to the progression from continuous dcpos (\cref{sec:continuous-dcpos})
to algebraic ones (\cref{sec:algebraic-dcpos}), we now turn to small
\emph{compact} bases.

\begin{definition}[Small compact basis]%
  \index{basis!compact|textbf}%
  For a \(\V\)-dcpo \(D\), a map \(\beta : B \to D\) with \(B : \V\) is a
  \emph{small compact basis} for \(D\) if the following conditions hold:
  \begin{enumerate}[(i)]
  \item for every \(b : B\), the element \(\beta(b)\) is compact in \(D\);
  \item for every \(x : D\), the family
    \(\pa*{\Sigma_{b : B}\pa*{\beta(b) \below x}} \xrightarrow{\beta \circ \fst} D\)
    is directed and has supremum \(x\);
  \item for every \(x : D\) and \(b : B\), the proposition \(\beta(b) \below x\) is
    \(\V\)-small.
  \end{enumerate}
  We will write \(\dset_\beta x\) for the type
  \(\Sigma_{b : B}\pa*{\beta(b) \below x}\) and conflate this type with the
  canonical map \(\dset_\beta x \xrightarrow {\beta \circ \fst} D\).
  \nomenclature[dsetbetax]{$\dset_\beta x$}{type of compact elements below
    \(x\) (or its associated family)}
\end{definition}

\begin{remark}
  If \(\beta : B \to D\) is a small compact basis for a \(\V\)-dcpo \(D\), then the
  type \(\dset_\beta x\) is small. Similarly to \cref{tacitly-use-small-basis}, we
  will use this tacitly and write as if the type \(\dset_\beta x\) is actually a
  type in \(\V\).
\end{remark}

\begin{lemma}\label{structural-algebraicity-if-small-compact-basis}%
  \index{algebraicity}\index{algebraicity!structural}%
  If a dcpo comes equipped with a small compact basis, then it is structurally
  algebraic. Hence, if a dpco has an unspecified small compact basis, then it is
  algebraic.
\end{lemma}
\begin{proof}
  For every element \(x\) of a dcpo \(D\), the family \(\dset_\beta x \to D\)
  consists of compact elements and approximates \(x\), so the assignment
  \(x \mapsto \dset_\beta x\) makes \(D\) structurally continuous.
\end{proof}

\begin{lemma}\label{small-and-compact-basis}
  A map \(\beta : B \to D\) is a small compact basis for a dcpo \(D\) if and
  only if \(\beta\) is a small basis for \(D\) and \(\beta(b)\) is compact for
  every \(b : B\).
\end{lemma}
\begin{proof}
  If \(\beta(b)\) is compact for every \(b : B\), then \(\beta(b) \below x\) if
  and only if \(\beta(b) \ll x\) for every \(b : B\) and \(x : D\), so that
  \(\ddset_\beta x \simeq \dset_\beta x\) for every \(x : D\).
  In particular, \(\ddset_\beta x\) approximates \(x\) if and only if
  \(\dset_\beta x\) does, which completes the proof.
\end{proof}

\begin{proposition}
  A small compact basis contains every compact element. That is, if
  \(\beta : B \to D\) is a small compact basis for a dcpo \(D\) and \(x : D\) is
  compact, then there exists \(b : B\) such that \(\beta(b) = x\).
\end{proposition}
\begin{proof}
  Suppose we have a compact element \(x : D\). By compactness of \(x\) and the
  fact that \(x = \dset_\beta x\), there exists \(b : B\) with \(\beta(b) \ll x\)
  such that \(x \below \beta(b)\). But then \(\beta(b) = x\) by antisymmetry.
\end{proof}

\section{Examples of dcpos with small compact bases}%
\label{sec:algebraic-examples}

Armed with the theory of small bases we turn to examples illustrating small
bases in practice. Our examples will involve small \emph{compact} bases and an
example of a dcpo with a small basis that is not compact will have to wait until
\cref{sec:dyadics} when we have developed the ideal completion.

\begin{example}\label{Omega-small-compact-basis}%
  \index{basis!compact!for the type of subsingletons}\index{type!of subsingletons}
  The map \(\beta : \Two \to \Omega_{\U}\) defined by \(0 \mapsto \Zero_{\U}\)
  and \(1 \mapsto \One_{\U}\) is a small compact basis for \(\Omega_{\U}\). In
  particular, \(\Omega_{\U}\) is (structurally) algebraic.
\end{example}

The basis \(\beta : \Two \to \Omega_{\U}\) defined above has the special
property that it is \emph{dense} in the sense
of~\cite[\mkTTurl{TypeTopology.Density}]{TypeTopology}: its image has empty
complement, i.e.\ the type
\(\Sigma_{P : \Omega_{\U}}\lnot\pa*{\Sigma_{b : \Two}\,\beta(b) = P}\) is empty.

\begin{proof}[Proof of~\cref{Omega-small-compact-basis}]
  By \cref{compact-elements-in-Omega}, every element in the image of \(\beta\)
  is compact. Moreover, \(\Omega_{\U}\) is locally small, so we only need to
  prove that for every \(P : \Omega_{\U}\) the family \(\dset_{\beta} P\) is
  directed with supremum \(P\).
  The domain of the family is inhabited, because \(\beta(0)\) is the least
  element. Semidirectedness also follows easily, since \(\Two\) has only two
  elements for which we have \(\beta(0) \below \beta(1)\).
  Finally, the supremum of \(\dset_{\beta} P\) is obviously below
  \(P\). Conversely, if \(P\) holds, then
  \(\bigsqcup {\dset_{\beta} P} = \One = P\).
  The final claim follows from
  \cref{structural-algebraicity-if-small-compact-basis}.
\end{proof}

\begin{example}\label{lifting-has-small-compact-basis}%
  \index{basis!compact!for the lifting}\index{lifting}
  For a set \(X : \U\), the map \(\beta : \pa*{\One + X} \to \lifting_{\U}(X)\)
  given by \(\inl(\star) \mapsto \bot\) and \(\inr(x) \mapsto \eta(x)\) is a
  small compact basis for \(\lifting_{\U}(X)\). In particular,
  \(\lifting_{\U}(X)\) is (structurally) algebraic.
\end{example}

Similar to~\cref{Omega-small-compact-basis}, the basis
\(\beta : (\One + X) \to \lifting_{\U}(X)\) defined above is also dense.

\begin{proof}[Proof of~\cref{lifting-has-small-compact-basis}]
  By \cref{compact-elements-in-lifting}, every element in the image of \(\beta\)
  is compact. Moreover, the lifting is locally small by
  \cref{lifting-is-pointed-dcpo}, so we only need to prove that for every
  partial element \(l\), the family \(\dset_{\beta} l\) is directed with
  supremum~\(l\).
  The~domain of the family is inhabited, because \(\beta(\inl(\star))\) is the
  least element. Semidirectedness also follows easily: First of all,
  \(\beta(\inl(\star))\) is the least element. Secondly, if we have \(x,x' : X\)
  such that \(\beta(\inr(x)),\beta(\inr(x')) \below l\), then because
  \(\beta(\inr(x)) \equiv \eta(x)\) is defined, we must have
  \(\beta(\inr(x)) = l = \beta(\inr(x'))\) by definition of the order.
  Finally, the supremum of \(\dset_{\beta} l\) is obviously a partial element
  below \(l\). Conversely, if \(l\) is defined, then \(l = \eta(x)\) for some
  \(x : X\), and hence, \(l = \eta(x) \below \bigsqcup \dset_{\beta} l\).
  The final claim follows from
  \cref{structural-algebraicity-if-small-compact-basis}.
\end{proof}

\begin{example}\label{powerset-small-compact-basis}%
  \index{basis!compact!for the powerset}\index{powerset}%
  For a set \(X : \U\), the map \(\beta : \List(X) \to \powerset_{\U}(X)\) from
  \cref{def:list-to-powerset} (whose image is the type of Kuratowski finite
  subsets of \(X\)) is a small compact basis for \(\powerset_{\U}(X)\). In
  particular, \(\powerset_{\U}(X)\) is (structurally) algebraic.
\end{example}

Notice that the map \(\beta : \List(X) \to \powerset(X)\) is not an
embedding, as two lists can represent the same Kuratowski finite subset.
Of course, an embedding is given by the inclusion of the Kuratowski finite
subsets into the powerset, and its codomain is small if we assume set
replacement, because it is the image of \(\beta\).

\begin{proof}[Proof of~\cref{powerset-small-compact-basis}]
  By \cref{Kuratowski-finite-iff-list,compact-elements-in-powerset}, all
  elements in the image of \(\beta\) are compact. Moreover,
  \(\powerset_{\U}(X)\) is locally small, so we only need to prove that for
  every \(A : \powerset(X)\) the family \(\dset_{\beta} A\) is directed with
  supremum \(A\), but this was also proven in
  \cref{compact-elements-in-powerset}.
  The final claim follows from
  \cref{structural-algebraicity-if-small-compact-basis}.
\end{proof}

At this point the reader may ask whether we have any examples of dcpos that are
structurally algebraic but that do not have a small compact basis.
The following example shows that this can happen in our predicative setting.

\begin{example}\label{lifting-structurally-algebraic-but-no-small-basis}%
  \index{algebraicity!structural}%
  The lifting \(\lifting_{\V}(P)\) of a proposition \(P : \U\) is structurally
  algebraic, but has a small compact basis if and only if \(P\) is \(\V\)-small.
\end{example}

Thus, requiring that \(\lifting_{\V}(P)\) has a small basis for every
proposition \(P : \U\) is equivalent to \(\Propresizing{\U}{\V}\).

\begin{proof}[Proof of~\cref{lifting-structurally-algebraic-but-no-small-basis}]
  Note that \(\lifting_{\V}(P)\) is simply the type of propositions in \(\V\)
  that imply \(P\). It is structurally algebraic, because given such a
  proposition \(Q\), the family
  \begin{align*}
    Q + \One_{\V} &\to \lifting_{\V}(P) \\
    \inl(q) &\mapsto \One_{\V} \\
    \inr(\star) &\mapsto \Zero_{\V}
  \end{align*}
  is directed, has supremum \(Q\) and consists of compact elements.
  But if \(\lifting_{\V}(P)\) had a small compact basis \(\beta : B \to D\),
  then we would have \(P \simeq \exists_{b : B}\pa*{\beta(b) \simeq \One_{\V}}\) and
  the latter is \(\V\)-small.
\end{proof}

\section{The rounded ideal completion}\label{sec:rounded-ideal-completion}

We have seen that in continuous dcpos, the basis essentially ``generates'' the
whole dcpo, because the basis suffices to approximate any of its elements.
It is natural to ask whether one can start from a more abstract notion of basis
and ``complete'' it to a continuous dcpo.
This is exactly what we do here using the notion of an \emph{abstract basis} and
the \emph{rounded ideal completion}.

\begin{definition}[Abstract basis]%
  \index{basis!abstract|textbf}%
  \index{interpolation|textbf}%
  An \emph{abstract \(\V\)-basis} is a type \(B : \V\) with a binary relation
  \({\prec} : B \to B \to \V\) that is proposition-valued, transitive and
  satisfies
  \begin{description}
  \item[\emph{nullary interpolation}:]\label{abstract-nullary-interpolation}
    for every \(a : B\), there exists \(b : B\) with \(b \prec a\), and
  \item[\emph{binary interpolation}:\phantom{(}]\label{abstract-binary-interpolation}
    for every \(a_1,a_2 \prec b\), there exists \(a : B\) with
    \({a_1,a_2} \prec a \prec b\).\qedhere
  \end{description}
\end{definition}

\begin{definition}[Ideal, (rounded) ideal completion, \(\Idl{V}(B,\prec)\)]%
  \hfill%
  \begin{enumerate}[(i)]
  \item A subset \(I : B \to \Omega_{\V}\) of an abstract \(\V\)-basis
    \((B,{\prec})\) is a \emph{\ideal{V}} if it is a directed lower set
    with respect to \({\prec}\).\index{ideal}
    That it is a lower set means: if \(b \in I\) and \(a \prec b\), then
    \(a \in I\) too.%
    \index{lower set}%
  \item We write \(\Idl{V}(B,{\prec})\) for the type of \ideal{V}s of an
    abstract \(\V\)-basis \((B,{\prec})\) and call \(\Idl{V}({B,\prec})\) the
    \emph{(rounded) ideal completion} of \((B,\prec)\). \qedhere%
    \index{ideal!completion|textbf}%
    \nomenclature[Idl]{$\Idl{V}({B,\prec})$}{rounded \(\V\)-ideal completion of
      an abstract \(\V\)-basis \((B,{\prec})\)}
  \end{enumerate}
\end{definition}

The name rounded ideal completion is justified
by~\cref{Idl-is-dcpo,ideals-are-rounded} below.

\begin{definition}[Union of ideals, \(\bigcup \mathcal I\)]%
  \label{def:union-of-ideals}\index{ideal!union}
  Given a family \(\mathcal I : S \to \Idl{V}(B,{\prec})\) of ideals, indexed
  by \(S : \V\), we write
  \[
    \textstyle\bigcup \mathcal I \colonequiv \set{b \in B \mid \exists_{s : S}\pa*{b \in
        \mathcal I_s}}
  \]%
  \nomenclature[cup]{$\bigcup \mathcal I$}{union of a family of ideals}%
  for the set-theoretic union of the ideals indexed by \(\mathcal I\).
\end{definition}

\begin{lemma}\label{union-of-directed-family-of-ideals-is-ideal}
  If \(\mathcal I : S \to \Idl{V}(B,{\prec})\) is directed, then
  \(\bigcup \mathcal I\) is an ideal.
\end{lemma}
\begin{proof}
  The subset \(\bigcup \mathcal I\) is easily seen to be a lower set, for if
  \(a \prec b \in \bigcup \mathcal I\), then there exists \(s : S\) such that
  \(a \prec b \in \mathcal I_s\), so \(a \in \mathcal I_s\) as \(\mathcal I_s\)
  is a lower set, but then \(a \in \bigcup \mathcal I\).
  Moreover, \(\bigcup \mathcal I\) is inhabited: Since \(\mathcal I\) is
  directed, there exists \(s : S\), but \(\mathcal I_s\) is an ideal and
  therefore inhabited, so there exists \(b \in \mathcal I_s\) which implies
  \(b \in \bigcup \mathcal I\).
  Finally, suppose we have \(b_1,b_2 \in \bigcup \mathcal I\). By definition,
  there exist \(s_1,s_2 : S\) such that \(b_1 \in \mathcal I_{s_1}\) and
  \(b_2 \in \mathcal I_{s_2}\).
  By directedness of \(S\), there exists \(s : S\) such that
  \(\mathcal I_{s_1}, \mathcal I_{s_2} \subseteq \mathcal I_s\). Hence,
  \(b_1,b_2 \in \mathcal I_s\), which is an ideal, so there exists
  \(b \in \mathcal I_s\) with \(b_1,b_2 \prec b\). But then also
  \(b \in \bigcup \mathcal I\), which proves that \(\bigcup\mathcal I\) is
  directed and hence an ideal, completing the proof.
\end{proof}

\begin{lemma}\label{Idl-is-dcpo}
  The rounded ideal completion of an abstract \(\V\)-basis \((B,{\prec})\) is a
  \(\V\)-dcpo when ordered by subset inclusion.
\end{lemma}
\begin{proof}
  Since taking unions yields the least upper bound in the powerset, we only have
  to prove that the union of ideals is again an ideal, but this
  is~\cref{union-of-directed-family-of-ideals-is-ideal}.
\end{proof}

Paying attention to the universe levels, the ideals form a large but locally
small \(\V\)-dcpo because \(\Idl{V}(B,{\prec}) : \DCPO{V}{V^+}{V}\).
For the remainder of this section, we will fix an abstract \(\V\)-basis
\((B,{\prec})\) and consider its \ideal{V}s.

\begin{lemma}[Roundedness]\label{ideals-are-rounded}%
  \index{ideal!roundedness}
  The ideals of an abstract basis are \emph{rounded}: for every element \(a\) of
  an ideal~\(I\), there exists \(b \in I\) such that \(a \prec b\).
\end{lemma}
\begin{proof}
  Because ideals are semidirected.
\end{proof}

Roundedness makes up for the fact that we have not required an abstract basis to
be reflexive. If it is, then (\cref{sec:rounded-ideal-completion-reflexive}) the
ideal completion is structurally algebraic.

\begin{definition}[Principal ideal, \(\dset b\)]%
  \index{ideal!principal}%
  The \emph{principal ideal} of an element \(b : B\) is defined as the subset
  \(\dset b \colonequiv \set{a \in B \mid a \prec b}\).%
  \nomenclature[dasetb]{$\dset b$}{principal ideal of a basis element \(b\)}
  Observe that the principal ideal is indeed an ideal: it is a lower set by
  transitivity of \({\prec}\), and inhabited and semidirected precisely by
  nullary and binary interpolation, respectively.
\end{definition}

\begin{lemma}\label{principal-ideal-is-monotone}
  The assignment \(b \mapsto \dset b\) is monotone, i.e.\ if \(a \prec b\), then
  \(\dset a \subseteq \dset b\).
\end{lemma}
\begin{proof}
  By transitivity of \({\prec}\).
\end{proof}

\begin{lemma}\label{directed-sup-of-principal-ideals}
  Every ideal is the directed supremum of its principal ideals. That is, for an
  ideal \(I\), the family
  \(\pa*{\Sigma_{b : B}\pa*{b \in I}} \xrightarrow{b \mapsto \dset b}
  \Idl{V}(B,{\prec})\) is directed and has supremum \(I\).
\end{lemma}
\begin{proof}
  Since ideals are lower sets, we have \(\dset b \subseteq I\) for every
  \(b \in I\). Hence, the union \(\bigcup_{b \in I}\dset b\) is a subset of
  \(I\). Conversely, if \(a \in I\), then by roundedness of \(I\) there exists
  \(a' \in I\) with \(a \prec a'\), so that \(a \in \bigcup_{b \in I}\dset b\).
  So it remains to show that the family is directed. Notice that it is
  inhabited, because \(I\) is an ideal.
  Now suppose that \(b_1,b_2 \in I\). Since \(I\) is directed, there exists
  \(b \in I\) such that \({b_1,b_2} \prec b\). But this implies
  \({\dset b_1,\dset b_2} \subseteq \dset b\)
  by~\cref{principal-ideal-is-monotone}, so the family is semidirected, as
  desired.
\end{proof}

\begin{lemma}\label{Idl-way-below-characterization}%
  \index{way-below relation!in the rounded ideal completion}%
  The following are equivalent for every two ideals \(I\) and \(J\):
  \begin{enumerate}[(i)]
  \item\label{Idl-way-below} \(I \ll J\);
  \item\label{Idl-way-below-1} there exists \(b \in J\) such that
    \(I \subseteq \dset b\);
  \item\label{Idl-way-below-2} there exist \(a \prec b\) such that
    \(I \subseteq \dset a \subseteq \dset b \subseteq J\).
  \end{enumerate}
  In particular, if \(b\) is an element of an ideal \(I\), then
  \(\dset b \ll I\).
\end{lemma}
\begin{proof}
  We show that \ref{Idl-way-below} \(\Rightarrow\) \ref{Idl-way-below-1}
  \(\Rightarrow\) \ref{Idl-way-below-2} \(\Rightarrow\) \ref{Idl-way-below}.
  So suppose that \(I \ll J\). Then \(J\) is the directed supremum of its
  principal ideals by \cref{directed-sup-of-principal-ideals}. Hence, there
  exists \(b \in J\) such that \(I \subseteq \dset b\) already, which is exactly
  \ref{Idl-way-below-1}.
  Now suppose that we have \(a \in J\) with \(I \subseteq \dset a\). By
  roundedness of \(J\), there exists \(b \in J\) with \(a \prec b\). But then
  \(I \subseteq \dset a \subseteq \dset b \subseteq J\) by
  \cref{principal-ideal-is-monotone} and the fact that \(J\) is a lower set,
  establishing \ref{Idl-way-below-2}.
  Now suppose that condition \ref{Idl-way-below-2} holds and that \(J\) is a
  subset of some directed join of ideals \(\mathcal J\) indexed by a type
  \(S : \V\).
  Since \(a \in \dset b \subseteq J\), there exists \(s : S\) such that
  \(a \in \mathcal J_s\). In particular, \(\dset a \subseteq \mathcal J_s\)
  because ideals are lower sets. Hence, if \(a' \in I \subseteq \dset a\), then
  \(a' \in \mathcal J_s\), so \(I \subseteq \mathcal J_s\), which proves that
  \(I \ll J\).

  Finally, if \(b\) is an element of an ideal \(I\), then \(\dset b \ll I\),
  because \ref{Idl-way-below-1} implies \ref{Idl-way-below} and
  \(\dset b \subseteq \dset b\) obviously holds.
\end{proof}

\begin{theorem}\label{Idl-has-small-basis}%
  \index{basis}%
  The principal ideals \(\dset{(-)} : B \to \Idl{V}(B,{\prec})\) yield a small
  basis for \(\Idl{V}(B,{\prec})\). In particular, \(\Idl{V}(B,{\prec})\) is
  (structurally) continuous.
\end{theorem}
\begin{proof}
  First of all, note that the way-below relation on \(\Idl{V}(B,{\prec})\) is
  small-valued because of \cref{Idl-way-below-characterization}. So it remains
  to show that for every ideal \(I\), the family
  \(\pa*{\Sigma_{b : B}\pa*{\dset b \ll I}} \xrightarrow{b \mapsto \dset b}
  \Idl{V}(B,{\prec})\) is directed with supremum \(I\).
  That the domain of this family is inhabited follows from
  \cref{Idl-way-below-characterization} and the fact that \(I\) is inhabited.
  For semidirectedness, suppose we have \(b_1,b_2 : B\) with
  \({\dset b_1,\dset b_2} \ll I\). By \cref{Idl-way-below-characterization}
  there exist \({c_1,c_2} \in I\) such that \(\dset b_1 \subseteq \dset c_1\)
  and \(\dset b_2 \subseteq \dset c_2\).
  Since \(I\) is directed, there exists \(b \in I\) with \({c_1,c_2} \prec b\).
  But now \(\dset b_1 \subseteq \dset c_1 \subseteq \dset b \ll I\) by
  \cref{principal-ideal-is-monotone,Idl-way-below-characterization} and
  similarly, \(\dset b_2 \subseteq \dset b \ll I\). Hence, the family is
  semidirected, as we wished to show.
  Finally, we show that \(I\) is the supremum of the family. If \(b \in I\),
  then, since \(I\) is rounded, there exists \(c \in I\) with \(b \prec
  c\). Moreover, \(\dset c \ll I\) by
  \cref{Idl-way-below-characterization}. Hence, \(b\) is included in the join of
  the family.
  Conversely, if we have \(b : B\) with \(\dset b \ll I\), then
  \(\dset b \subseteq I\), so \(I\) is also an upper bound for the family.
\end{proof}

\subsection{The rounded ideal completion of a reflexive abstract basis}%
\label{sec:rounded-ideal-completion-reflexive}

\begin{lemma}\label{abstract-basis-if-reflexive}%
  \index{reflexivity}%
  If \({\prec} : B \to B \to \V\) is proposition-valued, transitive and
  reflexive, then \((B,{\prec})\) is an abstract basis.
\end{lemma}
\begin{proof}
  The interpolation properties for \({\prec}\) are easily proved when it is
  reflexive.
\end{proof}

\begin{lemma}\label{principal-ideal-below-characterization}
  If an element \(b : B\) is reflexive, i.e.\ \(b \prec b\) holds, then
  \(b \in I\) if and only if \(\dset b \subseteq I\) for every ideal \(I\).
\end{lemma}
\begin{proof}
  The left-to-right implication holds because \(I\) is a lower set and the
  converse holds because \(b \in \dset b\) as \(b\) is assumed to be reflexive.
\end{proof}

\begin{lemma}\label{principal-ideals-are-compact}
  If \(b : B\) is reflexive, then its principal ideal \(\dset b\) is compact.
\end{lemma}
\begin{proof}
  Suppose that we have \(b : B\) such that \(b \prec b\) holds and that
  \(\dset b \subseteq \bigcup \mathcal I\) for some directed family
  \(\mathcal I\) of ideals. By \cref{principal-ideal-below-characterization}, we
  have \(b \in \bigcup\mathcal I\), which means that there exists \(s\) in the
  domain of \(\mathcal I\) such that \(b \in \mathcal I_s\). Using
  \cref{principal-ideal-below-characterization} once more, we see that
  \(\dset b \subseteq \mathcal I_s\), proving that \(\dset b\) is compact.
\end{proof}

\begin{theorem}\label{Idl-has-small-compact-basis}%
  \index{basis!compact}%
  If \({\prec}\) is reflexive, then the principal ideals
  \(\dset{(-)} : B \to \Idl{V}(B,{\prec})\) yield a small \emph{compact} basis
  for \(\Idl{V}(B,{\prec})\). In particular, \(\Idl{V}(B,{\prec})\) is
  (structurally) algebraic.
\end{theorem}
\begin{proof}
  This follows from
  \cref{Idl-has-small-basis,principal-ideals-are-compact,small-and-compact-basis}.
\end{proof}

\begin{theorem}\label{Idl-mediating-map}
  If \(f : B \to D\) is a monotone map to a \(\V\)-dcpo \(D\), then the map
  \(\bar{f} : \Idl{V}(B,{\prec}) \to D\) defined by taking an ideal \(I\) to the
  supremum of the directed family
  \({f \circ \fst} : \pa*{\Sigma_{b : B}\pa*{b \in I}} \to D\) is Scott
  continuous.
  Moreover, if \({\prec}\) is reflexive, then the diagram
  \[
    \begin{tikzcd}
      B \ar[rr,"f"] \ar[dr,"\dset{(-)}"'] & & D \\
      & \Idl{V}(B,{\prec}) \ar[ur,"\bar{f}"']
    \end{tikzcd}
  \]
  commutes.
\end{theorem}
\begin{proof}
  Note that \({f \circ \fst} : \pa*{\Sigma_{b : B}\pa*{b \in I}} \to D\) is
  indeed a directed family, because \(I\) is a directed subset of \(B\) and
  \(f\) is monotone.
  For Scott continuity of \(\bar{f}\), assume that we have a directed family
  \(\mathcal I\) of ideals indexed by \(S : \V\).
  We first show that \(\bar{f}\pa*{\bigcup \mathcal I}\) is an upper bound of
  \(\bar{f} \circ \mathcal I\). So let \(s : S\) be arbitrary and note that
  \(\bar{f}\pa*{\bigcup\mathcal I} \supseteq \bar{f}\pa*{\mathcal I_s}\) as soon
  as \(\bar{f}\pa*{\bigcup \mathcal I} \aboveorder f(b)\) for every
  \(b \in \mathcal I_s\). But for such \(b\) we have
  \(b \in \bigcup \mathcal I\), so this holds.
  Now suppose that \(y\) is an upper bound of \(\bar{f} \circ \mathcal I\). To
  show that \(\bar{f}\pa*{\bigcup\mathcal I} \below y\), it is enough to prove
  that \(f(b) \below y\) for every \(b \in \mathcal I\). But for such \(b\),
  there exists \(s : S\) such that \(b \in \mathcal I_s\) and hence,
  \(f(b) \below \bar{f}\pa*{\mathcal I_s} \below y\).

  Finally, if \({\prec}\) is reflexive, then we prove that
  \(\bar{f}\pa*{\dset b} = f(b)\) for every \(b : B\) by antisymmetry.  Since
  \({\prec}\) is assumed to be reflexive, we have \(b \in \dset b\) and
  therefore, \(f(b) \below \bar{f}\pa*{\dset b}\).
  Conversely, for every \(c \prec b\) we have \(f(c) \below f(b)\) by
  monotonicity of \(f\) and hence, \(\bar{f}(\dset b) \below f(b)\), as desired.
\end{proof}

\subsection{Example: the ideal completion of dyadics}\label{sec:dyadics}

We end this section by describing an example of a continuous dcpo, built using
the ideal completion, that is not algebraic. In fact, this dcpo has no compact
elements at all.

We inductively define a type and an order representing dyadic rationals
\(m / 2^n\) in the interval \((-1,1)\) for integers \(m,n\).
This type is similar to the lower Dedekind reals but with dyadics instead of
rationals and is extended with a point at \(+\infty\).%
\index{Dedekind real}\index{rational number}
We prefer to work with this type, because working with lower Dedekind reals
would require us to develop and formalise the theory of integers, rational
numbers, etc.

The~intuition for the upcoming definitions is the following. Start with the
point~\(0\) in the middle of the interval (represented by \emph{middle}
below). Then consider the two functions (respectively represented by \emph{left}
and \emph{right} below)
\begin{align*}
  l,r &: (-1,1) \to (-1,1) \\
  l(x) &\colonequiv (x-1)/2 \\
  r(x) &\colonequiv (x+1)/2
\end{align*}
that generate the dyadic rationals. Observe that \(l(x) < 0 < r(x)\) for every
\(x : (-1,1)\). Accordingly, we inductively define the following types.

\begin{definition}[{Dyadics}, \(\dyadics\), \(\prec\)]%
  \index{dyadics|textbf}%
  The type of \emph{dyadics} \(\dyadics : \U_0\) is the inductive type with
  these three constructors
  \[
    \dyadicmiddle : \dyadics \quad \dyadicleft : {\dyadics \to \dyadics}
    \quad \dyadicright : {\dyadics \to \dyadics}.
  \]
  \nomenclature[D]{$\dyadics$}{type of dyadics}%
  \nomenclature[left]{$\dyadicleft$}{left constructor of the type of dyadics}%
  \nomenclature[right]{$\dyadicright$}{right constructor of the type of dyadics}%
  \nomenclature[middle]{$\dyadicmiddle$}{middle constructor of the type of dyadics}%
  We also inductively define \({\prec} : {\dyadics \to \dyadics \to \U_0}\) as
  \begin{alignat*}{6}
    \dyadicmiddle & \prec \dyadicmiddle && \colonequiv \Zero &\quad
    {\dyadicleft (x)} & \prec {\dyadicmiddle} && \colonequiv \One
    &\quad
    {\dyadicright (x)} & \prec {\dyadicmiddle} && \colonequiv \Zero \\
    \dyadicmiddle & \prec {\dyadicleft (y)} && \colonequiv \Zero &
    {\dyadicleft (x)} & \prec {\dyadicleft (y)} && \colonequiv {x \prec y} &
    {\dyadicright (x)} & \prec {\dyadicleft (y)} && \colonequiv \Zero \\
    \dyadicmiddle & \prec {\dyadicright (y)} && \colonequiv \One &
    {\dyadicleft (x)} & \prec {\dyadicright (y)} && \colonequiv \One &
    {\dyadicright (x)} & \prec {\dyadicright (y)} && \colonequiv {x \prec y}.
    \!\!\!\!\qedhere
  \end{alignat*}
\end{definition}

\begin{lemma}
  The type of dyadics is a set with decidable equality.
\end{lemma}
\begin{proof}
  Sethood follows from having decidable equality by Hedberg's
  Theorem.
  To see that \(\dyadics\) has decidable equality, one can use a standard
  inductive proof.
\end{proof}

\begin{definition}[Trichotomy, density, having no endpoints]%
  \index{trichotomy}\index{density}\index{endpoints}%
  We say that a binary relation \({<}\) on a type \(X\) is
  \begin{itemize}
  \item \emph{trichotomous} if exactly one of \(x < y\), \(x = y\) or
    \(y < x\) holds.
  \item \emph{dense} if for every \(x,y : X\), there exists some \(z : X\)
    such that \(x \prec z \prec y\).
  \item \emph{without endpoints} if for every \(x : X\), there exist some
    \(y,z : X\) with \(y \prec x \prec z\).\qedhere
  \end{itemize}
\end{definition}

\begin{lemma}\label{dyadics-order-properties}
  The relation \({\prec}\) on the dyadics is proposition-valued, transitive,
  irreflexive, trichotomous, dense and without endpoints.
\end{lemma}
\begin{proof}
  That \({\prec}\) is proposition-valued, transitive, irreflexive and
  trichotomous is all proven by a straightforward induction on the definition on
  \(\dyadics\).
  That it has no endpoints is witnessed by the fact that for every
  \(x : \dyadics\), we have
  \begin{equation}\label{left-and-right-in-order}\tag{\(\dagger\)}
    {\dyadicleft x} \prec x \prec {\dyadicright x}
  \end{equation}
  which is proven by induction on \(\dyadics\) as well.
  We spell out the inductive proof that it is dense, making use of
  \eqref{left-and-right-in-order}. Suppose that \(x \prec y\). Looking at the
  definition of the order, we see that we need to consider five cases.
  \begin{itemize}
  \item If \(x = \dyadicmiddle\) and \(y = \dyadicright y'\), then we have
    \(x \prec \dyadicright(\dyadicleft(y')) \prec y\).
  \item If \(x = \dyadicleft(x')\) and \(y = \dyadicmiddle\), then we have
    \(x \prec \dyadicleft(\dyadicright(x')) \prec y\).
  \item If \(x = \dyadicleft(x')\) and \(y = \dyadicright y'\), then we have
    \(x \prec \dyadicmiddle \prec y\).
  \item If \(x = \dyadicright (x')\) and \(y = \dyadicright y'\), then we have
    \(x' \prec y'\) and therefore, by induction hypothesis, there exists
    \(z' : \dyadics\) such that \(x' \prec z' \prec y'\).
    Hence, \(x \prec \dyadicright(z') \prec y\).
  \item If \(x = \dyadicleft (x')\) and \(y = \dyadicleft (y')\), then
    \(x' \prec y'\) and hence, by induction hypothesis, there exists
    \(z' : \dyadics\) such that \(x' \prec z' \prec y'\).
    Thus, \({x \prec \dyadicleft(z') \prec y}\). \qedhere
  \end{itemize}
\end{proof}

\begin{proposition}\label{dyadics-form-abstract-basis}%
  \index{basis!abstract}%
  The pair \((\dyadics,{\prec})\) is an abstract \(\U_0\)-basis.
\end{proposition}
\begin{proof}
  By \cref{dyadics-order-properties} the relation \({\prec}\) is
  proposition-valued and transitive.
  Moreover, that it has no endpoints implies unary interpolation.
  For binary interpolation, suppose that we have \(x \prec z\) and
  \(y \prec z\). Then by trichotomy there are three cases.
  \begin{itemize}
  \item If \(x = y\), then using density and our assumption that \(x \prec z\),
    there exists \(d : \dyadics\) with \(y = x \prec d \prec z\), as desired.
  \item If \(x \prec y\), then using density and our assumption that
    \(y \prec z\), there exists \(d : \dyadics\) with \(y \prec d \prec z\), but
    then also \(x \prec d\) since \(x \prec y\), so we are done.
  \item If \(x \prec y\), then the proof is similar to that of the second case.
    \qedhere
  \end{itemize}
\end{proof}

\begin{proposition}\index{ideal!completion}\index{compactness}\index{algebraicity}
  The ideal completion \(\Idlnum{0}(\dyadics,{\prec}) : \DCPOnum{0}{1}{0}\) is
  structurally continuous with small basis
  \(\dset (-) : \dyadics \to \Idlnum{0}(\dyadics,{\prec})\). Moreover, it cannot be
  algebraic, because none of its elements are compact.
\end{proposition}
\begin{proof}
  The first claim follows from \cref{Idl-has-small-basis}. Now suppose for a
  contradiction that we have a compact ideal \(I\). By
  \cref{Idl-way-below-characterization}, there exists \(x \in I\) with
  \(I \subseteq \dset x\). But this implies \(x \prec x\), which is impossible
  as \({\prec}\) is irreflexive.
\end{proof}

\section{Ideal completions of small bases}%
\label{sec:ideal-completions-of-small-bases}%
\index{basis}\index{basis!abstract}\index{ideal!completion}%

Given a \(\V\)-dcpo \(D\) with a small basis \(\beta : B \to D\), we show that
there are two natural ways of turning \(B\) into an abstract basis. Either
define \(b \prec c\) by \(\beta(b) \ll \beta(c)\), or by
\(\beta (b) \below \beta(c)\).
Taking their \ideal{V} completions we show that the former yields a
continuous dcpo isomorphic to \(D\), while the latter yields an algebraic dcpo
(with a small compact basis) in which \(D\) can be embedded.
The latter fact will find application in
\cref{sec:exponentials-with-small-bases}, while the former gives us a
presentation theorem: every dcpo with a small basis is isomorphic to a dcpo of
ideals. In particular, if \(D : \DCPO{V}{U}{T}\) has a small basis, then it is
isomorphic to a dcpo with simpler universe parameters, namely
\(\Idl{V}\pa*{B,{\ll_{\beta}}} : \DCPO{V}{V^+}{V}\).
Of course a similar result holds for dcpos with a small compact basis. In
studying these variations, it is helpful to first develop some machinery that
all of them have in common.

Fix a \(\V\)-dcpo \(D\) with a small basis \(\beta : B \to D\). In what follows
we conflate the family
\(\ddset_{\beta} x : \pa*{\Sigma_{b : B}\pa*{\beta(b) \ll x}} \xrightarrow{\beta
  \circ \fst} D\) with its associated subset
\(\set{b \in B \mid \beta(b) \ll x}\), formally given by the map
\(B \to \Omega_{\V}\) defined as
\(b \mapsto \exists_{b : B}\pa*{\beta(b) \ll x}\).

\begin{lemma}\label{ddsets-is-continuous}
  The assignment \(x : D \mapsto \ddset_{\beta} x : \powerset(B)\) is Scott
  continuous.
\end{lemma}
\begin{proof}
  Note that \(\ddset_{\beta}(-)\) is monotone: if \(x \below y\) and \(b : B\)
  is such that \(\beta(b) \ll x\), then also \(\beta(b) \ll y\).  So it suffices
  to prove that
  \(\ddset_{\beta} \pa*{\bigsqcup \alpha} \subseteq \bigcup_{i :
    I}\ddset_{\beta}{\alpha_i}\). So suppose that \(b : B\) is such that
  \(\beta(b) \ll \bigsqcup \alpha\). By \cref{unary-interpolation-basis}, there
  exists \(c : B\) with \(\beta(b) \ll \beta(c) \ll \bigsqcup \alpha\).
  Hence, there exists \(i : I\) such that
  \(\beta(b) \ll \beta(c) \below \alpha_i\) already, and therefore,
  \(b \in \bigcup_{j : J}\ddset_{\beta} {\alpha_j}\), as desired.
\end{proof}

By virtue of the fact that \(\beta\) is a small basis for \(D\), we know that
taking the directed supremum of \(\ddset_{\beta} x\) equals \(x\) for every
\(x : D\). In other words, \(\ddset_{\beta}{(-)}\) is a section of
\(\bigsqcup{(-)}\). The following lemma gives conditions for the other composite
to be an inflation or a deflation.

\begin{lemma}\label{inflation-deflation-criteria}
  Let \(I : B \to \Omega_{\V}\) be a subset of \(B\) such that its associated
  family
  \({\bar{I} : \pa*{\Sigma_{b : B}\pa*{b \in I}} \xrightarrow{\beta \circ \fst}
  D}\) is directed.
  \begin{enumerate}[(i)]
  \item\label{inflation-criterion} If the conjunction of
    \(\beta(b) \below \beta(c)\) and \(c \in I\) implies \(b \in I\), then
    \(\ddset_{\beta} \bigsqcup\bar{I} \subseteq I\).
  \item\label{deflation-criterion} If for every \(b \in I\) there exists
    \(c \in I\) such that \(\beta(b) \ll \beta(c)\), then
    \(I \subseteq \ddset_{\beta} \bigsqcup\bar{I}\).
  \end{enumerate}
  In particular, if both conditions hold, then
  \(I = \ddset_{\beta}\bigsqcup \bar{I}\).
\end{lemma}

Note that the first condition says that \(I\) is a lower set with respect to the
order of \(D\), while the second says that \(I\) is rounded with respect to the
way-below relation.

\begin{proof}
  \ref{inflation-criterion} Suppose that \(I\) is a lower set and let \(b : B\)
  be such that \(\beta(b) \ll \bigsqcup \bar{I}\). Then there exists \(c \in I\)
  with \(\beta(b) \below \beta(c)\), which implies \(b \in I\) as desired,
  because \(I\) is assumed to be a lower set.
  \ref{deflation-criterion} Assume that \(I\) is rounded and let \(b \in I\) be
  arbitrary. By roundedness of \(I\), there exists \(c \in I\) such that
  \(\beta(b) \ll \beta(c)\). But then
  \(\beta(b) \ll \beta(c) \below \bigsqcup\bar{I}\), so that
  \(b \in \ddset_{\beta}\bigsqcup \bar{I}\), as we wished to show.
\end{proof}

\begin{lemma}\label{semidirected-lower-set-criteria}
  Suppose that we have \({\prec} : B \to B \to \V\) and let \(x : D\) be
  arbitrary.
  \begin{enumerate}[(i)]
  \item\label{lowerset-criterion} If \(b \prec c\) implies
    \(\beta(b) \below \beta(c)\) for every \(b,c : B\), then
    \(\ddset_{\beta} x\) is a lower set w.r.t.\ \({\prec}\).
  \item\label{semidirected-criterion} If \(\beta(b) \ll \beta(c)\) implies
    \(b \prec c\) for every \(b,c : B\), then \(\ddset_{\beta} x\) is
    semidirected w.r.t.\ \(\prec\).
  \end{enumerate}
\end{lemma}
\begin{proof}
  \ref{lowerset-criterion} This is immediate, because \(\ddset{\beta} x\) is a
  lower set with respect to the order relation on \(D\). %
  \ref{semidirected-criterion} Suppose that the condition holds and that we have
  \({b_1,b_2} : B\) such that \(\beta(b_1),\beta(b_2) \ll x\). Using binary
  interpolation in the basis, there exist \({c_1,c_2} : B\) with
  \(\beta(b_1) \ll \beta(c_1) \ll x\) and \(\beta(b_2) \ll \beta(c_2) \ll
  x\). Hence, \(c_1,c_2 \in \ddset_{\beta x}\) and moreover, by assumption we
  have \(b_1 \prec c_1\) and \(b_2 \prec c_2\), as desired.
\end{proof}

\subsection{Ideal completion with respect to the way-below relation}%
\index{way-below relation|(}

\begin{lemma}\label{way-below-abstract-basis}
  If \(\beta : B \to D\) is a small basis for a \(\V\)-dcpo \(D\), then
  \(\pa*{B,\ll_{\beta}}\) is an abstract \(\V\)-basis where \(b \ll_{\beta} c\)
  is defined as \(\beta(b) \ll \beta(c)\).
  \nomenclature[llbeta]{$b \ll_{\beta} c$}{way-below relation restricted to a basis \(\beta\)}
\end{lemma}

\begin{remark}
  The definition of an abstract \(\V\)-basis requires the relation on it to be
  \(\V\)-valued. Hence, for the lemma to make sense we appeal to the fact that
  \(\beta\) is a \emph{small} basis which tells us that we can substitute
  \(\beta(b) \ll \beta(c)\) by an equivalent type in \(\V\).
\end{remark}

\begin{proof}[Proof of \cref{way-below-abstract-basis}]
  The way-below relation is proposition-valued and transitive. Moreover,
  \({\ll_{\beta}}\) satisfies nullary and binary interpolation precisely because
  we have nullary and binary interpolation in the basis for the way-below
  relation by \cref{nullary-interpolation-basis,binary-interpolation-basis}.
\end{proof}

The following theorem is a presentation result for dcpos with a small basis:
every such dcpo can be presented as the rounded ideal completion of its small
basis.
\begin{theorem}\label{Idl-iso-continuous}%
  \index{dcpo!isomorphism}
  The map \(\ddset_{\beta}{(-)} : D \to \Idl{V}\pa*{B,{\ll_{\beta}}}\) is an
  isomorphism of \(\V\)-dcpos.
\end{theorem}
\begin{proof}
  First of all, we should check that the map is well-defined, i.e.\ that
  \(\ddset_{\beta} x\) is an \(\pa*{B,{\ll_{\beta}}}\)-ideal. It is an inhabited
  subset by nullary interpolation in the basis and a semidirected lower set
  because the criteria of \cref{semidirected-lower-set-criteria} are satisfied
  when taking \({\prec}\) to be \({\ll_{\beta}}\).
  Secondly, the map \(\ddset_{\beta}{(-)}\) is Scott continuous by
  \cref{ddsets-is-continuous}.

  Now notice that the map \(\beta : \pa*{B,{\ll_\beta}} \to D\) is monotone and
  that the Scott continuous map it induces by \cref{Idl-mediating-map} is
  exactly the map \(\bigsqcup : \Idl{V}\pa*{B,{\ll_\beta}} \to D\) that takes an
  ideal \(I\) to the supremum of its associated directed family
  \(\beta \circ \fst : \pa*{\Sigma_{b : B}\pa*{b \in I}} \to D\).

  Since \(\beta\) is a basis for \(D\), we know that
  \(\bigsqcup{\ddset_{\beta} x} = x\) for every \(x : D\). So it only
  remains to show that \(\ddset_{\beta} \circ \bigsqcup\) is the identity on
  \(\Idl{V}\pa*{B,{\ll_{\beta}}}\), for which we will use
  \cref{inflation-deflation-criteria}.
  So suppose that \(I : \Idl{V}\pa*{B,{\ll_{\beta}}}\) is arbitrary. Then we
  only need to prove that
  \begin{enumerate}[(i)]
  \item\label{to-prove-lower} the conjunction of \(\beta(b) \below \beta(c)\)
    and \(c \in I\) implies \(b \in I\) for every \(b,c : B\);
  \item\label{to-prove-rounded} for every \(b \in I\), there exists \(c \in I\)
    such that \(\beta(b) \ll \beta(c)\).
  \end{enumerate}
  Note that \ref{to-prove-rounded} is just saying that \(I\) is a rounded ideal
  w.r.t.\ \({\ll_{\beta}}\), so this holds.
  For \ref{to-prove-lower}, suppose that \(\beta(b) \below \beta(c)\) and
  \(c \in I\). By roundedness of \(I\), there exists \(c' \in I\) such that
  \(c \ll_{\beta} c'\). But then \(\beta(b) \below \beta(c) \ll \beta(c')\), so
  that \(b \ll_{\beta} c'\) which implies that \(b \in I\), because ideals are
  lower sets.
\end{proof}%
\index{way-below relation|)}

\subsection{Ideal completion with respect to the order relation}

\begin{lemma}\label{below-abstract-basis}
  If \(\beta : B \to D\) is a small basis for a \(\V\)-dcpo \(D\), then
  \(\pa*{B,\below_{\beta}}\) is an abstract \(\V\)-basis where
  \(b \below_{\beta} c\) is defined as \(\beta(b) \below \beta(c)\).
  \nomenclature[sqsubseteqbeta]{$b \below_{\beta} c$}{order relation restricted to a basis \(\beta\)}
\end{lemma}
\begin{proof}
  The relation \({\below_{\beta}}\) is reflexive, so this follows from
  \cref{abstract-basis-if-reflexive}.
\end{proof}

\begin{remark}
  The definition of an abstract \(\V\)-basis requires the relation on it to be
  \(\V\)-valued. Hence, for the lemma to make sense we appeal to
  \cref{locally-small-if-small-basis} to know that \(D\) is locally small which
  tells us that we can substitute \(\beta(b) \below \beta(c)\) by an equivalent
  type in \(\V\).
\end{remark}

\begin{theorem}\label{Idl-iso-algebraic}%
  \index{retract!Scott continuous}\index{isomorphism (of dcpos)}%
  The map \(\ddset_{\beta}{(-)} : D \to \Idl{V}\pa*{B,{\below_{\beta}}}\) is the
  embedding in an embedding-projection pair.
  In particular, \(D\) is a Scott continuous retract of the algebraic dcpo
  \(\Idl{V}\pa*{B,\below_{\beta}}\) that has a small compact basis.
  Moreover, if \(\beta\) is a small \emph{compact} basis, then the map is an
  isomorphism.
\end{theorem}
\begin{proof}
  First of all, we should check that the map is well-defined, i.e.\ that
  \(\ddset_{\beta} x\) is an \(\pa*{B,{\below_{\beta}}}\)-ideal. It is an
  inhabited subset by nullary interpolation in the basis and a semidirected
  lower set because the criteria of \cref{semidirected-lower-set-criteria} are
  satisfied when taking \({\prec}\) to be \({\below_{\beta}}\).
  Secondly, the map \(\ddset_{\beta}{(-)}\) is Scott continuous by
  \cref{ddsets-is-continuous}.

  Now notice that the map \(\beta : \pa*{B,{\below_\beta}} \to D\) is monotone
  and that the continuous map it induces by \cref{Idl-mediating-map} is
  exactly the map \(\bigsqcup : \Idl{V}\pa*{B,{\below_\beta}} \to D\) that takes
  an ideal \(I\) to the supremum of its associated directed family
  \({\beta \circ \fst : {\pa*{\Sigma_{b : B}\pa*{b \in I}} \to D}}\).

  Since \(\beta\) is a basis for \(D\), we know that
  \(\bigsqcup{\ddset_{\beta} x} = x\) for every \(x : D\). So it only remains to
  show that \(\ddset_{\beta} \circ \bigsqcup\) is a deflation, for which we will
  use \cref{inflation-deflation-criteria}.
  So suppose that \(I : \Idl{V}\pa*{B,{\below_{\beta}}}\) is arbitrary. Then we
  only need to prove that the conjunction of \(\beta(b) \below \beta(c)\) and
  \(c \in I\) implies \(b \in I\), but this holds, because \(I\) is a lower set
  with respect to \({\below_{\beta}}\).

  Finally, assume that \(\beta\) is a small compact basis. We show that
  \(\ddset_{\beta} \circ \bigsqcup\) is also inflationary in this case. So let
  \(I\) be an arbitrary ideal. By \cref{inflation-deflation-criteria} it is
  enough to show that for every \(b \in I\), there exists \(c \in I\) such that
  \(\beta(b) \ll \beta(c)\). But by assumption, \(\beta(b)\) is compact, so we
  can simply take \(c\) to be \(b\).
\end{proof}

Combining Theorems~\ref{Idl-has-small-basis}~and~\ref{Idl-iso-continuous} and
Theorems~\ref{Idl-iso-algebraic}~and~\ref{Idl-has-small-compact-basis}, we
obtain the following result.
\begin{corollary}\hfill%
  \index{dcpo!isomorphism}%
  \index{basis}\index{basis!compact}\index{reflexivity}%
  \begin{enumerate}[(i)]
  \item A \(\V\)-dcpo \(D\) has a small basis if and only if it is isomorphic
    to \(\Idl{V}(B,{\prec})\) for an abstract basis \((B,{\prec})\).
  \item A \(\V\)-dcpo \(D\) has a small compact basis if and only if it is
    isomorphic to \(\Idl{V}(B,{\prec})\) for an abstract basis \((B,{\prec})\)
    where \({\prec}\) is reflexive.
  \end{enumerate}
  In particular, every \(\V\)-dcpo with a small basis is isomorphic to one whose
  order takes values in \(\V\) and whose carrier lives in \(\V^+\).
\end{corollary}

\section{Structurally continuous and algebraic bilimits}%
\label{sec:structurally-continuous-and-algebraic-bilimits}%
\index{bilimit|(}

We show that bilimits are closed under structural
continuity/algebraicity.
For the reminder of this section, fix a directed diagram of \(\V\)-dcpos
\((D_i)_{i : I}\) with embedding-projection pairs
\(\pa*{\varepsilon_{i,j},\pi_{i,j}}_{i \below j \text{ in } I}\) between them,
as in \cref{sec:bilimits}.

Now suppose that for every \(i : I\), we have \(\alpha_i : J_i \to D_i\) with
each \(J_i : \V\). Then we define \(J_\infty \colonequiv \Sigma_{i : I}J_i\) and
\(\alpha_\infty : J_\infty \to D_\infty\) by
\((i,j) \mapsto \varepsilon_{i,\infty}(\alpha_i(j))\), where
\(\varepsilon_{i,\infty}\) is as in \cref{epsilon-infty}.

\begin{lemma}\label{infty-family-directed-sup}
  If every \(\alpha_i\) is directed and we have \(\sigma : D_\infty\) such that
  \(\alpha_i\) approximates \(\sigma_i\), then \(\alpha_\infty\) is directed and
  approximates \(\sigma\).
\end{lemma}
\begin{proof}
  Observe that \(\alpha_\infty\) is equal to the supremum, if it exists, of the
  directed families \(\pa*{\varepsilon_{i,\infty} \circ \alpha_i}_{i : I}\) in
  the ind-completion of \(D_\infty\), cf.\ the proof of
  \cref{ind-completion-is-directed-complete}.
  Hence, for directedness of \(\alpha_\infty\), it suffices to prove that the
  family \(i \mapsto \varepsilon_{i,\infty} \circ \alpha_i\) is directed with
  respect to cofinality.
  The index type \(I\) is inhabited, because we are working with a directed
  diagram of dcpos. For semidirectedness, we will first prove that if
  \(i \below i'\), then \(\varepsilon_{i,\infty} \circ \alpha_i\) is cofinal in
  \(\varepsilon_{i',\infty} \circ \alpha_{i'}\).

  So suppose that \(i \below i'\) and \(j : J_i\). As \(\alpha_{i}\)
  approximates \(\sigma_{i}\), we have \(\alpha_{i}(j) \ll \sigma_i\).
  Because \(\varepsilon_{i,i'}\) is an embedding, it preserves the way-below
  relation (\cref{embedding-preserves-and-reflects-way-below}), so that we get
  \(\varepsilon_{i,i'}(\alpha_i(j)) \ll \varepsilon_{i,i'}(\sigma_i) \below
  \sigma_{i'} = \bigsqcup \alpha_{i'}\). Hence, there exists \(j' : J_{i'}\)
  with \(\varepsilon_{i,i'}(\alpha_i(j)) \below \alpha_{i'}(j')\) which yields
  \(\varepsilon_{i,\infty}(\alpha_i(j)) =
  \varepsilon_{i',\infty}\pa*{\varepsilon_{i,i'}(\alpha_i(j))} \below
  \varepsilon_{i',\infty}(\alpha_{i'}(j'))\), completing the proof that
  \(\varepsilon_{i,\infty} \circ \alpha_i\) is cofinal in
  \(\varepsilon_{i',\infty} \circ \alpha_{i'}\).

  Now to prove that the family
  \(i \mapsto \varepsilon_{i,\infty} \circ \alpha_i\) is semidirected with
  respect to cofinality, suppose we have \(i_1,i_2 : I\). Since \(I\) is a
  directed preorder, there exists \(i : I\) such that \(i_1,i_2 \below i\). But
  then \(\varepsilon_{i_1,\infty} \circ \alpha_{i_1}\) and
  \(\varepsilon_{i_2,\infty} \circ \alpha_{i_2}\) are both cofinal in
  \(\varepsilon_{i,\infty} \circ \alpha_i\) by the above.

  Thus, \(\alpha_\infty\) is directed. To see that its supremum is \(\sigma\),
  observe that
  \begin{align*}
    \sigma &= \textstyle\bigsqcup_{i : I} \varepsilon_{i,\infty}(\sigma_i)
    &&\text{(by \cref{sigma-sup-of-epsilon-pis})} \\
    &= \textstyle\bigsqcup_{i : I} \varepsilon_{i,\infty}\pa*{\textstyle\bigsqcup \alpha_i}
    &&\text{(since \(\alpha_i\) approximates \(\sigma_i\))} \\
    &= \textstyle\bigsqcup_{i : I}\bigsqcup \varepsilon_{i,\infty} \circ \alpha_i
    &&\text{(by Scott continuity of \(\varepsilon_{i,\infty}\))} \\
    &= \textstyle\bigsqcup_{(i,j) : J_\infty} \alpha_\infty(i,j),
  \end{align*}
  as desired.

  Finally, we wish to show that \(\alpha_\infty(i,j) \ll \sigma\) for every
  \((i,j) : J_\infty\). But \(\varepsilon_{i,\infty}\) is an embedding and
  therefore preserves the way-below relation while \(\alpha_i(j)\)
  approximates~\(\sigma_i\), so we get
  \(\alpha_\infty(i,j) \equiv \varepsilon_{i,\infty}(\alpha_i(j)) \ll
  \varepsilon_{i,\infty}(\sigma_i) \below \sigma\) where the final inequality
  holds because \(\varepsilon_{i,\infty} \circ \pi_{i,\infty}\) is a deflation.
\end{proof}

\begin{lemma}\label{infty-family-compact}
  If \(\alpha_i(j)\) is compact for every \(i : I\) and \(j : J_i\), then all
  the values of \(\alpha_\infty\) are compact too.
\end{lemma}
\begin{proof}
  Let \((i,j) : J_\infty\) be arbitrary. Since \(\varepsilon_{i,\infty}\) is an
  embedding it preserves compact elements, so
  \(\alpha_\infty(i,j) \equiv \varepsilon_{i,\infty}(\alpha_i(j))\) is compact.
\end{proof}

\begin{theorem}\label{structurally-continuous-bilimit}%
  \index{algebraicity!structural!of the bilimit}%
  \index{continuity!structural!of the bilimit}%
  If each \(D_i\) is structurally continuous, then so is \(D_\infty\).
  Furthermore, if each \(D_i\) is structurally algebraic, then so is
  \(D_\infty\).
\end{theorem}
\begin{proof}
  Let \(\sigma : D_\infty\) be arbitrary. By structural continuity of each
  \(D_i\), we have a directed family \(\alpha_i : J_i \to D_i\) approximating
  \(\sigma_i\). Hence, by \cref{infty-family-directed-sup}, the family
  \(\alpha_\infty\) is directed and approximates \(\sigma\), proving the
  structural continuity of \(D_\infty\).
  Now if each \(D_i\) is structurally algebraic, then \(D_\infty\) is
  structurally algebraic by \cref{infty-family-compact} and the above.
\end{proof}

Note that we do not expect to be able to prove that \(D_\infty\) is continuous
if each \(D_i\) is, because it would require an instance of the axiom of choice
to get the continuity structures on each \(D_i\) and without those we have
nothing to operate on.

\begin{theorem}\label{bilimit-has-small-basis}%
  \index{basis!for the bilimit}%
  \index{basis!compact!for the bilimit}%
  If each \(D_i\) has a small basis \(\beta_i : B_i \to D_i\), then the map
  \(\beta_\infty\) defined by
  \(\beta_\infty : \pa*{B_\infty \colonequiv \pa*{\Sigma_{i : I}B_i}}
  \xrightarrow{(i,b) \mapsto \varepsilon_{i,\infty}(\beta_i(b))} D_\infty\) is a
  small basis for \(D_\infty\).
  Furthermore, if each \(\beta_i\) is a small compact basis, then
  \(\beta_\infty\) is a small compact basis too.
\end{theorem}
\begin{proof}
  First of all, we must show that \(\beta_\infty(i,b) \ll \sigma\) is small for
  every \(i : I\), \(b : B_i\) and \(\sigma : D_\infty\). We claim that this is
  the case as the way-below relation on \(D_\infty\) has small values. By
  \cref{cont-loc-small-iff-way-below-small} and
  \cref{structurally-continuous-bilimit}, it suffices to prove that \(D_\infty\)
  is locally small. But this holds by \cref{locally-small-bilimit} as each
  \(D_i\) is locally small by \cref{locally-small-if-small-basis}.

  It remains to prove that, for an arbitrary element \(\sigma : D_\infty\), the
  family \(\ddset_{\beta_\infty} \sigma\) given by
  \(\pa*{\Sigma_{(i,b) : B_\infty}\beta_\infty(i,b) \ll \sigma}
  \xrightarrow{\beta_\infty \circ \fst} D_\infty\) is directed with supremum
  \(\sigma\).
  Note that for every \(i : I\) and \(b : B_i\), we have that
  \(\beta_i(b) \ll \sigma_i\) implies
  \[
    \beta_\infty(i,b) \equiv \varepsilon_{i,\infty}(\beta_i(b)) \ll
    \varepsilon_{i,\infty}(\sigma_i) \below \sigma,
  \]
  since \cref{embedding-preserves-and-reflects-way-below} tells us that the
  embedding \(\varepsilon_{i,\infty}\) preserves the way-below relation.
  Hence, the identity induces a well-defined map
  \[
    \iota : \pa*{\Sigma_{i : I}\Sigma_{b :B_i}\beta_i(b) \ll \sigma_i} \to
    \pa*{\Sigma_{(i,b) : B_\infty}\beta_\infty(i,b) \ll \sigma}.
  \]
  \cref{subbasis-lemma} now tells us that we only need to show that
  \(\ddset_{\beta_\infty}\sigma \circ \iota\) is directed and has supremum
  \(\sigma\).
  But if we write
  \(\alpha_i : \pa*{\Sigma_{b : B_i}\beta_i(b) \ll \sigma_i} \to D_i\) for the
  map \(b \mapsto \beta_i(b)\), then we see that
  \(\ddset_{\beta_\infty} \sigma \circ \iota\) is given by \(\alpha_\infty\), as
  defined at the start of this section. But then \(\alpha_\infty\) is indeed
  seen to be directed with supremum \(\sigma\) by virtue of
  \cref{infty-family-directed-sup} and the fact that \(\alpha_i\) approximates
  \(\sigma_i\).

  Finally, if every \(\beta_i\) is a small compact basis, then \(\beta_\infty\)
  is also a small compact basis because by \cref{small-and-compact-basis} all we
  need to know is that
  \(\beta_\infty(i,b) \equiv \varepsilon_{i,\infty}(\beta_i(b))\) is compact for
  every \(i : I\) and \(b : B_i\). But this follows from the fact that
  embeddings preserve compactness and that each \(\beta_i(b)\) is compact.
\end{proof}%
\index{bilimit|)}

\section{Exponentials with small (compact) bases}%
\label{sec:exponentials-with-small-bases}%
\index{exponential}%

Just as in the classical, impredicative setting, the exponential of two
continuous dcpos need not be continuous~\cite{Jung1988}.
However, with some work, we are able to show that \(E^D\) has a small basis
provided that both \(D\) and \(E\) do and that \(E\) has all (not necessarily
directed) \(\V\)-suprema.
We first establish this for small compact bases using \emph{step functions} and
then derive the result for compact bases using \cref{Idl-iso-algebraic}.

\subsection{Single step functions}

Suppose that we have a dcpo \(D\) and a pointed dcpo
\(E\). Classically~\cite[Exercise~II-2.31]{GierzEtAl2003}, the single step
function given by \(d : D\) and \(e : E\) is defined as
\begin{align*}
  \ssf{d}{e} : D &\to E \\
  x &\mapsto
  \begin{cases}
    e &\text{if \(d \below x\);} \\
    \bot &\text{otherwise}.
  \end{cases}
\end{align*}

Constructively, we can't expect to make this case distinction, so we define
single step functions using subsingleton suprema instead.

\begin{definition}[Single step function, \(\ssf{d}{e}\)]%
  \index{single step function}\index{supremum!subsingleton}
  The \emph{single step function} given by two elements \(d : D\) and \(e : E\),
  where \(D\) is a locally small \(\V\)-dcpo and \(E\) is a pointed \(\V\)-dcpo,
  is the function \(\ssf{d}{e} : D \to E\) given by mapping \(x : D\) to the
  supremum of the family indexed by the subsingleton \(d \below x\) that is
  constantly \(e\).
  \nomenclature[arrowzzz]{$\ssf{d}{e}$}{single step function}
\end{definition}

\begin{remark}
  Recall from \cref{pointed-dcpos-sups} that the supremum of a
  subsingleton-indexed family \(\alpha : {P \to E}\) is given by the supremum of
  the directed family \(\One + P \to E\) defined by \(\inl(\star) \mapsto \bot\)
  and \(\inr(p) \mapsto \alpha(p)\).
  Note that we need \(D\) to be locally small, because we need \(d \below x\) to
  be a subsingleton in \(\V\) to use the \(\V\)-directed-completeness of \(E\).
\end{remark}

\begin{lemma}\label{single-step-function-continuous-if-compact}
  If \(d : D\) is compact, then \(\ssf{d}{e}\) is Scott continuous for every
  \(e : E\).
\end{lemma}
\begin{proof}
  Suppose that \(d : D\) is compact and that \(\alpha : I \to D\) is a directed
  family. We first show that \(\ssf{d}{e}\pa*{\bigsqcup \alpha}\) is an upper
  bound of \(\ssf{d}{e} \circ \alpha\). So let \(i : I\) be arbitrary. Then we
  have to prove
  \(\bigsqcup_{d \below \alpha_i} e \below \bigsqcup {\ssf{d}{e} \circ
    \alpha}\). Since the supremum gives a lower bound of upper bounds, it
  suffices to prove that \(e \below \bigsqcup {\ssf{d}{e} \circ \alpha}\)
  whenever \(d \below \alpha_i\). But in this case we have
  \(e = \ssf{d}{e}(\alpha_i) \below \bigsqcup {\ssf{d}{e} \circ \alpha}\),
  so we are done.

  To see that \(\ssf{d}{e}\pa*{\bigsqcup \alpha}\) is a lower bound of upper
  bounds, suppose that we are given \(y : E\) such that \(y\) is an upper bound
  of \(\ssf{d}{e} \circ \alpha\). We are to prove that
  \(\pa*{\bigsqcup_{d \below \bigsqcup \alpha} e} \below y\).
  Note that it suffices for \(d \below \bigsqcup \alpha\) to imply \(e \below y\).
  So assume that \(d \below \bigsqcup \alpha\). By compactness of \(d\) there
  exists \(i : I\) such that \(d \below \alpha_i\) already. But then
  \(e = \ssf{d}{e}(\alpha_i) \below y\), as desired.
\end{proof}

\begin{lemma}\label{above-single-step-function}%
  \index{compactness!of single step functions}%
  A Scott continuous function \(f : D \to E\) is above the single step function
  \(\ssf{d}{e}\) with \(d : D\) compact if and only if \(e \below f(d)\).
\end{lemma}
\begin{proof}
  Suppose that \(\ssf{d}{e} \below f\). Then \(\ssf{d}{e}(d) = e \below f(d)\),
  proving one implication.
  Now assume that \(e \below f(d)\) and let \(x : D\) be arbitrary. To prove
  that \(\ssf{d}{e}(x) \below f(x)\), it suffices that \(e \below f(x)\)
  whenever \(d \below x\). But if \(d \below x\), then
  \(e \below f(d) \below f(x)\) by monotonicity of \(f\).
\end{proof}

\begin{lemma}\label{single-step-function-is-compact}
  If \(d\) and \(e\) are compact, then so is \(\ssf{d}{e}\) in the exponential
  \(E^D\).
\end{lemma}
\begin{proof}
  Suppose that we have a directed family \(\alpha : I \to E^D\) such that
  \(\ssf{d}{e} \below \bigsqcup \alpha\). Then we consider the directed family
  \(\alpha^d : I \to E\) given by \(i \mapsto \alpha_i(d)\).
  We claim that \(e \below \bigsqcup \alpha^d\). Indeed, by
  \cref{above-single-step-function} and our assumption that
  \(\ssf{d}{e} \below \bigsqcup\alpha\) we get
  \(e \below \pa*{\bigsqcup\alpha}(d) = \bigsqcup \alpha^d\).
  Now by compactness of \(e\), there exists \(i : I\) such that
  \(e \below \alpha^d(i) \equiv \alpha_i(d)\) already. But this implies
  \(\ssf{d}{e} \below \alpha_i\) by \cref{above-single-step-function} again,
  finishing the proof.
\end{proof}

\subsection{Exponentials with small compact bases}

Fix \(\V\)-dcpos \(D\) and \(E\) with small compact bases
\(\beta_D : B_D \to D\) and \(\beta_E : B_E \to D\) and moreover assume that
\(E\) has suprema for all (not necessarily directed) families indexed by types
in \(\V\). We are going to construct a small compact basis on the exponential
\(E^D\).%
\index{sup-completeness}

\begin{lemma}\label{sup-of-single-step-functions}
  If \(E\) is sup-complete, then every Scott continuous function
  \(f : D \to E\) is the supremum of the collection of single step functions
  \(\pa*{\ssf{\beta_D(b)}{\beta_E(c)}}_{b : B_D , c : B_E}\) that are below
  \(f\).
\end{lemma}
\begin{proof}
  Note that \(f\) is an upper bound by definition, so it remains to prove that
  it is the least. Therefore suppose we are given an upper bound
  \(g : D \to E\). We have to prove that \(f(x) \below g(x)\) for every
  \(x : D\), so let \(x : D\) be arbitrary. Now
  \(x = \bigsqcup \ddset_{\beta_D} x\), because \(\beta_D\) is a small compact
  basis for \(D\), so by Scott continuity of \(f\) and \(g\), it suffices to
  prove that \(f(\beta_D(b)) \below g(\beta_D(b))\) for every \(b : B_D\).
  So let \(b : B_D\) be arbitrary. Since \(\beta_E\) is a small compact basis
  for \(E\), we have
  \(f(\beta_D(b)) = \bigsqcup \dset_{\beta_E} f(\beta_D(b))\). So to prove
  \(f(\beta_D(b)) \below g(\beta_D(b))\) it is enough to know that
  \(\beta_E(c) \below g(\beta_D(b))\) for every \(c : B_E\) with
  \(\beta_E(c) \below f(\beta_D(b))\).
  But for such \(c : B_E\) we have \(\ssf{\beta_D(b)}{\beta_E(c)} \below f\) and
  therefore \(\ssf{\beta_D(b)}{\beta_E(c)} \below g\) because \(f\) is an upper
  bound of such single step functions, and hence
  \(\beta_E(c) \below g(\beta_D(b))\) by \cref{above-single-step-function}, as
  desired.
\end{proof}

\begin{definition}[Directification]\label{directification}%
  \index{directification}\index{type!of lists}\index{join!binary}%
  In a \(\V\)-sup-complete poset \(P\), the \emph{directification} of a family
  \(\alpha : I \to P\) is the family \(\bar\alpha : \List(I) \to P\) inductively
  defined by \([] \mapsto \bot\) and
  \(i :: l \mapsto \alpha_i \vee \bar{\alpha}(l)\), where \(\bot\) denotes the
  least element of \(P\) and \({\vee}\) the binary join.
  \nomenclature[vee]{$x \vee y$}{binary join}
  It is clear that \(\bar\alpha\) has the same supremum as \(\alpha\), and by
  concatenating lists, one sees that the directification yields a directed
  family, hence the name.
\end{definition}

\begin{lemma}\label{directification-is-compact}
  If each element of a family into a sup-complete dcpo is compact, then so are
  all elements of its directification.
\end{lemma}
\begin{proof}
  By induction, \cref{least-element-is-compact,binary-join-is-compact}.
\end{proof}

Let us write \(\sigma : B_D \times B_E \to E^D\) for the map that takes
\((b,c)\) to the single step function \(\ssf{\beta_D(b)}{\beta_E(c)}\) and
\(\beta : B \colonequiv \List(B_D \times B_E) \to E^D\) for its directification,
which exists because \(E^D\) is \(\V\)-sup-complete as \(E\) is and suprema are
calculated pointwise.

\begin{theorem}\label{exponential-has-small-compact-basis}%
  \index{basis!compact!for the exponential}%
  The map \(\beta\) is a small compact basis for the exponential \(E^D\), where
  \(E\) is assumed to be sup-complete.
\end{theorem}
\begin{proof}
  Firstly, every element in the image of \(\beta\) is compact by
  \cref{directification-is-compact,single-step-function-is-compact}.
  Secondly, for every \(b : B\) and Scott continuous map \(f : D \to E\), the
  type \(\beta(b) \below f\) is small, because \(E^D\) is locally small by
  \cref{exponential-is-locally-small}.  Thirdly, for every such \(f\), the
  family
  \(\pa*{\Sigma_{b : B}\pa*{\beta(b) \below f}} \xrightarrow{{\beta} \circ
    {\fst}} E^D\) is directed because \(\beta\) is the directification of
  \(\sigma\).  Finally, this family has supremum \(f\) because
  of~\cref{sup-of-single-step-functions}.
\end{proof}

\subsection{Exponentials with small bases}

We now present a variation of~\cref{exponential-has-small-compact-basis} but for
(sup-complete) dcpos with small bases. In fact, we will prove it using
\cref{exponential-has-small-compact-basis} and the theory of Scott continuous
retracts (\cref{Idl-iso-algebraic} in particular).

\begin{definition}[Closure under finite joins]%
  \index{join!finite}%
  A small basis \(\beta : B \to D\) for a sup-complete poset is \emph{closed
    under finite joins} if we have \(b_\bot : B\) with \(\beta(b_\bot) = \bot\) and a
  map \({\vee} : B \to B \to B\) such that
  \(\beta(b \vee c) = \beta(b) \vee \beta(c)\) for every \(b,c : B\).
\end{definition}

\begin{lemma}\label{close-basis-under-finite-joins}
  If \(D\) is a sup-complete dcpo with a small basis \(\beta : B \to D\), then
  the directification of \(\beta\) is also a small basis for \(D\). Moreover, by
  construction, it is closed under finite joins.
\end{lemma}
\begin{proof}
  Since \(\beta\) is a small basis for \(D\), it follows by
  \cref{locally-small-if-small-basis} that the way-below relation on \(D\) is
  small-valued. Hence, writing \(\bar\beta\) for the directification of
  \(\beta\), it remains to prove that \(\ddset_{\bar\beta} x\) is directed with
  supremum \(x\) for every \(x : D\). But this follows easily from
  \cref{subbasis-lemma}, because \(\ddset_\beta x\) is directed with supremum
  \(x\) and this family is equal to the composite
  \(\pa*{\Sigma_{b : B}\pa*{\beta(b) \ll x}} \xhookrightarrow{b \mapsto [b]}
  \pa*{\Sigma_{l : \List(B)}\pa*{\bar\beta(l) \ll x}} \xrightarrow{\bar\beta
    \circ \fst} D\).
\end{proof}

\begin{lemma}\label{sup-complete-ideal-completion}
  If \(D\) is a \(\V\)-sup-complete poset with a small basis \(\beta : B \to D\)
  that is closed under finite joins, then the ideal-completion
  \(\Idl{V}(B,\below)\) is \(\V\)-sup-complete too.
\end{lemma}
\begin{proof}
  Since the \(\V\)-ideal completion is \(\V\)-directed complete, it suffices to
  show that \(\Idl{V}(B,\below)\) has finite joins. Since \(\beta : B \to D\) is
  closed under finite joins, we have \(b_\bot : B\) with
  \(\beta(b_\bot) = \bot\) and we easily see that \(\set{b_\bot}\) is the least
  element of \(\Idl{V}(B,\below)\).
  Now suppose that we have two ideals \(I , J : \Idl{V}(B,\below)\) and consider
  the subset
  \[K \colonequiv \set{b \in B \mid \exists_{b_0,\dots,b_{n-1} \in I}\exists_{c_0,\dots,c_{m-1}
      \in J}\pa*{\beta(b) \below \beta\pa*{b_0 \vee \dots \vee b_{n-1} \vee c_0 \vee \dots
        \vee c_{m-1}}}}.
  \]
  Observe that \(K\) is a lower set and that it is directed as \(B\) is closed
  under finite joins. Thus, \(K \in \Idl{V}(B,\below)\).
  We claim that \(K\) is the join of \(I\) and \(J\). First of all, \(I\) and
  \(J\) are both subsets of \(K\), so it remains to prove that \(K\) is the
  least upper bound. To this end, suppose that we have an ideal \(L\) that
  includes \(I\) and \(J\), and let \(b \in K\) be arbitrary. We show that
  \(b \in L\).
  Since \(b \in K\), there exist \(b_0,\dots,b_{n-1} \in I\) and
  \(c_0,\dots,c_{m-1}\in J\) such that
  \(\beta(b) \below \beta(b_0 \vee \dots \vee b_{n-1} \vee c_0 \vee \dots \vee
  c_{m-1})\).
  Then \(b_0,\dots,b_{n-1} \in L\) and there exists \(b \in L\) such that
  \(\beta(b_0),\dots,\beta(b_{n-1}) \below \beta(b)\) as \(L\) is directed. But
  \(L\) is a lower set, so \(L\) must also contain
  \(\beta(b_0 \vee \dots \vee b_{n-1})\). Similarly, \(L\) contains
  \(\beta(c_0\vee\dots\vee c_{m-1})\). Finally, using a similar argument once
  again, we get that \(L\) contains
  \(\beta(b_0\vee\dots\vee b_{n-1} \vee c_0 \vee \dots \vee c_{m-1})\). But
  \(\beta(b)\) is below this element, so \(L\) must contain it, finishing the
  argument that \(K \subseteq L\). Hence, \(K\) is the least upper bound of
  \(I\) and \(J\), completing the proof.
\end{proof}

\begin{theorem}\label{exponentials-with-small-bases-via-retract}%
  \index{basis!for the exponential}%
  The exponential \(E^D\) of dcpos has a specified small basis if \(D\) and
  \(E\) have specified small bases and \(E\) is sup-complete.
\end{theorem}
\begin{proof}
  Suppose that \(\beta_D : B_D \to D\) and \(\beta_E : B_E \to E\) are small
  bases and that \(E\) is sup-complete.  By
  \cref{close-basis-under-finite-joins} we can assume that
  \(\beta_E : B_E \to E\) is closed under finite joins.
  We will write \(D'\) and \(E'\) for the
  respective ideal completions \(\Idl{V}{(B_D,\below)}\) and
  \(\Idl{V}(B_E,\below)\).
  Then~\cref{Idl-iso-algebraic} tells us that we have Scott continuous retracts
  \(\retractalt{D}{D'}{s_D}{r_D}\) and \(\retractalt{E}{E'}{s_E}{r_E}\).
  Composition yields a Scott continuous retract \(\retract{E^D}{{E'}^{D'}}\)
  where \(s(f) \colonequiv s_E \circ f \circ r_D\) and
  \(r(g) \colonequiv r_E \circ g \circ s_D\).
  Now \(D'\) and \(E'\) have small compact basis
  by~\cref{Idl-has-small-compact-basis} and \(E'\) is sup-complete by
  \cref{sup-complete-ideal-completion}. Therefore, \({E'}^{D'}\) has a small
  basis by \cref{exponential-has-small-compact-basis}.
  Finally, \cref{small-basis-closed-under-continuous-retracts} tells us that the
  retraction \(r\) yields a small basis on \(E^D\), as desired.
\end{proof}

Note how, unlike~\cref{exponential-has-small-compact-basis}, the above theorem
does not give an explicit description of the small basis for the exponential. It
may be possible to do so using function graphs, as is done in the classical
setting of effective domain theory in~\cite[Section~4.1]{Smyth1977}.

\section{Notes}\label{sec:continuous-and-algebraic-notes}

This chapter is largely based on our work~\cite{deJongEscardo2021a}. In
particular, \cref{sec:way-below,sec:rounded-ideal-completion} are expanded and
revised versions of parts of that paper.
Some of the ideas for the arguments
in~\cref{sec:structurally-continuous-and-algebraic-bilimits,sec:exponentials-with-small-bases}
were already present in the expanded version of~\cite{deJongEscardo2021a}.
The present treatment of continuous and algebraic dcpos and small (compact)
bases is significantly different from that of~\cite{deJongEscardo2021a}. In the
published work, our definition of continuous dcpo was an amalgamation of
pseudocontinuity and having a small basis, although it did not imply local
smallness.%
\index{continuity!of a dcpo}\index{algebraicity}%
\index{basis}\index{basis!compact}%
\index{continuity!pseudo-}\index{dcpo!locally small}
In this chapter we have disentangled the two notions and based our definition of
continuity on that of~\cite{JohnstoneJoyal1982} without making any reference to
a basis. The current notion of a small basis is simpler and slightly stronger
than that of~\cite{deJongEscardo2021a}, which allows us to prove that having a
small basis is equivalent to being presented by ideals.

This idea of a small basis is similar, but different to \citeauthor{Aczel2006}'s
notion of a ``set-generated'' dcpo~\cite[Section~6.4]{Aczel2006} in the context
of constructive set theory, and a similar smallness criterion in a categorical
context also appears in \cite[Proposition~2.16]{JohnstoneJoyal1982}.%
\index{smallness}\index{set-generated}
While Aczel requires the set \(\{b \in B \mid b \below x\}\) to be directed, we
use the way-below relation (or compact elements) in line with the usual
definition of a basis~\cite[Section~2.2.6]{AbramskyJung1994}.
The particular case of a dcpo with a small compact basis is similar to the
notion of an accessible category~\cite{MakkaiPare1989}.

As mentioned before, our treatment of continuous dcpos is based on the
work~\cite{JohnstoneJoyal1982} of~\citeauthor{JohnstoneJoyal1982}, but we use
the propositional truncation to ensure that the type of continuous dcpos is a
subtype of the type of dcpos. Moreover, our discussion of pseudocontinuity
(\cref{sec:pseudocontinuity}) is new.

Abstract bases were introduced by~\citeauthor{Smyth1977}~\cite{Smyth1977} under
the name ``R-structures'', but our treatment of them and the rounded ideal
completion is closer to that of~\cite[Section~2.2.6]{AbramskyJung1994}, although
ours is based on families and avoids impredicative, set-theoretic constructions.%
\index{R-structure}

Finally, the example of the ideal completion of the dyadics in
\cref{sec:dyadics} and \cref{exponentials-with-small-bases-via-retract} were
suggested to me by Mart\'in Escard\'o.%
\index{dyadics}\index{basis!for the exponential}


%% file: mainmatter/applications.tex
\chapter{Applications in semantics of programming languages}\label{chap:applications}%
\index{semantics}\index{programming language}\index{lambda-calculus@\(\lambda\)-calculus}%
\index{exponential}\index{Scott model!of the untyped \(\lambda\)-calculus}

We present two applications of domain theory in the semantics of programming
languages. The first application is Scott's~\cite{Scott1972} famous
\(D_\infty\): a construction of a nontrivial pointed dcpo \(D\) that is
isomorphic to its self-exponential \(D^D\). This allows one to view functions as
elements (and vice versa) and hence to give a genuine model of the untyped
\(\lambda\)-calculus where self-application is fundamental.
The construction works by instantiating the bilimit machinery
from~\cref{sec:bilimits} with a particular diagram.
It is noteworthy that this construction is possible in our predicative setting
given that Scott's \(D_\infty\) is obtained by iterating exponentials where it
is not obvious that this does not lead to ever-increasing universes.
While the applications can be developed fully using the basic theory of
dcpos (as set out in~\cref{chap:basic-domain-theory}) the exposition in
\cref{chap:continuous-and-algebraic-dcpos} allows us to conclude that Scott's
\(D_\infty\) is algebraic. In fact, it has a small compact basis.

The second application is a constructive and predicative account of the
Scott~\cite{Scott1993} model of the typed programming language
PCF~\cite{Plotkin1977}, including the fundamental soundness and computational
adequacy theorems, formulated and proved originally by
\citeauthor{Plotkin1977}~\cite{Plotkin1977}.
The lifting from \cref{sec:lifting} is the essential ingredient in our
constructive treatment and the model also illustrates the usefulness of the
least fixed point theorem (\cref{least-fixed-point}).
We employ computational adequacy in our investigations into semidecidability and
this is also where the theory on indexed
\(\WW\)-types~(\cref{sec:indexed-W-types}) will find application.

\section{Scott's \texorpdfstring{\(D_\infty\)}{D-infinity} model of the untyped
  \texorpdfstring{\(\lambda\)}{lambda}-calculus}\label{sec:Scott-D-infty}

We construct Scott's \(D_\infty\)~\cite{Scott1972} predicatively. Formulated
precisely, we construct a pointed \(D_\infty : \DCPOnum{0}{1}{1}\) such that
\(D_\infty\) is isomorphic to its self-exponential~\(D_\infty^{D_\infty}\),
employing the machinery from~\cref{sec:bilimits}.

\begin{definition}[\(D_n\)]
  We inductively define pointed dcpos \(D_n : \DCPOnum{0}{1}{1}\) for every
  natural number \(n\) by setting
  \(D_0 \colonequiv \lifting_{\U_0}\pa*{\One_{\U_0}}\) and
  \(D_{n+1} \colonequiv D_n^{D_n}\).
\end{definition}

In light of~\cref{universe-levels-of-lifting-and-exponentials} we highlight the
fact that every \(D_n\) is a \(\U_0\)-dcpo with carrier in \(\U_1\) by the
discussion of universe parameters of exponentials
in~\cref{exponential-universe-parameters}.%
\index{universe!parameters}

\begin{definition}[\(\varepsilon_n\), \(\pi_n\)]
  We inductively define for every natural number \(n\), two Scott continuous maps
  \(\varepsilon_n : D_n \to D_{n+1}\) and \({\pi_n : D_{n+1} \to D_n}\):
  \begin{enumerate}[(i)]
  \item
    \begin{itemize}
    \item \(\varepsilon_0 : D_0 \to D_1\) is given by mapping \(x : D_0\) to the
      continuous function that is constantly~\(x\),
    \item \(\pi_0 : D_1 \to D_0\) is given by evaluating a continuous function
      \(f : D_0 \to D_0\) at~\(\bot\) which is itself continuous
      by~\cref{exponential-universal-property},
    \end{itemize}
  \item
    \begin{itemize}
    \item \(\varepsilon_{n+1} : D_{n+1} \to D_{n+2}\) takes a continuous
      function \(f : D_n \to D_n\) to the continuous composite
      \(D_{n+1} \xrightarrow{\pi_n} D_n \xrightarrow{f} D_n
      \xrightarrow{\varepsilon_n} D_{n+1}\), and
    \item \(\pi_{n+1} : D_{n+2} \to D_{n+1}\) takes a continuous function
      \(f : D_{n+1} \to D_{n+1}\) to the continuous composite
      \(D_n \xrightarrow{\varepsilon_n} D_{n+1} \xrightarrow{f} D_{n+1}
      \xrightarrow{\pi_n} D_n\). \qedhere
    \end{itemize}
  \end{enumerate}
\end{definition}

\begin{lemma}
  The maps \(\pa*{\varepsilon_n,\pi_n}\) form an embedding-projection pair for
  every natural number \(n\).%
  \index{embedding-projection pair}
\end{lemma}
\begin{proof}
  We prove this by induction on \(n\). For \(n \equiv 0\) and arbitrary
  \(x : D_0\), we have
  \[
    \pi_0(\varepsilon_0(x)) \equiv \pi_0(\const_x) \equiv \const_x(\bot) \equiv
    x,
  \]
  so \(\varepsilon_0\) is indeed a section of \(\pi_0\).
  Moreover, for arbitrary \(f : D_1\), we have
  \[
    \varepsilon_0(\pi_0(f)) \equiv \varepsilon_0(f(\bot)) \equiv \const_{f(\bot)},
  \]
  so that for arbitrary \(x : D_0\) we get
  \(\pa*{\varepsilon_0(\pi_0(f))}(x) \equiv f(\bot) \below f(x)\) by
  monotonicity of~\(f\), proving that \(\varepsilon_0 \circ \pi_0\) is
  deflationary.

  Now suppose that the result holds for a natural number \(n\); we prove it for
  \(n+1\). For arbitrary \(f : D_n \to D_n\), we calculate that
  \[
    \pi_{n+1}\pa*{\varepsilon_{n+1}(f)}
    \equiv \pi_n \circ \varepsilon_{n+1}(f) \circ \varepsilon_n
    \equiv \pi_n \circ \varepsilon_n \circ f \circ \pi_n \circ \varepsilon_n
    = f,
  \]
  as \(\varepsilon_n\) is a section of \(\pi_n\) by induction hypothesis.
  The proof that \(\varepsilon_{n+1} \circ \pi_{n+1}\) is a deflation is
  similar.
\end{proof}

In order to apply the tools from \cref{sec:bilimits}, we will need
embedding-projection pairs \(\pa*{\varepsilon_{n,m},\pi_{n,m}}\) from \(D_n\) to
\(D_m\) whenever \(n \leq m\).

\begin{definition}[\(\varepsilon_{n,m},\pi_{n,m}\)]
  We define Scott continuous maps \(\varepsilon_{n,m} : D_n \to D_m\) and
  \({\pi_{n,m} : D_m \to D_n}\) for every two natural numbers \(n \leq m\) as
  follows:
  \begin{enumerate}[(i)]
  \item \(\varepsilon_{n,m}\) and \(\pi_{n,m}\) are both defined to be the
    identity if \(n = m\);
  \item if \(n < m\), then we define \(\varepsilon_{n,m}\) as the composite
    \[
      D_n \xrightarrow{\varepsilon_n} D_{n+1} \to \cdots \to D_{m-1}
      \xrightarrow{\varepsilon_{m-1}} D_m
    \]
    and \(\pi_{n,m}\) as the composite
    \[
      D_m \xrightarrow{\pi_m} D_{m-1} \to \cdots \to D_{n+1}
      \xrightarrow{\pi_{n}} D_n
    \]
  \end{enumerate}
  which yields embedding-projection pairs as they are compositions of them.
\end{definition}

Instantiating the framework of~\cref{sec:bilimits} with the above diagram of
objects \(D_n : \DCPOnum{0}{1}{1}\), we arrive at the construction of
\(D_\infty\) and appropriate embedding-projection pairs. Observe that
\(D_\infty\) is a \(\U_0\)-dcpo with carrier and order taking values
in~\(\U_1\), just like each \(D_n\), as was also mentioned
in~\cref{bilimit-universe-parameters}.%
\index{universe!parameters}

\begin{definition}[\(D_\infty\)]\index{bilimit}%
  \nomenclature[Dinfty]{$D_\infty$}{Scott model of the untyped \(\lambda\)-calculus}%
  Applying \cref{def:D-infty,pi-infty,epsilon-infty} to the above diagram yields
  \(D_\infty : \DCPOnum{0}{1}{1}\) with embedding-projection pairs
  \(\pa*{\varepsilon_{n,\infty},\pi_{n,\infty}}\) from \(D_n\) to \(D_\infty\)
  for every natural number \(n\).
\end{definition}

\begin{lemma}\label{pi-is-strict}
  The function \(\pi_n : D_{n+1} \to D_n\) is strict for every natural number
  \(n\). Hence, so is \(\pi_{n,m}\) whenever \(n \leq m\).
\end{lemma}
\begin{proof}
  Both statements are proved by induction.
\end{proof}

\begin{proposition}
  The dcpo \(D_\infty\) is pointed.
\end{proposition}
\begin{proof}
  Since every \(D_n\) is pointed, we can consider the function
  \(\sigma : \Pi_{n : \Nat} D_n\) given by \(\sigma(n) \colonequiv
  \bot_{D_n}\). Then \(\sigma\) is an element of \(D_{\infty}\) by
  \cref{pi-is-strict} and it is the least, so \(D_\infty\) is indeed pointed.
\end{proof}

We now work towards showing that \(D_\infty\) is isomorphic to the exponential
\(D_\infty^{D_\infty}\). Note that this exponential is again an element of
\(\DCPOnum{0}{1}{1}\) by~\cref{exponential-universe-parameters}, so the universe
parameters do not increase.

\begin{definition}[\(\Phi_n\)]
  For every natural number \(n\), we define the continuous maps
  \begin{align*}
    \Phi_{n+1} : D_{n+1} &\to  D_\infty^{D_{\infty}} \\
    f &\mapsto \pa{D_\infty \xrightarrow{\pi_{n,\infty}} D_n
    \xrightarrow{f} D_n \xrightarrow{\varepsilon_{n,\infty}} D_\infty}
  \end{align*}
  and \(\Phi_0 : D_0 \to D_{\infty}^{D_{\infty}}\) as
  \(\Phi_1 \circ \varepsilon_0\).
\end{definition}

\begin{lemma}\label{Phi-commutes-with-epsilons}
  For every two natural numbers \(n \leq m\), the diagram

  \[
    \begin{tikzcd}
      D_n \ar[dr,"\varepsilon_{n,m}"'] \ar[rr,"\Phi_n"]
      & & D_\infty^{D_\infty} \\
      & D_m \ar[ur,"\Phi_m"']
    \end{tikzcd}
  \]
  commutes.
\end{lemma}
\begin{proof}
  By induction on the difference of the two natural numbers, it suffices to
  prove that for every natural number \(n\), the diagram
  \[
    \begin{tikzcd}
      D_n \ar[dr,"\varepsilon_n"'] \ar[rr,"\Phi_n"]
      & & D_\infty^{D_\infty} \\
      & D_{n+1} \ar[ur,"\Phi_{n+1}"']
    \end{tikzcd}
  \]
  commutes. But this follows from~\cref{epsilon-pi-infty-commutes} and unfolding
  the definition of \(\Phi_n\).
\end{proof}

\begin{definition}[\(\Phi\)]
  The map \(\Phi : D_\infty \to D_\infty^{D_\infty}\) is defined as the unique
  Scott continuous map induced by the \(\Phi_n\) via~\cref{colimit}.
\end{definition}

\begin{lemma}\label{Phi-alt}
  For \(\sigma : D_\infty\) we have %
  \(\Phi(\sigma) = \textstyle\bigsqcup_{n : \Nat}\Phi_{n+1}\pa*{\sigma_{n+1}}\).
\end{lemma}
\begin{proof}
  Recalling the proof of~\cref{colimit} we have
  \(\Phi(\sigma) \equiv \textstyle\bigsqcup_{n : \Nat}\Phi_{n}\pa*{\sigma_{n}}\),
  from which the claim follows easily.
\end{proof}

We now define a map in the other direction using that \(D_\infty\) is also the
limit.

\begin{definition}[\(\Psi_n\)]
  For every natural number \(n\), we define the continuous maps
  \begin{align*}
    \Psi_{n+1} : D_\infty^{D_\infty} &\to D_{n+1}  \\
    f &\mapsto \pa{D_n \xrightarrow{\varepsilon_{n,\infty}} D_\infty
    \xrightarrow{f} D_\infty \xrightarrow{\pi_{n,\infty}} D_n}
  \end{align*}
  and \(\Psi_0 : D_\infty^{D_\infty} \to D_0\) as \(\pi_0 \circ \Psi_1\).
\end{definition}

\begin{lemma}\label{Psi-commutes-with-pis}
  For every two natural numbers \(n \leq m\), the diagram
  \[
    \begin{tikzcd}
      D_\infty^{D_\infty} \ar[dr,"\Psi_m"'] \ar[rr,"\Psi_n"]
      & & D_n \\
      & D_m \ar[ur,"\pi_{n,m}"']
    \end{tikzcd}
  \]
  commutes.
\end{lemma}
\begin{proof}
  Similar to \cref{Phi-commutes-with-epsilons}.
\end{proof}

\begin{definition}[\(\Psi\)]
  The map \(\Psi : D_\infty^{D_\infty} \to D_\infty\) is defined as the unique
  Scott continuous map induced by the \(\Psi_n\) via~\cref{limit}.
\end{definition}

\begin{lemma}\label{Psi-alt}
  For \(f : D_\infty^{D_\infty}\) we have
  \(\Psi(f) = \textstyle\bigsqcup_{n : \Nat}\varepsilon_{n+1,\infty}\pa*{\Psi_{n+1}(f)}\).
\end{lemma}
\begin{proof}
  Notice that
  \begin{align*}
    \Psi(f)
    &= \textstyle\bigsqcup_{n : \Nat}\varepsilon_{n,\infty}\pa*{\pi_{n,\infty}\pa*{\Psi(f)}}
    &&\text{(by \cref{sigma-sup-of-epsilon-pis})} \\
    &= \textstyle\bigsqcup_{n : \Nat}\varepsilon_{n,\infty}\pa*{\Psi_n(f)}
    &&\text{(by \cref{f-infty})},
  \end{align*}
  from which the claim follows easily.
\end{proof}

\begin{theorem}\label{isomorphic-to-self-exponential}%
  \index{dcpo!isomorphism}\index{exponential}%
  The maps \(\Phi\) and \(\Psi\) are inverses and hence, \(D_\infty\) is
  isomorphic to \(D_\infty^{D_\infty}\).
\end{theorem}
\begin{proof}
  For arbitrary \(\sigma : D_\infty\) we calculate that
  \begin{align*}
    \Psi(\Phi(\sigma))
    &= \Psi\pa*{\textstyle\bigsqcup_{n : \Nat}\Phi_{n+1}\pa*{\sigma_{n+1}}}
    &&\text{(by \cref{Phi-alt})} \\
    &= \textstyle\bigsqcup_{n : \Nat}\Psi\pa*{\Phi_{n+1}\pa*{\sigma_{n+1}}}
    &&\text{(by Scott continuity of \(\Psi\))} \\
    &= \textstyle\bigsqcup_{n : \Nat}\textstyle\bigsqcup_{m : \Nat}
      \varepsilon_{m+1,\infty}\pa*{\Psi_{m+1}\pa*{\Phi_{n+1}\pa*{\sigma_{n+1}}}}
    &&\text{(by \cref{Psi-alt})} \\
    &= \textstyle\bigsqcup_{n : \Nat}\varepsilon_{n+1,\infty}\pa*{\Psi_{n+1}\pa*{\Phi_{n+1}\pa*{\sigma_{n+1}}}} \\
    &\equiv \textstyle\bigsqcup_{n : \Nat}\varepsilon_{n+1,\infty}\pa*{\pi_{n,\infty} \circ \varepsilon_{n,\infty} \circ \sigma_{n+1} \circ \pi_{n,\infty} \circ \varepsilon_{n,\infty}}
    &&\text{(by definition)} \\
    &= \textstyle\bigsqcup_{n : \Nat} \varepsilon_{n+1,\infty}\pa*{\sigma_{n+1}}
    &&\text{(since \(\pi_{n,\infty} \circ \varepsilon_{n,\infty} = \id\))} \\
    &= \sigma
    &&\text{(by \cref{sigma-sup-of-epsilon-pis})},
  \end{align*}
  so \(\Phi\) is indeed a section of \(\Psi\).
  Moreover, for arbitrary \(f : D_\infty^{D_\infty}\) we calculate that
  \begin{align*}
    \Phi\pa*{\Psi(f)}
    &= \Phi\pa*{\textstyle\bigsqcup_{n : \Nat}\varepsilon_{n+1,\infty}\pa*{\Psi_{n+1}(f)}}
    &&\text{(by \cref{Psi-alt})} \\
    &= \textstyle\bigsqcup_{n : \Nat}\Phi\pa*{\varepsilon_{n+1,\infty}\pa*{\Psi_{n+1}(f)}}
    &&\text{(by Scott continuity of \(\Phi\))} \\
    &= \textstyle\bigsqcup_{n : \Nat}\textstyle\bigsqcup_{m : \Nat}
      \Phi_{m+1}\pa*{\pi_{m+1,\infty}\pa*{\varepsilon_{n+1,\infty}\pa*{\Psi_{n+1}(f)}}}
    &&\text{(by \cref{Phi-alt})} \\
    &= \textstyle\bigsqcup_{n : \Nat}\Phi_{n+1}\pa*{\pi_{n+1,\infty}\pa*{\varepsilon_{n+1,\infty}\pa*{\Psi_{n+1}(f)}}} \\
    &= \textstyle\bigsqcup_{n : \Nat}\Phi_{n+1}\pa*{\Psi_{n+1}(f)}
    &&\text{(since \(\pi_{n+1,\infty} \circ \varepsilon_{n+1,\infty} = \id\))} \\
    &\equiv \textstyle\bigsqcup_{n : \Nat}\pa*{\varepsilon_{n,\infty} \circ \pi_{n,\infty} \circ f \circ \varepsilon_{n,\infty} \circ \pi_{n,\infty}}
    &&\text{(by definition)} \\
    &= \pa*{\textstyle\bigsqcup_{n : \Nat}\varepsilon_{n,\infty} \circ \pi_{n,\infty}} \circ f \circ \pa*{\textstyle\bigsqcup_{m : \Nat}\varepsilon_{m,\infty} \circ \pi_{m,\infty}} \\
    &= f
    &&\text{(by \cref{epsilon-pi-sup})},
  \end{align*}
  finishing the proof.
\end{proof}

\begin{remark}
  Of course, \cref{isomorphic-to-self-exponential} is only interesting when
  \(D_\infty\) is not the trivial one-element dcpo. Fortunately, \(D_\infty\)
  has (infinitely) many elements besides the least
  element\(\bot_{D_\infty}\). For instance, we can consider
  \(x \colonequiv \eta(\star) : D_0\) and observe that
  \(\varepsilon_{0,\infty}(x)\) is an element of \(D_\infty\) not equal
  to~\(\bot_{D_\infty}\), since \(x \neq \bot_{D_0}\).
\end{remark}

\begin{theorem}
  Scott's \(D_\infty\) has a small compact basis and in particular is
  (structurally) algebraic.%
  \index{basis!compact}\index{algebraicity!structural}\index{algebraicity}%
\end{theorem}
\begin{proof}
  By \cref{lifting-has-small-compact-basis} the \(\U_0\)-dcpo \(D_0\) has a
  small compact basis. Moreover, it is not just a \(\U_0\)-dcpo as it has
  suprema for all (not necessarily directed) families indexed by types in
  \(\U_0\), as \(D_0\) is isomorphic to \(\Omega_{\U_0}\).
  Hence, by induction it follows that
  each \(D_n\) is \(\U_0\)-sup-complete.
  Therefore, by induction and \cref{exponential-has-small-compact-basis} we get
  a small compact basis for each \(D_n\).
  Thus, by \cref{bilimit-has-small-basis}, the bilimit \(D_\infty\) has a small
  basis too.
\end{proof}

\section{Scott's model of the programming language PCF}%
\label{sec:Scott-model-of-PCF}

PCF~\cite{Plotkin1977} is a typed \(\lambda\)-calculus with a base type for
natural numbers and additional constants.%
\index{PCF|seealso {Scott model of PCF}}\index{PCF!types!base}
The full syntax of PCF and its reduction rules (operational semantics) are
described in~\cref{sec:PCF}.
For example, we have numerals \(\underline{n}\) of the base type \(\iota\)
corresponding to natural numbers and basic operations on them, such as a
predecessor term \(\pred\) and a term \(\ifz\) that allows us to perform case
distinction on whether an input is zero or not. The most striking feature of
PCF is its fixed point combinator \(\fix_\sigma\) for every PCF type
\(\sigma\). The idea is that for a term \(t\) of function type
\(\sigma \Rightarrow \sigma\), the term \(\fix_\sigma t\) of type \(\sigma\)
is a fixed point of \(t\). The use of \(\fix\) is that it gives us general
recursion.%
\index{PCF!fixed point}
The operational semantics of PCF is a reduction strategy that allows us to
compute in PCF. We write \(s \smallstep t\) for \(s\) \emph{reduces to}
\(t\). We show a few examples below:
\[
  {\pred \underline{0}} \smallstep \underline{0}; \quad%
  \pred \underline{n + 1} \smallstep \underline{n}; \quad%
  \ifz s\,t\,\underline{0} \smallstep s; \quad%
  \ifz s\,t\,\underline{n + 1} \smallstep t; \quad%
  \fix t \smallstep t (\fix  t).
\]
We see that \(\pred\) indeed acts as a predecessor function and that \(\ifz\)
performs case distinction on whether its third argument is zero or not. The
reduction rule for \(\fix\) reflects that \(\fix t\) is a fixed point of \(t\)
and may be seen as an unfolding (of a recursive definition).

Another way to give meaning to the PCF terms is through denotational semantics,
that is, by giving a model of PCF. A model of PCF assigns to every PCF type
\(\sigma\) some mathematical structure \(\densem{\sigma}\) and to every PCF term
\(t\) of type \(\sigma\) an element \(\densem{t}\) of \(\densem{\sigma}\).%
\index{semantics}
In Scott's model~\cite{Scott1993}, we interpret the PCF types as pointed
dcpos. Specifically, in our constructive and predicative setting, we interpret
the base type \(\iota\) as the lifting of the type of natural numbers and
function types using exponentials.%
\index{Scott model!of PCF}\index{lifting}\index{exponential}
The least element serves as an interpretation of a nonterminating computation,
as is made precise by computational adequacy discussed below.
As mentioned in the introduction to this chapter, the fact that Scott continuous
maps have least fixed points (\cref{least-fixed-point}) will be fundamental in
giving a sound meaning to PCF's fixed point combinator.

Soundness and computational adequacy are important properties that any model of
PCF should have.  Soundness states that if a PCF term \(s\) reduces to \(t\)
(according to the operational semantics), then their interpretations are equal
in the model (symbolically, \(\densem{s} = \densem{t}\)).  Computational
adequacy is completeness at the base type \(\iota\). It says that for every term
\(t\) of type \(\iota\) and every natural number \(n\), if
\(\densem{t} = \densem{\underline {n}}\), then \(t\) reduces to \(\underline{n}\).%
\index{Scott model!of PCF!computational adequacy}\index{Scott model!of PCF!soundness}%
\index{soundness|see {Scott model of PCF}}%
\index{computational adequacy|see {Scott model of PCF}}
The Scott model of PCF was originally proved to be sound and computationally
adequate by \citeauthor{Plotkin1977}~\cite{Plotkin1977} and we prove these
results in our foundational type-theoretic setup too.

Since the base type is interpreted using the lifting, every PCF term of the base
type is interpreted as a partial element of \(\Nat\), and hence gives rise to a
proposition (the domain of the partial element).
Motivated by constructive issues involving countable choice
(see~\cref{sec:semidec-pcf}), we use computational adequacy in a syntactic
approach to establishing that all such propositions are semidecidable.

\subsection{PCF and its operational semantics}\label{sec:PCF}
We precisely define the types and terms of PCF as well as the small-step
operational semantics.%
\index{PCF}
Instead of the formulation by \citeauthor{Plotkin1977}~\cite{Plotkin1977}, which
features variables and \(\lambda\)-abstraction, we revert to the original,
combinatory, formulation of the terms of LCF by
\citeauthor{Scott1993}~\cite{Scott1993} in order to simplify the technical
development.%
\index{combinator}
We note that it is possible to represent every closed \(\lambda\)-term in terms
of combinators by a well-known technique~\cite[Section~2C]{HindleySeldin2008}.

\begin{definition}[PCF types, \(\iota\), \(\sigma \Rightarrow \tau\)]%
  \index{PCF!types}\index{PCF!types!base|textbf}\index{PCF!types!function}%
  The \emph{PCF types} are inductively defined as
  \begin{enumerate}[(i)]
  \item $\iota$ is a type, the \emph{base type}, and
    \nomenclature[iota]{$\iota$}{base type of natural numbers in PCF}
  \item for every two types $\sigma$ and $\tau$, there is a \emph{function type}
    $\sigma \Rightarrow \tau$.
    \nomenclature[arrowzz]{$\sigma \Rightarrow \tau$}{function type in PCF}
  \end{enumerate}
  As usual, $\Rightarrow$ will be right associative, so we write
  $\sigma \Rightarrow \tau \Rightarrow \rho$ for
  $\sigma \Rightarrow (\tau \Rightarrow \rho)$.
\end{definition}

\begin{definition}[PCF terms, \(\zeroo\), \(\succc\), \(\pred\), \(\ifz\), %
    \(\PCFk_{\sigma,\tau}\), \(\PCFs_{\sigma,\tau,\rho}\), \(\fix_\sigma\)]%
  \label{def:PCFterms}%
  \index{PCF!terms}
  The \emph{PCF terms of PCF type $\sigma$} are inductively generated by
  \begin{center}
    \AxiomC{\phantom{$t : \iota$}}
    \UnaryInfC{${\zeroo} \text{ of type } \iota$}
    \DisplayProof \quad
    \AxiomC{\phantom{$t : \iota$}}
    \UnaryInfC{${\succc} \text{ of type } \iota \Rightarrow \iota$}
    \DisplayProof\vspace{0.5cm}\\
    \AxiomC{\phantom{$t : \iota$}}
    \UnaryInfC{$\pred \text{ of type } \iota \Rightarrow \iota$}
    \DisplayProof \quad
    \AxiomC{\phantom{$t : \iota$}}
    \UnaryInfC{${\ifz} \text{ of type } {\iota \Rightarrow \iota \Rightarrow \iota
        \Rightarrow \iota}$}
    \DisplayProof \vspace{0.5cm}\\
    \AxiomC{\phantom{$t : \iota$}}
    \UnaryInfC{${\PCFk_{\sigma,\tau}} \text{ of type } \sigma \Rightarrow \tau \Rightarrow
      \sigma$}
    \DisplayProof \quad \AxiomC{\phantom{$t : \iota$}}
    \UnaryInfC{${\PCFs_{\sigma,\tau,\rho}} \text{ of type } (\sigma \Rightarrow \tau
      \Rightarrow \rho) \Rightarrow \break (\sigma \Rightarrow \tau) \Rightarrow
      \sigma \Rightarrow \rho$}
    \DisplayProof\vspace{0.5cm}\\
    \AxiomC{\phantom{$t : \iota$}}
    \UnaryInfC{${\fix_\sigma} \text{ of type } (\sigma \Rightarrow \sigma)
      \Rightarrow \sigma$}
    \DisplayProof \quad
    \AxiomC{$s \text{ of type } \sigma \Rightarrow \tau$}
    \AxiomC{$t \text{ of type } \sigma$}
    \BinaryInfC{$(st) \text{ of type } \tau$}
    \DisplayProof
  \end{center}
  \nomenclature[zero]{$\zeroo$}{constant for the natural number 0 in PCF}%
  \nomenclature[succ]{$\succc$}{successor function on the natural numbers in PCF}%
  \nomenclature[pred]{$\pred$}{predecessor function on the natural numbers in PCF}%
  \nomenclature[ifz]{$\ifz$}{conditional in PCF}%
  \nomenclature[k]{$\PCFk$}{k-combinator in PCF}%
  \nomenclature[s]{$\PCFs$}{s-combinator in PCF}%
  \nomenclature[fix]{$\fix$}{fixed point combinator in PCF}%
  We will often drop the parentheses in the final clause, as well as the PCF
  type subscripts in \({\PCFk_{\sigma,\tau}}\), \({\PCFs_{\sigma,\tau,\rho}}\)
  and \({\fix_\sigma}\). Finally, we employ the convention that the parentheses
  associate to the left, i.e.\ we write $rst$~for~$(rs)t$.
\end{definition}

\begin{definition}[PCF numerals, \(\underline{n}\)]\label{def:PCFnumerals}%
  \index{PCF!numeral|textbf}%
  For any $n : \Nat$, let us write $\underline n$ for the
  $n$\textsuperscript{th} PCF numeral, defined inductively by
  \(\underline 0 \colonequiv {\zeroo}\) and
  \(\underline {n+1} \colonequiv \succc\,\underline n\).
  \nomenclature[underline]{$\underline{n}$}{numeral in PCF}
\end{definition}

To define the small-step operational semantics of PCF, we first define an
auxiliary type family by induction and then we propositional truncate its values
to obtain the small-step relation.

\begin{definition}[Small-step (pre-)relation, \({\smallsteppre}\), \(\smallstep\)]%
  \label{def:small-step}%
  \index{PCF!operational semantics}\index{PCF!small-step relation}%
  Define the \emph{small-step pre-relation} $\smallsteppre$ of type
  \[
    \Pi_{\sigma : \text{PCF types}}\pa*{\text{PCF terms of type $\sigma$} \to
    \text{PCF terms of type $\sigma$} \to \mathcal{U}_0}
  \]
  \nomenclature[triangle']{$s \smallsteppre t$}{small-step pre-relation of PCF}
  as the inductive family generated by
  \begin{center}
    \AxiomC{\phantom{$f \smallsteppre g$}}
    \UnaryInfC{$\pred\underline 0\smallsteppre\underline 0$}
    \DisplayProof \quad \AxiomC{\phantom{$f \smallsteppre g$}}
    \UnaryInfC{$\pred\underline {n+1}\smallsteppre\underline n$}
    \DisplayProof \quad \AxiomC{\phantom{$f \smallsteppre g$}}
    \UnaryInfC{$\ifz s\,t\,\underline 0\smallsteppre s$} \DisplayProof
    \quad \AxiomC{\phantom{$f \smallsteppre g$}}
    \UnaryInfC{$\ifz s\,t\,\underline {n+1}\smallsteppre t$}
    \DisplayProof \vspace{0.5cm}
    \\
    \AxiomC{\phantom{$f \smallsteppre g$}}
    \UnaryInfC{${\PCFk}\,s\,t\smallsteppre s$} \DisplayProof \quad
    \AxiomC{\phantom{$f \smallsteppre g$}}
    \UnaryInfC{${\PCFs}f\,g\,t\smallsteppre f\,t(g\,t)$} \DisplayProof \quad
    \AxiomC{\phantom{$f\smallsteppre g$}}
    \UnaryInfC{$\fix f\smallsteppre f(\fix f)$} \DisplayProof
    \quad \AxiomC{$f\smallsteppre g$} \UnaryInfC{$f\,t\smallsteppre g\,t$}
    \DisplayProof \vspace{0.5cm}
    \\
    \AxiomC{$s \smallsteppre t$}
    \UnaryInfC{$\succc s \smallsteppre \succc t$} \DisplayProof
    \quad \AxiomC{$s \smallsteppre t$}
    \UnaryInfC{$\pred s \smallsteppre \pred t$} \DisplayProof
    \quad \AxiomC{$r \smallsteppre r'$}
    \UnaryInfC{$\ifz s\,t\,r \smallsteppre \ifz s\,t\,r'$}
    \DisplayProof
  \end{center}

  We have been unable to prove that $s \smallsteppre t$ is a proposition for all
  PCF terms $s$~and~$t$ of the same type. The difficulty is that one cannot
  perform induction on \emph{both} \(s\) and \(t\), because the reduction
  relation is not defined by induction on terms. However, conceptually,
  $s \smallsteppre t$ should be a proposition, as (by inspection of the
  definition), there is at most one way by which we obtained
  $s \smallsteppre t$. Moreover, for technical reasons that will become apparent
  later, we need $\smallsteppre$ to be proposition-valued.

  We solve the problem by defining the \emph{small-step relation} $\smallstep$ as
  the propositional truncation of ${\smallsteppre}$, i.e.\
  $s \smallstep t \colonequiv \squash{s \smallsteppre t}$.%
  \index{propositional truncation}%
  \nomenclature[triangle]{$s \smallstep t$}{small-step relation of PCF}
\end{definition}

\begin{remark}
  Benedikt Ahrens pointed out that in an impredicative framework, one could use
  propositional resizing and an impredicative encoding, i.e.\ by defining
  $\smallstep$ as a $\Pi$-type of all suitable proposition-valued
  relations. This is similar to the situation in set theory, where one would
  define $\smallstep$ as an intersection. Specifically, say that a relation
  \[
    R : \Pi_{\sigma : \text{PCF types}} \pa*{\text{PCF terms of type
        \(\sigma\)} \to \text{PCF terms of type \(\sigma\)} \to \Omega_{\mathcal
        U_0}}
  \] is \emph{suitable} if it closed under all the clauses of
  \cref{def:small-step}, that is, we would want to have elements of
  \[
    R\pa*{\iota,\pred \underline 0,\underline 0},
    R\pa*{\iota,\pred \underline{n+1},\underline n},
    R\pa*{\iota,\ifz\,s\,t\,\underline 0,s},
    \text{etc.}
  \]
  We could then define
  \(s \smallstep_{\text{impred}} t \colonequiv \Pi_{R \text{ suitable}}
  R(\sigma,s,t)\). But notice the increase in universes:
  \[
    \smallstep_{\text{impred}} : \Pi_{\sigma : \text{PCF types}}
    \pa*{\text{PCF terms of type \(\sigma\)} \to \text{PCF terms of type
        \(\sigma\)} \to \Omega_{\mathcal U_1}}.
  \]
  So because of this increase, \(\smallstep_{\text{impred}}\) itself is not one
  of the suitable relations.
  Moreover, the definition \({\smallstep_{\text{impred}}}\) only quantifies over
  \(\U_0\)-valued relations, and so has no bearing on relations valued in other
  universes.
  Therefore \(\smallstep_{\text{impred}}\) does not
  satisfy the appropriate universal property of being the least relation closed
  under the clauses in \cref{def:small-step}.
  Notice that this is exactly the same phenomenon as discussed in
  \cref{sec:prop-trunc-resizing} for the Voevodsky propositional truncation.
  With propositional resizing we could resize the relations to all have values
  in \(\U_0\) and obtain the appropriate universal property.
  The advantage of using the propositional truncation above is that it does
  satisfy the right universal property even without propositional resizing.
\end{remark}

\begin{samepage} 
  Let $R : X \to X \to \Omega$ be a relation on a type $X$. We might try to
  define the reflexive transitive closure $R_\ast$ of $R$ as an inductive
  family, generated by three constructors:
  \nopagebreak%
\begin{align*}
  \operatorname{extend} &: \Pi_{x,y : X}\pa*{x \relR y \to x \relR_\ast y}; \\
  \operatorname{refl}   &: \Pi_{x : X}\pa*{x \relR_\ast x}; \\
  \operatorname{trans} &: \Pi_{x,y,z : X}\pa*{x \relR_\ast y \to y \relR_\ast z \to x \relR_\ast z}.
\end{align*}
\end{samepage}
But $R_\ast$ is not necessarily proposition-valued, even though $R$ is. This is
because we might add a pair $(x,y)$ to $R_\ast$ in more than one way, for
example, once by an instance of $\operatorname{extend}$ and once by an instance of
$\operatorname{trans}$. Thus, we are led to the following definition.

\begin{definition}\label{def:refltransclos}%
  \index{reflexive transitive closure|textbf}%
  Let $R : X \to X \to \Omega$ be a relation on a type $X$. We define the
  \emph{reflexive transitive closure} $R^\ast$ of $R$ by
  $x \relR^\ast y \colonequiv \squash*{x \relR_\ast y}$, where $R_\ast$ is as above.
\end{definition}
It is not hard to show that $R^\ast$ is the least reflexive and transitive
proposition-valued relation that extends $R$, so $R^\ast$ satisfies the
appropriate universal property.

Some properties of $\smallstep$ reflect onto $\smallstepclos$ as the following
lemma shows.
\nomenclature[triangleast]{$s \smallstepclos t$}{reflexive transitive closure
  of the small-step relation of PCF}
\begin{lemma}\label{refltransreduce}
  Let $r', r, s$ and $t$ be PCF terms of type $\iota$. If
  $r' \smallstepclos r$, then
  \begin{enumerate}[(i)]
  \item\label{refltransreduce-1} $\succc r' \smallstepclos \succc r$,
  \item $\pred r' \smallstepclos \pred r$, and
  \item $\ifz s\,t\,r' \smallstepclos \ifz s\,t\,r$.
  \end{enumerate}
  Moreover, if $f$ and $g$ are PCF terms of type $\sigma\Rightarrow\tau$ and
  $f \smallstepclos g$, then $f t \smallstepclos g t$ for any PCF term $t$ of
  type $\sigma$.
\end{lemma}
\begin{proof}
  We only prove~\ref{refltransreduce-1} the rest is similar. Suppose
  $r' \smallstepclos r'$.  Since $\succc r' \smallstepclos \succc r$ is a
  proposition, we may assume that we actually have an element $p$ of type
  $r' \smallstep_\ast r'$. Now we can perform induction on $p$. The cases were
  $p$ is formed using $\operatorname{refl}$ or $\operatorname{trans}$ are easy. If $p$ is
  formed by $\operatorname{extend}$, then we get an element of type
  ${r \smallstep r'} \equiv {r \smallsteppre r'}$. Again, as we are proving a
  proposition, we may suppose the existence of an element of type
  $r \smallsteppre r'$. By \cref{def:small-step}, we then get
  $\succc r \smallsteppre \succc r'$. This in turn yields,
  $\succc r'\smallstep\succc r$ and finally we use $\operatorname{extend}$ to get the
  desired $\succc r'\smallstepclos\succc r$.
\end{proof}

\subsection{The Scott model of PCF}

We proceed by defining our constructive version of the Scott model of PCF using
pointed (exponentials of) \(\U_0\)-dcpos.
Recall from~\cref{flat-gives-LPO} that interpreting the base type by naively
adding a least element to the type of natural numbers is constructively
inadequate, which is why we use the lifting (as defined in~\cref{sec:lifting})
of the type of natural numbers.

\begin{definition}[Interpretation of PCF types, \(\densem{\sigma}\)]%
  \index{Scott model!of PCF|(textbf}%
  We inductively define a map
  \[
    \densem{-} : \text{PCF types} \to \DCPOnum{0}{1}{1}
  \]
  \nomenclature[brackets]{$\densem{\sigma}$}{interpretation of a PCF type as a
    pointed dcpo}
  interpreting a PCF type as a pointed \(\U_0\)-dcpo as follows:
  \begin{enumerate}[(i)]
  \item $\densem{\iota} \colonequiv \lifting(\Nat)$;
  \item
    $\densem{\sigma \Rightarrow \tau} \colonequiv
    \densem{\tau}^{\densem{\sigma}}$. \qedhere
  \end{enumerate}
\end{definition}

We recall~\cref{exponential-universe-parameters} on universe parameters and we
remark that is fortunate that the interpretation function \(\densem{-}\) takes
values in \(\DCPOnum{0}{1}{1}\) and does not require ever-increasing universe
parameters.

Next, we interpret PCF terms as elements of these pointed dcpos, for which we
will need that \(\lifting\) is a monad (with unit \(\eta\)) and (in particular)
a functor (recall \cref{lifting-is-monad,def:lifting-functor}).

\begin{definition}[Interpretation of PCF terms, \(\densem{t}\)]%
  Define for each PCF term $t$ of PCF type $\sigma$ a term $\densem{t}$ of type
  $\densem{\sigma}$, by the following inductive clauses:
  \nomenclature[brackets]{$\densem{t}$}{interpretation of a PCF term as an
    element of a pointed dcpo}
  \begin{enumerate}[(i)]
  \item $\densem{\zeroo} \colonequiv \eta(0)$;
  \item $\densem{\succc} \colonequiv \lifting(s)$, where
    $s : \Nat \to \Nat$ is the successor function;
  \item $\densem{\pred} \colonequiv \lifting(p)$, where
    $p : \Nat \to \Nat$ is the predecessor function with the convention that 0 is mapped to 0;
  \item
    $\densem{\ifz} :
      \densem{\iota\Rightarrow\iota\Rightarrow\iota\Rightarrow\iota}$ is defined
      using the Kleisli extension (\cref{lifting-is-monad}) as:
      $\lambda x,y.\pa*{\chi_{x,y}}^{\#}$, where
    \[
      \chi_{x,y}(n) \colonequiv
      \begin{cases}
        x &\text{if } n = 0; \\
        y &\text{else};
      \end{cases}
    \]
  \item $\densem{\PCFk} \colonequiv \lambda x, y . x$;
  \item $\densem{\PCFs} \colonequiv \lambda f, g, x . (f (x)) (g (x))$;
  \item $\densem{\fix} \colonequiv \mu$, where $\mu$ is the least fixed
    point operator from \cref{least-fixed-point};%
    \index{least fixed point}%
  \item \(\densem{st} \colonequiv \densem{s}\pa*{\densem{t}}\) for \(s\) of
    type \(\sigma \Rightarrow \tau\) and \(t\) of type \(\sigma\).
    \qedhere
  \end{enumerate}
\end{definition}
\index{Scott model!of PCF|)textbf}%

\begin{remark}
  Of course, there are some things to be proved here. Namely,
  \(\densem{\succc}\), \(\densem{\pred}\), \dots, \(\densem{\fix}\),
  \(\densem{s\,t}\) all need to be Scott continuous. In the case of
  $\densem{\succc}$ and $\densem{\pred}$, we simply appeal to
  \cref{lifting-extension-is-continuous,def:lifting-functor}.
  For $\densem{\fix}$, this is guaranteed by \cref{least-fixed-point}.
  The continuity of \(\densem{\PCFk}\), \(\densem{\PCFs}\) and \(\densem{\ifz}\)
  can be verified directly; the details are omitted here.
  Finally, the interpretation of application is continuous
  by~\cref{exponential-universal-property}.
  \end{remark}

As a first result about our denotational semantics, we show that the PCF
numerals have a canonical interpretation in the denotational semantics.
This basic result is fundamental and finds application in the proof of
soundness.

\begin{proposition}\label{densemnumeral}
  For every natural number $n$, we have $\densem{\underline n} = \eta (n)$.
\end{proposition}
\begin{proof}
  We proceed by induction on $n$. The $n\equiv 0$ case is by definition of
  $\densem{\underline 0}$. Suppose $\densem{\underline m} = \eta (m)$ for a
  natural number $m$. Then,
  \begin{align*}
    \densem{\underline {m+1}}
    &\equiv \densem{\succc}(\densem{\underline m}) \\
    &= \lifting(s)(\eta (m))&&\text{(by induction hypothesis)} \\
    &= \eta(m + 1) &&\text{(by definition of the lift functor)},
  \end{align*}
  as desired.
\end{proof}

\subsection{Soundness and computational adequacy}%
\label{sec:soundnesscompadequacy}

Having defined the Scott model of PCF we show that it respects the operational
semantics by proving soundness and computational adequacy.

\begin{theorem}[Soundness]%
  \index{Scott model!of PCF!soundness|textbf}
  If $s \smallstepclos t$, then $\densem{s} = \densem{t}$ for every two PCF
  terms \(s\) and \(t\) (necessarily of the same type).
\end{theorem}
\begin{proof}
  Since the carriers of dcpos are sets, the type \(\densem{s} = \densem{t}\) is
  a proposition. Therefore, we can use induction on the derivation of
  \(s \smallstepclos t\). We use the Kleisli monad laws in proving some of the
  cases. For example, one step is to prove that
  \[
    \densem{\ifz s\,t\,\underline{n+1}} = \densem{t}.
  \]
  This may be proved by the following chain of equalities:
  \begin{align*}
    \densem{\ifz s\,t\,\underline{n+1}}
    &\equiv
      \densem{\ifz s\,t}(\densem{\underline{n+1}}) \\
    &= \densem{\ifz s\,t}(\eta(n+1))
    &&\text{(by \cref{densemnumeral})} \\
    &\equiv \pa*{\chi_{\densem{s},\densem{t}}}^\#(\eta(n+1))
    &&\text{(by definition of $\densem{\ifz}$)} \\
    &= \chi_{\densem{s},\densem{t}}(n+1)
    &&\text{(by \cref{lifting-is-monad})} \\
    &= \densem{t}.
  \end{align*}
  The other cases are proved similarly.
\end{proof}

Ideally, we would like a converse to soundness. However, this is not possible,
as for example,
${\densem{\PCFk{\zeroo}}} = {\densem{\PCFk (\succc (\pred \zeroo))}}$, but
neither ${\PCFk \zeroo} \smallstepclos {\PCFk (\succc (\pred \zeroo))}$ nor
${\PCFk (\succc (\pred \zeroo))} \smallstepclos {\PCFk \zeroo}$ holds. We do,
however, have the following.

\begin{theorem}[Computational adequacy]\label{Adequacy}%
  \index{Scott model!of PCF!computational adequacy|textbf}%
  \index{PCF!types!base}\index{PCF!numeral}%
  For a PCF term \(t\) of the base type, if the partial element \(\densem{t}\)
  is defined, then \(t\) reduces to the numeral given by the value of
  \(\densem{t}\).
\end{theorem}

Equivalently, for every $n : \Nat$, it holds that
$\densem{t} = \densem{\underline{n}}$ implies $t \smallstepclos {\underline n}$.
Another useful rephrasing is: for every PCF term \(t\) of the base type, we have
\(t \smallstepclos \underline{\liftvalue(\densem{t},p)}\) for every
\(p : \isdefined(\densem{t})\).

We do not prove computational adequacy directly, as, unlike soundness, it does
not allow for a straightforward proof by induction on terms. Instead, we use the
standard technique of logical relations~\cite[Chapter 7]{Streicher2006} which
goes back to~\cite{Tait1967} and obtain the result as a direct corollary of
\cref{mainlemma}.

\begin{definition}[Logical relation, \(R_\sigma\)]\label{def:R}%
  \index{logical relation}%
  For every PCF type $\sigma$, define a relation
  \[
    \logicalR{\sigma} : \textup{PCF terms of type
      $\sigma$} \to \densem{\sigma} \to \Omega_{\U_0}
  \]
  \nomenclature[R_sigma]{$s \logicalR{\sigma} t$}{logical relation used to prove
    computational adequacy}
  by induction on $\sigma$ as
  \begin{enumerate}[(i)]
  \item\label{R-iota}
    $t \logicalR{\iota} d \colonequiv \Pi_{p : \isdefined (d)}\pa*{t \smallstepclos
      \underline {\liftvalue(d,p)}}$, and
  \item\label{R-arrow}
    $s \logicalR{\tau \Rightarrow \rho} f \colonequiv \Pi_{t : \text{PCF terms
        of type $\tau$}} \Pi_{d : \densem{\tau}} \pa*{t \logicalR{\tau} d \to st
      \logicalR{\rho} f(d)}$.
  \end{enumerate}
  We sometimes omit the type subscript $\sigma$ in $\logicalR{\sigma}$.
\end{definition}

Note that~\ref{R-iota} in \cref{def:R} is the statement of computational
adequacy. By generalising, we can prove properties of \(R\) by induction on
types.

\begin{lemma}\label{logicalrelationsmall-step}
  If \(s \smallstepclos t\) and \(t \logicalR{\sigma} d\), then
  \(s \logicalR{\sigma} d\), for all PCF types \(\sigma\) and PCF terms
  \(s\)~and~\(t\) of type \(\sigma\) and elements \(d : \densem{\sigma}\).
\end{lemma}
\begin{proof}
  By induction on $\sigma$, making use of the last part of
  \cref{refltransreduce}.
\end{proof}

\begin{lemma}\label{logicalrelationconstants}
  For $t$ equal to \(\zeroo\), \({\succc}\), \({\pred}\), \({\ifz}\),
    \({\PCFk}\) or \({\PCFs}\), we have $t \logicR \densem{t}$.
\end{lemma}
\begin{proof}
  By the previous lemma and \cref{refltransreduce}.
\end{proof}

Next, we wish to extend the previous lemma to the case where
$t \equiv {\fix_\sigma}$ for any PCF type $\sigma$. This is slightly more
complicated and we need two intermediate lemmas.

\begin{lemma}\label{logicalrelationbottom}
  For every PCF type \(\sigma\) and term \(t\) of type \(\sigma\) it holds that
  \(t \logicalR{\sigma} \bot\).
\end{lemma}
\begin{proof}
  By induction on $\sigma$: for the base type, this holds vacuously; for
  function types, it follows by induction hypothesis and the pointwise ordering.
\end{proof}

\begin{lemma}
  The logical relation is closed under directed suprema. That is, for every PCF
  type \(\sigma\) and every PCF term $t$ of type $\sigma$ and every directed
  family $d : I \to \densem{\sigma}$, if $t \logicalR{\sigma} d_i$ for every
  $i : I$, then $t \logicalR{\sigma} \bigsqcup_{i : I}d_i$.
\end{lemma}
\begin{proof}
  This proof is somewhat different from the classical proof, so we spell out the
  details. We prove the lemma by induction on $\sigma$.
  The case when $\sigma$ is a function type is easy, because least upper bounds
  are calculated pointwise and so it reduces to an application of the induction
  hypothesis. We concentrate on the case when $\sigma \equiv \iota$ instead.
  Recall that $\bigsqcup_{i : I} d_i$ is given by
  $\pa*{{\exists_{i : I} \isdefined(d_i)}, \phi}$, where $\phi$ is the
  factorisation of
  \[
    \pa*{\Sigma_{i : I} \isdefined(d_i)}
    \xrightarrow{(i,p) \mapsto \liftvalue\pa*{d_i,p}} \lifting(\Nat)
  \]
  through ${\exists_{i : I} \isdefined(d_i)}$.
  We are tasked with proving that $t \smallstepclos \underline{\phi (p)}$
  whenever \(p\) witnesses that \(\bigsqcup_{i : I} d_i\) is defined. Since we
  are trying to prove a proposition (as $\smallstepclos$~is proposition-valued),
  we may assume that we have $(j,p) : \Sigma_{i : I}\isdefined(d_i)$.  By
  definition of $\phi$ we have:
  $\phi\pa*{\tosquash{(j,p)}} = \liftvalue\pa*{d_{j},p}$ and by assumption we
  know that $t \smallstepclos \underline{\liftvalue\pa*{d_{j},p}}$, so we are
  done.
\end{proof}

\begin{lemma}
  For every PCF type $\sigma$, we have
  $\fix_\sigma \logicalR{(\sigma \Rightarrow \sigma)\Rightarrow \sigma}
  \densem{\fix_\sigma}$.
\end{lemma}
\begin{proof}
  Suppose that \(t \logicalR{\sigma\Rightarrow\sigma} f\); we are to prove that
  $\fix t \logicalR{\sigma} \mu(f)$.
  By the previous lemma, it suffices to prove that
  $\fix t \logicalR{\sigma} f^n(\bot)$ for every natural number $n$, which we do
  by induction on $n$.
  The base case is an application of \cref{logicalrelationbottom}.
  Now suppose that $\fix t \logicalR{\sigma} f^m(\bot)$. Then, using
  $t \logicalR{\sigma\Rightarrow\sigma} f$, we find
  $t(\fix t) \logicalR{\sigma} f(f^m(\bot))$. Hence, by
  \cref{logicalrelationsmall-step}, we obtain the
  $\fix t \logicalR{\sigma} f^{m+1}(\bot)$, completing our inductive proof.
\end{proof}

\begin{lemma}[Fundamental Theorem]\label{mainlemma}%
  \index{fundamental theorem (of the logical relation)}
  We have $t \logicR \densem{t}$ for every PCF term \(t\).
\end{lemma}
\begin{proof}
  The proof is by induction on $t$. The base cases are taken care of by
  \cref{logicalrelationconstants} and the previous lemma. For the inductive
  step, suppose $t$ is a PCF term of type $\sigma\Rightarrow\tau$. By induction
  hypothesis, $ts \logicalR\tau \densem{ts}$ for every PCF term $s$ of type
  $\sigma$, but $\densem{ts} \equiv \densem{t}\densem{s}$, so we are done.
\end{proof}

Computational adequacy is now a direct corollary of \cref{mainlemma}.
\begin{proof}[Proof of computational adequacy]
  Take $\sigma$ to be the base type $\iota$ in \cref{mainlemma}.
\end{proof}

\paragraph*{Using computational adequacy to compute.} An interesting use of
computational adequacy is that it allows one to argue semantically to obtain
results about termination (i.e.\ reduction to a numeral) in PCF. Classically,
every PCF program of type $\iota$ either terminates or it does not.  From a
constructive point of view, we wait for a program to terminate, with no a priori
knowledge of termination. The waiting could be indefinite. Less naively, we
could limit the number of computation steps to avoid indefinite waiting, with an
obvious shortcoming: how many steps are enough? Instead, one could use
computational adequacy to compute as we describe now.%
\index{constructivity}\index{Scott model!of PCF!computational adequacy}

For a PCF type \(\sigma\), a \emph{functional of type $\sigma$} is an element of
$\densem{\sigma}$. By induction on PCF types, we define when a functional is
said to be \emph{total}:%
\index{PCF!functional}\index{PCF!functional!total}%
\begin{enumerate}[(i)]
\item a functional $i$ of type $\iota$ is total if $i = \densem{\underline n}$
  for some natural number $n$;
\item a functional $f$ of type $\sigma\Rightarrow\tau$ is total if it maps total
  functionals to total functionals, viz.\ $f(d)$ is a total functional of type
  $\tau$ for every total functional $d$ of type $\sigma$.
\end{enumerate}
Now, let $s$ be a PCF term of type
$\sigma_1 \Rightarrow \sigma_2 \Rightarrow \dots \Rightarrow \sigma_n
\Rightarrow \iota$. If we can prove that $\densem{s}$ is total, then
computational adequacy lets us conclude that for all total inputs
$\densem{t_1} : \densem{\sigma_1},\dots,\densem{t_n} : \densem{\sigma_n}$, the
term $s(t_1,\dots,t_n)$ reduces to the numeral representing
$\densem{s}(\densem{t_1},\dots,\densem{t_n})$.
Thus, instead of e.g.\ giving a number of steps as a timeout for the
computation, we supply a proof of totality to computational adequacy and we are
guaranteed to obtain a result.

Of course, this approach still requires us to prove that \(\densem{s}\) is
total, which may be challenging. But note that we can use domain-theoretic
arguments to prove this about the denotation~\(\densem{s}\), whereas in a direct
proof of termination we would only have the operational semantics available for
our argument.

\subsection{Semidecidability and PCF terms of the base type}\label{sec:semidec-pcf}%
\index{semidecidability}\index{PCF!types!base}\index{PCF!numeral}

Given a PCF term \(t\) of the base type, we intuitively expect it to be
semidecidable whether \(t\) will compute to a numeral, as we can reduce \(t\)
one step at a time following the operational semantics of PCF and stop when we
have obtained a numeral.\index{PCF!operational semantics}

A rather slick way of proving this would be to argue that we could have
restricted to semidecidable propositions in the lifting of the natural numbers,
i.e.\ we modify the Scott model of PCF and set
\(\densem{\iota} \colonequiv \liftingsd(\Nat)\) with%
\index{Scott model!of PCF}\index{lifting}%
\[
  \liftingsd(X) \colonequiv \Sigma_{P : \Omega_{\U_0}}\pa*{P \text{ is
      semidecidable}}\times\pa*{P \to X}.
\]
\nomenclature[L_sd]{$\liftingsd(X)$}{lifting of a type \(X\) with respect to the
  semidecidable propositions}
This restricted lifting of a set does not necessarily yield a dcpo, but observe
that only used suprema of \(\omega\)-chains in the Scott model of PCF.%
\index{omega-completeness@\(\omega\)-completeness}
Thus it would suffice for \(\liftingsd(X)\) to be a \(\omega\)-cpo. Escard\'o
and Knapp~\cite[Corollary~5]{EscardoKnapp2017} observed that this is indeed the
case \emph{if} countable choice is assumed. And in this case it coincides with
the other
constructions~\cite{ChapmanUustaluVeltri2019,AltenkirchDanielssonKraus2017} of
\(\omega\)-cpos, as discussed in the~\nameref{sec:basic-domain-theory-notes}
of~\cref{chap:basic-domain-theory}.%
\index{choice!axiom of countable}
Countable choice is used to obtain witnessing sequences \(\alpha_n\) for a given
\(\Nat\)-indexed sequence \(P_n\) of semidecidable propositions.
Indeed, at least some weak form of countable choice is indeed necessary for
\(\liftingsd(X)\) to be closed under countable joins as shown
by~\cite{deJong2022}.
But since countable choice is not provable in constructive univalent foundations
(cf.~the~\nameref{sec:basic-domain-theory-notes}
of~\cref{chap:basic-domain-theory}), we resort to a different, syntactic
approach.
That~is, without using countable choice, we use soundness and computational
adequacy to prove that \(\isdefined(\densem{t})\) is semidecidable for every PCF
term of the base type.%
\index{Scott model!of PCF!soundness}%
\index{Scott model!of PCF!computational adequacy}
In other words, although we cannot define the Scott model using \(\liftingsd\),
the map
\[
  \densem{-} : \text{PCF terms of the base type} \to \lifting(\Nat)
\]
factors through \(\liftingsd(\Nat)\), even in the absence of countable choice.

Specifically, we prove semidecidability of the proposition
\(\isdefined(\densem{t})\) by appealing to \cref{semidecidable-criterion} and
showing that \(\isdefined(\densem{t})\) is logically equivalent to
$\exists_{n : \Nat} \exists_{k : \Nat} \,t \smallstep^k {\underline n}$, where
$t \smallstep^k {\underline n}$ says that $t$ reduces to $\underline n$ in at
most $k$ steps.%
\nomenclature[trianglek]{$s \smallstep^k t$}{\(k\)-step reflexive transitive
  closure of the small-step relation of PCF}
We prove this notion to be decidable by using that the terms of PCF have
decidable equality which is an application of the theory on indexed
\(\WW\)-types developed in~\cref{sec:indexed-W-types}.%
\index{W-type@\(\WW\)-type!indexed}\index{decidability!of equality}
It turns out to be helpful to study the \(k\)-step reflexive transitive closure
of an arbitrary relation more generally and isolate criteria for it to be
decidable.%
\index{reflexive transitive closure}

\subsubsection{Decidability of the \(k\)-step reflexive transitive closure of a
  relation}

Fix an arbitrary type \(X\) and a proposition-valued binary relation \(R\) on
\(X\).
We define the \(k\)-step reflexive transitive closure of \(R\). As
in~\cref{def:refltransclos}, we want this relation to be proposition-valued
again, so we proceed with an auxiliary definition that we truncate.

\begin{definition}[\(k\)-step reflexive transitive closure, \(R_k\), \(R^k\)]\label{kstep-refltransclos}
  Define \(x \relR_k y\) by induction on the natural number \(k\) as
  \begin{enumerate}[(i)]
  \item \(x \relR_0 y \colonequiv x = y\), and
  \item
    \(x \relR_{k+1} z \colonequiv \Sigma_{y : X} \pa{x \relR y} \times \pa{y
      \relR_k z}\).
  \end{enumerate}
  The \emph{\(k\)-step reflexive transitive closure} \(\relR^k\) of \(R\) is
  given by propositionally truncating:
  \(x \relR^k y \colonequiv \squash*{x \relR_k y}\).
\end{definition}

The following proposition relates the reflexive transitive closure of \(R\) and
its \(k\)-step reflexive transitive closure.

\begin{proposition}\label{refltransclos-k-step-refltransclos}
  For every \(x\) and \(y\) we have \(x \relR_\ast y\) if and only if
  \(\Sigma_{k : \Nat}\pa*{x \relR_k y}\), where \(R_\ast\) is the untruncated
  reflexive transitive closure from just before~\cref{def:refltransclos}.
  Hence, for every \(x\) and \(y\) it holds that
  \(x \relR^\ast y\) if and only if \(\exists_{k : \Nat}\pa*{x \relR^k y}\).
\end{proposition}
\begin{proof}
  We define the auxiliary binary relation \(R'\) on \(X\) inductively: for every
  \(x : X\) we have \(x \relR' x\); and if \(x \relR y\) and \(y \relR' z\),
  then \(x \relR' z\).
  Then \(R'\) is reflexive, transitive and it extends \(R\). It follows that
  \(R'\) and \(R_\ast\) are equivalent, so it remains to show that
  \(x \relR' y\) and \(\Sigma_{k : \Nat}\pa*{x \relR_k y}\) are logically
  equivalent.
  In one direction, induction on \(k\) yields a proof that \(x \relR_k y\)
  implies \(x \relR' y\) for every natural number \(k\), and hence that
  \(\Sigma_{k : \Nat}\pa*{x \relR_k y}\) implies \(x \relR' y\).
  The other direction is obtained by induction on the constructors of \(R'\).

  The final claim follows from the functoriality of the truncation
  and the general fact that \(\squash*{\Sigma_{x : X}A(x)}\) and
  \(\squash*{\Sigma_{x : X}\squash*{A(x)}}\) are
  equivalent~\cite[Theorem~7.3.9]{HoTTBook}.
\end{proof}

\begin{definition}[Singe-valuedness and decidability of a relation]%
  \index{decidability}\index{single-valuedness}%
  The relation \(R\) is said to be
  \begin{enumerate}[(i)]
  \item \emph{single-valued} if for every \(x\), \(y\) and \(z\) with
    \(x \relR y\) and \(x \relR z\), we have \(y = z\), and
  \item \emph{decidable} if the type \(x \relR y\) is decidable for every \(x\)
    and \(y\) in \(X\). %
    \qedhere
  \end{enumerate}
\end{definition}

\begin{proposition}\label{decidable-k-step-refl-trans-clos}
  The \(k\)-step reflexive transitive closure \(R^k\) is decidable for every
  natural number \(k\) if
  \begin{enumerate}[(i)]
  \item\label{carrier-dec-eq} the type \(X\) has decidable equality,
  \item\label{single-valued-rel} the relation \(R\) is single-valued, and
  \item\label{left-ext-decidable} the type \(\Sigma_{y : X}\pa*{x \relR y}\) is
    decidable for every \(x : X\). \qedhere
  \end{enumerate}
\end{proposition}
\begin{proof}
  Since the propositional truncation of a type is decidable as soon the type
  itself is, it suffices to prove that the relation \(R_k\) is decidable, which we
  do by induction on~\(k\).
  For \(k = 0\), this means decidability of \(x = y\) for every \(x\) and \(y\)
  which we have by assumption~\ref{carrier-dec-eq}.
  Now suppose that \(R^k\) is decidable and let \(x\) and \(z\) be arbitrary
  elements of \(X\). We prove that \(x \relR_{k+1} z\) is decidable.
  By definition, this means proving that
  \begin{equation}\tag{\(\dagger\)}\label{to-prove-decidable}
    \Sigma_{y : X}\pa*{x \relR y}\times\pa*{y \relR_k z}
  \end{equation}
  is decidable.
  Using~\ref{left-ext-decidable} we can decide whether
  \(\Sigma_{y : X}\pa*{x \relR y}\) has an element or not. If it does not, then
  obviously~\eqref{to-prove-decidable} has no elements either.
  So assume that we have \(y : X\) with \(x \relR y\). By induction hypothesis,
  the type \(y \relR_k z\) is decidable. If it has an element, then so
  does~\eqref{to-prove-decidable}.
  If it does not, then we claim that~\eqref{to-prove-decidable} must be empty
  too. For if it isn't, then we get \(y' : X\) with \(x \relR y'\) and
  \(y' \relR_k z\). But the relation \(R\) is assumed to be single-valued, so
  \(y' = y\) and hence \(y \relR_k z\), contradicting our assumption.
\end{proof}

\subsubsection{Semidecidability at the base type}

After completing the generalities above, we are now ready to complete the proof
of the strategy outlined at the start of this section.
The application of~\cref{decidable-k-step-refl-trans-clos} to the small-step
reduction relation of PCF requires us to prove that the syntax of PCF has
decidable equality, which follows from the results on indexed \(\WW\)-types
featured in~\cref{sec:indexed-W-types}.%
\index{decidability!of equality}\index{W-type@\(\WW\)-type!indexed}%
\index{W-type@\(\WW\)-type!indexed!encoding PCF terms}

\begin{theorem}\label{semidecidability-of-PCF-props}%
  \index{semidecidability}\index{PCF!types!base}%
  For every PCF term \(t\) of the base type, the proposition
  \(\isdefined(\densem{t})\) is semidecidable as witnessed by the logical
  equivalence
  \[
    \isdefined(\densem{t}) \iff \exists_{n : \Nat}\exists_{k : \Nat} \, t
      \smallstep^k \underline n
  \]
  and the decidability of \(t \smallstep^k \underline n\).
\end{theorem}
\begin{proof}
  First of all, observe that \(\densem{t}\) is defined if and only if there
  exists \(n : \Nat\) such that \(t \smallstepclos \underline n\).
  Indeed, if we have \(p : \isdefined(\densem{t})\), then
  \(t \smallstepclos \underline{\liftvalue(\densem{t},p)}\) by computational
  adequacy.
  Conversely, if there exists a natural number \(n\) such that
  \(t \smallstepclos \underline n\), then \(\densem{t} = \eta(n)\) by soundness
  and \cref{densemnumeral}, so that \(\densem{t}\) must be defined.
  Furthermore, by \cref{refltransclos-k-step-refltransclos}, we have that
  \(t \smallstepclos \underline n\) is logically equivalent to
  \(\exists_{k : \Nat}\,t \smallstep^k \underline n\), so this proves the
  logical equivalence
  \(\isdefined(\densem{t}) \iff \exists_{n : \Nat}\exists_{k : \Nat} \, t
    \smallstep^k \underline n\).
  Hence, it only remains to prove that \(t \smallstep^k \underline n\) is
  decidable, for which we use~\cref{decidable-k-step-refl-trans-clos}.
  Accordingly, we need to check its three conditions. First of all, the type of
  PCF terms should have decidable equality, which is guaranteed by a modest
  extension of~\cref{PCF-terms-have-decidable-equality}.
  The other two conditions can be proved by inspection of the operational
  semantics of PCF using decidability of equality of PCF terms.
\end{proof}

\section{Notes}\label{sec:applications-notes}

The treatment of Scott's \(D_\infty\) model of the untyped \(\lambda\)-calculus
in~\cref{sec:Scott-D-infty} is an expanded account of Section~5.2 of our
paper~\cite{deJongEscardo2021a}, while~\cref{sec:Scott-model-of-PCF} is a slight
revision of the exposition in our publication~\cite{deJong2021a}.

The proof that \(D_\infty\) is isomorphic to \(D_\infty^{D_\infty}\) largely
follows that of~\cite[Theorem~4.4]{Scott1972}, although we instantiate the
general framework involving directed bilimits set out in~\cref{sec:bilimits},
rather than working with sequential bilimits directly.

The Scott model was proved sound and computationally adequate by
\citeauthor{Plotkin1977}~\cite{Plotkin1977}, and the techniques of Scott and
Plotkin have been extended to many other programming
languages~\cite{Plotkin1983}.
Our proof follows the modern
presentation given by \citeauthor{Streicher2006}~\cite{Streicher2006}, although,
instead of formulating PCF with variables and \(\lambda\)-abstraction, we revert
to the original, combinatory, formulation of the terms of LCF by
\citeauthor{Scott1993}~\cite{Scott1993} in order to simplify the technical
development.

The formulation of computational adequacy in terms of
\(\isdefined(\densem{t})\), and the suggestion that it could be leveraged to
prove semidecidability of these propositions are due to Mart\'in Escard\'o.


%% file: mainmatter/predicativity-in-order-theory.tex
\chapter{Predicativity in order theory}\label{chap:predicativity-in-order-theory}

In the preceding chapters we gave a type-theoretic account of constructive and
predicative domain theory including many familiar constructions and notions,
such as Scott's \(D_\infty\) model of the untyped \(\lambda\)-calculus and the
theory of continuous dcpos. In this chapter we complement this by exploring what
cannot be done predicatively.

\section{Introduction}

The work in this chapter is in the spirit of constructive reverse
mathematics~\cite{Ishihara2006} and amounts to predicative reverse mathematics:
we show certain statements to crucially rely on resizing axioms in the sense
that they are equivalent to them. Such arguments are important in constructive
mathematics. For example, the constructive failure of trichotomy on the real
numbers is shown~\cite{BridgesRichman1987} by reducing it to a nonconstructive
instance of excluded middle. As another example, note that
in~\cref{flat-gives-LPO} we used a reduction to the limited principle of
omniscience (LPO) to show that \(\Nat_\bot\) cannot be a dcpo constructively.%
\index{constructivity}

Our first main result is that nontrivial (directed or bounded) complete posets
are necessarily large. All examples of dcpos that we have seen have large
carriers, in the sense that all examples of \(\V\)-dcpos have carriers that live
in \(\V^+\) or some higher universe.%
\index{smallness}\index{resizing!propositional}\index{predicativity}
We show here that this is no coincidence, but rather a necessity, in the sense
that if such a nontrivial poset is small, then weak propositional resizing
holds. It is possible to derive full propositional resizing if we strengthen
nontriviality to positivity in the sense of~\cite{Johnstone1984}. The
distinction between nontriviality and positivity is analogous to the distinction
between nonemptiness and inhabitedness.%
\index{poset!nontrivial}%
\index{poset!delta-complete@\deltacomplete{\V}!positive|see {positivity}}%
\index{positivity}
We prove our results for a general class of posets, which includes directed
complete posets, bounded complete posets and sup-lattices, using a technical
notion of a \deltacomplete{\V} poset.%
\index{poset!delta-complete@\deltacomplete{\V}}
We also show that nontrivial locally small
\deltacomplete{\V} posets necessarily lack decidable equality.%
\index{decidability!of equality}
Specifically, we
can derive weak excluded middle from assuming the existence of a nontrivial
locally small \deltacomplete{\V} poset with decidable equality.%
\index{excluded middle!weak}
Moreover, if we assume positivity instead of nontriviality, then we can derive
full excluded middle.\index{excluded middle}

The fact that nontrivial dcpos are necessarily large has the important
consequence that \emph{Tarski's Theorem} (and similar results, such as
Pataraia's Lemma) cannot be applied in nontrivial instances, even though it has
a predicative proof.%
\index{Tarski's Theorem}
Furthermore, we explain that generalisations of Tarski's Theorem that allow for
large structures are provably false. Specifically, we show that the ordinal of
ordinals in a univalent universe does not have a maximal element, but does have
small suprema in the presence of small set quotients or set replacement,
illustrating the abstract theory
of~\cref{sec:quotients-replacement-prop-trunc-revisited}.%
\index{ordinal}\index{supremum!of ordinals}\index{set quotient}\index{set replacement}

Finally, we clarify, in our predicative setting, the relation between the
traditional definition of sup-lattice that requires suprema for all subsets and
our definition that asks for suprema of all small families, further explaining
our choice to work with families in our development of domain theory.

\section{Large posets without decidable equality}
\label{sec:large-posets}
A well-known result of Freyd in classical mathematics says that every
complete small category is a preorder~\cite[Exercise~D of
Chapter~3]{Freyd1964}. In other words, complete categories are necessarily large
and only complete preorders can be small, at least impredicatively.
Predicatively, by contrast, we show that many weakly complete posets (including
directed complete posets, bounded complete posets and sup-lattices) are
necessarily large.
We capture these structures by a technical notion of a \deltacomplete{\V} poset
in~\cref{sec:delta-complete-posets}. In~\cref{sec:nontrivial-and-positive} we
define when such structures are nontrivial and introduce the constructively
stronger notion of positivity. \cref{sec:retract-lemmas} and
\cref{sec:small-completeness-with-resizing} contain the two fundamental
technical lemmas and the main theorems, respectively. Finally, we consider
alternative formulations of being nontrivial and positive that ensure that these
notions are properties rather than data and shows how the main theorems remain
valid, assuming univalence.

\subsection{\texorpdfstring{\(\delta_{\V}\)}{delta\_V}-complete posets}
\label{sec:delta-complete-posets}
We start by introducing a class of weakly complete posets that we call
\deltacomplete{\V} posets. The notion of a \deltacomplete{\V} poset is a
technical and auxiliary notion sufficient to make our main theorems go
through. The important point is that many familiar structures (dcpos, bounded
complete posets, sup-lattices) are \deltacomplete{\V} posets
(see~\cref{examples-of-delta-complete-posets}).

\begin{definition}[\deltacomplete{\V} poset, \(\delta_{x,y,P}\), \(\bigvee \delta_{x,y,P}\)]%
  \index{delta-completeness@\(\delta_{\V}\)-completeness|see {poset}}
  \index{poset!delta-completeness@\deltacomplete{\V}}
  A poset \(X\) is \emph{\(\delta_\V\)-complete} for a universe~\(\V\) if for
  every pair of elements \(x,y : X\) with \(x \below y\) and every subsingleton
  \(P\) in \(\V\), the family
  \nomenclature[deltaxyP]{$\delta_{x,y,P}$}{family indexed by \(\One + P\)
    sending the left component to \(x\) and the right component to \(y\)}
  \begin{align*}
    \delta_{x,y,P} : 1 + P &\to X \\
    \inl(\star) &\mapsto x, \text{ and} \\
    \inr(p) &\mapsto y
  \end{align*}
  has a supremum \(\bigvee \delta_{x,y,P}\) in \(X\).
\end{definition}
\begin{remark}[Classically, every poset is \deltacomplete{\V}]
  \label{classically-every-poset-is-delta-complete}
  Consider a pair of elements \(x \below y\) of a poset. If \(P : \V\) is a
  decidable proposition, then we can define the supremum of \(\delta_{x,y,P}\)
  by case analysis on whether \(P\) holds or not. For if it holds, then the
  supremum is \(y\), and if it does not, then the supremum is \(x\). Hence, if
  excluded middle holds in \(\V\), then the family \(\delta_{x,y,P}\) has a
  supremum for every \(P : \V\). Thus, if excluded middle holds in \(\V\), then
  every poset (with carrier in any universe) is \deltacomplete{\V}.
\end{remark}
The above remark naturally leads us to ask whether the converse also holds,
i.e.\ if every poset is \deltacomplete{\V}, does excluded middle in \(\V\) hold?
As far as we know, we can only get weak excluded middle in \(\V\), as we will
later see in~\cref{Two-is-not-delta-complete}.
This proposition also shows that in the absence of excluded middle, the notion
of \(\delta_{\V}\)-completeness isn't trivial. For now, we focus on the fact
that, also constructively and predicatively, there are many examples of
\deltacomplete{\V} posets.

\begin{example}[\(\delta_{\V}\)-complete posets]\hfill
  \label{examples-of-delta-complete-posets}%
  \index{sup-lattice|textbf}\index{poset!bounded complete|textbf}\index{dcpo}
  \begin{enumerate}[(i)]
  \item Every \emph{\(\V\)-sup-lattice} is \deltacomplete{\V}. That is, if a
    poset \(X\) has suprema for all families \(I \to X\) with \(I\) in the
    universe \(\V\), then \(X\) is \deltacomplete{\V}.
    Hence, in particular, \(\Omega_{\V}\) is \deltacomplete{\V}, as is the
    \(\V\)-powerset \(\powerset_{\V}(X)\) for a type \(X\) in any universe.
  \item Every \emph{\(\V\)-bounded complete} poset is \deltacomplete{\V}. That
    is, if \(X\) is a poset with suprema for all bounded families \(I \to X\)
    with \(I : \V\), then \(X\) is \deltacomplete{\V}.
    A~family \(\alpha : I \to X\) is \emph{bounded} if there exists some
    \(x : X\) with \(\alpha(i) \below x\) for every \(i : I\). For example, the
    family \(\delta_{x,y,P}\) is bounded by \(y\).
  \item Every \(\V\)-dcpo is \deltacomplete{\V}, since the family
    \(\delta_{x,y,P}\) is directed. %
    \qedhere
  \end{enumerate}
\end{example}

\subsection{Nontrivial and positive posets}
\label{sec:nontrivial-and-positive}
In \cref{classically-every-poset-is-delta-complete} we saw that if we can decide
a proposition \(P\), then we can define \(\bigvee \delta_{x,y,P}\) by case
analysis. What about the converse? That is, if \(\delta_{x,y,P}\) has a supremum
and we know that it equals \(x\) or \(y\), can we then decide \(P\)?  Of course,
if \(x = y\), then \(\bigvee \delta_{x,y,P} = x = y\), so we don't learn
anything about \(P\). But what if we add the assumption that \(x \neq y\)? It
turns out that constructively we can only expect to derive decidability of
\(\lnot P\) in that case. This is due to the fact that \(x \neq y\) is a negated
proposition, which is rather weak constructively, leading us to later define
(see~\cref{def:strictly-below}) a constructively stronger notion for elements of
\deltacomplete{\V} posets.

\begin{definition}[Nontriviality]%
  \index{poset!nontrivial|textbf}%
  A poset \((X,\below)\) is \emph{nontrivial} if we have specified \(x,y : X\)
  with \(x \below y\) and \(x \neq y\).
\end{definition}

\begin{lemma}\label{delta-sup-weak-em}
  For a nontrivial poset \((X,{\below},x,y)\) and a proposition \(P : \V\), we
  have the following two implications:
  \begin{enumerate}[(i)]
  \item\label{delta-sup-weak-em-1} if the supremum of \(\delta_{x,y,P}\) exists and
    \(x = \bigvee \delta_{x,y,P}\), then \(\lnot P\) is the case;
  \item\label{delta-sup-weak-em-2} if the supremum of \(\delta_{x,y,P}\) exists and
    \(y = \bigvee \delta_{x,y,P}\), then \(\lnot\lnot P\) is the case.
  \end{enumerate}
\end{lemma}
\begin{proof}
  Let \(P : \V\) be an arbitrary proposition. For~\ref{delta-sup-weak-em-1},
  suppose that \(x = \bigvee \delta_{x,y,P}\) and assume for a contradiction
  that we have \(p : P\). Then
  \( y \equiv \delta_{x,y,P}(\inr(p)) \below \bigvee \delta_{x,y,P} = x, \)
  which is impossible by antisymmetry and our assumptions that \(x \below y\)
  and \(x \neq y\).
  For~\ref{delta-sup-weak-em-2}, suppose that \(y = \bigvee \delta_{x,y,P}\) and
  assume for a contradiction that \(\lnot P\) holds. Then
  \(x = \bigvee \delta_{x,y,P} = y\), contradicting our assumption that
  \(x \neq y\). \qedhere
\end{proof}

\begin{proposition}\label{Two-is-not-delta-complete}%
  \index{excluded middle!weak}%
  If the poset \(\Two\) with exactly two elements \(0 \below 1\) is
  \deltacomplete{\V}, then weak excluded middle in \(\V\) holds.
\end{proposition}
\begin{proof}
  Suppose that \(\Two\) were \deltacomplete{\V} and let \(P : \V\) be an
  arbitrary subsingleton. We must show that \(\lnot P\) is decidable. Since
  \(\Two\) has exactly two elements, the supremum \(\bigvee \delta_{0,1,P}\)
  must be \(0\) or \(1\). But then we apply \cref{delta-sup-weak-em} to get
  decidability of \(\lnot P\).
\end{proof}
Combining~\cref{classically-every-poset-is-delta-complete,Two-is-not-delta-complete}
yields that excluded middle implies that every poset is \deltacomplete{\V},
which in turns implies weak excluded middle.
We do not know whether these implications can be reversed.
That the conclusion of the implication in
\cref{delta-sup-weak-em}\ref{delta-sup-weak-em-2} cannot be strengthened to
say that \(P\) must hold is shown by the following observation.
\begin{proposition}\label{delta-sup-em}%
  \index{excluded middle}%
  If for every two propositions \(Q\) and \(R\) in \(\Omega_{\V}\) with
  \(Q \below R\) and \(Q \neq R\) we have that the equality
  \(R = \bigvee \delta_{Q,R,P}\) in \(\Omega_{\V}\) implies \(P\) for every
  proposition \(P : \V\), then excluded middle in \(\V\) follows.
\end{proposition}
\begin{proof}
  Assume the hypothesis in the proposition. We are going to show that
  \(\lnot\lnot P \to P\) for every proposition \(P : \V\), from which excluded
  middle in \(\V\) follows. So let \(P\) be a proposition in \(\V\) such that
  its double negation holds. This yields \(\Zero \neq P\), so by assumption the
  equality \(P = \bigvee \delta_{\Zero,P,P}\) implies~\(P\). But this equality
  holds, by construction of suprema in \(\Omega_\V\).
\end{proof}

Thus, having a pair of elements \(x \below y\) with \(x \neq y\) is rather weak
constructively in that we can only derive \(\lnot\lnot P\) from
\(y = \bigvee\delta_{x,y,P}\).
As promised in the introduction of this section, we now introduce and motivate a
constructively stronger notion.

\begin{definition}[Strictly below, \(x \sbelow y\)]\label{def:strictly-below}%
  \index{strictly below relation}%
  We~say that \(x\) is \emph{strictly below} \(y\) in a
  \(\delta_{\V}\)\nobreakdash-complete poset if \(x \below y\) and, moreover,
  for every \(z \aboveorder y\) and every proposition \(P : \V\), the equality
  \(z = \bigvee \delta_{x,z,P}\) implies \(P\).
  \nomenclature[sqsubset]{$x \sbelow y$}{strictly below relation}
\end{definition}
Note that with excluded middle, \(x \sbelow y\) is equivalent to the conjunction
of \(x \below y\) and \(x \neq y\). But constructively, the former is much
stronger, as the following examples and proposition illustrate.

\begin{example}[Strictly below in \(\Omega_{\V}\) and \(\powerset_{\V}(X)\)]%
  \label{examples-strictly-below}%
  \hfill
  \begin{enumerate}[(i)]
  \item\label{strictly-below-in-Omega}%
    \index{strictly below relation!in the type of subsingletons}%
    We illustrate the notion of strictly below in \(\Omega_{\V}\). For an
    arbitrary proposition \(P : \V\), we have that \(\Zero_\V \neq P\) holds
    precisely when \(\lnot\lnot P\) does. However, \(\Zero_\V\) is strictly
    below \(P\) if and only if \(P\) holds.
    More generally, for any two propositions \(Q,P : \V\), we have
    \((Q \below P) \times (Q\neq P)\) if and only if
    \(\lnot Q \times \lnot\lnot P\) holds.
    But, \(Q \sbelow P\) holds if and only if \(\lnot Q \times P\) holds.
  \item\label{strictly-below-in-powerset}%
    \index{strictly below relation!in the powerset}%
    In the powerset \(\powerset_{\V}(X)\) of a type \(X : \V\) the situation is
    slightly more involved, but still illustrative.
    If we have two subsets \(A \below B\) of \(X\), then
    \(A \neq B\) if and only if
    \(\lnot\pa*{\forall_{x : X}\pa*{x \in B \to x \in A}}\).
    However, if \(A \sbelow B\) and \(y \in A\) is decidable for every
    \(y : X\), then we get the stronger
    \(\exists_{x : X}\pa*{x \in B \times x \not\in A}\).
    For we can take \(P : \V\) to be
    \(\exists_{x : X}\pa*{x \in B \times x \not\in A}\) and observe that
    \(\bigvee\delta_{A,B,P} = B\), because if \(x \in B\), either \(x \in A\) in
    which case \(x \in \bigvee\delta_{A,B,P}\), or \(x \not\in A\) in which case
    \(P\) must hold and \(x \in B = \bigvee\delta_{A,B,P}\).
    Conversely, if we have \(A \below B\) and an element \(x \in B\) with
    \(x \not\in A\), then \(A \sbelow B\). For if \(C \aboveorder B\) is a
    subset and \(P : \V\) a proposition such that \(\bigvee\delta_{A,C,P} = C\),
    then \(x \in C = \bigvee\delta_{A,C,P} = A \cup \{y \in C \mid P\}\), so
    either \(x \in A\) or \(P\) must hold. But \(x \not\in A\) by assumption, so
    \(P\) must be true, proving \(A \sbelow B\). %
    \qedhere
  \end{enumerate}
\end{example}

\begin{proposition}\label{sbelow-below-neq}%
  \index{excluded middle}%
  For elements \(x\) and \(y\) of a \deltacomplete{\V} poset, we have that
  \(x \sbelow y\) implies both \(x \below y\) and \(x \neq y\). However, if the
  conjunction of \(x \below y\) and \(x \neq y\) implies \(x \sbelow y\) for
  every \(x,y : \Omega_\V\), then excluded middle in \(\V\) holds.
\end{proposition}
\begin{proof}
  Note that \(x \sbelow y\) implies \(x \below y\) by definition. Now suppose
  that \(x \sbelow y\). Then the equality
  \(y = \bigvee \delta_{x,y,\Zero_{\V}}\) implies that \(\Zero_{\V}\) holds. But
  if \(x = y\), then this equality holds, so \(x \neq y\), as desired.

  For \(P : \Omega_{\V}\) we observed that \(\Zero_\V \neq P\) is equivalent to
  \(\lnot\lnot P\) and that \(\Zero_\V \sbelow P\) is equivalent to \(P\), so if
  we had \(\pa*{\pa*{x \below y} \times \pa*{x \neq y}} \to x \sbelow y\) in
  general, then we would have \(\lnot\lnot P \to P\) for every proposition \(P\)
  in \(\V\), which is equivalent to excluded middle in~\(\V\).
\end{proof}

\begin{lemma}\label{sbelow-trans}
  The following transitivity properties hold for all elements \(x\), \(y\) and
  \(z\) of a \deltacomplete{\V} poset:
  \begin{enumerate}[(i)]
  \item\label{below-sbelow-sbelow} if \(x \below y \sbelow z\), then
    \(x \sbelow z\);
  \item\label{sbelow-below-sbelow} if \(x \sbelow y \below z\), then
    \(x \sbelow z\).
  \end{enumerate}
\end{lemma}
\begin{proof}
  \ref{below-sbelow-sbelow} Assume \(x \below y \sbelow z\), let \(P\) be an
  arbitrary proposition in \(\V\) and suppose that \(z \below w\). We must show
  that \(w = \bigvee \delta_{x,w,P}\) implies \(P\). But \(y \sbelow z\), so we
  know that the equality \(w = \bigvee \delta_{y,w,P}\) implies \(P\). Now
  observe that \(\bigvee \delta_{x,w,P} \below \bigvee \delta_{y,w,P}\), so if
  \(w = \bigvee \delta_{x,w,P}\), then \(w = \bigvee \delta_{y,w,P}\), finishing
  the proof.
  \ref{sbelow-below-sbelow} Assume \(x \sbelow y \below z\), let \(P\) be an
  arbitrary proposition in \(\V\) and suppose that \(z \below w\). We must show
  that \(w = \bigvee \delta_{x,w,P}\) implies \(P\). But \(x \sbelow y\) and
  \(y \below w\), so this follows immediately.
\end{proof}

\begin{proposition}\label{positive-element-equivalent}
  The following are equivalent for an element \(y\) of a \(\V\)-sup-lattice~\(X\):
  \begin{enumerate}[(i)]
  \item\label{positive-1} the least element of \(X\) is strictly below \(y\);
  \item\label{positive-2} for every family \(\alpha : I \to X\) with \(I : \V\),
    if \(y \below \bigvee \alpha\), then \(I\) is inhabited;
  \item\label{positive-3} there exists some \(x : X\) with \(x \sbelow y\).
  \end{enumerate}
\end{proposition}
\begin{proof}
  Write \(\bot\) for the least element of \(X\). By~\cref{sbelow-trans} we have:
  \[
    \bot \sbelow y
    \iff \exists_{x : X}\pa*{\bot \below x \sbelow y}
    \iff \exists_{x : X}\pa*{x \sbelow y},
  \]
  which proves the equivalence of \ref{positive-1} and \ref{positive-3}. It
  remains to prove that \ref{positive-1} and \ref{positive-2} are
  equivalent. Suppose that \(\bot \sbelow y\) and let \(\alpha : I \to X\) with
  \(y \below \bigvee \alpha\). Using \(\bot \sbelow y \below \bigvee \alpha\)
  and~\cref{sbelow-trans}, we have \(\bot \sbelow \bigvee \alpha\). Hence, we
  only need to prove
  \(\bigvee \alpha \below \bigvee \delta_{\bot,\bigvee \alpha,\exists {i :
      I}}\), but
  \(\alpha_j \below \bigvee \delta_{\bot,\bigvee\alpha,\exists {i : I}}\) for
  every \(j : I\), so this is true indeed.
  For the converse, assume that \(y\) satisfies \ref{positive-2}, suppose
  \(z \aboveorder y\) and let \(P : \V\) be a proposition such that
  \(z = \bigvee \delta_{\bot,z,P}\). We must show that \(P\) holds. But notice
  that
  \(y \below z = \bigvee \delta_{\bot,z,P} = \bigvee \pa*{(p : P)\mapsto z}\),
  so \(P\) must be inhabited as \(y\) satisfies~\ref{positive-2}.
\end{proof}
\cref{positive-2}~in~\cref{positive-element-equivalent} says exactly that \(y\)
is a positive element in the sense of~\cite[p.~98]{Johnstone1984}.%
\index{positivity}
Observe that \ref{positive-2} makes sense for any poset, not just
\(\V\)-sup-lattices: we don't need to assume the existence of suprema to
formulate condition~\ref{positive-2}, because we can rephrase
\(y \below \bigvee \alpha\) as ``for every \(x : X\), if \(x\) is an upper bound
of \(\alpha\) and \(x\) is below any other upper bound of \(\alpha\), then
\(y \below x\)''. Similarly, the notion of being strictly below makes sense for
any poset.
What \cref{positive-element-equivalent} shows is that strictly below generalises
Johnstone's positivity from a \emph{unary} relation to a \emph{binary} one.
Another binary generalisation of positivity in a different direction is that of
a positivity relation in formal
topology~\cite{Sambin2003,CirauloSambin2018,CirauloVickers2016}. For a formal
topology \(S\), one considers a binary relation \(\ltimes\) between \(S\) and
its powerclass. Then \(a \ltimes S\) implies that \(a\) is
positive~\cite[p.~764]{CirauloSambin2018}, while sets of the form
\(\{a \in S \mid a \ltimes U\}\) are thought of as formal closed
subsets~\cite{CirauloVickers2016}.

Looking to strengthen the notion of a nontrivial poset, we make the following
definitions.

\begin{definition}[Positivity; cf.\ {\cite[p.~98]{Johnstone1984}}]\hfill%
  \index{positivity|textbf}%
  \begin{enumerate}[(i)]
  \item An element of a \deltacomplete{\V} poset is \emph{positive} if it
    satisfies \cref{positive-element-equivalent}\ref{positive-3}.
  \item A \deltacomplete{\V} poset \(X\) is \emph{positive} if we have specified
    \(x,y : X\) with \(x\) strictly below \(y\).%
    \qedhere
  \end{enumerate}
\end{definition}

\begin{example}[Nontriviality and positivity in \(\Omega_{\V}\) and
  \(\powerset_{\V}(X)\)]\label{examples-nontrivial-positive}\hfill
  \begin{enumerate}[(i)]
  \item\label{nontrivial-positive-Omega}%
    \index{positivity!in the type of subsingletons}%
    Consider an element \(P\) of the \deltacomplete{\V} poset \(\Omega_\V\). The
    pair \(\pa*{\Zero_\V , P}\) witnesses nontriviality of \(\Omega_\V\) if and
    only if \(\lnot\lnot P\) holds, while it witnesses positivity if and only if
    \(P\) holds.
  \item\label{nontrivial-positive-powerset}%
    \index{positivity!in the powerset}%
    Say that a subset \(A : \powerset_{\V}(X)\) is \emph{nonempty} if the type
    \(\Sigma_{x : X}\pa*{x\in A}\) is nonempty, and \emph{inhabited} if this type is
    inhabited.
    The pair \((\emptyset , A)\) witnesses nontriviality of \(\powerset_\V(X)\)
    if and only if \(A\) is nonempty, while it witnesses positivity if and only
    if \(A\) is inhabited.%
    \qedhere
  \end{enumerate}
\end{example}

We describe how the notion of strictly below relates to compactness and the
way-below relation from domain theory.

\begin{proposition}%
  \index{compactness}%
  If \(x \below y\) are unequal elements of a \(\V\)-dcpo \(D\) and \(y\) is
  compact, then \(x \sbelow y\) without needing to assume excluded middle.
  In particular, a compact element \(x\) of a \(\V\)-dcpo with a least element
  \(\bot\) is positive if and only if \(x \neq \bot\).
\end{proposition}
\begin{proof}
  Suppose that \(x \below y\) are unequal and that \(y\) is compact. We are to
  show that \(x \sbelow y\). So assume we have \(z \aboveorder y\) and a
  proposition \(P : \V\) such that \(y \below z = \bigvee\delta_{x,z,P}\). By
  compactness of \(y\), there exists \(i : \One + P\) such that
  \(y \below \delta_{x,z,P}(i)\) already. But \(i\) can't be equal to
  \(\inl(\star)\), since \(x \neq y\) is assumed. Hence, \(i = \inr(p)\) and
  \(P\) must hold.
\end{proof}

Note that \(x \sbelow y\) does not imply \(x \ll y\) in general, because with
excluded middle, \(x \sbelow y\) is simply the conjunction of \(x \below y\) and
\(x \neq y\), which does not imply \(x \ll y\) in general.
Also, the conjunction of \(x \ll y\) and \(x \neq y\) does not imply
\(x \sbelow y\), as far as we know.

We end this section by summarising why we consider strictly below to be a
suitable notion in our constructive framework.
First of all, \(x \sbelow y\) coincides with \((x \below y) \times (x \neq y)\)
in the presence of excluded middle, so it is compatible with classical logic.
Secondly, we've seen in~\cref{examples-strictly-below} that strictly below works
well in the poset of truth values and in powersets, yielding familiar
constructive strengthenings.
Thirdly, being strictly below generalises Johnstone's notion of positivity from
a unary to a binary relation.
And finally, as we will see shortly, the derived notion of positive poset is
exactly what we need to derive \(\Omegaresizingalt{\V}\) rather than the weaker
\(\Omeganotnotresizingalt{\V}\) in \cref{positive-impredicativity}.

\subsection{Retract lemmas}\label{sec:retract-lemmas}
We show that the type of propositions in \(\V\) is a retract of any positive
\deltacomplete{\V} poset and that the type of \(\lnot\lnot\)-stable
propositions in \(\V\) is a retract of any nontrivial \deltacomplete{\V} poset.

\begin{definition}[\(\Delta_{x,y}\)]
  For a nontrivial \deltacomplete{\V} poset \(\pa*{X,{\below},x,y}\), we define
  \(\Delta_{x,y} : \Omega_{\V} \to X\) by the assignment
  \(P \mapsto \bigvee \delta_{x,y,P}\).
  \nomenclature[Deltaxy]{$\Delta_{x,y}$}{map sending a proposition \(P\) to the
      supremum of \(\delta_{x,y,P}\)}
\end{definition}
We will often omit the subscripts in \(\Delta_{x,y}\) when it is clear from the
context.
We extend the definition of local smallness~(\cref{def:local-smallness}) from
\(\V\)-dcpos to \deltacomplete{\V} posets.

\begin{definition}[Local smallness, \({\below_{\V}}\)]%
  \index{poset!delta-complete@\deltacomplete{\V}!locally small}%
  A \deltacomplete{\V} poset is \emph{locally small} if its order has
  \(\V\)-small values, in which case we often denote the order with values in
  \(\V\) by \({\below_{\V}}\).
  \nomenclature[sqsubseteqV]{$x \below_{\V} y$}{\(\V\)-valued order relation of
    a locally small \deltacomplete{\V} poset}
\end{definition}

\begin{lemma}\label{nontrivial-retract}%
  \index{poset!nontrivial}%
  \index{retract}%
  A locally small \deltacomplete{\V} poset \(X\) is nontrivial, witnessed by
  elements \(x \below y\), if~and only if the composite
  \(\Omeganotnot{\V} \hookrightarrow \Omega_{\V} \xrightarrow{\Delta_{x,y}} X\)
  is a section.
\end{lemma}
\begin{proof}
  Suppose first that \((X,\below,x,y)\) is nontrivial and locally small. We define
  \begin{align*}
    r : X &\to \Omeganotnot{\V} \\
    z &\mapsto z \not\below_{\V} x.
  \end{align*}
  Note that negated propositions are \(\lnot\lnot\)-stable, so \(r\) is
  well-defined. Let \(P : \V\) be an arbitrary \(\lnot\lnot\)-stable
  proposition. We want to show that \(r (\Delta_{x,y}(P)) = P\). By~propositional
  extensionality, establishing logical equivalence suffices.
  Suppose first that \(P\) holds. Then
  \(\Delta_{x,y}(P) \equiv \bigvee \delta_{x,y,P} = y\), so
  \(r(\Delta_{x,y}(P)) = r(y) \equiv \pa*{y \not\below_{\V} x}\) holds by
  antisymmetry and our assumptions that \(x \below y\) and \(x \neq y\).
  Conversely, assume that \(r(\Delta_{x,y}(P))\) holds, i.e.\ that we have
  \(\bigvee \delta_{x,y,P} \not\below_{\V} x\). Since \(P\) is
  \(\lnot\lnot\)-stable, it suffices to derive a contradiction from~\(\lnot
  P\). So assume~\(\lnot P\). Then \(x = \bigvee \delta_{x,y,P}\), so
  \(r(\Delta_{x,y}(P)) = r(x) \equiv x \not\below_{\V} x\), which is false by
  reflexivity.

  For the converse, assume that
  \(\Omeganotnot{\V} \hookrightarrow \Omega_{\V} \xrightarrow{\Delta_{x,y}} X\)
  has a retraction \(r : \Omeganotnot{\V} \to X\). Then
  \(\Zero_{\V} = r(\Delta_{x,y}(\Zero_{\V})) = r(x)\) and
  \(\One_{\V} = r(\Delta_{x,y}(\One_{\V})) = r(y)\),
  where we used that \(\Zero_{\V}\) and \(\One_{\V}\) are \(\lnot\lnot\)-stable.
  Since \(\Zero_{\V} \neq \One_{\V}\), we get \(x \neq y\), so
  \((X,\below,x,y)\) is nontrivial, as desired.
\end{proof}
The appearance of the double negation in the above lemma is due to the
definition of nontriviality. If we instead assume a positive poset \(X\), then
we can exhibit all of \(\Omega_{\V}\) as a retract of \(X\).
\begin{lemma}\label{positive-retract}%
  \index{positivity}\index{retract}%
  A locally small \deltacomplete{\V} poset \(X\) is positive, witnessed by
  elements \(x \sbelow y\), if~and only if for every \(z \aboveorder y\), the map
  \(\Delta_{x,z} : \Omega_{\V} \to X\) is a section.
\end{lemma}
\begin{proof}
  Suppose first that \((X,\below,x,y)\) is positive and locally small and let
  \(z \aboveorder y\) be arbitrary. We define
  \begin{align*}
    r_z : X &\mapsto \Omega_{\V} \\
    w &\mapsto z \below_{\V} w.
  \end{align*}
  Let \(P : \V\) be arbitrary proposition. We want to show that
  \(r_z(\Delta_{x,z}(P)) = P\). Because of propositional extensionality, it
  suffices to establish a logical equivalence between \(P\) and
  \(r_z(\Delta_{x,z}(P))\).
  If \(P\) holds, then \(\Delta_{x,z}(P) = z\), so
  \(r_z(\Delta_{x,z}(P)) = r_z(z) \equiv \pa*{z \below_{\V} z}\) holds as well
  by reflexivity.
  Conversely, assume that \(r_z(\Delta_{x,z}(P))\) holds, i.e.\ that we have
  \(z \below_{\V} \bigvee \delta_{x,z,P}\). Since
  \({\bigvee \delta_{x,z,P} \below z}\) always holds, we get
  \(z = \bigvee \delta_{x,z,P}\) by antisymmetry. But by assumption
  and~\cref{sbelow-trans}, the element \(x\) is strictly~below~\(z\), so \(P\)
  must hold.

  For the converse, assume that for every \(z \aboveorder y\), the map
  \(\Delta_{x,z} : \Omega_{\V} \to X\) has a retraction
  \(r_z : X \to \Omega_{\V}\). We must show that the equality
  \(z = \Delta_{x,z}(P)\) implies~\(P\) for every \(z \aboveorder y\) and
  proposition \(P : \V\). Assuming \(z = \Delta_{x,z}(P)\), we have
  \(\One_{\V} = r_z(\Delta_{x,z}(\One_{\V})) = r_z(z) = r_z(\Delta_{x,z}(P)) =
  P\), so \(P\) must hold indeed. Hence, \((X,\below,x,y)\) is positive, as
  desired.
\end{proof}

\subsection{Small completeness with resizing}
\label{sec:small-completeness-with-resizing}
We present our main theorems here, which show that, constructively and
predicatively, nontrivial \deltacomplete{\V} posets are necessarily large and
necessarily lack decidable equality.

\begin{definition}[Smallness]%
  \index{poset!delta-complete@\deltacomplete{\V}!small}%
  A \deltacomplete{\V} poset is \emph{small} if it is locally small and its
  carrier is \(\V\)-small.
\end{definition}

\begin{theorem}\label{nontrivial-impredicativity}\label{positive-impredicativity}\hfill
  \begin{enumerate}[(i)]
  \item\label{nontrivial-impredicativity-1} There is a nontrivial small
    \deltacomplete{\V} poset if and only if \(\Omeganotnotresizingalt{\V}\)
    holds.
  \item\label{positive-impredicativity-2} There is a positive small
    \deltacomplete{\V} poset if and only if \(\Omegaresizingalt{\V}\) holds.
  \end{enumerate}%
  \index{resizing!of the type of propositions}
\end{theorem}
\begin{proof}
  \ref{nontrivial-impredicativity-1} Suppose that \((X,\below,x,y)\) is a
  nontrivial small \deltacomplete{\V} poset. Using \cref{nontrivial-retract}, we
  can exhibit \(\Omeganotnot{\V}\) as a retract of \(X\). But \(X\) is
  \(\V\)-small by assumption, so by~\cref{is-small-retract} the type
  \(\Omeganotnot{\V}\) is \(\V\)-small as well.
  For the converse, note that
  \(\pa*{\Omeganotnot{\V},\to,\Zero_{\V},\One_{\V}}\) is a nontrivial locally
  small \(\V\)-sup-lattice with \(\bigvee \alpha\) given by
  \(\lnot\lnot\exists_{i : I}\alpha_i\). And if we assume
  \(\Omeganotnotresizingalt{\V}\), then it is small.
  \ref{positive-impredicativity-2} Suppose that \((X,\below,x,y)\) is a positive
  small poset. By \cref{positive-retract}, we can exhibit \(\Omega_{\V}\) as a
  retract of \(X\). But \(X\) is \(\V\)-small by assumption, so
  by~\cref{is-small-retract} the type \(\Omega_{\V}\) is \(\V\)-small as well.
  For the converse, note that \(\pa*{\Omega_{\V},\to,\Zero_\V,\One_\V}\) is a
  positive locally small \(\V\)-sup-lattice. And if we assume
  \(\Omegaresizingalt{\V}\), then it is small. \qedhere
\end{proof}

\begin{lemma}[{\cite[\mkTTurl{TypeTopology.DiscreteAndSeparated}]{TypeTopology}}]%
  \label{equality-retract}\hfill%
  \index{retract}\index{decidability!of equality}%
  \index{not-not-stability@\(\lnot\lnot\)-stability}%
  \begin{enumerate}[(i)]
  \item\label{dec-eq-retract} Types with decidable equality are closed under
    retracts.
  \item\label{not-not-stable-eq-retract} Types with \(\lnot\lnot\)-stable
    equality are closed under retracts.
  \end{enumerate}
\end{lemma}
\begin{proof}
  \ref{dec-eq-retract} Let \(s : X \to Y\) be a section with retraction \(r\),
  assume that \(Y\) has decidable equality and let \(x,y : X\) be
  arbitrary. Then \(s(x) = s(y)\) is decidable by assumption. If
  \(s(x) = s(y)\), then \(x = r(s(x)) = r(s(y)) = y\); and if
  \(s(x) \neq s(y)\), then certainly \(x \neq y\). Thus, \(x = y\) is decidable,
  as desired.
  \ref{not-not-stable-eq-retract} Using the same notation as before, the type
  \(s(x) = s(y)\) is assumed to be \(\lnot\lnot\)-stable. But then
  \[
    \lnot\lnot(x = y)
    \xrightarrow{\text{functoriality of \(\lnot\lnot\)}}
    \lnot\lnot(s(x) = s(y))
    \xrightarrow{\text{\(\lnot\lnot\)-stability}}
    s(x) = s(y)
    \xrightarrow{\text{apply \(r\)}}
    x = y,
  \]
  so \(x = y\) is \(\lnot\lnot\)-stable, completing the proof.
\end{proof}

\begin{example}[Types with \(\lnot\lnot\)-stable equality]%
  \label{types-with-not-not-stable-equality}%
  \index{type!of natural numbers}%
  \index{Dedekind real}%
  The simple types \(\Nat\), \({\Nat \to \Nat}\), \({\Nat \to \Nat \to \Nat}\),
  etc., see~{\cite[\mkTTurl{TypeTopology.SimpleTypes}]{TypeTopology}}, and the
  type of Dedekind real
  numbers~{\cite[\mkTTurl{Various.Dedekind}]{TypeTopology}} all have
  \(\lnot\lnot\)-stable equality, as does the type \(\Omeganotnot{\U}\) of
  \(\lnot\lnot\)\nobreakdash-stable propositions in any universe \(\U\).
\end{example}

\begin{theorem}\label{nontrivial-weak-em}%
  \label{excluded middle!weak}\index{poset!nontrivial}%
  There is a nontrivial locally small \deltacomplete{\V} poset with decidable
  equality if and only if weak excluded middle in \(\V\) holds.
\end{theorem}
\begin{proof}
  Suppose that \((X,\below,x,y)\) is a nontrivial locally small
  \deltacomplete{\V} poset with decidable equality. Then by
  \cref{nontrivial-retract,equality-retract}, the type \(\Omeganotnot{\V}\) must
  have decidable equality too. But negated propositions are
  \(\lnot\lnot\)-stable, so this yields weak excluded middle in \(\V\). For the
  converse, note that \(\pa*{\Omeganotnot{\V},\to,\Zero_\V,\One_\V}\) is a
  nontrivial locally small \(\V\)-sup-lattice that has decidable equality if and
  only if weak excluded middle in \(\V\) holds.
\end{proof}

\begin{theorem}\label{positive-em} The following are equivalent:%
  \index{positivity}\index{excluded middle}%
  \index{not-not-stability@\(\lnot\lnot\)-stability}%
  \begin{enumerate}[(i)]
  \item\label{positive-loc-small-not-not-stable-eq}%
    there is a positive locally small \deltacomplete{\V} poset with
    \(\lnot\lnot\)-stable equality;
  \item\label{positive-loc-small-dec-eq}%
    there is a positive locally small \deltacomplete{\V} poset with decidable
    equality;
  \item\label{positive-em-em} excluded middle in \(\V\) holds.
  \end{enumerate}
\end{theorem}
\begin{proof}
  Note that {\ref{positive-loc-small-dec-eq}} \(\Rightarrow\)
  {\ref{positive-loc-small-not-not-stable-eq}}, so we are left to show that
  {\ref{positive-em-em}} \(\Rightarrow\) {\ref{positive-loc-small-dec-eq}} and
  that {\ref{positive-loc-small-not-not-stable-eq}} \(\Rightarrow\)
  {\ref{positive-em-em}}.
  For the first implication, note that
  \(\pa*{\Omega_{\V},\to,\Zero_\V,\One_\V}\) is a positive locally small
  \(\V\)-sup-lattice that has decidable equality if and only if excluded middle
  in \(\V\) holds.
  To see that
  {\ref{positive-loc-small-not-not-stable-eq}}~implies~{\ref{positive-em-em}},
  suppose that \((X,\below,x,y)\) is a positive locally small \deltacomplete{\V}
  poset with \(\lnot\lnot\)-stable equality. Then by
  \cref{positive-retract,equality-retract}, the type \(\Omega_{\V}\) must have
  \(\lnot\lnot\)-stable equality. But this implies that \(\lnot\lnot P \to P\)
  for every proposition \(P\) in \(\V\) which is equivalent to excluded middle
  in \(\V\).
\end{proof}

In particular, \cref{positive-em}\ref{positive-loc-small-not-not-stable-eq}
shows that, constructively, none of the types
from~\cref{types-with-not-not-stable-equality} can be equipped with the
structure of a positive \deltacomplete{\V} poset.
Hence, we cannot expect the Dedekind reals to form a positive bounded complete
poset.

Lattices, bounded complete posets and dcpos are necessarily large and
necessarily lack decidable equality in our predicative constructive setting.%
\index{smallness}\index{decidability!of equality}%
\index{poset!bounded complete}\index{sup-lattice}\index{dcpo}
More precisely:
\begin{corollary}\hfill
  \begin{enumerate}[(i)]
  \item There is a nontrivial small \(\V\)-sup-lattice (or \(\V\)-bounded complete
    poset or \(\V\)-dcpo) if~and~only~if \(\Omeganotnotresizingalt{\V}\) holds.
  \item There is a positive small \(\V\)-sup-lattice (or \(\V\)-bounded complete
    poset or \(\V\)-dcpo) if~and~only~if \(\Omegaresizingalt{\V}\) holds.
  \item There is a nontrivial locally small \(\V\)-sup-lattice (or
    \(\V\)-bounded complete poset or \(\V\)-dcpo) with decidable equality if and
    only if weak excluded middle in \(\V\) holds.
  \item There is a positive locally small \(\V\)-sup-lattice (or \(\V\)-bounded
    complete poset or \(\V\)-dcpo) with decidable equality if and only if
    excluded middle in \(\V\) holds.
  \end{enumerate}
\end{corollary}

The above notions of non-triviality and positivity are data rather than
property. Indeed, a nontrivial poset \(X\) is (by definition) equipped with two
designated points \(x,y : X\) such that \(x \below y\) and \(x \neq y\). It is
natural to wonder if the propositionally truncated versions of these two notions
yield the same conclusions. We show that this is indeed the case if we assume
univalence. The need for the univalence assumption comes from the fact that
smallness is a property precisely if univalence holds, as shown in
\cref{is-small-is-prop,is-small-univalence}.%
\index{univalence}

\begin{definition}[Nontrivial/positive in an unspecified way]%
  \index{existence!unspecified}\index{poset!nontrivial}\index{positivity}%
  A poset \(X\) is \emph{nontrivial in an unspecified way} if there exist some
  elements \(x,y : X\) such that \(x \below y\) and \(x \neq y\), i.e.\
  \(\exists_{x,y: X}\pa*{\pa*{x \below y} \times \pa*{x \neq y}}\).  Similarly,
  we can define when a poset is \emph{positive in an unspecified way} by
  truncating the notion of positivity.
\end{definition}

\begin{theorem}
  Suppose that the universes \(\V\) and \(\V^+\) are univalent.
  \begin{enumerate}[(i)]
  \item\label{unspecified-1} There is a small \deltacomplete{\V} poset that is
    nontrivial in an unspecified way if and only if
    \(\Omeganotnotresizingalt{\V}\) holds.
  \item\label{unspecified-2} There is a small \deltacomplete{\V} poset that is
    positive in an unspecified way if and only if \(\Omegaresizingalt{\V}\)
    holds.
  \end{enumerate}
\end{theorem}
\begin{proof}
  \ref{unspecified-1} Suppose that \(X\) is a \deltacomplete{\V} poset that is
  nontrivial in an unspecified way.  By~\cref{is-small-is-prop} and univalence
  of \(\V\) and \(\V^+\), the type ``\(\Omeganotnot{\V}\issmalltype{V}\)'' is a
  proposition. By the universal property of the propositional truncation, in
  proving that the type \(\Omeganotnot{\V}\) is \(\V\)-small we can therefore
  assume that are given points \(x,y : X\) with \(x \below y\) and \(x \neq
  y\). The result then follows
  from~\cref{nontrivial-impredicativity}\ref{nontrivial-impredicativity-1}.
  \ref{unspecified-2} By reduction to
  \cref{positive-impredicativity}\ref{positive-impredicativity-2}.
\end{proof}
Similarly, we can prove the following theorem by reduction to
\cref{nontrivial-weak-em,positive-em}.

\begin{theorem}\hfill
  \begin{enumerate}[(i)]
  \item There is a locally small \deltacomplete{\V} poset with decidable
    equality that is nontrivial in an unspecified way if and only if weak
    excluded middle in \(\V\) holds.
  \item There is a locally small \deltacomplete{\V} poset with decidable
    equality that is positive in an unspecified way if and only if excluded
    middle in \(\V\) holds.
  \end{enumerate}
\end{theorem}

\section{Maximal points and fixed points}\label{sec:maximal-and-fixed-points}

As is well known, in impredicative mathematics, a poset has suprema of all
subsets if and only if it has infima of all subsets.
Perhaps counter-intuitively, this ``duality'' theorem can be proved
predicatively. However, in the absence of impredicativity, it is not possible to
fulfil its hypotheses when trying to apply it, because there are no nontrivial
examples.%
\index{duality}\index{sup-lattice}

To explain this, we first have to make the statement of the duality theorem
precise. A~single universe formulation is ``every \(\V\)-small
\(\V\)-sup-lattice has all infima of families indexed by types in \(\V\)''.
The usual proof, adapted from subsets to families, shows that this formulation
is predicatively provable, but in our predicative setting
\cref{nontrivial-impredicativity} tells us that there are no nontrivial examples
to apply it to.

It is natural to wonder whether the single universe formulation can be
generalised to \emph{locally} small \(\V\)-sup-lattices (with large carriers),
resulting in a predicatively useful result.
However, one of the anonymous reviewers of our
submission~\cite{deJongEscardo2022} pointed out that this generalisation is
false and suggested the ordinals as a counterexample in a set-theoretic setting:
it is a class with suprema for all subsets but has no greatest element.
This led us to prove~(\cref{sec:small-suprema-of-ordinals}) in our
type-theoretic context that the locally small, but large type of ordinals in a
universe~\(\V\) is a \(\V\)-sup-lattice (with no greatest element).

Similarly, consider a generalised formulation of Tarski's
theorem~\cite{Tarski1955} that allows for multiple universes, i.e.\ we define
\(\text{\emph{Tarski's-Theorem}}_{\V,\U,\T}\) as the assertion that every
monotone endofunction on a \(\V\)-sup-lattice with carrier in a universe \(\U\)
and order taking values in a universe \(\T\) has a greatest fixed point.%
\index{Tarski's Theorem|textbf}
Then \(\Tarski{\V}{\V}{\V}\) corresponds to the original formulation and,
moreover, is provable predicatively, but not useful predicatively
because~\cref{nontrivial-impredicativity} shows that its hypotheses can only be
fulfilled for trivial posets.
On~the other hand, \(\Tarski{\V}{\V^+}{\V}\) is provably false because the
identity map on the \(\V\)-sup-lattice of ordinals in \(\V\) is a
counterexample.
Analogous considerations could be made for a lemma due to
Pataraia~\cite{Pataraia1997,Escardo2003} saying that every dcpo has a greatest
monotone inflationary endofunction.%
\index{Pataraia's Lemma}

\subsection{A predicative counterexample}\label{sec:predicative-counterexample}

Because the type of ordinals in \(\V\) is not \(\V\)-small even impredicatively,
the above does not rule out the possibility that a \(\V\)-sup-lattice \(X\) has
all \(\V\)-infima provided \(X\) is \(\V\)\nobreakdash-small impredicatively.
To address this, we present an example of a \(\V\)-sup-lattice that is
\(\V\)-small impredicatively, but predicatively does not necessarily have a
maximal element. In particular, it need not have a greatest element or all
\(\V\)-infima.

Fix a proposition \(P_\U\) in a universe \(\U\). We consider its lifting (in the
sense of~\cref{sec:lifting}) with respect to a universe \(\V\), i.e.\ we
consider
\(\lifting_{\V}\pa*{P_\U} \equiv \Sigma_{Q : \Omega_\V} \pa*{Q \to P_\U}\), just
like in~\cref{lifting-structurally-algebraic-but-no-small-basis}.
This is a subtype of \(\Omega_\V\) and it is closed under \(\V\)-suprema (in
particular, it contains the least~element).%
\index{lifting}

\begin{example}\hfill
  \begin{enumerate}[(i)]
  \item If \(P_\U \colonequiv \Zero_\U\), then
    \( \lifting_{\V}(P_\U) \simeq \pa*{\Sigma_{Q : \Omega_\V} \lnot Q} \simeq
    \pa*{\Sigma_{Q : \Omega_\V} \pa*{Q = \Zero_{\V}}} \simeq \One \).
  \item If \(P_\U \colonequiv \One_\U\), then
    \( \lifting_{\V}(P_\U) \equiv \pa*{\Sigma_{Q : \Omega_\V} \pa*{Q \to \One_\U}}
    \simeq \Omega_\V \). \qedhere
  \end{enumerate}
\end{example}
What makes \(\lifting_{\V}(P_\U)\) useful is the following observation.
\begin{lemma}\label{maximal-iff-resize}%
  \index{maximality}%
  Suppose that the poset \(\lifting_{\V}(P_\U)\) has a maximal element
  \(Q : \Omega_\V\). Then \(P_\U\) is equivalent to \(Q\), which is the
  greatest element of \(\lifting_{\V}(P_\U)\). In particular,
  \(P_\U\) is \(\V\)-small.
  Conversely, if \(P_\U\) is equivalent to a proposition \(Q : \Omega_\V\), then
  \(Q\) is the greatest element of~\(\lifting_{\V}(P_\U)\).
\end{lemma}
\begin{proof}
  Suppose that \(\lifting_{\V}(P_\U)\) has a maximal element \(Q :
  \Omega_\V\). We wish to show that \(Q \simeq P_\U\). By definition of
  \(\lifting_{\V}(P_\U)\), we already have that \(Q \to P_\U\). So only the
  converse remains. Therefore suppose that \(P_\U\) holds. Then, \(\One_\V\) is
  an element of \(\lifting_{\V}(P_\U)\). Obviously \(Q \to 1_\V\), but \(Q\) is
  maximal, so actually \(Q = 1_\V\), that is, \(Q\) holds, as
  desired. Thus, \(Q \simeq P_\U\). It~is then straightforward to see that \(Q\)
  is actually the greatest element of \(\lifting_{\V}(P_\U)\), since
  \(\lifting_{\V}(P_\U) \simeq \Sigma_{Q' : \Omega_\V}(Q' \to Q)\).
  For the converse, assume that \(P_\U\) is equivalent to a proposition
  \(Q : \Omega_\V\). Then, as before,
  \(\lifting_{\V}(P_\U) \simeq \Sigma_{Q' : \Omega_\V}(Q' \to Q)\), which shows
  that \(Q\) is indeed the greatest element of~\(\lifting_{\V}(P_\U)\).
\end{proof}

\begin{corollary}%
  \index{infimum}%
  The \(\V\)-sup-lattice \(\lifting_{\V}(P_{\U})\) has all \(\V\)-infima if and
  only if \(P_{\U}\) is \(\V\)-small.
\end{corollary}
\begin{proof}
  Suppose first that \(\lifting_{\V}(P_{\U})\) has all \(\V\)-infima. Then it
  must have an infimum for the empty family
  \(\Zero_{\V} \to \lifting_{\V}(P_{\U})\). But this infimum must be the greatest
  element of \(\lifting_{\V}(P_{\U})\). So by \cref{maximal-iff-resize} the
  proposition \(P_{\U}\) must be \(\V\)-small.

  Conversely, suppose that \(P_{\U}\) is equivalent to a proposition \(Q : \V\).
  Then the infimum of a family \(\alpha : I \to \lifting_{\V}(P_{\U})\) with
  \(I : \V\) is given by \(\pa*{Q \times \Pi_{i : I} \alpha_i} : \V\).
\end{proof}

In~\cite{deJongEscardo2021a} we used~\cref{maximal-iff-resize} to conclude that
a version of Zorn's Lemma that says that every pointed dcpo has a maximal
element is predicatively unavailable, as \(\lifting_{\V}\pa*{P_{\U}}\) is a
pointed \(\V\)-dcpo, but has a maximal element if and only if \(P_{\U}\) is
\(\V\)-small.
But, as in our above discussion of the duality theorem and Tarski's Theorem, we
must pay attention to the universes here.
Zorn's Lemma restricted to \(\V\)-small \(\V\)-sup-lattices is, assuming
excluded middle~\cite{Bell1997}, equivalent to the axiom of choice, as usual.%
\index{Zorn's Lemma}
Disregarding its constructive status for a moment, the predicative issue is
that there are no nontrivial \(\V\)-small \(\V\)-sup-lattices
(\cref{nontrivial-impredicativity}).
But the generalisation of Zorn's Lemma to \emph{locally} small
\(\V\)-sup-lattices is false (even if we assume the axiom of choice and hence,
excluded middle), because the \(\V\)-sup-lattice of ordinals in \(\V\), having
no maximal element, is a counterexample.

\subsection{Small suprema of small ordinals}
\label{sec:small-suprema-of-ordinals}%
\index{ordinal|(}%
\index{supremum!of ordinals|(}

We now show that the ordinal \(\Ord_{\V}\) of ordinals in a fixed
\emph{univalent} universe \(\V\) has suprema for all families indexed by types
in~\(\V\) and that it has no maximal element.%
\index{univalence}%
\nomenclature[Ord_V]{$\Ord_{\V}$}{ordinal of ordinals in a univalent universe \(\V\)}
The~latter is implied by~\cite[Lemma~10.3.21]{HoTTBook}, but we were not able to
find a proof of the former in the literature:
Theorem~9 of \cite{KrausForsbergXu2021} only proves \(\Ord_{\V}\) to have joins
of increasing sequences, while \cite[Lemma~10.3.22]{HoTTBook} shows that every
family indexed by a type in \(\V\) has some upper bound, but does not prove it
to be the least (although least upper bounds are required for
\cite[Exercise~10.17(ii)]{HoTTBook}).
We present two proofs: one based on~\cite[Lemma~10.3.22]{HoTTBook} using small
set quotients and an alternative one using small images.

Following~\cite[Section~10.3]{HoTTBook}, we define an ordinal to be a type
equipped with a proposition-valued, transitive, extensional and (inductive)
well-founded relation.%
\index{well-founded!relation}
In \cite{HoTTBook} the underlying type of an ordinal is required to be a set,
but this actually follows from the other axioms,
see~\cite[\mkTTurl{Ordinals.Type}]{TypeTopology}.
The type of ordinals, denoted by \(\Ord_{\V}\), in a given \emph{univalent}
universe \(\V\) can itself be equipped with such a
relation~\cite[Theorem~10.3.20]{HoTTBook} and thus is an ordinal again. However,
it is not an ordinal in \(\V\), but rather in the next universe \(\V^+\), and
this is necessary, because it is contradictory for \(\Ord_{\V}\) to be
isomorphic to an ordinal in \(\V\), see~\cite{TypeTopologyBuraliForti}. Before
we can prove that \(\Ord_{\V}\) has \(\V\)-suprema, we need to recall a few
facts.
The well-order on \(\Ord_{\V}\) is given by: \(\alpha \prec \beta\) if and only
if there exists a (necessarily) unique \(y : \beta\) such that \(\alpha\) and
\(\beta \initseg y\) are isomorphic ordinals.
\nomenclature[prec]{$\alpha \prec \beta$}{well-order of an ordinal}
Here \(\beta \initseg y\) denotes
the ordinal of elements \(b : \beta\) satisfying \(b \prec y\).
\nomenclature[dsetzzz]{$\alpha \initseg x$}{initial segment of an ordinal}

\begin{lemma}[{\cite[\mkTTurl{Ordinals.OrdinalOfOrdinals}]{TypeTopology}}]%
  \label{initseg-order}%
  \index{simulation}%
  For every two points \(x\) and \(y\) of an ordinal \(\alpha\), we have
  \(x \prec y\) in \(\alpha\) if and only if
  \(\alpha \initseg x \prec \alpha \initseg y\) as ordinals.
\end{lemma}

\begin{definition}[Simulation; {\cite[Section~10.3]{HoTTBook}}]
  A \emph{simulation} between two ordinals \(\alpha\)~and~\(\beta\) is a map
  \(f : \alpha \to \beta\) satisfying the following conditions:
  \begin{enumerate}[(i)]
  \item for every \(x,y : \alpha\), if \(x \prec y\), then \(f(x) \prec f(y)\);
  \item for every \(x : \alpha\) and \(y : \beta\), if \(y \prec f(x)\), then there
    exists a (necessarily unique) \(x' : \alpha\) such that \(x' \prec x\) and
    \(f(x') = y\). %
    \qedhere
  \end{enumerate}
\end{definition}

\begin{lemma}[{\cite[\mkTTurl{Ordinals.OrdinalOfOrdinals}]{TypeTopology}}]\label{Ord-order}
  For ordinals \(\alpha\) and \(\beta\), the following are equivalent:
  \begin{enumerate}[(i)]
  \item there exists a (necessarily unique) simulation from \(\alpha\) to
    \(\beta\);
  \item for every ordinal \(\gamma\), if \(\gamma \prec \alpha\), then
    \(\gamma \prec \beta\).
  \end{enumerate}
  We write \(\alpha \preceq \beta\) if the equivalent conditions above hold.
\end{lemma}

Recall from~\cref{existence-of-small-set-quotients} what it means to have small
set quotients. If these are available, then the type of ordinals has all small
suprema.

\begin{theorem}[Extending~{\cite[Lemma~10.3.22]{HoTTBook}}]%
  \index{set quotient}%
  Assuming small set quotients, the large ordinal \(\Ord_{\V}\) has suprema of
  families indexed by types in \(\V\).
\end{theorem}
\begin{proof}
  Given \(\alpha : I \to \Ord_{\V}\), define \(\hat{\alpha}\) as the quotient of
  \(\Sigma_{i : I}\,\alpha_i\) by the \(\V\)-valued equivalence relation
  \({\approx}\) where \((i,x) \approx (j,y)\) if and only if \(\alpha_i \initseg x\)
  and \(\alpha_j \initseg y\) are isomorphic as ordinals.
  By our assumption, the quotient \(\hat\alpha\) lives in~\(\V\).  Next, we
  use~\cite[Lemma~10.3.22]{HoTTBook} which tells us that
  \(\pa*{\hat{\alpha},\prec}\) with
  \[
    [(i,x)] \prec [(j,y)] \colonequiv (\alpha_i \initseg x) \prec (\alpha_j \initseg y).
  \]
  is an ordinal that is an upper bound of \(\alpha\).
  So we show that \(\hat{\alpha}\) is a lower bound of upper bounds of
  \(\alpha\). To this end, suppose that \(\beta : \Ord_\V\) is such that
  \(\alpha_i \preceq \beta\) for every \(i : I\). In light of~\cref{Ord-order},
  this assumption yields two things:
  \begin{enumerate}[(1)]
  \item\label{assum1} for every \(i : I\) and \(x : \alpha_i\) there exists a
    unique \(b_i^x : \beta\) such that
    \(\alpha_i \initseg x = \beta \initseg b_i^x\);
  \item\label{assum2} for every \(i : I\), a simulation
    \(f_i : \alpha_i \to \beta\) such that for every \(x : \alpha_i\), we have
    \(f_i(x) = b_i^x\).
  \end{enumerate}
  We are to prove that \(\hat{\alpha} \preceq \beta\). We start by defining
  \begin{align*}
    f : \pa*{\Sigma_{i : I}\,\alpha_i} &\to \beta \\
    (i,x) &\mapsto b_i^x.
  \end{align*}
  Observe that \(f\) respects \(\approx\), for if \((i,x) \approx (j,y)\), then by
  univalence,
  \[(\beta \initseg b_i^x) = (\alpha_i \initseg x) = (\alpha_j \initseg y)
    = (\beta \initseg b_j^y),\]
  so \(b_i^x = b_j^y\) by uniqueness of \(b_i^x\).
  Thus, \(f\) induces a map \(\hat{f} : \hat{\alpha} \to \beta\) satisfying the equality
  \(\hat{f}([(i,x)]) = f(i,x)\) for every \((i,x) : \Sigma_{j : J}\,\alpha_j\).

  It remains to prove that \(\hat{f}\) is a simulation. Because the defining
  properties of a simulation are propositions, we can use set quotient induction
  and it suffices to prove the following two things:

  \begin{enumerate}[(I)]
  \item\label{cond-1} if \(\alpha_i \initseg x \prec \alpha_j \initseg y\), then
    \(b_i^x \prec b_j^y\);
  \item\label{cond-2} if \(b \prec b_i^x\), then there exists \(j : I\) and
    \(y : \alpha_j\) such that \(\alpha_i \initseg y \prec \alpha_j \initseg x\)
    and \(b_j^y = b\).
  \end{enumerate}
  For~\ref{cond-1}, observe that if \(\alpha_i \initseg x \prec \alpha_j \initseg y\),
  then \(\beta \initseg b_i^x \prec \beta \initseg b_j^y\), from which
  \(b_i^x \prec b_j^y\) follows using~\cref{initseg-order}.
  For~\ref{cond-2} suppose that \(b \prec b_i^x\). Because \(f_i\) (see
  item~\ref{assum2} above) is a simulation, there exists \(y : \alpha_i\) with
  \(y \prec x\) and \(f_i(y) = b\). By~\cref{initseg-order}, we get
  \(\alpha_i \initseg y \prec \alpha_i \initseg x\). Moreover,
  \(b_i^y = f_i(y) = b\), finishing the proof of~\ref{cond-2}.
\end{proof}

In~\cref{sec:set-replacement} we saw that set replacement is equivalent to the
existence of small set quotients, so the following result immediately follows
from the theorem above. But the point is that an alternative construction
without set quotients is available, if set replacement is assumed.

\begin{theorem}\index{set replacement}%
  Assuming set replacement, the large ordinal \(\Ord_{\V}\) has suprema of
  families indexed by types in \(\V\).
\end{theorem}

\begin{proof}
  Given \(\alpha : I \to \Ord_{\V}\), consider the image of the map
  \(e : \Sigma_{i : I}\,\alpha_i \to \Ord_{\V}\) given by
  \(e(i,x) \colonequiv \alpha_i \initseg x\).
  The image of \(e\) is conveniently equivalent to the type
  \(\Sigma_{\gamma : \Ord_{\V}}\exists_{i : I}\, \gamma \prec \alpha_i\),
  i.e.\ the type of ordinals that are initial segments of some \(\alpha_i\).
  One can prove that \(\image(e)\) with the induced order from \(\Ord_{\V}\) is
  again a well-order and that for every \(i : I\), the canonical map
  \(\alpha_i \to \image(e)\) is a simulation.
  Moreover, if \(\beta\) is an ordinal such that \(\alpha_i \preceq \beta\) for
  every \(i : I\), then for every \(i : I\) and every \(x : \alpha_i\) there
  exists a unique \(b_i^x : \beta\) such that
  \(\alpha_i \initseg x = \beta \initseg b_i^x\).
  Now observe that for every \(\gamma : \Ord_{\V}\), the map
  \(\pa*{\Sigma_{i : I}\Sigma_{x : \alpha_i}\,\pa*{\alpha_i \initseg x =
      \gamma}} \to \beta\) defined by the assignment
  \((i , x , p) \mapsto b_i^x\) is a constant function to a set.
  Hence, by~\cref{constant-map-to-set-factors-through-truncation}, this map
  factors through the propositional truncation
  \(\exists_{i : I}\Sigma_{x : \alpha_i}\,\pa*{{\alpha \initseg x} = \gamma}\).
  This yields a map \(\image(e) \to \beta\) which can be proved to be a
  simulation, as desired.
  Finally, we use set replacement and the fact that \(\Ord_{\V}\) is locally
  \(\V\)-small (by univalence) to get an ordinal in \(\V\) equivalent to
  \(\image(e)\), finishing the proof.
\end{proof}
\index{ordinal|)}%
\index{supremum!of ordinals|)}

\section{Families and subsets}
\label{sec:families-and-subsets}%
\index{supremum!of subsets|(}%
\index{subset|(}%

In traditional impredicative foundations, completeness of posets is usually
formulated using subsets. For instance, dcpos are defined as posets \(D\) such
that every directed subset of \(D\) has a supremum in
\(D\). \cref{examples-of-delta-complete-posets} are all formulated using small
families instead of subsets. While subsets are primitive in set theory, families
are primitive in type theory, so this could be an argument for using families
above. However, that still leaves the natural question of how the family-based
definitions compare to the usual subset-based definitions, especially in our
predicative setting, unanswered. This section addresses this question. We first
study the relation between subsets and families predicatively and then clarify
our definitions in the presence of impredicativity.
In our answers we will consider sup-lattices, but similar arguments could be
made for posets with other sorts of completeness, such as dcpos.

We first show that simply asking for completeness with respect to all subsets is
not satisfactory from a predicative viewpoint. In fact, we will now see that
even asking for completeness with respect to all elements of
\(\powerset_{\T}(X)\) for some fixed universe \(\T\) is problematic from a
predicative standpoint, where we recall from~\cref{def:powerset} that we refer
to the elements of \(\powerset_{\T}(X) \equiv (X \to \Omega_{\T})\) as
\(\T\)-valued subsets of \(X\).

\begin{theorem}
  \label{all-T-subsets-resizing}
  Let \(\U\) and \(\V\) be universes, fix a proposition \(P_{\U} : \U\) and
  recall \(\lifting_{\V}(P_{\U})\) defined in~\cref{sec:predicative-counterexample},
  which has \(\V\)-suprema.
  If \(\lifting_{\V}(P_{\U})\) has suprema for all \(\T\)-valued subsets, then
  \(P_{\U}\) is \(\V\)-small independently of the choice of the type universe
  \(\T\).
\end{theorem}
\begin{proof}
  Let \(\T\) be a type universe and consider the subset \(S\) of
  \(\lifting_{\V}(P_\U)\) given by \(Q \mapsto \One_\T\).  Note that \(S\) has a
  supremum in \(\lifting_{\V}(P_\U)\) if and only if \(\lifting_{\V}(P_\U)\) has
  a greatest element, but by~\cref{maximal-iff-resize}, the latter is equivalent
  to \(P_\U\) being \(\V\)-small.
\end{proof}

The proof above illustrates that if we have a subset \(S : \powerset_{\T}(X)\),
then there is no reason why the total space \({\Sigma_{x : X} (x \in S)}\)
(recall~\cref{def:total-space}) should be \(\T\)-small. In fact, for
\(S(x) \colonequiv \One_{\T}\) as above, the latter is equivalent to asking that
\(X\) is \(\T\)-small.

In an attempt to solve the problem described in \cref{all-T-subsets-resizing},
we look to impose size restrictions on the total space of a subset. There are
two natural such restrictions and they are reminiscent of Bishop and Kuratowski
finite subsets.%
\index{finiteness!Bishop}

\begin{definition}[\(\V\)-small and \(\V\)-covered subsets]%
  \index{subset!small}\index{subset!covered}%
  An element \(S : \powerset_{\T}(X)\) is
  \begin{enumerate}[(i)]
  \item \emph{\(\V\)-small} if its total space is \(\V\)-small, and
  \item \emph{\(\V\)-covered} if we have \(I : \V\) with a surjection
    \(e : I \surj \totalspace(S)\). %
    \qedhere
  \end{enumerate}
\end{definition}

Observe that every \(\V\)-small subset is \(\V\)-covered, because every
equivalence is a surjection.
But the converse does not hold: We can emulate the well-known argument used to
show that, constructively, Kuratowski finiteness does not necessarily imply
Bishop finiteness to show that, predicatively, being \(\V\)-covered does not
necessarily imply being \(\V\)-small.

\begin{proposition}\label{covered-small-resizing}%
  \index{resizing!propositional}%
  For every two universes \(\U\) and \(\V\), if every \(\V\)-covered element of
  \(\powerset_{\U}\pa*{\Omega_{\U}}\) is \(\V\)-small, then
  \(\Propresizing{\U}{\V}\) holds.
\end{proposition}
\begin{proof}
  Suppose that every \(\V\)-covered \(\U\)-valued subset of \(\Omega_{\U}\) is
  \(\V\)-small and let \(P : \U\) be an arbitrary proposition. Consider the
  subset \(S_P : \Omega_{\U} \to \Omega_{\U}\) given by
  \(S_P(Q) \colonequiv \pa*{Q = P} \vee \pa*{Q = \One_{\U}}\). Notice that this
  is \(\V\)-covered as witnessed by
  \begin{align*}
    \pa*{\One_{\V} + \One_{\V}} &\surj \totalspace(S_P) \\
    \inl(\star) &\mapsto \pa*{P\hspace{1.25mm},\tosquash*{\inl(\refl)}} \\
    \inr(\star) &\mapsto \pa*{\One_{\U},\tosquash*{\inr(\refl)}},
  \end{align*}
  so by assumption \(\totalspace(S_P)\) is \(\V\)-small. But observe that \(P\)
  holds if and only if \(\totalspace(S_P)\) is a subsingleton, but the latter
  type is \(\V\)-small by assumption, hence so is \(P\).
\end{proof}

In the case where we restrict our attention to a single universe \(\V\) and a
locally \(\V\)-small set \(X\), the two notions coincide if and only if we have
set replacement for maps into \(X\) with \(\V\)-small domain.%
\index{set replacement}
\begin{proposition}\label{covered-small-set-replacement}
  If \(X\) is a locally \(\V\)-small set, then every \(\V\)-covered element of
  \(\powerset_{\V}(X)\) is \(\V\)-small if and only if the image of any map into
  \(X\) with a \(\V\)-small domain is \(\V\)-small.
\end{proposition}
\begin{proof}
  Suppose first that every \(\V\)-covered subset \(S : X \to \Omega_{\V}\) is
  \(\V\)-small and let \(f : I \to X\) be a map such that \(I\) is \(\V\)-small.
  Without loss of generality, we may assume that \(I : \V\), because we can
  always precompose \(f\) with the equivalence witnessing that \(I\) is
  \(\V\)-small.
  Now consider the subset \(S : X \to \Omega_{\V}\) given by
  \(S(x) \colonequiv \exists_{i : I}\pa*{f(i) =_{\V} x}\), where \({=_{\V}}\)
  has values in \(\V\) and is provided by our assumption that \(X\) is locally
  \(\V\)-small.
  Then \(S\) is \(\V\)-covered, because we have
  \(I \surj \image(f) \simeq \totalspace(S)\), where the first map is the
  corestriction of \(f\).
  So by assumption \(\totalspace(S)\) is \(\V\)-small, which means that
  \(\image(f)\) must be \(\V\)-small too.

  Conversely, assume the set replacement principle and let
  \(S : X \to \Omega_{\V}\) be \(\V\)-covered by \(e : I \surj \totalspace(S)\).
  Define the subset \(S' : X \to \Omega_{\V}\) by
  \(S'(x) \colonequiv \exists_{i : I}\pa*{x =_{\V} \fst(e_i)}\).
  By the assumed set replacement principle for \(X\), the subset \(S'\) is
  \(\V\)-small because of the equivalence
  \(\totalspace(S') \simeq \image({\fst} \circ {e})\).
  Finally, it follows from the surjectivity of \(e\) that \(S\) and \(S'\) are
  equal as subsets, and therefore that
  \(\totalspace(S) \simeq \totalspace(S')\). Hence, \(S\) is a \(\V\)-small
  subset, as desired.
\end{proof}

So, predicatively, and in the absence of a set replacement principle, the notion
of a \(\V\)\nobreakdash-small subset is strictly stronger than that of a
\(\V\)-covered subset. Hence, in this setting, having suprema for all
\(\V\)-small subsets is strictly weaker than having suprema for all
\(\V\)\nobreakdash-covered subsets.
Meanwhile, \cref{family-subset-sup-equiv} will imply that there are plenty of
examples of posets with suprema for all \(\V\)-covered subsets, even
predicatively.
So we prefer the stronger, but predicatively reasonable requirement of asking
for suprema of all \(\V\)-covered subsets.

Form a practical viewpoint, \(\V\)-covered subsets also give us an easy handle
on examples like the following: Let \(X\) be a poset with suprema for all
(directed) \(\U_0\)-covered subsets. Then the least fixed point of a Scott
continuous endofunction \(f\) on \(X\) can be computed as the supremum of the
subset \(\{\bot, f(\bot) , f^2(\bot) , \dots\}\), which is covered by \(\Nat\).
But it is not clear that this subset is \(\U_0\)-small, at least not in the
absence of set replacement.

Our preference for \(\V\)-covered subsets over \(\V\)-small subsets also makes
it clear why we do not impose an injectivity condition on families, because for
every type \(X : \U\) there is an equivalence between embeddings
\(I \hookrightarrow X\) with \(I : \V\) and \(\pa*{\U \sqcup \V}\)-valued
subsets of \(X\) whose total spaces are \(\V\)-small,
cf.~\cite[\mkTTurl{Slice.Slice}]{TypeTopology}.

\begin{theorem}
  \label{family-subset-equiv}
  For \(X : \U\) and any universe \(\V\) we have an equivalence between
  \(\V\)-covered \(\pa*{\U \sqcup \V}\)-valued subsets of \(X\) and families
  \(I \to X\) with \(I : \V\).
\end{theorem}
\begin{proof}
  The forward map \(\varphi\) is given by \((S,I,e) \mapsto (I,{\fst} \circ {e})\).
  In the other direction, we define \(\psi\) by mapping \((I,\alpha)\) to the
  triple \((S,I,e)\) where \(S\) is the subset of \(X\) given by
  \(S(x) \colonequiv \exists_{i : I}\,(x = \alpha(i))\) and
  \(e : I \surj \totalspace(S)\) is defined as
  \(e (i) \colonequiv \pa*{\alpha(i),\tosquash*{(i,\refl)}}\).
  The composite \(\varphi \circ \psi\) is easily seen to be equal to the
  identity. To show that \(\psi \circ \varphi\) equals the identity, we need the
  following intermediate result, which is proved using function extensionality
  and path induction.
  \begin{claim}
    Let \(S,S' : X \to \Omega_{\U\sqcup\V}\), \(e : I \to \totalspace(S)\)
    and \(e' : I \to \totalspace(S')\). If \(S = S'\) and \({{\fst} \circ e \sim
    {\fst} \circ e'}\), then \((S,e) = (S',e')\).
  \end{claim}
  The result follows from the claim using function and propositional
  extensionality.
\end{proof}

\begin{corollary}
  \label{family-subset-sup-equiv}
  A poset with carrier in a universe \(\U\) has suprema for all \(\V\)-covered
  \(\pa*{\U \sqcup \V}\)\nobreakdash-valued subsets if and only if it has
  suprema for all families indexed by types in \(\V\).
\end{corollary}
\begin{proof}
  This is because the supremum of a \(\V\)-covered subset equals the supremum of
  the corresponding family and vice versa by inspecting the proof
  of~\cref{family-subset-equiv}.
\end{proof}

We conclude by comparing our family-based approach to the subset-based approach
in the presence of impredicativity.

\begin{theorem}\label{impred-comparison}
  Assuming \(\Omegaresizing{\T}{\U_0}\) for every universe \(\T\). Then the
  following are equivalent for a poset with carrier in a universe \(\U\):
  \begin{enumerate}[(i)]
  \item\label{all-subsets} the poset has suprema for all subsets;
  \item\label{all-covered-subsets} the poset has suprema for all \(\U\)-covered
    subsets;
  \item\label{all-small-subsets} the poset has suprema for all \(\U\)-small
    subsets;
  \item\label{all-small-families} the poset has suprema for all families indexed
    by types in \(\U\).
  \end{enumerate}
\end{theorem}
\begin{proof}
  Clearly
  \ref{all-subsets}~\(\Rightarrow\)~\ref{all-covered-subsets}~\(\Rightarrow\)~\ref{all-small-subsets}.
  We show that \ref{all-small-subsets} implies \ref{all-subsets}, which proves
  the equivalence of \ref{all-subsets}--\ref{all-small-subsets}. Assume that a
  poset \(X\) has suprema for all \(\U\)-small subsets and let
  \(S : X \to \Omega_{\T}\) be any subset of \(X\). Using
  \(\Omegaresizing{\T}{\U_0}\), the total space \(\totalspace(S)\) is
  \(\U\)-small. So~\(X\)~has a supremum for \(S\) by assumption, as
  desired. Finally, \ref{all-covered-subsets}~and~\ref{all-small-families} are
  equivalent in the presence of \(\Omegaresizing{\T}{\U_0}\)
  by~\cref{family-subset-sup-equiv}.
\end{proof}
If condition~\ref{all-small-families} of~\cref{impred-comparison} holds, then
the poset has suprema for all families indexed by types in \(\V\) provided that
\(\V \sqcup \U \equiv \U\).
Typically, in the examples in our account of domain theory for~instance,
\(\U \colonequiv \U_1\) and \(\V \colonequiv \U_0\), so that
\(\V \sqcup \U \equiv \U\) holds. Thus, our \(\V\)-families-based approach
generalises the traditional subset-based approach.%
\index{supremum!of subsets|)}%
\index{subset|)}%

\section{Notes}\label{sec:predicativity-in-order-theory-notes}

This chapter is based on an extended and revised
version~\cite[Sections~4--6]{deJongEscardo2022} of our
paper~\cite{deJongEscardo2021b}. %
We would like to thank the anonymous reviewers of~\cite{deJongEscardo2022} for
their valuable and complementary suggestions. We~are particularly grateful to
the reviewer who pointed out that one of our results can be strengthened
to~\cref{is-small-retract} and for their insights and questions
on~\cref{sec:maximal-and-fixed-points,sec:families-and-subsets} that have
considerably improved the exposition.


%% file: mainmatter/formalisation.tex
\chapter{Formalisation}\label{chap:formalisation}%
\index{formalisation|(}

This research started with formalising the Scott model of PCF in
\Coq~\cite{Coq} using the \UniMath~\cite{UniMath} library.%
\index{proof assistant!Coq@\Coq}\index{UniMath@\UniMath}\index{Scott model!of PCF}
The \UniMath\ project was started in 2014 by merging repositories from Vladimir
Voevodsky, Benedikt Ahrens and Daniel Grayson. The current \UniMath\
Coordinating Committee members are: Benedikt Ahrens, Daniel Grayson, Ralph
Matthes and Niels van der Weide and the library has accepted
contributions\footnote{\url{https://github.com/UniMath/UniMath/graphs/contributors}}
from over 60 people at the time of writing.

For historical reasons~\cite{Voevodsky2015}, the \verb|Type-in-Type| flag is
enabled in \UniMath, so that it is not possible to have \Coq\ automatically
check the universes for us.
Since we were interested in developing domain theory predicatively, having the
proof assistant carefully track universes was an important feature.%
\index{universe!parameters}
Hence, we decided to continue our formalisation efforts in \Agda~\cite{Agda}
using Mart\'in Escard\'o's (and collaborators') \TypeTopology~\cite{TypeTopology}
development which explicitly keeps track of universes.%
\index{proof assistant!Agda@\Agda}\index{TypeTopology@\TypeTopology}
This is the reason some parts are formalised both in \Coq/\UniMath\ and in
\Agda/\TypeTopology. But because we did not wish to duplicate all our efforts,
some parts are only formalised in \Coq/\UniMath.

Formalising our efforts has helped to experiment with and has structured and
guided our development of domain theory as set out in this thesis.
Moreover~\cite{Hart2020} has shown that the \Agda\
formalisation~\cite{TypeTopologyDomainTheory} can be taken as a starting point
for a further formal development of domain theory in univalent
foundations. Since then, the code base has been extended and improved
considerably, hopefully making our formalisation more convenient to work with.

Both the \Coq/\UniMath~and \Agda/\TypeTopology~proofs and their renderings in
HTML (for presentation and reading) are archived by the University of Birmingham at
\href{https://doi.org/10.25500/edata.bham.00000912}{\texttt{doi:10.25500/edata.bham.00000912}}.

\section{Overview}

We give a comprehensive overview per chapter of what is and what isn't
formalised.

\paragraph{\cref{chap:univalent-foundations}}%

All the background material up to and
including~\cref{sec:impredicativity-resizing-axioms} is fairly well-known and
all formalised between~\cite{TypeTopology} and~\cite{Escardo2019}, except
for~\cref{is-small-retract} because we only recently learned that this stronger
result was possible, see~\cref{small-retract-improvement}. The weaker result for
sections that are embeddings is formalised however and is sufficient for our
applications.

\cref{sec:quotients-replacement-prop-trunc-revisited} on set quotients, set
replacement and propositional truncations is formalised in \Agda/\TypeTopology,
see~\cite{TypeTopologyQuotientLarge,%
  TypeTopologyQuotientF,%
  TypeTopologyQuotientReplacement,%
  TypeTopologyQuotient}.%
\index{set quotient}\index{set replacement}\index{propositional truncation}%

\cref{sec:indexed-W-types} on indexed \(\WW\)-types is formalised in
\Coq/\UniMath,
see~\cite[{\href{https://tdejong.com/Scott-PCF-UniMath/UniMath.MoreFoundations.Wtypes.html}{\texttt{MoreFoundations.Wtypes}}}]{UniMathScottModelOfPCF},
as it was aimed at the application to PCF.%
\index{W-type@\(\WW\)-type!indexed}

\paragraph{\cref{chap:basic-domain-theory}} %
All of~\cref{chap:basic-domain-theory} is formalised in \Agda/\TypeTopology, %
see~\cite{TypeTopologyDomainTheory}, except for
\begin{itemize}
\item products of dcpos (although this was included in~\cite{Hart2020}),
\item \cref{flat-gives-LPO,lifting-is-plus-one-iff-em}, and
\item \cref{lifting-is-free2} (but the
  similar~\cref{lifting-is-free,lifting-is-free3} are formalised).
\end{itemize}

\paragraph{\cref{chap:continuous-and-algebraic-dcpos}} %
All of~\cref{chap:continuous-and-algebraic-dcpos} is formalised in
\Agda/\TypeTopology, see~\cite{TypeTopologyDomainTheory}, with
\cref{lifting-structurally-algebraic-but-no-small-basis}, and
\cref{close-basis-under-finite-joins,sup-complete-ideal-completion} as the only
exceptions.

\paragraph{\cref{chap:applications}} %
\cref{sec:Scott-D-infty} is fully formalised in \Agda/\TypeTopology, %
see the dedicated
file~\cite[\mkTTurl{Bilimits.Dinfinity}]{TypeTopologyDomainTheory}.
Regarding \cref{sec:Scott-model-of-PCF}, we would like to mention three formalisations.

The Scott model of PCF, including soundness and computational adequacy, is fully
formalised in \Coq/\UniMath, %
see~\cite[{\href{https://tdejong.com/Scott-PCF-UniMath/UniMath.Partiality.PCF.html}{\texttt{Partiality.PCF}}}]{UniMathScottModelOfPCF}.%
\index{Scott model!of PCF}\index{Scott model!of PCF!soundness}%
\index{Scott model!of PCF!computational adequacy}
This also includes the general~\cref{decidable-k-step-refl-trans-clos} in%
~\cite[{\href{https://tdejong.com/Scott-PCF-UniMath/UniMath.MoreFoundations.ClosureOfHrel.html}{\texttt{MoreFoundations.ClosureOfHrel}}}]{UniMathScottModelOfPCF}, %
and the logical equivalence of~\cref{semidecidability-of-PCF-props}, but the
application of~\cref{decidable-k-step-refl-trans-clos} to obtain decidability of
\({\smallstep^k}\) is not formalised.

To check the predicative validity and the universes involved, which is not
possible in \Coq/\UniMath\ because it uses \verb|Type-in-Type| (as mentioned
above), we also formalised the syntax of PCF and the definition of the Scott
model of PCF in \Agda/\TypeTopology, %
see~\cite[\mkTTurl{ScottModelOfPCF.ScottModelOfPCF}]{TypeTopologyDomainTheory}.

This \Agda\ development was extended by Brendan Hart~\cite{Hart2020} to a proof
of soundness and computational adequacy for PCF with variables and
\(\lambda\)\nobreakdash-abstraction for a final year MSci project supervised by
Mart\'in Escard\'o and myself.

\paragraph{\cref{chap:predicativity-in-order-theory}} %
Only the results of~\cref{sec:small-suprema-of-ordinals} on suprema of ordinals
have been formalised in \Agda/\TypeTopology,
see~\cite{TypeTopologyOrdinalSuprema}.%
\index{ordinal}

\section{Future work}
It would be desirable to expand~\cite{TypeTopologyDomainTheory} with a
formalisation of products and
\cref{close-basis-under-finite-joins,sup-complete-ideal-completion}.%
\index{dcpo!product}
Products are not included because they were not needed to develop the
applications in~\cref{chap:applications}, while the lemmas are missing due to a
lack of time.
Including products could potentially be achieved by merging~\cite{Hart2020}.
The other results are less pressing to have formalised, because they do not
obstruct a further computer-verified development of domain theory, for example
because they are results showing that some statement implies a constructive or
predicative taboo.

\section{Statistics}%
\index{formalisation!statistics|(}

To give the reader an impression of the (relative) sizes of our formalisations
of domain theory, \cref{Coq-code-stats} and \cref{Agda-code-stats} show the
number of lines in
respectively~\cite{UniMathScottModelOfPCF} and \cite{TypeTopologyDomainTheory}.
Notice that the files listed in these tables depend on auxiliary files that
develop univalent foundations (just
as~\cref{chap:basic-domain-theory,chap:continuous-and-algebraic-dcpos,chap:applications}
depend on~\cref{chap:univalent-foundations}), but that these auxiliary files are
not included in the statistics.\vspace{1em}

\begin{table}[h!]
  \centering
  \begin{tabular}{lr@{}}\toprule
    File & Number of lines \\ \midrule
    \texttt{DCPO.v}            & 637 \\
    \texttt{LiftMonad.v}       & 151 \\
    \texttt{PartialElements.v} & 369 \\
    \texttt{PCF.v}             & 769 \\
    \cmidrule(l){2-2}
         & 1953 \\
    \bottomrule
  \end{tabular}
  \caption{Number of lines (including comments and blank lines) per file in our
    \Coq/\UniMath\ formalisation~\cite{UniMathScottModelOfPCF}. It should be
    noted that the file~\texttt{PCF.v} in~\cref{Coq-code-stats} includes the
    soundness and computational adequacy of the Scott model of PCF.}
  \label{Coq-code-stats}
\end{table}

\begin{center}
\begin{longtable}{lr@{}}\toprule
  \!\!\!File & Number of lines \\ \midrule
  \multicolumn{2}{@{}l}{\textbf{Basics}}\\
  \texttt{Dcpo.lagda}            & 303 \\
  \texttt{Exponential.lagda}     & 346 \\
  \texttt{LeastFixedPoint.lagda} & 309 \\
  \texttt{Miscelanea.lagda}      & 668 \\
  \texttt{Pointed.lagda}         & 341 \\
  \texttt{SupComplete.lagda}     & 294 \\
  \texttt{WayBelow.lagda}        & 289 \\
  \cmidrule(l){2-2}
  & 2550 \\
  \midrule
  \multicolumn{2}{@{}l}{\textbf{BasesAndContinuity}}\\
  \texttt{Bases.lagda}                & 721 \\
  \texttt{Continuity.lagda}           & 502 \\
  \texttt{ContinuityDiscussion.lagda} & 375 \\
  \texttt{IndCompletion.lagda}        & 379 \\
  \texttt{StepFunctions.lagda}        & 509 \\
  \cmidrule(l){2-2}
  & 3716 \\
  \midrule
  \multicolumn{2}{@{}l}{\textbf{Bilimits}}\\
  \texttt{Dinfinity.lagda}  & 962 \\
  \texttt{Directed.lagda}   & 1229 \\
  \texttt{Sequential.lagda} & 495 \\
  \cmidrule(l){2-2}
  & 2686 \\
  \midrule
  \multicolumn{2}{@{}l}{\textbf{Examples}}\\
  \texttt{IdlDyadics.lagda} & 80 \\
  \texttt{Omega.lagda}      & 219 \\
  \texttt{Powerset.lagda}   & 216 \\
  \cmidrule(l){2-2}
  & 515  \\
  \midrule
  \multicolumn{2}{@{}l}{\textbf{IdealCompletion}}\\
  \texttt{IdealCompletion.lagda} & 198 \\
  \texttt{Properties.lagda}      & 528 \\
  \texttt{Retracts.lagda}        & 472 \\
  \cmidrule(l){2-2}
  & 1198 \\
  \midrule
  \multicolumn{2}{@{}l}{\textbf{Lifting}}\\
  \texttt{LiftingDcpo.lagda}         & 466 \\
  \texttt{LiftingSet.lagda}          & 395 \\
  \texttt{LiftingSetAlgebraic.lagda} & 208 \\
  \cmidrule(l){2-2}
  & 1069 \\
  \midrule
  \multicolumn{2}{@{}l}{\textbf{ScottModelOfPCF}}\\
  \texttt{ScottModelOfPCF.lagda} & 71 \\
  \texttt{PCF.lagda}             & 118 \\
  \texttt{PCFCombinators.lagda}  & 474 \\
  \cmidrule(l){2-2}
  & 663 \\
  \midrule
  \!\!\!\textbf{All components combined} & 11167 \\
  \bottomrule
  \caption{Number of lines (including comments and blank lines) per component of
    our \Agda/\TypeTopology\ formalisation~\cite{TypeTopologyDomainTheory}}
  \label{Agda-code-stats}
\end{longtable}
\end{center}%
\index{formalisation!statistics|)}
\index{formalisation|)}


%% file: mainmatter/conclusion.tex
\chapter{Conclusion}\label{chap:conclusion}

We provide a summary of our contributions and our approach to developing domain
theory in constructive and predicative univalent foundations. Furthermore, we
briefly describe various directions for future research.

\section{Summary}
We have demonstrated how constructive and predicative univalent foundations
provides an adequate and sophisticated setting to develop domain theory.
Since higher inductive types can be seen as specific instances of resizing
principles, it is noteworthy that the only higher inductive type needed
for our purposes is the propositional truncation.

Instead of working with information systems, abstract bases or formal
topologies, and approximable relations, we studied directed complete posets and
Scott continuous directly, using type universes and type equivalences to deal
with size issues in the absence of propositional resizing axioms.
Seeing a poset as a category in the usual way, we can say that dcpos are large,
but locally small, and have small filtered colimits.
By carefully keeping track of universe parameters, we showed that complex
constructions of dcpos, such as Scott's \(D_\infty\) model of the untyped
\(\lambda\)-calculus, which involves countable infinite iterations of
exponentials, are predicatively possible.
We further illustrated our domain-theoretic development by giving a predicative
and constructive account of the soundness and computational adequacy of the
Scott model of PCF. In particular, this illustrated the use of the
Escard\'o--Knapp lifting monad.

Taking inspiration from work in category theory by Johnstone and Joyal, we gave
predicatively adequate definitions of continuous and algebraic dcpos, and
discussed issues related to the absence of the axiom of choice.
We~also presented predicative adaptations of the notions of a basis and the
rounded ideal completion.
The theory was accompanied by several examples: we described small compact bases
for the lifting and the powerset, and considered the rounded ideal completion of
the dyadics.

The fact that nontrivial dcpos have large carriers is in fact unavoidable and
characteristic of our predicative setting, as we explained in a complementary
chapter on the constructive and predicative limitations of univalent
foundations. We proved no-go theorems regarding both constructivity and
predicativity for a general class of posets that includes dcpos, bounded
complete posets, sup-lattices and frames.
In particular, we showed that locally small nontrivial dcpos necessarily lack
decidable equality in our constructive setting.
The fact that nontrivial dcpos are necessarily large has the important
consequence that Tarski's theorem (and similar results) cannot be applied in
nontrivial instances, even though it has a predicative proof.
Further, we explained, by studying the large \(\V\)-sup-lattice of ordinals in a
univalent universe \(\V\), that generalisations of Tarski's theorem that allow
for large structures are provably false.
Finally, we elaborated on the connections between requiring suprema of families
and of subsets in our predicative setting.

Moreover, we contributed to the overall theory of (predicative) univalent
foundations by studying a set replacement principle and the connections and
universe levels of set quotients and propositional truncations. We also
presented a general criterion for decidable equality of indexed \(\WW\)-types
that we applied to the syntax of~PCF when proving that totality of PCF terms of
base type is semidecidable.

\section{Future work}

To prove that \(D_\infty\) had a small compact basis, we used that each \(D_n\)
is a \(\U_0\)-sup-lattice, so that we could apply the results
of~\cref{sec:exponentials-with-small-bases}.
\cref{lifting-has-small-compact-basis} tells us that \(\lifting_{\U_0}(\Nat)\)
has a small compact basis too, but to prove that the \(\U_0\)-dpcos in the Scott
model of PCF have small compact bases using the techniques
of~\cref{sec:exponentials-with-small-bases}, we would need
\(\lifting_{\U_0}(\Nat)\) to be a \(\U_0\)-sup-lattice, but it isn't.
However, it is complete for \emph{bounded} families indexed by types in \(\U_0\)
and we believe that is possible to generalise the results
of~\cref{sec:exponentials-with-small-bases} from sup-lattices to bounded
complete posets.%
\index{poset!bounded complete}
Classically, this is fairly straightforward, but from preliminary considerations
it appears that constructively one needs to impose certain decidability criteria
on the bases of the dcpos. For instance that the partial order is decidable when
restricted to basis elements. We also studied such decidability conditions in
our paper~\cite{deJong2021b} discussed below.%
\index{decidability}\index{constructivity}\index{basis}
These conditions should be satisfied by the bases of the dcpos in the Scott
model of PCF, but we leave a full treatment of bounded complete dcpos with bases
satisfying such conditions for future investigations.

It would be worthwhile to further develop domain-theoretic
applications. Specifically, it would be interesting to give a fully rigorous
formalisation of the surprising domain-theoretic algorithms that exhaustively
search infinite sets in finite time due to
\citeauthor{Berger1990}~\cite{Berger1990} and
\citeauthor{Escardo2008}~\cite{Escardo2008}.%
\index{algorithm}
The fact that our development is constructive might then pay off as we could use
our constructive proofs of domain theoretic facts to directly compute the output
of such algorithms.

To complement the applications of domain theory in the semantics of programming
language, it would be desirable to explore applications in (pointfree)
topology. For example, can we predicatively replicate the
connection~\cite{Hyland1981} between exponentiable locales and continuous
lattices?%
\index{locale}\index{continuity!of a dcpo}

Although some results (see \cref{continuous-wrt-Scott-topology}) can be stated
in terms of opens of the Scott topology\index{topology}, we have not given a
constructive account of the Scott topology. We did study this topic and a
related apartness relation~\cite{BridgesVita2011} in our
paper~\cite{deJong2021b}.\index{apartness relation}
To simplify the development and in the interest of appealing to a broader
audience of constructivists, the work~\cite{deJong2021b} is situated in informal
constructive, but impredicative, set theory rather than univalent foundations.
For this reason (and for time and space limitations) this paper is not part of
this thesis, even though it shares the domain-theoretic theme.

Furthermore, we hope that our formalisation efforts (as discussed
in~\cref{chap:formalisation}) provide adequate support for those wishing to
further develop computer-verified domain theory in univalent foundations.
This hope is reinforced by the fact that a gifted MSci student, Brendan Hart,
supervised by Mart\'in Escard\'o and myself, was able to do so for a final year
project~\cite{Hart2020} using an earlier and rudimentary version of our \Agda\
code.

Finally, the most fundamental and pressing question regarding predicativity in
univalent foundations is whether propositional resizing can be given a
computational interpretation.\index{resizing!propositional}\index{predicativity}
